\documentclass[a4paper,10pt]{article}

\usepackage[utf8]{inputenc}     
\usepackage[T1]{fontenc}        
\usepackage[english]{babel}     
\usepackage{fullpage}

\usepackage{enumerate}
\usepackage{array}			    
\usepackage{graphicx}           
\usepackage{float}              
\usepackage{caption}            
\usepackage{subcaption}         
\usepackage{pifont}             
\usepackage{slashed}            

\usepackage{amssymb, amsfonts}
\usepackage{amsmath}
\usepackage{amsthm}             
\usepackage{dsfont}
\usepackage{bbold}              
\usepackage{mathabx}            
\usepackage{stmaryrd}           
\usepackage{mathrsfs}           
\usepackage{empheq}             
\numberwithin{equation}{section}
\allowdisplaybreaks

\usepackage{pgfplots}
\pgfplotsset{compat=1.15}
\usetikzlibrary[arrows,patterns,patterns.meta]

\usepackage{hyperref}           
\hypersetup{
    colorlinks=true,            
    linktocpage=false,          
    citecolor= green,
    urlcolor=magenta,
    linkcolor=blue
    }                           
\usepackage{cleveref}

\def \N {\mathbb N}
\def \R {\mathbb R}
\def \C {\mathbb C}
\def \M {\mathcal{M}}

\newcommand \dsphere[1]{d \sigma_{\mathbb{S}^{#1}}}
\newcommand \grando[1]{\mathcal{O}\left({#1}\right)}

\DeclareMathOperator{\Id}{Id}
\DeclareMathOperator{\Mat}{Mat}
\DeclareMathOperator{\im}{im}
\DeclareMathOperator{\sign}{sign}
\DeclareMathOperator{\supp}{supp}
\DeclareMathOperator{\cotan}{cot}
\DeclareMathOperator{\tf}{tf}
\DeclareMathOperator{\Ric}{Ric}

\DeclareMathOperator{\trace}{tr}
\DeclareMathOperator{\diver}{div}
\DeclareMathOperator{\grad}{grad}
\DeclareMathOperator{\vectorspan}{span}
\DeclareMathOperator{\Tan}{Tan}
\DeclareMathOperator{\Conf}{C}
\DeclareMathOperator{\Weylchg}{T}
\DeclareMathOperator{\Signchg}{S}

\theoremstyle{plain}      
\newtheorem{thm}{Theorem}[section]
\newtheorem{prop}[thm]{Proposition}
\newtheorem{lem}[thm]{Lemma}
\newtheorem{defi}[thm]{Definition}
\newtheorem{cor}[thm]{Corollary}

\theoremstyle{remark}
\newtheorem*{ex}{Example}
\AtBeginEnvironment{ex}{%
    \pushQED{\qed}}
\AtEndEnvironment{ex}{\popQED}
\newtheorem*{rem}{Remark}
\AtBeginEnvironment{rem}{%
    \pushQED{\qed}}
\AtEndEnvironment{rem}{\popQED}
\newtheorem*{rems}{Remarks}
\AtBeginEnvironment{rems}{%
    \pushQED{\qed}}
\AtEndEnvironment{rems}{\popQED}

\title{{\bf \large GEOMETRIC REFLECTIVE BOUNDARY CONDITIONS FOR ASYMPTOTICALLY ANTI-DE SITTER SPACES}}
\author{\\ \normalsize LUDOVIC SOU\^ETRE\thanks{ludovic.souetre@sorbonne-universite.fr} \\ \\ {\it \small Sorbonne Université, Université Paris Cité, CNRS, INRIA,} \\
{\it \small Laboratoire Jacques-Louis Lions, LJLL, F-75005 Paris, France}}
\date{}

\begin{document}
\maketitle

\begin{abstract}
    This article solves the initial boundary value problem for the vacuum Einstein equations with a negative cosmological constant in dimension $4$, giving rise to asymptotically Anti-de Sitter spaces. We introduce a new family of geometric reflective boundary conditions, which can be regarded as the homogeneous Robin boundary conditions, involving both the conformal class and the stress-energy tensor of the timelike conformal boundary. This family includes as a special case the homogeneous Neumann boundary condition, consisting of setting the boundary stress-energy tensor to zero. It also agrees, in a limit case, with the homogeneous Dirichlet boundary condition, where one fixes a locally conformally flat conformal class on the boundary, already covered in Friedrich’s pioneering work \cite{F95}.
    
    The proof of local existence and uniqueness for this family of boundary conditions relies notably on Friedrich's framework and his extended conformal Einstein equations. These are rewritten in a tensorial formalism, rather than a spinorial one, as it facilitates the comparison with the Fefferman-Graham expansion of asymptotically Anti-de Sitter metrics. Similarly to Friedrich's proof, our geometric boundary conditions are eventually inferred from gauge-dependent boundary conditions by means of an auxiliary system on the conformal boundary.
    
    In addition, our analysis also comprises new necessary and sufficient conditions for the unphysical fields associated to the initial data to be smooth up to the conformal boundary. Finally, the paper contains examples of asymptotically Anti-de Sitter spaces, with a focus on their conformal boundary data. This provides valuable insight into the possible shapes of boundary stress-energy tensors.
\end{abstract}

\tableofcontents

\section{Introduction}

\subsection{The vacuum Einstein equations with a negative cosmological constant}

The vacuum Einstein equations with a cosmological constant $\Lambda\in\R$ for a $(n+1)$-dimensional Lorentzian manifold $(\M,g)$ write as
\begin{equation}
    \tag{VE}
    \label{eq:VE}
    \Ric(g) = \frac{2\Lambda}{n-1} g \,,
\end{equation}
where $\Ric(g)$ is the Ricci tensor of the Lorentzian metric $g$. The sign of the cosmological constant $\Lambda$ determines not only the sign of the scalar curvature $R$ and thus local geometric properties, but also the nature of the potential conformal boundary at infinity. The simplest solutions of \eqref{eq:VE} are the maximally symmetric spaces which are the Minkowski space for $\Lambda = 0$, the de Sitter (abbreviated to dS) space for $\Lambda>0$ and the Anti-de Sitter (AdS) space for $\Lambda<0$. The latter has the remarkable property of not being globally hyperbolic due to the presence of a timelike conformal boundary. More generally, we are interested in this article in spaces with a similar timelike conformal boundary, called asymptotically Anti-de Sitter (aAdS).

The negative cosmological constant case was first investigated in the physics literature in the late 1970s and 1980s in order to determine the effects of a non-zero cosmological constant on the geometry, in particular in connection with supergravity theories. Linear scalar fields on AdS were studied by Avis, Isham and Storey \cite{AIS78}, Breitenlohner and Freedman \cite{BF82bis,BF82} while perturbations of the AdS space were examined by Ashtekar and Magnon \cite{AM84}, Henneaux and Teitelboim \cite{HT85}. In the late 1990s, a significant surge of interest in solutions of \eqref{eq:VE} with a negative cosmological constant was generated by the AdS/CFT correspondence proposed by Maldacena \cite{M98}. In broad terms, this correspondence establishes a link between a supergravity theory on the interior of aAdS spaces and a superconformal field theory on their conformal boundary.

\bigskip

From a mathematical perspective, the local theory of \eqref{eq:VE} with $\Lambda<0$ presents additional difficulties compared with the well-known case of $\Lambda = 0$ studied by Choquet-Bruhat \cite{CB52}. A first difficulty comes from the fact that the conformal boundary of aAdS spaces is situated at infinity. A conformal embedding allows to bring it at finite distance, however this leads a priori to a singular system of equations on a bounded domain, whose boundary location has yet to be determined. A second difficulty arises from being faced with an initial boundary value problem. Indeed, it is necessary to specify both Cauchy data on an initial hypersurface and boundary conditions on the conformal boundary in order to solve the equations in a unique way. These boundary conditions must satisfy the following key points:
\begin{itemize}
    \item[a)] After hyperbolic reduction through some gauge choices, the boundary conditions for the resulting analytic\footnote{In this paper, the opposite terms `analytic' and `geometric' are to be understood as `gauge-dependent' and `gauge-independent' respectively. It should be noted that `analytic' does not refer to real analytic or holomorphic functions.} initial boundary value problem must lead to a well-posed problem. This implies that the boundary conditions have to match the number of dynamical degrees of freedom of \eqref{eq:VE} which is $(n+1)(n-2)/2$, see for instance Christodoulou \cite[Section 2.1]{C08}.
    
    \item[b)] The Einstein equations~\eqref{eq:VE} entail the existence of constraints which have to be satisfied initially and propagated by the evolution. To go back from the analytic to the geometric setting and thus obtain geometric existence, the boundary conditions must enable this propagation.
    
    \item[c)] Furthermore, the boundary conditions must make sense in any gauge and therefore have a real geometric meaning in order to ensure geometric uniqueness, see Friedrich \cite{F09}.
\end{itemize}

\subsection{Previous works on local existence of aAdS spaces}

\bigskip

\noindent \textbf{Conformal methods}

\bigskip

The use of conformal methods for the Einstein equations dates back to the work of Penrose in the 1960s. In particular, he introduced the concept of asymptotically simple spacetimes in \cite{P60} allowing the attachment of infinities through a conformal embedding of the spacetime. The metric prior to (respectively following) the conformal mapping is said to be \emph{physical} (respectively \emph{unphysical}). A fundamental aspect of the construction is that the physical and unphysical spacetimes share the same global causality structure. The common smoothness requirement imposed on the unphysical metric is linked to the peeling property of gravitational radiation, which gives the asymptotic behaviour of the Weyl tensor near a smooth conformal boundary. This is generally stated in the case of a vanishing cosmological constant $\Lambda=0$. Nowadays, it is however understood that smoothness of null infinity is too strong a requirement in the asymptotic flat case, see Kehrberger \cite{K22}.

In the 1970s and 1980s, conformal mappings began to be exploited in the context of partial differential equations related to general relativity. In particular, they are used in order to solve the constraints by Lichnerowicz, Choquet-Bruhat and York, see for instance \cite{CB74}, and to prove global existence of some semilinear and then quasilinear wave equations, see for example Christodoulou \cite{C86}. Moreover, Friedrich derived in 1981 \cite{F81} a set of smooth tensorial equations on the unphysical spacetime encoding the Einstein equations for the physical spacetime, which are known as the \emph{conformal Einstein equations}. Indeed, these equations allow for conformal mappings with a conformal factor vanishing on subsets of the unphysical spacetime. The conformal Einstein equations led to many applications in the case of $\Lambda > 0$, see for instance \cite{F86}.

The first work establishing a local existence and uniqueness theory for \eqref{eq:VE} with $\Lambda<0$ is the celebrated article of Friedrich of 1995 \cite{F95}. It is based on an extension of the conformal Einstein equations, enlarging the connections used from Levi-Civita connections to \emph{Weyl connections}. The resulting tensorial equations are called the \emph{extended conformal Einstein equations}. The new gauge freedom, consisting of choosing the Weyl connection, is essential to Friedrich's local existence result, obtained for $n=3$ and a broad set of analytic boundary conditions. These boundary conditions are a priori gauge-dependent, meaning that they only make sense in a specific gauge. A second result of \cite{F95} establishes that among them, there exists one which can be interpreted geometrically as imposing the \emph{boundary conformal class} through an auxiliary evolution system on the boundary. This is commonly referred to as the \emph{Dirichlet boundary condition} for \eqref{eq:VE} with $\Lambda < 0$. Geometric local existence and uniqueness of solutions are thus ensured for this specific boundary condition.

For more details on conformal methods in general relativity and Friedrich's works, one can refer to the comprehensive book written by Valiente Kroon \cite{K16}.

\bigbreak

\noindent \textbf{Asymptotic expansions}

\bigskip

The seminal work of Fefferman and Graham on the ambient space construction \cite{FG85}, see also the extensive book \cite{FG12}, contains in particular an asymptotic expansion nearby the conformal boundary of aAdS metrics satisfying the vacuum Einstein equations \eqref{eq:VE} with $\Lambda<0$. This expansion is called the \emph{Fefferman-Graham expansion} (FG expansion). The Einstein equations leave two free data in the expansion: the leading term and part of a sub-leading term whose position in the expansion depends on $n$, the dimension of the boundary. The former corresponds to the \emph{conformal class} and imposing it constitutes the \emph{Dirichlet boundary condition} as seen previously. The latter is the trace-free and divergence-free part of the \emph{boundary stress-energy tensor} and imposing it is thereby typically interpreted as the \emph{Neumann boundary condition}.

In 2019, Enciso and Kamran \cite{EK19} proved local existence and uniqueness of solutions for the vacuum Einstein equations \eqref{eq:VE} with $\Lambda<0$ under the Dirichlet boundary condition for all boundary dimensions $n\geq3$, thus extending Friedrich's result to higher dimensions. Their work is inspired by the Riemannian case studied by Graham and Lee \cite{GL91}. The proof uses a peeling strategy to derive iteratively the first terms of an asymptotic expansion of the metric, akin to the FG expansion, as well as the analysis of linear fields on aAdS spaces developed in various articles listed in \Cref{sec:further_works}.

\subsection{Main result}
\label{sec:main_result}

The main goal of this paper is to derive new geometric boundary conditions for which local existence and uniqueness of solutions to the geometric initial boundary value problem, arising from the vacuum Einstein equations \eqref{eq:VE} with $\Lambda<0$ and $n=3$, hold. To state these new boundary conditions, recall that the FG expansion of a 4-dimensional aAdS metric $\widetilde{g}$ takes the following form
\[ \widetilde{g} = \frac{1}{x^2} \left( dx^2 + \mathfrak{h}_{ij}dy^idy^j - x^2 \mathfrak{l}_{ij}dy^idy^j + x^3 \mathfrak{t}_{ij} dy^i dy^j + \grando{x^4} \right) \,. \]
In this expansion, $x$ is a boundary defining function of the conformal boundary $\mathscr{I}$, $\mathfrak{h}$ is a Lorentzian metric on $\mathscr{I}$, $\mathfrak{l}$ is the Schouten tensor of $\mathfrak{h}$ and $\mathfrak{t}$ is the boundary stress-energy tensor. The tensor $\mathfrak{t}$ is trace-free and divergence-free for the metric $\mathfrak{h}$. The geometric free data on the conformal boundary is given by the class $[(\mathfrak{h},\mathfrak{t})]$ for the following equivalence relation
\[ (\mathfrak{h}',\mathfrak{t}') \sim (\mathfrak{h},\mathfrak{t}) \text{ if } \exists \, \Omega \in \mathcal{C}^\infty(\mathscr{I},\R_+) \text{ such that } \mathfrak{h}' = \Omega^2 \mathfrak{h} \text{ and } \mathfrak{t}' = \Omega^{-1} \mathfrak{t}' \,. \]
Recall furthermore that the Cotton-York tensor $\mathfrak{y}$ of $\mathfrak{h}$ is defined as the dual of the Cotton tensor $\mathfrak{c}$, that is
\[ \mathfrak{y}_{ij} := \frac{1}{2} \epsilon_i{}^{kl} \mathfrak{c}_{klj} := \epsilon_i{}^{kl} \mathfrak{D}_{[k} \mathfrak{l}_{l]j} \,, \]
where $\epsilon$, $\mathfrak{D}$ and $\mathfrak{l}$ denote respectively the volume form, the Levi-Civita connection and the Schouten tensor of $\mathfrak{h}$.

The new family of geometric reflective boundary conditions consist in imposing that, for any representative $(\mathfrak{h},\mathfrak{t})$ of the geometric free data on the conformal boundary, the boundary stress-energy tensor $\mathfrak{t}$ is proportional to the Cotton-York tensor of $\mathfrak{h}$ by a fixed real constant $\mu \in \R$, that is
\begin{equation}
    \label{eq:bc_robin}
    \mathfrak{t}_{ij} = \mu \, \mathfrak{y}_{ij} \,.
\end{equation}
This family can be interpreted as the homogeneous Robin boundary conditions. Indeed, when $\mu = 0$, \eqref{eq:bc_robin} reduces to the homogeneous Neumann boundary condition since it sets the boundary stress-energy tensor (class) to zero. Moreover, in the limit $\mu \to \pm \infty$, \eqref{eq:bc_robin} becomes $\mathfrak{y}_{ij} = 0$. This is equivalent to imposing that the conformal class $[\mathfrak{h}]$ on the conformal boundary is locally conformally flat, in agreement with the definition of the homogeneous Dirichlet boundary condition.

In this paper, we prove the local existence and uniqueness of solutions of \eqref{eq:VE} with $\Lambda<0$ and $n=3$ for the above family of geometric boundary conditions. A rough version of our main result is the following, see \Cref{thm:main_precise} for the precise formulation.

\begin{thm}[Main theorem, rough version]
    \label{thm:main_rough}
    For any 3-dimensional asymptotically hyperbolic smooth solution to the Vacuum Constraint equations \eqref{eq:VC} with cosmological constant $\Lambda=-3$, satisfying a set of geometric regularity conditions on its conformal boundary and a family of compatibility conditions with the boundary conditions \eqref{eq:bc_robin}, there exists a 4-dimensional asymptotically Anti-de Sitter space which is a smooth local solution to the Vacuum Einstein equations \eqref{eq:VE} with cosmological constant $\Lambda=-3$, agreeing with the initial data on some initial hypersurface, and satisfying the boundary conditions \eqref{eq:bc_robin}. Moreover, this space is locally unique.
\end{thm}

\bigbreak

\noindent \textbf{Comments on the main theorem}

\begin{itemize}
    \item This addresses the problem for all negative cosmological constants $\Lambda<0$ since one can always reduce to the case $\Lambda=-3$ by multiplying the fields by a constant. 
    
    \item The `geometric regularity conditions on the conformal boundary' imposed on the initial data refer to a set of novel conditions, described in $i)$ of \Cref{thm:smooth_unphysical_fields}, for the unphysical fields to be smooth up to the conformal boundary. The `compatibility conditions with the boundary conditions \eqref{eq:bc_robin}' are the corner-type conditions defined in \Cref{lem:geometric_compatibility_conditions}.
    
    \item The Birmingham-AdS spaces with $m=0$ (which includes the AdS space) defined in \Cref{sec:BAdS} and the aAdS spaces defined in \Cref{sec:special_boundaries} are examples of aAdS spaces verifying the homogeneous Neumann boundary condition, that is \eqref{eq:bc_robin} with $\mu=0$. To the best of our knowledge, there are very few known examples of aAdS that are solutions to \eqref{eq:VE} and satisfy \eqref{eq:bc_robin} with $\mu \neq 0$ outside of the real analytic class. In \Cref{sec:ppwave}, we give examples of such aAdS spaces constructed out of conformal boundary whose conformal class contains particular pp-wave metrics. Up to solving the geometric initial data problem, our result provides non-trivial solutions which are not real analytic. In particular, for $\mu=0$, these translate into non-real analytic solutions to the Poincaré problem, defined by Fefferman and Graham in \cite{FG85}, for 3-dimensional boundaries of Lorentzian signature.
    
    \item A physical interpretation of the boundary condition \eqref{eq:bc_robin} in terms of Conformal Field Theory on the conformal boundary can be found in de Haro \cite{dH09}.

    \item After completion of this work, we became aware of the recent\footnote{Note that our geometric boundary conditions \eqref{eq:bc_robin} were presented on June 10\textsuperscript{th}, 2025 during the conference \textit{Boundaries, stability, and singularities in general relativity - A meeting celebrating Helmut Friedrich’s 80th birthday}.} article \cite{FAS25} which proposes a similar type of boundary conditions.
\end{itemize}

\bigbreak

\noindent \textbf{Open problems}

\bigskip

Our result naturally leads to several open problems.
\begin{itemize}
    \item Understanding how to recover solutions such as Schwarzschild-AdS and Kerr-AdS by solving the geometric initial-boundary value problem with inhomogeneous Neumann or Robin-type boundary conditions.
    \item Constructing similar geometric initial data satisfying the compatibility conditions.
    \item Addressing problems in higher dimensions, that is for $n\geq 4$. This article is also written in the hope that it will be useful to that end. Thus, many definitions and computations are presented in any dimension, even though our main result holds in dimension $n=3$ (note that our boundary conditions \eqref{eq:bc_robin} do not make sense in higher dimensions).
\end{itemize}

\bigbreak

\noindent \textbf{Comments on the proof}

\bigskip

Our proof makes extensive use of Friedrich's framework and his extended conformal vacuum Einstein equations which we will review in detail in the article. We made some modifications and improvements to it which we highlight here:
\begin{itemize}
    \item The analysis is written in a tensorial formalism rather than a spinorial one. This facilitates the dialogue with the FG expansion presented in \Cref{sec:FG_expansion} and the consequent asymptotic expansions of the curvature tensors which we derive in \Cref{sec:expansions_curvature_tensors}. These are crucial to link the electromagnetic decomposition of the rescaled Weyl tensor and the boundary free data.
    
    \item Weyl connections are systematically used, notably to define the extended version of the Friedrich scalar in \Cref{sec:friedrich_scalar} and for the extended conformal vacuum constraint equations in \Cref{sec:CVC}. Similarly to the treatment of the (non-extended) conformal Einstein equations by Friedrich, see for example \cite{F82}, the extended Friedrich scalar enables to only work with equations which are regular up to the conformal boundary, not only for the evolution but also for the propagation of the constraints. This is actually a key feature of the (extended or not) conformal Einstein equations. Regarding the constraints, we find it better to write them down for Weyl connections since the gauge construction and thus the analytic initial data requires them.
    \item The gauge construction, see \Cref{prop:construction_gauge}, only refer to the extended conformal vacuum Einstein equations and not to the vacuum Einstein equations. This removes an unnecessary extra step which discriminates the conformal boundary.
    
    \item In \Cref{sec:derivation_evolution_system}, we follow a different strategy to find an evolution system enabling to show the propagation of the constraints. Indeed, we choose to exploit the algebraic dependencies rather than adding a multiple of certain constraints to the evolution equations. The two strategies are linked by the fact that the terms that need to be counteracted by a constraint either do not appear or are removed using the algebraic dependencies.
\end{itemize}
Similarly to the proof of Friedrich, our geometric boundary conditions \eqref{eq:bc_robin} are eventually inferred from a minimal set of analytic boundary conditions by means of an auxiliary evolution system on the boundary, see \Cref{sec:homogeneous_robin}.

\subsection{Further works on AdS and aAdS spaces}
\label{sec:further_works}

\bigskip

\noindent \textbf{Linear fields}

\bigskip

In 1978, Avis, Isham and Storey \cite{AIS78} showed the global well-posedness of the conformal wave equation on the AdS space under the homogeneous Dirichlet and Neumann boundary conditions. Their work was generalised to the Klein-Gordon equation -- also known as the massive wave equation and which is not conformally invariant -- by Breitenlohner and Freedman \cite{BF82bis,BF82} in 1982. The spectral properties of the Klein-Gordon operator on AdS and the boundary conditions leading to well-posedness were systematically studied by Ishibaldi and Wald \cite{IW04} in 2004. Further analysis was established by Bachelot for the Dirac equation in 2008 \cite{B08} and the Klein-Gordon equation on the Poincaré patch of the AdS space in 2011 \cite{B11}.

Moving away from the exact AdS space, linear fields on aAds spaces have been investigated since the 2010s. The well-posedness of the Klein-Gordon equation on an aAdS space under the homogeneous Dirichlet boundary condition was proved by Vasy \cite{V12} and Holzegel \cite{H12}. In 2013, Warnick \cite{W13} introduced a new analytic framework by considering twisted derivatives. This allowed him to prove the well-posedness under Robin boundary conditions. See also Enciso and Kamran \cite{EK15}. Decay for the Klein-Gordon equation on Kerr-Anti-de Sitter was studied by Holzegel and Smulevici \cite{HS13kads,HS14}. In \cite{HW14}, Holzegel and Warnick investigated the boundedness properties of the Klein-Gordon equation on aAdS black holes depending on the boundary conditions imposed. While these previous works focus on scalar equations, the Teukolsky equations were recently investigated by Holzegel and Graf focusing on mode stability on Kerr-Anti-de Sitter \cite{GH23} and decay properties on Schwarzschild-Anti-de Sitter \cite{GH24_part_II}.

Apart from reflective boundary conditions which are used in all above-mentioned works, Holzegel, Luk, Smulevici and Warnick \cite{HLSW20} studied the global dynamics and derived decay estimates for the conformal wave equation, the linearised Bianchi equation and the Maxwell equations on a fixed AdS background in the case of dissipative boundary conditions. 

In complement to initial boundary value problems, the existence of a boundary also give rise to unique continuation problems. A unique continuation statement was shown for various wave equations on aAdS spaces by Holzegel and Shao \cite{HS16,HS17}. A gauge-invariant unique continuation criterion was later derived by Chatzikaleas and Shao \cite{CS22}, which is also applicable in the non-linear setting. Moreover, counterexamples to unique continuation for waves were studied by Guisset and Shao \cite{GS24}.

\bigbreak

\noindent \textbf{Non-linear fields}

\bigskip

With regard to the conformal approach, Friedrich's framework was extended in the 2010s to several cases: coupling with the Maxwell equations by Lübbe and Valiente Kroon \cite{LK12}, system of conformal wave equations by Carranza and Valiente Kroon \cite{CK18}, coupling with trace-free matter by Carranza, Hursit and Valiente Kroon \cite{CHK19}.

The spherically symmetric Einstein-Klein-Gordon system with $\Lambda<0$ and under Dirichlet boundary conditions was studied by Holzegel and Smulevici who showed local existence and uniqueness of solutions in \cite{HS12} and the stability of Schwarzschild-Anti-de Sitter in \cite{HS13}. Holzegel and Warnick constructed solutions of the spherically symmetric Einstein-Klein-Gordon system under Robin boundary conditions in \cite{HW15}.

The non-linear (in)stability of the Anti-de Sitter space, originally conjectured by Dafermos and Holzegel \cite{DH06} and by Anderson \cite{A06}, was first addressed by Bizo\'n and Rostworowski \cite{BR11} who investigated by means of numerical methods the spherically symmetric Einstein-massless scalar field system with a negative cosmological constant and reflective boundary conditions. This work was completed by Maliborski and Rostworowski who exhibited in \cite{MR13} time-periodic solutions using formal series. By restricting to toy models on the Einstein cylinder, Chatzikaleas and Smulevici \cite{CS24} constructed non-formal time-periodic solutions to various field equations. In addition of the study of time-periodic solutions, a substantial contribution on the instability of AdS for reflective boundary conditions was made by Moschidis who proved the instability of the spherically symmetric Einstein-null dust and Einstein-massless Vlasov systems with $\Lambda<0$ in \cite{M20,M23}.

In 2021, Shao \cite{S21} showed that partial FG expansions hold when sufficient regularity is assumed on the rescaled metric. Finally, Holzegel and Shao \cite{HS23} addressed the unique continuation problem for the Einstein equations with $\Lambda<0$. They also proved the inheritance of symmetries from the conformal boundary to the bulk.

\subsection{Structure of the article}
\label{sec:structure_article}

In \Cref{sec:preliminaries}, we recall notions and tools which will be used later in the article. In particular, this includes manifolds with corners in \Cref{sec:manifolds_with_corners} and conformal geometry in \Cref{sec:conformal_geometry}.

\Cref{sec:aAdS_spaces} is dedicated to the definitions of aAdS spaces and a presentation of their fundamental properties. This section includes standard examples of aAdS spaces, giving some insight of what boundary stress-energy tensors can look like.

\Cref{sec:local_existence} is the main section of the article and contains the proof of \Cref{thm:main_precise}. We introduced the extended Conformal Vacuum Einstein equations in \Cref{sec:CVE} before using them to prove local existence of solutions for some analytic boundary conditions in \Cref{sec:geometric_existence}. A subset of these analytic boundary conditions are then interpreted geometrically in \Cref{sec:geometric_bc}. In \Cref{sec:geometric_initial_data}, we discuss the geometric initial data problem. Finally, the different parts are combined together in \Cref{sec:main_thm} leading to \Cref{thm:main_precise}.

\subsection{Acknowledgements}

The author would like to thank J. Smulevici for his advice and careful reading of this paper. 

\section{Preliminaries}
\label{sec:preliminaries}

In this section, we introduce the notations, notions and tools which will be used in the whole article. Most of the proofs are omitted since they are standard and can be found in the references. Only a few are detailed either because we could not find a suitable reference or due to their importance in the proof of local existence of aAdS spaces.

\subsection{Notations and conventions}

The symbol $:=$ in an equation denotes a definition. The sign convention for the signature of Lorentzian metrics is $(-+ \dots +)$. The action of a 1-form $\kappa$ on a vector field $X$ is denoted by $\langle \kappa , X \rangle$ and the action of a vector field $X$ on a smooth function $f$ by $X\cdot f$ or $X(f)$.

\bigbreak

The symbols $\nabla$, $\widetilde{\nabla}$, $\overline{\nabla}$, $\breve{\nabla}$ (respectively $D$, $\widetilde{D}$, $\mathfrak{D}$ and $\slashed{D}$, $\slashed{\mathfrak{D}}$) denote the Levi-Civita of the $(n+1)$-dimensional metrics $g$, $\widetilde{g}$, $\overline{g}$, $\breve{g}$ (respectively $n$-dimensional metrics $h$, $\widetilde{h}$, $\mathfrak{h}$ and $(n-1)$-dimensional metrics $\slashed{h}$, $\slashed{\mathfrak{h}}$) whereas $\widehat{\nabla}$, $\widecheck{\nabla}$ denote Weyl connections on $(n+1)$-dimensional and $n$-dimensional manifolds. Tensor fields associated to a connection are indicated with the symbol of the connection. For example, $\widehat{W}$ is the Weyl tensor of $\widehat{\nabla}$.

\bigbreak
 
The components of tensor fields with respect to a coordinate system are written with Greek alphabet indices for $(n+1)$-dimensional manifolds and with Latin alphabet indices for $n$-dimensional ones. Note that the Latin indices range either from $1$ to $n$ for Riemaniann manifolds or from $0$ to $n-1$ for Lorentzian manifolds. Capital Latin alphabet indices are reserved for $(n-1)$-dimensional Riemannian manifolds and range from $1$ to $n-1$. 

A frame field $(e_{\bf a})$ defines a dual coframe field $(\omega^{\bf a})$ by $\langle \omega^{\bf a},e_{\bf b}\rangle = \delta^{\bf a}{}_{\bf b}$. The connection coefficients of a connection $\nabla$ with respect to a frame field are defined by $\Gamma_{\bf a}{}^{\bf c}{}_{\bf b} := \langle \omega^{\bf c},\nabla_{e_{\bf a}} e_{\bf b}\rangle$. The frame components of tensor fields are written with the following bold indices ${\bf a,b,c,d,\dots}$ (respectively ${\bf i,j,k,l,\dots}$ and ${\bf A,B,C,D,\dots}$) for $(n+1)$-dimensional (respectively $n$-dimensional and $(n-1)$-dimensional) manifolds.

The Einstein summation convention is used, that is a repeated index in superscript and subscript is implicitly summed over its range. Symmetric and antisymmetric parts of a tensor field are denoted by $(\;)$ and $[\;]$, respectively. Some indices can be excluded using $|\;|$. For instance,
\[ T_{(i_1 i_2|j|i_3)} := \frac{1}{6} \sum_{\sigma \in S_3} \, T_{i_{\sigma(1)} i_{\sigma(2)} j i_{\sigma(3)}} \,. \]

\subsection{Manifolds with corners}
\label{sec:manifolds_with_corners}

We recall the definitions of smooth manifolds with corners as well as their basic properties. The presentation follows the lines of Melrose \cite{M96}. The main concepts are boundary hypersurfaces (see \Cref{def:boundary_hypersurface}), smooth manifolds with corners (\Cref{def:manifold_with_corners}) and boundary defining functions (\Cref{def:bdf}). 

Corners arise naturally in the context of initial boundary value problems as the intersection of the initial hypersurface and the boundary. Moreover, boundary defining functions are essential to define the conformal boundary and aAdS spaces.

\subsubsection{Definitions}

Let $N \in \N$ and $k,l \in \llbracket0,N\rrbracket$. Introduce
\begin{align*}
    \R^N_k &:= [0,+\infty)^k \times \R^{N-k} = \{ x = (x^1,\dots,x^N) \in \R^N \mid \forall \, j \in \llbracket1,k\rrbracket, \, x_j \geq 0 \} \,, \\
    \partial_l \R^N_k &:= \{ x \in \R^N_k \mid \#\{j \in \llbracket1,k\rrbracket \mid x_j = 0\} = l\} \,. 
\end{align*}
For any open set $U \subset \R^N_k$, define
\[ \partial_l U := U \cap \partial_l \R^N_k \,. \]

\begin{defi}
    \label{def:smooth_functions}
    The algebra of smooth functions on an open set $U \subset \R^N_k$ is
    \[ \mathcal{C}^\infty(U) := \{ f|_U \mid f \in \mathcal{C}^\infty(V) \} \, \]
    where $V \subset \R^N$ is an open set such that $U = V \cap \R^N_k$.
\end{defi}

\begin{defi}
    A map $\phi$ from an open set $U \subset \R^N_k$ to an open set $V \subset \R^N_j$ is said to be a diffeomorphism if it is a homeomorphism and if the components of $\phi$ and $\phi^{-1}$ are smooth functions in $\mathcal{C}^\infty(U)$ and $\mathcal{C}^\infty(V)$ respectively.
\end{defi}

\begin{prop}
    Let $\phi : U \subset \R^N_k \to V \subset \R^N_j$ be a diffeomorphism. Then
    \[ \forall \, l \in \llbracket 0,N \rrbracket, \quad \phi(\partial_l U) = \partial_l V \,.\]
\end{prop}

\begin{defi}
    Let $\M$ be a paracompact Hausdorff topological space.
    \begin{itemize}
        \item A (N-dimensional) \emph{chart} $(\phi,\mathcal{U},k)$ on $\M$ is a homeomorphism $\phi$ from an open subset $\mathcal{U} \subset \M$ to an open subset $V \subset \R^N_k$ for some $k \in \llbracket0,N\rrbracket$.
        \item Two charts $(\phi_1,\mathcal{U}_1,k_1)$, $(\phi_2,\mathcal{U}_2,k_2)$ on $\M$ are said to be \emph{compatible} if either $\mathcal{U}_1 \cap \mathcal{U}_2 = \varnothing$ or $\phi_2 \circ (\phi_1)^{-1} : \phi_1(\mathcal{U}_1 \cap \mathcal{U}_2) \to \phi_2(\mathcal{U}_1 \cap \mathcal{U}_2)$ is a diffeomorphism.
        \item A (N-dimensional) \emph{atlas} on $\M$ is a family of (N-dimensional) charts $(\phi_a,\mathcal{U}_a,k_a)_{a \in A}$ such that
        \begin{itemize}
            \item[i)] any two charts in the family are compatible,
            \item[ii)] the family covers $\M$ that is $\M = \bigcup_{a \in A} \mathcal{U}_a$.
        \end{itemize}
        \item An atlas on $\M$ is said to be \emph{maximal} if it contains all charts compatible with each chart of the atlas.
    \end{itemize}
\end{defi}

\begin{defi}
    A N-dimensional \emph{t-manifold} is a connected, paracompact and Hausdorff topological space endowed with a maximal N-dimensional atlas.
\end{defi}

\begin{defi}
    \label{def:boundary_hypersurface}
    Let $\M$ be a t-manifold.
    \begin{itemize}
        \item A \emph{coordinate system at} $p\in\M$ is a chart $(\phi,\mathcal{U},k)$ such that $\phi(p) = 0$.
        \item The \emph{set of $l$-corner points} $\partial_l \M$ is defined by
        \[ \partial_l \M := \{ p \in \M \mid \exists \, (\phi,\mathcal{U},k) \text{ coordinate system at $p$ such that } k=l  \} \,. \]
        For $l=0$, the notation $\mathring{\M}$ is used instead of $\partial_0 \M$.
        \item The \emph{boundary} of $\M$ is defined by
        \[ \partial \M := \overline{\partial_1\M} = \bigcup_{l \geq 1} \partial_l \M = \M \setminus \mathring{\M} \,, \]
        \item A \emph{boundary hypersurface} $\mathcal{S} \subset \partial \M$ is the closure of a connected component of $\partial_1 \M$,
        \item The algebra of smooth function on $\M$ is defined by
        \[ \mathcal{C}^\infty(\M) := \{ f : \M \to \R \mid \forall \text{ chart } (\phi,\mathcal{U},k), f \circ \phi^{-1} \in \mathcal{C}^\infty(\phi(\mathcal{U})) \} \,. \]
    \end{itemize}
\end{defi}

\begin{ex}
    The cube $[0,1]^3$ is a t-manifold. The sets of $0$-corner points, $1$-corner points, $2$-corner points and $3$-corner points are respectively the interior $(0,1)^3$, the faces, the edges and the corners.
\end{ex}

\begin{defi}
    \label{def:submanifold}
    Let $\M$ be a N-dimensional t-manifold. A connected subset $\mathcal{S} \subset \M$ is said to be a \emph{submanifold} if for all $p \in \mathcal{S}$, there exists a coordinate system $(\phi,\mathcal{U},k)$ at $p$, an invertible matrix $G \in GL_N(\R)$ and an open neighbourhood $V \subset \R^N$ of $0_{\R^N}$ such that
    \[ \phi(\mathcal{S} \cap \mathcal{U}) = G(\R^{N'}_{k'} \times \{0_{\R^{N-N'}}\}) \cap V \]
    for some $k' \leq N' \leq N$.
\end{defi}

\begin{defi}
    \label{def:manifold_with_corners}
    A \emph{smooth manifold with corners} is a t-manifold $\M$ such that all boundary hypersurfaces of $\M$ are submanifolds in the sense of \Cref{def:submanifold}.

    In particular, a \emph{smooth manifold with boundary} (respectively \emph{smooth manifold without boundary}) is a smooth manifold with corners $\M$ such that $\partial_l \M = \varnothing$ for $l \geq 2$ and $\partial_1\M \neq \varnothing$ (respectively $\partial_l \M = \varnothing$ for $l \geq 1$). 
\end{defi}

\subsubsection{Boundary defining functions}

Let $\M$ be a smooth manifold with corners and $\mathcal{I}$ be a non-empty union of disjoint boundary hypersurfaces.

\begin{defi}
    \label{def:bdf}
      A \emph{local boundary defining function} of $\mathcal{I}$ is a smooth function $x$ on an open set $\mathcal{U}$ with $\mathcal{I}\cap\mathcal{U}\neq\varnothing$ such that
    \begin{itemize}
        \item[i)] $x \geq 0$ on $\mathcal{U}$,
        \item[ii)] $\mathcal{I}\cap\mathcal{U} = \{p \in \mathcal{U} \mid x(p) = 0\}$,
        \item[iii)] for all $p \in \mathcal{I}\cap\mathcal{U}$, there exists a coordinate system at $p$ with $x$ as first coordinate.
    \end{itemize}
    If $\, \mathcal{U} = \M$, $x$ is said to be a \emph{global boundary defining function} of $\mathcal{I}$.
\end{defi}

\begin{rem}
    If $x$ is a local boundary defining function of $\mathcal{I}$ on an open set $\mathcal{U}$ then one deduces from $iii)$ that
    \begin{equation}
        \label{eq:bdf_iii)_modified}
        \forall \, p \in \mathcal{I} \cap \mathcal{U}, \quad dx|_p \neq 0 \,.
    \end{equation}
    This is the common version of $iii)$.
\end{rem}

\begin{prop} 
    There exists a global boundary defining function of $\mathcal{I}$.
\end{prop}

\begin{lem}
    \label{lem:dividing_bdf}
    Let $x \in \mathcal{C}^\infty(\mathcal{U})$ be a local boundary defining function of $\mathcal{I}$ and $f \in \mathcal{C}^\infty(\mathcal{U},\R)$ be a smooth function on $\mathcal{U}$. Then the function $f/x$ extends smoothly on $\mathcal{U} \cap \mathcal{I}$ if and only if $f$ vanishes on $\mathcal{U} \cap \mathcal{I}$.
\end{lem}

\begin{rem}
    Since $f,x$ are smooth functions on $\mathcal{U}\setminus\mathcal{I}$ and $x$ nowhere vanishes on $\mathcal{U}\setminus\mathcal{I}$ by $ii)$ of \Cref{def:bdf}, the function $f/x$ is always well-defined and smooth on $\mathcal{U}\setminus\mathcal{I}$.
\end{rem}

\begin{proof}
    $\Longrightarrow$ This is evident by evaluating $f = x (f/x)$ on $\mathcal{I}$.

    $\Longleftarrow$ Let $p \in \mathcal{U} \cap \mathcal{I}$. By $iii)$ of \Cref{def:bdf}, there exists a coordinate system $(\phi,\mathcal{V},k)$ at $p$ such that $\phi = (x,y^1,\dots,y^{N-1})$. One can assume that $\mathcal{V} \subset \mathcal{U}$. Since $x$ and $f$ are smooth, there exist smooth extensions $\overline{x}$ and $\overline{f}$ of respectively $x \circ \phi^{-1}$ and $f \circ \phi^{-1}$ on an open set $W$ of $\R^N$ with $\phi(\mathcal{V}) = W \cap \R^N_1$ (see \Cref{def:smooth_functions}). Thanks to \eqref{eq:bdf_iii)_modified}, one can assume that $\overline{x}$ nowhere vanishes on $W\setminus\phi(\mathcal{V}\cap\mathcal{I})$ up to reducing $W$. Hence $\overline{f}/\overline{x}$ is well-defined and smooth on $W\setminus\phi(\mathcal{V}\cap\mathcal{I})$. Finally, $\overline{f}/\overline{x}$ extend smoothly on $\phi(\mathcal{V}\cap\mathcal{I})$ using Taylor expansions of $f$ and $f$ vanishing on $\mathcal{U}\cap\mathcal{I}$. Hence $\overline{f}/\overline{x}$ is a smooth extension on $W$ of $f/x \circ \phi^{-1}$, which implies that $f/x$ is well-defined and smooth on $\mathcal{V}$.
\end{proof}

\noindent A consequence of the previous lemma and a key property of boundary defining functions is that they are all `linearly equivalent' as stated below.

\begin{cor}
    \label{lemma:bdf}
    Let $x \in \mathcal{C}^\infty(\mathcal{U})$ and $x' \in \mathcal{C}^\infty(\mathcal{U}')$ be two local boundary defining functions of $\mathcal{I}$ such that $\mathcal{U} \cap \mathcal{U}' \neq \varnothing$. Then the function $\Theta := x'/x$ is a smooth positive function on $\mathcal{U} \cap \mathcal{U}'$. 
\end{cor}

\begin{ex}
    If $x$ is a local boundary defining function then so is $x(2+x)$ but not $x^2$.
\end{ex}

\begin{proof}
    A fortiori, $x$ is a local boundary defining function on $\mathcal{U} \cap \mathcal{U}'$ and $x'$ vanishes on $\mathcal{U} \cap \mathcal{U}' \cap \mathcal{I}$. By application of \Cref{lem:dividing_bdf}, $x'/x$ extends smoothly on $\mathcal{U} \cap \mathcal{U}' \cap \mathcal{I}$. By $i)$ and $ii)$ of \Cref{def:bdf}, it is clear that $x'/x$ is positive on $(\mathcal{U} \cap \mathcal{U}')\setminus\mathcal{I}$. Using L'Hôpital's rule, one can deduce that $x'/x$ is also positive on $\mathcal{U} \cap \mathcal{U}' \cap \mathcal{I}$.
\end{proof}

\subsubsection{Polyhomogeneous functions}
\label{sec:polyhomogeneous_functions}

We define polyhomogeneous functions in a simpler version than Melrose in~\cite{M96} since it will be sufficient for our case. Polyhomogeneous functions is the right class of functions to consider for the metric components of aAdS spaces in arbitrary dimension as we will see in \Cref{thm:FG_exp}.

\begin{defi}
    An \emph{index set} $E$ is a subset of $\R \times \N_0$ such that
    \begin{itemize}
        \item[i)] $\forall \, s \in \R, \; E_{\leq s} := \{ (r,k) \in E \mid r \leq s\}$ is finite,
        \item[ii)] $(r,k) \in E \implies \forall \; 0 \leq l \leq k, \; (r,l) \in E$,
        \item[iii)] $(r,k) \in E \implies (r+1,k) \in E$\,.
    \end{itemize}
\end{defi}

\begin{ex}
    The minimal index set containing an element $(r,k) \in \R \times \N_0$ is $E = \{ (s,l) \in \R \times \N_0 \mid s-r \in \N_0, \, l \leq k\}$.
\end{ex}

\begin{defi}
    Let $\mathcal{M}$ be a smooth manifold with corners of dimension $N$, $\mathcal{I}$ be a union of disjoint boundary hypersurfaces and $E$ be an index set. A map $f : \mathcal{M} \to \R$ is said to be \emph{polyhomogeneous of class} $m \in \N_0 \cup \{+\infty\}$ with respect to $\mathcal{I}$ and with index set $E$ if
    \begin{itemize}
        \item[i)] $f \in \mathcal{C}^m(\mathcal{M}\setminus\mathcal{I},\R)$,
        \item[ii)] for all local boundary defining function $x$ of $\mathcal{I}$ and for all coordinate system $(\phi = (x,y),\mathcal{V},k)$ at a point $p\in \mathcal{I} \cap \mathcal{U}$ with $x$ as first coordinate, there exists a family of functions $(f_{(r,k)})_{(r,k)\in E} \in \mathcal{C}^m(\R^{N-1},\R)$ such that for all $0 \leq j \leq m$, for all $\alpha \in \N_0^{N-1}$ with $|\alpha| \leq m-j$ and for all $s \in \R$,
        \[ ( x\partial_x)^j \partial_y^\alpha \left( f \circ \phi^{-1} - \sum_{(r,k) \in E_{\leq s}} f_{(r,k)} x^r (\log x)^k\right) = o\left(x^s\right) \,. \]
    \end{itemize}
    The algebra of such functions is denoted by $\mathcal{C}^m_{\text{phg}}(\mathcal{M}\mid\mathcal{I},E)$.
\end{defi}

\begin{rem}
    Index sets are designed so that if $ii)$ holds for one boundary defining function then it remains true for all.
\end{rem}

\subsection{Curvature tensors}

Let $(\mathcal{M},g)$ be a smooth pseudo-Riemannian manifold with corners of dimension $N \geq 2$. Let us recall the definitions and the properties of the standard curvature tensors associated to the Levi-Civita connection $\nabla$ of $g$, namely the Riemann, Ricci, Weyl, Schouten, Cotton and Bach tensors. Smoothness is assumed for simplicity but note that the definitions hold for lower regularity: $\mathcal{C}^2$ for Riemann, Ricci, Schouten and Weyl; $\mathcal{C}^3$ for Cotton; $\mathcal{C}^4$ for Bach.

\subsubsection{The Riemann and Ricci tensors}

\begin{defi}
    The \emph{Riemann tensor} is defined by
    \begin{equation}
        R^\alpha{}_{\beta\mu\nu} X^\beta Y^\mu Z^\nu := \left( \nabla_Y \nabla_Z X - \nabla_Z \nabla_Y X - \nabla_{[Y,Z]} X\right)^\alpha
    \end{equation}
    for all smooth vector fields $X,Y,Z$ on $\mathcal{M}$. The \emph{Ricci tensor} and the \emph{scalar curvature} are the following contractions of the Riemann tensor
    \begin{equation}
        R_{\mu\nu} := R^\alpha{}_{\mu\alpha\nu}, \qquad
        R := g^{\mu\nu} R_{\mu\nu} \,.
    \end{equation}
\end{defi}

\noindent The Riemann tensor of a Levi-Civita connection:
\begin{itemize}
    \item enjoys the following symmetries
        \begin{align*}
            R^\alpha{}_{\beta(\mu\nu)} &= 0 \,, \\
            R_{(\alpha\beta)\mu\nu} &= 0 \,,
        \end{align*}
    \item verifies the first and second Bianchi identities
        \begin{subequations}
            \label{eq:Bianchi_Riemann}
            \begin{align}
                R^\alpha{}_{[\beta\mu\nu]} &= 0 \,, \\
                \nabla_{[\lambda} R^\alpha{}_{|\beta|\mu\nu]} &= 0 \,.
            \end{align}
        \end{subequations}
\end{itemize}
It follows that
\[ R_{\alpha\beta\mu\nu} = R_{\mu\nu\alpha\beta} \,, \qquad R_{[\mu\nu]} = 0 \,, \qquad \nabla^\mu R_{\mu\nu} = \frac{1}{2} \nabla_\nu R \,.\]

\begin{defi}
    For $N\geq3$, $(\mathcal{M},g)$ is said to be \emph{Einstein} if it is solution to the \emph{Vacuum Einstein equations (VE)}
    \begin{equation}
        \tag{\ref{eq:VE}}
        R_{\mu\nu} = \frac{2 \Lambda}{N-2} g_{\mu\nu} \,,
    \end{equation}
    where $\Lambda\in\R$ is the cosmological constant. If $N=2$, $(\M,g)$ is said to be Einstein if its scalar curvature $R$ is constant.
\end{defi}

\subsubsection{The Weyl, Schouten and Cotton tensors}
\label{sec:weyl_schouten_cotton}

Instead of working with the Riemann and Ricci tensors, one can use the Weyl and Schouten tensors. This enables to define a practical orthogonal decomposition of the Riemann tensor, see~\eqref{eq:decomposition_Riemann}. Moreover, these tensors are particularly useful in conformal geometry as we will see in \Cref{sec:conformal_geometry} where the definitions will be extended to Weyl connections. 

In what follows, it will be assumed that $N \geq 3$.

\begin{defi}
    The \emph{Weyl tensor}, whose components are denoted by $W^\alpha{}_{\beta\mu\nu}$, is defined as the totally trace-free part of the Riemann tensor $R^\alpha{}_{\beta\mu\nu}$. The \emph{Schouten tensor} is defined by
    \begin{equation}
        \label{def:schouten}
    	L_{\mu\nu} := \frac{1}{N-2} \left( R_{\mu\nu} - \frac{R}{2(N-1)} g_{\mu\nu} \right) \,.
    \end{equation}
\end{defi}

\noindent Note that the Weyl (respectively Schouten) tensor of a Levi-Civita connection has the same symmetries than its Riemann (respectively Ricci) tensor. Furthermore, the following decomposition holds
	\begin{equation}
        \label{eq:decomposition_Riemann}
		R^\alpha{}_{\beta\mu\nu} = W^\alpha{}_{\beta\mu\nu} + 2 S_{\beta[\mu}{}^{\alpha\sigma} L_{\nu]\sigma} \,,
	\end{equation}
	where
	\begin{equation}
		\label{eq:def_S}
		S_{\alpha\beta}{}^{\mu\nu} := \delta_\alpha{}^\mu \delta_\beta{}^\nu + \delta_\alpha{}^\nu \delta_\beta{}^\mu - g_{\alpha\beta} g^{\mu\nu} \,.
	\end{equation}

\begin{rem}
    Thanks to metric compatibility, one has $\nabla_\gamma S_{\alpha\beta}{}^{\mu\nu} = 0$.
\end{rem}

\begin{defi}
    The \emph{Cotton tensor} is defined by
    \begin{equation}
        \label{eq:def_cotton}
        C_{\lambda\mu\nu} := \nabla_\lambda L_{\mu\nu} - \nabla_\mu L_{\lambda\nu} \,.
    \end{equation}
\end{defi}

\noindent By definition, the Cotton tensor verifies
\begin{equation}
    \label{cotton_antisym_acyclique}
     C_{(\alpha\beta)\gamma} = 0 \,, \qquad C_{[\alpha\beta\gamma]} = 0 \,.
\end{equation}
The second property is called the third Bianchi identity.

\begin{rem}
    If $(\mathcal{M},g)$ is Einstein then its Cotton tensor vanishes identically.
\end{rem}

\noindent The first and second Bianchi identities~\eqref{eq:Bianchi_Riemann} rewrite as
\begin{subequations}
    \begin{align}
        \label{Bianchi1_Weyl}
        W^\alpha{}_{[\beta\mu\nu]} &= 0 \,, \\
        \nabla_{[\lambda} W^\alpha{}_{|\beta|\mu\nu]} &= S_{\beta[\lambda}{}^{\alpha\rho} C_{\mu\nu]\rho} \,.
    \end{align}
\end{subequations}
The contractions of the second Bianchi identity give
\begin{subequations}
    \begin{align}
        \nabla_\alpha W^\alpha{}_{\beta\mu\nu} &= (N-3) C_{\mu\nu\beta} \,, \\
        C_{\lambda\mu}{}^\mu &= 0 \,.
    \end{align}
\end{subequations}
Furthermore, one can prove the fourth Bianchi identity which writes
\begin{equation}
    \nabla_{[\lambda} C_{\mu\nu]\alpha} = - W^\beta{}_{\alpha[\lambda\mu} L_{\nu]\beta} \,.
\end{equation}
Its trace give
\begin{equation}
    \label{eq_fourth_Bianchi_id}
    \nabla^\alpha C_{\mu\nu\alpha} = 0 \,.
\end{equation}

\subsubsection{The Bach tensor}

\begin{defi}
    The \emph{Bach tensor} is defined by
    \begin{equation}
        B_{\mu\nu} := \nabla^\lambda C_{\lambda\mu\nu} + L^{\lambda\sigma} W_{\mu\lambda\nu\sigma} \,.
    \end{equation}
\end{defi}

\noindent By definition, one can check that
\[ B_{[\mu\nu]} = 0 \, \qquad B^\mu{}_\mu = 0 \,. \]

\begin{rem}
    If $(\mathcal{M},g)$ is Einstein then its Bach tensor vanishes identically.
\end{rem}

\noindent Furthermore, one has
\[ \nabla^\mu B_{\mu\nu} = (N-4) L^{\lambda\sigma} C_{\nu\lambda\sigma} \,. \]

\subsection{Conformal geometry}
\label{sec:conformal_geometry}

Conformal geometry is based on the freedom to positively rescale metrics. It plays an important role to study the conformal boundary of aAdS spaces and is at the heart of the (extended) conformal Einstein equations.

We start by reviewing the basic definitions of confomorphisms and conformal structures in \Cref{sec:confomorphims_and_conformal_structures},  before defining Weyl connections and their associated curvature tensors in \Cref{section_Weyl_conn}. Weyl connections are extensively used in the extended conformal Einstein equations presented in \Cref{sec:CVE}. Conformally invariant tensors are introduced in \Cref{sec:conformal_invariance} and we recall remarkable conformal structures in \Cref{sec:special_conf_structures}, stressing the importance of the Weyl and Cotton tensors.

We then introduce special fields of conformal structures such as conformal Killing fields in \Cref{sec:conformal_Killing_fields} and conformal geodesics in \Cref{sec:conf_geo}. The congruences of timelike conformal geodesics studied in \Cref{sec:congruence_conf_geo} will be used later to construct a gauge for the extended conformal vacuum Einstein equations, see \Cref{sec:gauge_construction}. Conformal Killing fields will only be used in connection with boundary stress-energy tensors in \Cref{sec:BSET}.

The main reference of this section is \cite[Section 5]{K16}, see also \cite{F03} for a detailed discussion on conformal geodesics.

\subsubsection{Confomorphisms and conformal structures}
\label{sec:confomorphims_and_conformal_structures}

\begin{defi}
    \label{def:confomorphism}
    A \emph{confomorphism} between two smooth pseudo-Riemannian manifolds $(\mathcal{M}_1,g_1)$ and $(\mathcal{M}_2,g_2)$ is a diffeomorphism $\phi : \mathcal{M}_1 \to \mathcal{M}_2$ such that there exists a smooth positive function $\Omega \in \mathcal{C}^\infty(\mathcal{M}_1,\R^\star_+)$ verifying $\phi^\star g_2 = \Omega^2 g_1$. The function $\Omega$ is called the \emph{conformal factor} associated to $\phi$.

    A confomorphism is said to be an \emph{isometry} if $\Omega = 1$ and a \emph{conformal rescaling} if $\mathcal{M}_1 = \mathcal{M}_2$ and $g_2 = \Omega^2 g_1$.
\end{defi}

\begin{defi}
	Two metrics $g$ and $\widetilde{g}$ on a smooth manifold $\mathcal{M}$ are said to be \emph{conformally related} if they are related by a conformal rescaling. This defines an equivalence relation on metrics on $\mathcal{M}$.
\end{defi}

\begin{defi}
	A \emph{conformal structure} is a pair $(\mathcal{M},[\widetilde{g}])$ where $\mathcal{M}$ is a smooth manifold and $[\widetilde{g}]$ is a conformal class of metrics on $\mathcal{M}$.
\end{defi}

\subsubsection{Weyl connections}
\label{section_Weyl_conn}

Weyl connections are the natural extension of Levi-Civita connections in the framework of conformal geometry. In what follows, consider a conformal structure $(\mathcal{M},[\widetilde{g}])$ of dimension $N\geq3$.

\begin{defi}
	 A \emph{Weyl connection} for $(\mathcal{M},[\widetilde{g}])$ is a connection $\widehat{\nabla}$ such that
    \begin{itemize}
        \item[i)] it is torsion-free: for all smooth vector fields $X$, $Y$ on $\M$
        \begin{equation}
            \label{eq_torsion_free}
            \widehat{\Sigma}(X,Y) := \widehat{\nabla}_X Y - \widehat{\nabla}_Y X - [X,Y] = 0 \,,
        \end{equation}
        \item[ii)] for all $g \in [\widetilde{g}]$, there exists a smooth covector field $\widehat{\kappa}_g \in T^\star \mathcal{M}$ verifying
    	\begin{equation}
    		\label{weyl_conn}
    		\widehat{\nabla}_\alpha g_{\beta\gamma} = -2 \, (\widehat{\kappa}_g)_\alpha \, g_{\beta\gamma} \,.
    	\end{equation}
    \end{itemize}
\end{defi}

\begin{rems} \,
	\begin{itemize}
		\item If~\eqref{weyl_conn} holds for a representative $g \in [\widetilde{g}]$ then it holds for all representatives. Indeed, if $\overline{g} = \Omega^2 g$ is another representative, it suffices to take $\widehat{\kappa}_{\overline{g}} = \widehat{\kappa}_g + d(\ln\Omega)$.
		\item The Levi-Civita connection $\nabla$ of a representative $g \in [\widetilde{g}]$ is a Weyl connection.
		\item From~\eqref{weyl_conn}, one deduces that $\widehat{\nabla}_\alpha g^{\beta\gamma} = 2(\widehat{\kappa}_g)_\alpha g^{\beta\gamma}$ and $\widehat{\nabla}_\sigma S_{\alpha\beta}{}^{\mu\nu} = 0$ where $S_{\alpha\beta}{}^{\mu\nu}$ is defined by \eqref{eq:def_S}. \qedhere
	\end{itemize}
\end{rems}

\begin{prop}
    \label{prop:exist&uniq_weylconn}
    Let $g \in [\widetilde{g}]$ be a representative and $\kappa$ be a smooth covector field on $\mathcal{M}$. There exists a unique Weyl connection $\widehat{\nabla}$ for $(\mathcal{M},[\widetilde{g}])$ such that the covector field associated to $\widehat{\nabla}$ with respect to $g$ by~\eqref{weyl_conn} is equal to $\kappa$.
\end{prop}

\begin{proof}
    The Weyl connection is entirely characterised by the following modified Koszul formula
    \begin{align}
        \label{extended_Koszul}
        2 g(\widehat{\nabla}_X Y,Z) &= X \cdot g(Y,Z) + g(Y,[Z,X]) + 2 \langle \kappa,X\rangle g(Y,Z) \nonumber \\
        &\quad + Y \cdot g(Z,X) + g(Z,[X,Y]) + 2 \langle \kappa,Y\rangle g(Z,X) \nonumber \\
        &\quad - Z \cdot g(X,Y) - g(X,[Y,Z]) - 2 \langle \kappa,Z\rangle g(X,Y) \,,
    \end{align}
     for all smooth vector fields $X,Y,Z$ on $\M$.
\end{proof}

\noindent The Riemann and Ricci tensors of a Weyl connection are defined similarly as these of a Levi-Civita connection.

\begin{defi}
    \label{def_riemann_hat}
    Let $\widehat{\nabla}$ be a Weyl connection for the conformal structure $(\mathcal{M},[\widetilde{g}])$. Its \emph{Riemann tensor} $\widehat{R}^\alpha{}_{\beta\mu\nu}$ is defined by
    \begin{equation}
        \label{eq:def_riemann_hat}
    	\widehat{R}^\alpha{}_{\beta\mu\nu} X^\beta Y^\mu Z^\nu := \left( \widehat{\nabla}_Y \widehat{\nabla}_Z X - \widehat{\nabla}_Z \widehat{\nabla}_Y X - \widehat{\nabla}_{[Y,Z]} X\right)^\alpha
    \end{equation}
    for all smooth vector fields $X,Y,Z$ on $\mathcal{M}$. Its \emph{Ricci tensor} is defined by
    \[ \widehat{R}_{\mu\nu} := \widehat{R}^\alpha{}_{\mu\alpha\nu} \,. \]
\end{defi}

\noindent The Riemann and Ricci tensors of a Weyl connection:
\begin{itemize}
    \item enjoy the following algebraic properties
        \begin{subequations}
            \begin{align}
                \widehat{R}^\alpha{}_{\beta(\mu\nu)} &= 0 \,, \\
                \label{eq:symetric_part_Riemann}
                g_{\alpha(\gamma} \widehat{R}^\alpha{}_{\beta)\mu\nu} &= \frac{2}{N} \widehat{R}_{[\mu\nu]} g_{\gamma\beta} \,, \\
                \label{eq:antisymetric_part_ricci}
                \widehat{R}_{[\alpha\beta]} &= N \widehat{\nabla}_{[\alpha} (\widehat{\kappa}_g)_{\beta]} \,,
            \end{align}
        \end{subequations}
        for any representative $g \in [\widetilde{g}]$,
    \item verify the first and second Bianchi identities
        \begin{subequations}
            \begin{align*}
                \widehat{R}^\alpha{}_{[\beta\mu\nu]} &= 0 \,, \\
                \widehat{\nabla}_{[\lambda} \widehat{R}^\alpha{}_{|\beta|\mu\nu]} &= 0 \,.
            \end{align*}
        \end{subequations}
\end{itemize}

\noindent Such as for Levi-Civita connections in \Cref{sec:weyl_schouten_cotton}, let us introduce the Weyl, Schouten and Cotton tensors of Weyl connections.

\begin{defi}
	The \emph{Weyl tensor} $\widehat{W}^\alpha{}_{\beta\mu\nu}$ of a Weyl connection $\widehat{\nabla}$ is defined as the totally trace-free part of its Riemann tensor $\widehat{R}^\alpha{}_{\beta\mu\nu}$ with respect to any representative metric $g \in [\widetilde{g}]$. 
\end{defi}

\begin{defi}
    \label{def:schouten_hat}
    The \emph{Schouten tensor} $\widehat{L}_{\alpha\beta}$ of a Weyl connection $\widehat{\nabla}$ is defined by
    \begin{equation}
		\widehat{L}_{\alpha\beta} := \frac{1}{N-2} \left( \widehat{R}_{(\alpha\beta)} - \frac{g^{\mu\nu}\widehat{R}_{\mu\nu}}{2(N-1)} g_{\alpha\beta} \right) - \frac{1}{N} \widehat{R}_{[\alpha\beta]} \,,
	\end{equation}
    for any representative $g \in [\widetilde{g}]$, the expression being independent of the choice of representative.
\end{defi}

\begin{prop}
    The Riemann tensor of a Weyl connection $\widehat{\nabla}$ verifies the following decomposition
    \begin{equation}
        \label{eq:decomposition_Riemann_Weyl}
		\widehat{R}^\alpha{}_{\beta\mu\nu} = \widehat{W}^\alpha{}_{\beta\mu\nu} + 2 S_{\beta[\mu}{}^{\alpha\sigma} \widehat{L}_{\nu]\sigma} \,,
	\end{equation}
    where $S_{\alpha\beta}{}^{\mu\nu}$ is defined by \eqref{eq:def_S}.
\end{prop}

\begin{defi}
    The \emph{Cotton tensor} of a Weyl connection $\widehat{\nabla}$ is defined by
    \begin{equation}
        \label{def:cotton_weylconn}
        \widehat{C}_{\lambda\mu\nu} := \widehat{\nabla}_\lambda \widehat{L}_{\mu\nu} - \widehat{\nabla}_\mu \widehat{L}_{\lambda\nu} \,.
    \end{equation}
\end{defi}

\noindent By definition, one has
\[ \widehat{C}_{(\alpha\beta)\gamma} = 0 \,. \]

\begin{prop}
    For any Weyl connection $\widehat{\nabla}$, one has
    \begin{subequations}
        \begin{align}
            \label{eq:antisym_part_schouten}
            \widehat{L}_{[\alpha\beta]} &= - \widehat{\nabla}_{[\alpha} (\widehat{\kappa}_g)_{\beta]} \,, \\
            \widehat{C}_{[\lambda\mu\nu]} &= 0 \,,
        \end{align}
    \end{subequations}
    where $\widehat{\kappa}_g$ is the covector field associated to $\widehat{\nabla}$ with respect to a representative $g \in [\widetilde{g}]$. The second equation is called the third Bianchi identity.
\end{prop}

Using the decomposition \eqref{eq:decomposition_Riemann_Weyl}, the first and second Bianchi identities for a Weyl connection $\widehat{\nabla}$ can be written under the form
\begin{subequations}
    \begin{align}
        \widehat{W}^\alpha{}_{[\beta\mu\nu]} &= 0 \,, \\
        \label{eq:second_Bianchi_Weyl}
        \widehat{\nabla}_{[\lambda} \widehat{W}^\alpha{}_{|\beta|\mu\nu]} &= S_{\beta[\lambda}{}^{\alpha\sigma} \widehat{C}_{\mu\nu]\sigma} \,.
    \end{align}
\end{subequations}

\noindent The contractions of the second Bianchi \eqref{eq:second_Bianchi_Weyl} identity give
\begin{subequations}
    \begin{align}
        \label{Bianchi2_Weyl}
        \widehat{\nabla}_\alpha \widehat{W}^\alpha{}_{\beta\mu\nu} &= (N-3) \widehat{C}_{\mu\nu\beta} \,, \\
        g^{\mu\nu} \widehat{C}_{\lambda\mu\nu} &= 0
    \end{align}
\end{subequations}
where $g \in [\widetilde{g}]$ is any representative. The fourth Bianchi identity also holds
\begin{equation}
    \label{eq:fourth_Bianchi_id_Weyl_con}
    \widehat{\nabla}_{[\lambda} \widehat{C}_{\mu\nu]\alpha} = - \widehat{W}^\beta{}_{\alpha[\lambda\mu} \widehat{L}_{\nu]\beta} \,.
\end{equation}

\subsubsection{Weyl connection changes}

Remember that the difference between any two connections $\acute{\nabla}$ and $\grave{\nabla}$ on $\M$ is measured by their transition tensor defined as follows
\[ Q_\mu{}^\sigma{}_\nu := \acute{\Gamma}_\mu{}^\sigma{}_\nu - \grave{\Gamma}_\mu{}^\sigma{}_\nu \,. \]
For instance, for all smooth (1,1)-tensor fields $X^\mu{}_\nu$, one has
\[ \acute{\nabla}_\alpha X^\beta{}_\gamma = \grave{\nabla}_\alpha X^\beta{}_\gamma + Q_\alpha{}^\beta{}_\mu X^\mu{}_\gamma  - Q_\alpha{}^\mu{}_\gamma X^\beta{}_\mu \,. \]
It follows that the torsion and Riemann tensors of $\acute{\nabla}$ and $\grave{\nabla}$, defined as in \eqref{eq_torsion_free} and \eqref{eq:def_riemann_hat}, are linked by
\begin{align*}
    \acute{\Sigma}_\alpha{}^\gamma{}_\beta &= \grave{\Sigma}_\alpha{}^\gamma{}_\beta + 2 Q_{[\alpha}{}^\gamma{}_{\beta]} \,, \\
    \acute{R}^\alpha{}_{\beta\mu\nu} &= \grave{R}^\alpha{}_{\beta\mu\nu} + 2 \acute{\nabla}_{[\mu} Q_{\nu]}{}^\alpha{}_\beta - 2 Q_{[\mu}{}^\alpha{}_{|\gamma|}Q_{\nu]}{}^\gamma{}_\beta + \acute{\Sigma}_\mu{}^\gamma{}_\nu Q_\gamma{}^\alpha{}_\beta \,.
\end{align*}
Let us now restrict to the case of two Weyl connections.

\begin{lem}
    \label{prop_transition_weyl_conn}
	The transition tensor between two Weyl connections $\widecheck{\nabla}$ and $\widehat{\nabla}$ is given by
	\[ Q_\alpha{}^\gamma{}_\beta = S_{\alpha\beta}{}^{\gamma\mu} \left( \widecheck{\kappa}_g-\widehat{\kappa}_g\right)_\mu \]
    where $\widecheck{\kappa}_g$ and $\widehat{\kappa}_g$ are the covector fields associated to respectively $\widecheck{\nabla}$ and $\widehat{\nabla}$ with respect to any representative $g \in [\widetilde{g}]$ by~\eqref{weyl_conn} and $S_{\alpha\beta}{}^{\mu\nu}$ is the tensor defined previously by~\eqref{eq:def_S}.
\end{lem}

\begin{prop}
    \label{prop:transformation_rules}
    The Schouten, Cotton and Weyl tensors of two Weyl connections $\widecheck{\nabla}$ and $\widehat{\nabla}$ are linked by
    \begin{subequations}
        \begin{align}
            \label{transfo_schouten}
            \widecheck{L}_{\mu\nu} &= \widehat{L}_{\mu\nu} - \widecheck{\nabla}_\mu (\widecheck{\kappa}_g-\widehat{\kappa}_g)_\nu - \frac{1}{2} S_{\mu\nu}{}^{\sigma\rho} (\widecheck{\kappa}_g-\widehat{\kappa}_g)_\sigma (\widecheck{\kappa}_g-\widehat{\kappa}_g)_\rho \,, \\
            \label{transfo_cotton}
            \widecheck{C}_{\lambda\mu\nu} &= \widehat{C}_{\lambda\mu\nu} + (\widecheck{\kappa}_g-\widehat{\kappa}_g)_\rho \widecheck{W}^\rho{}_{\nu\lambda\mu} \,, \\
            \label{transfo_weyl}
            \widecheck{W}^\alpha{}_{\beta\mu\nu} &= \widehat{W}^\alpha{}_{\beta\mu\nu} \,,
        \end{align}
    \end{subequations}
    where $\widecheck{\kappa}_g$ and $\widehat{\kappa}_g$ are the covector fields associated to $\widecheck{\nabla}$ and $\widehat{\nabla}$ with respect to a representative $g \in [\widetilde{g}]$ by~\eqref{weyl_conn}.
\end{prop}

\begin{rem}
    Thus all the Weyl tensors of Weyl connections for $(\M,[\widetilde{g}])$, and in particular of the Levi-Civita connection of representatives $g \in [\widetilde{g}]$, coincide. In dimension $N = 3$, the same property holds for the Cotton tensor since the Weyl tensor vanishes identically. Hence the Cotton tensor plays the role of the Weyl tensor in dimension $N=3$.
\end{rem}

\subsubsection{Conformal invariance}
\label{sec:conformal_invariance}

\begin{defi}
    Let $\mathcal{M}$ be a smooth manifold and consider a map $X : g \mapsto X_g$ which associates a tensor field $X_g$ to any metric $g$ on $\mathcal{M}$. The map $X$ is said to be \emph{conformally invariant of weight $k$} if
    \[ \forall \, \Omega \in \mathcal{C}^\infty(\M,\R_+^\star), \quad X_{\Omega^2 g} = \Omega^{-k} X_g \]
    The conformally invariant maps of weight $0$ will simply be called conformally invariant. In that case, the map $[\widetilde{g}] \mapsto X_g$ where $g \in [\widetilde{g}]$ is a representative is well defined.
\end{defi}

\begin{prop}
    In any dimension, the tensor $S_{\alpha\beta}{}^{\mu\nu}$ and the Weyl tensor $W^\alpha{}_{\beta\mu\nu}$ are conformally invariant. In dimension $N=3$, the Cotton tensor $C_{\alpha\beta\gamma}$ is conformally invariant. In dimension $N=4$, the Bach tensor $B_{\mu\nu}$ is conformally invariant of weight $2$. 
\end{prop}

\subsubsection{Conformally Einstein and conformally flat spaces}
\label{sec:special_conf_structures}

One can distinguish special conformal structures such as conformally Einstein spaces and conformally flat spaces. 

\begin{defi}
    A pseudo-Riemannian manifold $(\M,\widetilde{g})$ is said to be \emph{locally conformally Einstein} if for all $p \in \M$, there exists an open neighbourhood $\mathcal{U}$ of $p$ in $\M$ and a local conformal factor $\Omega \in \mathcal{C}^\infty(\mathcal{U},\R_+^\star)$ such that $(\mathcal{U},\Omega^2 \widetilde{g})$ is Einstein. If the conformal factor $\Omega$ can be defined globally on $\M$ then it is said to be \emph{globally conformally Einstein}.
\end{defi}

\begin{defi}
    A pseudo-Riemannian manifold $(\M,\widetilde{g})$ is said to be \emph{locally conformally flat} if for all $p \in \M$, there exists an open neighbourhood $\mathcal{U}$ of $p$ in $\M$ and a local conformal factor $\Omega \in \mathcal{C}^\infty(\mathcal{U},\R_+^\star)$ such that $(\mathcal{U},\Omega^2 \widetilde{g})$ is Riemann-flat. If the conformal factor $\Omega$ can be defined globally on $\M$ then it is said to be \emph{globally conformally flat}.
\end{defi}

\begin{rem}
    Since the definitions depend only on the conformal class $[\widetilde{g}]$, it extends to conformal structures.
\end{rem}

\noindent Locally conformally flat spaces are characterised by the following theorem.

\begin{thm}[Weyl-Schouten theorem]
    \label{Weyl_Schouten_theorem}
    Let $(\M,\widetilde{g})$ be a pseudo-Riemannian manifold of dimension $N \geq 2$. Then
    \begin{itemize}
        \item[i)] if $N=2$, it is locally conformally flat,
        \item[ii)] if $N=3$, it is locally conformally flat if and only if the Cotton tensor vanishes identically,
        \item[iii)] if $N \geq 4$, it is locally conformally flat if and only if the Weyl tensor vanishes identically.
    \end{itemize}
\end{thm}

\begin{rems} \, 
    \begin{itemize}
        \item If $\M$ is contractible, one can pass from locally to globally conformally flat.
        \item In dimension $N=3$, an Einstein manifold is always locally conformally flat. \qedhere
    \end{itemize}
\end{rems}

\subsubsection{Conformal Killing fields}
\label{sec:conformal_Killing_fields}

Conformal Killing fields are the generalisation of Killing fields in conformal geometry. Let $(\mathcal{M},[\widetilde{g}])$ be a conformal structure.

\begin{defi}
    Let $K$ be a smooth vector field on $\mathcal{M}$. The following statements are equivalent
    \begin{itemize}
        \item[i)] $\forall \, g \in [\widetilde{g}] \,, \exists \, \phi_{g,K} \in \mathcal{C}^\infty(\mathcal{M}), \quad \mathcal{L}_K g = 2 \phi_{g,K} \, g$,
        \item[ii)] $\exists \, g \in [\widetilde{g}] \,, \exists \, \phi_{g,K} \in \mathcal{C}^\infty(\mathcal{M}), \quad \mathcal{L}_K g = 2 \phi_{g,K} \, g$,
        \item[iii)] $K$ generates a flow of confomorphisms locally on $\mathcal{M}$,
        \item[iv)] $\forall$ Weyl connection $\widehat{\nabla}, \quad Z_\alpha{}^{\beta\mu}{}_\nu \widehat{\nabla}_\mu K^\nu = 0$,
        \item[v)] $\exists$ Weyl connection $\widehat{\nabla}, \quad Z_\alpha{}^{\beta\mu}{}_\nu \widehat{\nabla}_\mu K^\nu = 0$,
    \end{itemize}
    where $\mathcal{L}$ is the Lie derivative and $Z_\alpha{}^{\beta\mu}{}_\nu$ is the projection operator on symmetric trace-free 2-tensors defined by
    \begin{equation}
        \label{def_Z}
        Z_\alpha{}^{\beta\mu}{}_\nu := \delta_\alpha{}^\mu \delta_\nu{}^\beta + g_{\alpha\nu} g^{\beta\mu} - \frac{2}{N} \delta_\nu{}^\mu \delta_\alpha{}^\beta \,,
    \end{equation}
    for any $g \in [\widetilde{g}]$ (it is conformally invariant). If one statement (and thus all) is verified, $K$ is said to be a conformal Killing vector field of $(\mathcal{M},[\widetilde{g}])$.
\end{defi}

\begin{rem}
    The Killing fields of a representative $g \in [\widetilde{g}]$ are the conformal Killing fields which are divergence-free for the Levi-Civita connection $\nabla$ of $g$.
\end{rem}

\begin{prop}
    The set of all conformal Killing fields of $(\mathcal{M},[\widetilde{g}])$ forms a Lie subalgebra of smooth vector fields on $\mathcal{M}$. The latter is infinite dimensional if $N=2$ and has dimension at most $(N+1)(N+2)/2$ if $N\geq3$.
\end{prop}

\subsubsection{Conformal geodesics}
\label{sec:conf_geo}

Let $(\mathcal{M},[\widetilde{g}])$ be a conformal structure. Only the null geodesics of representatives $g \in [\widetilde{g}]$ are invariant up to a reparametrisation under conformal rescalings. However, one can still define a type of curves which enjoys nice conformal properties and includes all geodesic curves of representatives. These are called conformal geodesics.

\begin{defi}
	Let $\widehat{\nabla}$ be a Weyl connection for $(\M,[\widetilde{g}])$. A $\widehat{\nabla}$-\emph{conformal geodesic (curve)} is a pair $(x,\widehat{\beta})$ consisting of a curve $x : \tau \in I \subset \R \mapsto x(\tau) \in \mathcal{M}$, whose tangent vector is denoted by $\dot{x}$, and a covector field $\widehat{\beta}$ along $x$ satisfying the following equations
	\begin{subequations}
		\begin{align}
			\label{eq_conf_geo_x}
			\left(\widehat{\nabla}_{\dot{x}} \dot{x}\right)^\mu & = - S_{\lambda\sigma}{}^{\mu\nu} \dot{x}^\lambda \dot{x}^\sigma \widehat{\beta}_\nu \,, \\
			\label{eq_conf_geo_b}
			\left(\widehat{\nabla}_{\dot{x}} \widehat{\beta}\right)_\lambda & = \frac{1}{2} S_{\lambda\sigma}{}^{\mu\nu} \dot{x}^\sigma \widehat{\beta}_\mu \widehat{\beta}_\nu + \widehat{L}_{\sigma\lambda} \dot{x}^\sigma \,.
		\end{align}
	\end{subequations}
\end{defi}

\begin{rems} \,
	\begin{itemize}
		\item The initial data $(x_\star, \dot{x}_\star, \widehat{\beta}_\star)$ of a $\widehat{\nabla}$-conformal geodesic consists of an initial position, an initial tangent vector and an initial covector
		\[ x_\star \in \mathcal{M} \,, \qquad \dot{x}_\star \in T_{x_\star} \mathcal{M} \,, \qquad \widehat{\beta}_\star \in T^\star_{x_\star} \mathcal{M} \,.\]
		\item Along any $\widehat{\nabla}$-conformal geodesic $(x,\widehat{\beta})$ and for any representative $g \in [\widetilde{g}]$,
        \begin{equation}
    		\label{derivee_pseudonorme}
           \dot{x} \cdot g(\dot{x},\dot{x}) = - 2\langle \widehat{\kappa}_g + \widehat{\beta},\dot{x}\rangle g(\dot{x},\dot{x})
		\end{equation}
		where $\widehat{\kappa}_g$ is the covector field associated to $\widehat{\nabla}$ with respect to $g$ by~\eqref{weyl_conn}. In particular, $g(\dot{x},\dot{x})$ has constant sign along a conformal geodesic. This allows to classify conformal geodesics into timelike, null or spacelike categories. \qedhere
	\end{itemize}
\end{rems}

\noindent The following proposition links conformal geodesics for different Weyl connections.

\begin{prop}
    \label{prop:confgeo_changweylconn}
	Let $\widehat{\nabla}$ and $\widecheck{\nabla}$ be two Weyl connections for $(\M,[\widetilde{g}])$. Then $(x,\widehat{\beta})$ is a $\widehat{\nabla}$-conformal geodesic if and only if $(x,\widecheck{\beta} := \widehat{\beta} + \widehat{\kappa}_g-\widecheck{\kappa}_g)$ is a $\widecheck{\nabla}$-conformal geodesic,  where $\widehat{\kappa}_g$ and $\widecheck{\kappa}_g$ are the covector fields associated to $\widehat{\nabla}$ and $\widecheck{\nabla}$ with respect to a representative $g \in [\widetilde{g}]$ by~\eqref{weyl_conn}.
\end{prop}

\begin{rem}
    Note that the curve $x$ (and its parametrisation) is left unchanged.
\end{rem}

\subsubsection{Congruences of timelike conformal geodesics}
\label{sec:congruence_conf_geo}

Congruences of timelike conformal geodesics are crucial in the gauge construction presented in \Cref{sec:gauge_construction}. Let $(\mathcal{M},[\widetilde{g}])$ be a Lorentzian conformal structure of dimension $N= n+1$ and $\widehat{\nabla}$ be a Weyl connection. 

\begin{defi}
    A \emph{congruence of timelike $\widehat{\nabla}$-conformal geodesics} on an open set $\mathcal{U} \subset \mathcal{M}$ is a smooth family of timelike $\widehat{\nabla}$-conformal geodesics $(x,\widehat{\beta})$ such that at any point $p \in \mathcal{U}$ there is one and only one conformal geodesic passing through $p$.
\end{defi}

\begin{rem}
    A congruence of geodesics is usually defined as integral curves of a non-vanishing vector field satisfying the geodesic equation. The vector field is not unique and the parametrisation of the curves is not imposed. Here, we choose to fix the vector field, the Weyl connection and the parametrisation in the definition.
\end{rem}

\begin{ex}
    Consider the Einstein cylinder $(\M_{EC},g_{EC})$ where
    \begin{equation}
        \label{eq:EC}
        \M_{EC} := \R \times \mathbb{S}^3 \quad \text{and} \quad g_{EC} := -dt^2 + \dsphere{3}^2 \,.
    \end{equation}
    The Schouten tensor of the associated Levi-Civita connection $\nabla$ is given by
    \[ L = \frac{1}{2} \left( dt^2 + \dsphere{n}^2 \right) = \frac{1}{2} g_{EC} + dt^2 \,. \]
    Take the spacelike hypersurface
    \[ \mathcal{S}_\star := \left\{ (t,\psi,\theta,\varphi) \in \R \times (0,\pi)^2 \times (0,2\pi) \; \middle| \; t = \frac{2}{3} \sin\psi \right\} \,. \]
    The $\nabla$-conformal geodesics starting on $\mathcal{S}_\star$, with initial tangent vector $\dot{x}_\star$ the unit normal vector to $\mathcal{S}_\star$ and initial covector $\beta_\star = 0$, form locally a congruence of timelike conformal geodesics. Note that each curve remain on a great circle.
    
    In \Cref{fig:congruence_conformal_geodesics}, these conformal geodesics are represented in dotted or dashed lines while $\mathcal{S}_\star$ is depicted with a solid line. Here, $\psi$ goes from $-\pi$ to $\pi$ to fully cover a great circle: $(\psi,\theta,\varphi)$ with a negative $\psi$ shall be understood as $(-\psi,\pi-\theta,\varphi+\pi \mod 2\pi)$. Eventually, caustics form at $\psi = -\pi/2$ and at least some conformal geodesics present a finite time blow-up. These two behaviours result in the end of the congruence.
\end{ex}

\begin{figure}
    \centering
    \includegraphics[width=0.75\linewidth]{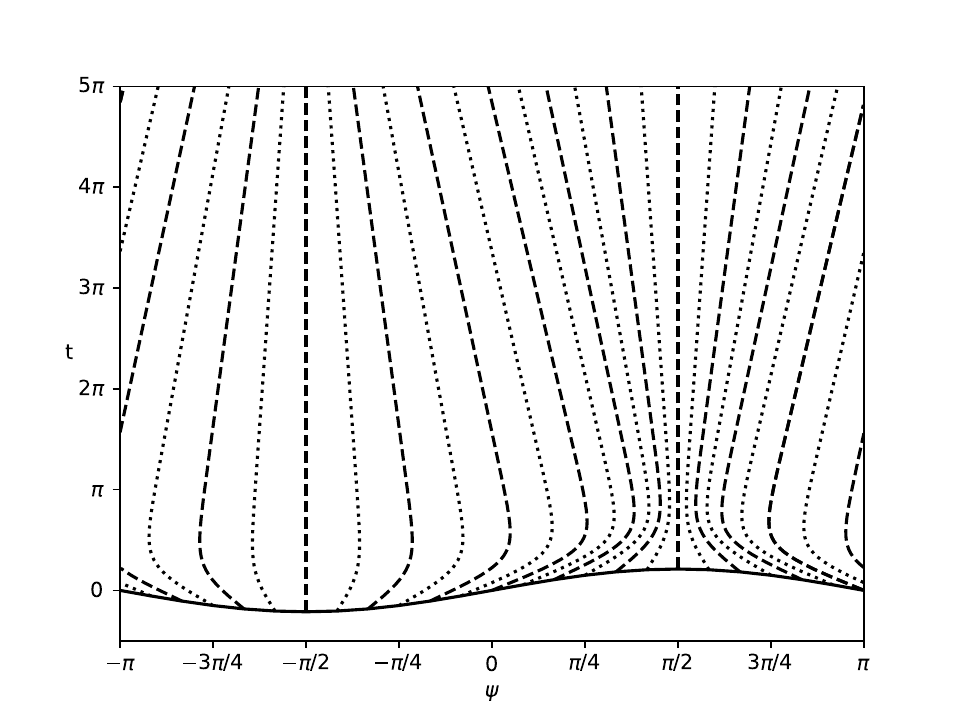}
    \caption{A congruence of timelike conformal geodesics in the Einstein cylinder}
    \label{fig:congruence_conformal_geodesics}
\end{figure}

\begin{defi}
    A \emph{conformal Gaussian coordinate system} on an open set $\mathcal{U} \subset \mathcal{M}$ is a coordinate system $(\tau,y^1,\dots,y^n)$ on $\mathcal{U}$ such that there exists a congruence of timelike conformal geodesics on $\mathcal{U}$ satisfying
    \begin{itemize}
        \item $\tau$ is the parameter function of the conformal geodesics in the congruence,
        \item the functions $(y^i)_{1 \leq i \leq n}$ are constant along the conformal geodesics.
    \end{itemize}
\end{defi}

\begin{rem}
    The Weyl connection is not specified here because the two properties above are independent of it thanks to \Cref{prop:confgeo_changweylconn}.
\end{rem}

\noindent In the next proposition, we recall that a congruence of timelike $\widehat{\nabla}$-conformal geodesics on an open set $\mathcal{U}$ singles out a canonical Weyl connection $\widecheck{\nabla}$ on $\mathcal{U}$ and a canonical representative $\breve{g}$ of the conformal class $[\widetilde{g}]$ on $\mathcal{U}$.

\begin{prop}
    \label{prop_canonical_congruence}
    Consider a congruence of timelike $\widehat{\nabla}$-conformal geodesics $(x,\widehat{\beta})$ on an open set $\mathcal{U}$. Then there exists a unique Weyl connection $\widecheck{\nabla}$ on $\mathcal{U}$ and a unique representative $\breve{g}$ of the conformal class $[\widetilde{g}]$ on $\mathcal{U}$ such that
    \begin{subequations}
        \begin{align}
            \widecheck{\nabla}_{\dot{x}} \dot{x} &= 0 \,, \\
            \widecheck{L}(\dot{x},.) &= 0 \,, \\
            \label{normalisation}
            \breve{g}(\dot{x},\dot{x}) &= -1 \,,
        \end{align}
    \end{subequations}
\end{prop}

\begin{rem}
    This means that one can find a Weyl connection $\widecheck{\nabla}$ for which $\widecheck{\beta} = 0$ and a conformally related metric $\breve{g}$ for which the parameter $\tau$ of the curves $x$ corresponds to the $\breve{g}$-proper time.
\end{rem}

\begin{proof}
    Fix a representative $g \in [\widetilde{g}]$. Since the conformal geodesics in the congruence are timelike, $g(\dot{x},\dot{x})$ is a smooth negative function on $\mathcal{U}$. Define
    \begin{equation}
        \label{Omega_gbreve}
        \Omega := (-g(\dot{x},\dot{x}))^{-1/2} \in \mathcal{C}^\infty(\mathcal{U},\R_+^\star) \,.
    \end{equation}
    Then the metric $\breve{g} := \Omega^2 \, g$ on $\mathcal{U}$ does not depend on the choice of $g$ and $\breve{g}(\dot{x},\dot{x}) = -1$. This defines uniquely the metric $\breve{g}$.
    
    Let $\widehat{\kappa}_{\breve{g}}$ be the smooth covector field on $\mathcal{U}$ associated to $\widehat{\nabla}$ with respect to $\breve{g}$ by~\eqref{weyl_conn}. Take the Weyl connection $\widecheck{\nabla}$ associated to the smooth covector field $\widecheck{\kappa}_{\breve{g}} := \widehat{\kappa}_{\breve{g}} + \widehat{\beta}$ on $\mathcal{U}$ with respect to $\breve{g}$ given by \Cref{prop:exist&uniq_weylconn}. It follows from \Cref{prop:confgeo_changweylconn} that the curves $(x,0)$ are $\widecheck{\nabla}$-conformal geodesics and this is the only choice of Weyl connection which gives $\widecheck{\beta} = 0$.
\end{proof}

\noindent When the conformal structure on $\mathcal{U}$ is globally conformally Einstein, one can determine a simple expression for the conformal factor between the metric $\breve{g}$ given by \Cref{prop_canonical_congruence} and the Einstein metric. This property, used in Friedrich's original proof, will be obtained through an alternative derivation directly from the extended conformal vacuum Einstein equations in the gauge construction, see \Cref{prop:construction_gauge} and in particular \eqref{gauge_Theta}.

\begin{prop}[Friedrich \protect{\cite[Lemma 3.1]{F95}}]
    \label{prop_conf_factor}
    Consider a congruence of timelike $\widehat{\nabla}$-conformal geodesics $(x,\widehat{\beta})$ on an open set $\mathcal{U}$. Introduce the hypersurface $\mathcal{S}_\star := \{x \in \mathcal{U} \mid \tau(x) = \tau_\star \}$ for some fixed parameter value $\tau_\star \in \R$. Assume that there exists a representative $\widetilde{g}$ in the conformal class such that $(\mathcal{U},\widetilde{g})$ is solution to the \eqref{eq:VE}. Then the metric $\breve{g}$ on $\mathcal{U}$ defined through~\eqref{normalisation} is conformally related to $\widetilde{g}$ on $\mathcal{U}$ by a conformal factor $\Omega$ (that is $\breve{g} = \Omega^2 \widetilde{g}$ on $\mathcal{U}$) expressed in a Gaussian coordinate system $(\tau,y)$ associated to the congruence as follows
    \begin{equation}
        \Omega(\tau,y) = \Omega_\star(y) \left( 1 + v_\star(y) (\tau-\tau_\star) - f_\star(y) \frac{(\tau-\tau_\star)^2}{2} \right) \,,
    \end{equation}
    where $\Omega_\star$, $q_\star$ and $f_\star$ are the following smooth functions on $\mathcal{S}_\star$
    \begin{subequations}
        \label{expr_coeff}
        \begin{align}
            \Omega_\star &:= (-\widetilde{g}(\dot{x}_\star,\dot{x}_\star))^{-1/2} \,,\\
            v_\star &:= \langle \widetilde{\beta}_\star,\dot{x}_\star\rangle \,,\\
            f_\star &:= \frac{1}{2\Omega_\star^2} \left(\widetilde{g}^{-1}(\widetilde{\beta}_\star,\widetilde{\beta}_\star) +\frac{2\Lambda}{n(n-1)}\right)  \,,
        \end{align}
    \end{subequations}
     with $\widetilde{\beta} := \widehat{\beta} + \widehat{\kappa}_{\widetilde{g}}$.
\end{prop}

\begin{rem}
    A more detailed analysis can be made with the hypotheses of \Cref{prop_conf_factor} such as in the case of a zero cosmological constant presented in \cite{F03}. Indeed, the following auxiliary variables
    \[ u := \widetilde{g}(\dot{x},\dot{x}) \,, \qquad v := \langle \widetilde{\beta}, \dot{x} \rangle \,, \qquad w := \frac{1}{2} \left(\widetilde{g}^{-1}(\widetilde{\beta},\widetilde{\beta}) +\frac{2\Lambda}{n(n-1)}\right) \,, \]
    satisfy a system of ordinary differential equations given by
    \begin{equation}
        \label{auxiliary_system_conf_geod}
        \dot{u} = -2vu \,, \qquad \dot{v} = -v^2 + uw \,, \qquad \dot{w} = vw \,.
    \end{equation}
    After solving for these auxiliary variables, the $\widetilde{\nabla}$-conformal geodesic equations reduce to a linear system of ordinary differential equations. Thus finite time blow-ups are entirely controlled by the auxiliary system \eqref{auxiliary_system_conf_geod} in this case. Yet integrating \eqref{auxiliary_system_conf_geod} gives
    \[ u(\tau) = \frac{u_\star}{P(\tau)^2} \,, \qquad v(\tau) = \frac{P'(\tau)}{P(\tau)} \,, \qquad w(\tau) = w_\star P(\tau) \,, \]
    where
    \[ P(\tau) := 1+v_\star(\tau-\tau_\star)+\frac{u_\star w_\star}{2} (\tau-\tau_\star)^2 \,. \]
    One concludes that a finite time blow-up occurs if and only if
    \begin{itemize}
        \item $u_\star w_\star = 0$ and $v_\star \neq 0$,
        \item $u_\star w_\star \neq 0$ and $v_\star^2 \geq 2 u_\star w_\star$. \qedhere
    \end{itemize}
\end{rem}

\begin{proof}
    For completeness, we include Friedrich's proof presented in \cite[Section 3.3]{F95}. Let $\widehat{\kappa}_{\widetilde{g}}$ be the covector associated to $\widehat{\nabla}$ with respect to $\widetilde{g}$. By \Cref{prop:confgeo_changweylconn}, $(x,\widetilde{\beta} := \widehat{\beta}+\widehat{\kappa}_{\widetilde{g}})$ is solution to the $\widetilde{\nabla}$-conformal geodesic equations where $\widetilde{\nabla}$ is the Levi-Civita connection of $\widetilde{g}$. With~\eqref{derivee_pseudonorme} and~\eqref{Omega_gbreve}, one finds
    \[ \dot{\Omega} := \dot{x} \cdot \Omega = \widetilde{\nabla}_{\dot{x}} \Omega = \langle \widetilde{\beta},\dot{x}\rangle \Omega\,. \]
    By further differentiating with respect to $\tau$ and using the $\widetilde{\nabla}$-conformal geodesic equations, one deduces that
    \begin{align*}
        \ddot{\Omega} &= \left( \frac{1}{2} \widetilde{g}(\dot{x},\dot{x})\widetilde{g}^{-1}(\widetilde{\beta},\widetilde{\beta}) + \widetilde{L}(\dot{x},\dot{x}) \right) \Omega \,,\\
        \dddot{\Omega} &= \left( (\widetilde{\nabla}_{\dot{x}} \widetilde{L})(\dot{x},\dot{x}) + 3 \left( \widetilde{L}(\dot{x},\widetilde{\beta}^\sharp) \widetilde{g}(\dot{x},\dot{x}) - \langle \widetilde{\beta},\dot{x} \rangle \widetilde{L}(\dot{x},\dot{x})\right) \right) \Omega \,,
    \end{align*}
    where $\widetilde{\beta}^\sharp$ is the vector field obtained by raising the index with the metric $\widetilde{g}$. Since $(\mathcal{U},\widetilde{g})$ is solution to the (VE), one has $\dddot{\Omega} = 0$. It follows that $\Omega$ is a second order polynomial in $\tau$. Let $(\dot{x}_\star,\widetilde{\beta}_\star)$ be the values of $(\dot{x},\widetilde{\beta})$ on $\mathcal{S}_\star$. By evaluating the previous equations and \eqref{Omega_gbreve} on $\mathcal{S}_\star$, one gets the expressions~\eqref{expr_coeff}.
\end{proof}

\subsection{Candidates and dual tensors}

In this section, we introduce tensors sharing the same algebraic properties than the Weyl or the Cotton tensor, which we call Weyl or Cotton candidates. The former are sometimes referred to as Weyl fields \cite{CK94} or spin-2 fields \cite{CK90}. These candidates can verify covariant differential equations similar to the Bianchi identities. In small dimensions, such equations simplify greatly by mean of duality. 

These notions will be exploited multiple times in \Cref{sec:local_existence}. In particular, let us highlight \Cref{def:cotton_york} of the Cotton-York tensor which plays a key role with regard to the boundary conditions in \Cref{sec:geometric_bc}.

\subsubsection{Weyl and Cotton candidates}

Let $(\M,g)$ be a smooth pseudo-Riemannian manifold.

\begin{defi}
    \label{def_Weyl_candidate}
     A \emph{Weyl candidate} on $(\M,g)$ is a tensor field $Q^\alpha{}_{\beta\mu\nu}$ enjoying the same algebraic properties than a Weyl tensor, that is
    \[ Q_{(\alpha\beta)\mu\nu} = 0 \,, \qquad Q^\alpha{}_{\beta(\mu\nu)} = 0 \,, \qquad Q^\alpha{}_{[\beta\mu\nu]} = 0 \,, \qquad Q^\alpha{}_{\beta\alpha\nu} = 0 \,. \]
\end{defi}

\begin{defi}
    \label{def:cotton_candidate}
    A \emph{Cotton candidate} on $(\M,g)$ is a tensor field $J_{\lambda\mu\nu}$ satisfying the same algebraic properties than a Cotton tensor, that is
    \[ J_{(\lambda\mu)\nu} = 0 \,, \qquad J_{[\lambda\mu\nu]} = 0 \,, \qquad J_{\lambda\mu}{}^{\mu} = 0 \,. \]
\end{defi}

\subsubsection{Hodge dual}

Let $(\M,g)$ be a smooth oriented pseudo-Riemannian manifold of signature\footnote{With the convention that $(N,0)$ is the Riemannian signature and $(N-1,1)$ is the Lorentzian.} $(p,q)$ and dimension $N=p+q$. Its volume form is the $N$-form defined locally by 
\[ \epsilon := \sqrt{|\det g|} \; dy^1 \wedge \dots \wedge dy^N \,, \]
for any local coordinate system $(y^1,\dots,y^N)$ and where $\det g$ is the determinant of the metric in this coordinate system. It is clear that $\epsilon_{\alpha_1,\dots,\alpha_N}$ is conformally invariant of weight $-N$. Moreover, if $\widehat{\nabla}$ is a Weyl connection for the conformal structure $(\M,[g])$ then
\[ \widehat{\nabla}_\alpha \epsilon_{\beta_1 \dots \beta_N} = - N (\widehat{\kappa}_g)_\alpha \epsilon_{\beta_1\dots\beta_N} \,, \]
where $\widehat{\kappa}_g$ is the covector field defined by \eqref{weyl_conn}.

Recall that the Hodge dual of a $k$-form $\omega$ is the $(N-k)$-form $\star\omega$ defined by
\[ (\star\omega)_{\alpha_1\dots\alpha_{N-k}} := \frac{1}{k!} \epsilon_{\alpha_1\dots\alpha_{N-k}}{}^{\beta_1\dots\beta_k} \omega_{\beta_1\dots\beta_k} \,. \]
From the identity, see \cite[Section 2.5.3]{K16} and references therein,
\begin{equation}
    \label{eq:identity_volume_form}
    \epsilon^{\alpha_1\dots \alpha_k \, \beta_1\dots \beta_{N-k}} \epsilon_{\alpha_1\dots \alpha_k \, \gamma_1\dots \gamma_{N-k}} = (-1)^q \, k! \, (N-k)! \; \delta_{[\gamma_1} {}^{\beta_1} \dots \, \delta_{\gamma_{N-k}]} {}^{\beta_{N-k}} \,,
\end{equation}
one deduces that for all $k$-form $\omega$,
\[ \star (\star\omega) = (-1)^{q+k(N-k)} \omega \,.\]

\subsubsection{Duality in dimension 4}
\label{sec:duality_dim_4}

Let $(\M,g)$ be an oriented smooth Lorentzian manifold of dimension $4$. The right and left dual tensors of a Weyl candidate $Q^\alpha{}_{\beta\mu\nu}$ are defined by
\begin{align*}
    (Q\star)_{\alpha\beta\lambda\xi} := \frac{1}{2} \epsilon_{\lambda\xi}{}^{\mu\nu} Q_{\alpha\beta\mu\nu} \,, \\
    (\star Q)_{\alpha\beta\lambda\xi} := \frac{1}{2} \epsilon_{\alpha\beta}{}^{\mu\nu} Q_{\mu\nu\lambda\xi} \,.
\end{align*}
It is classical that the left and right duals coincide and that they are Weyl candidates as well, see for instance \cite[Section 4]{CK90}. The left dual tensor of a Cotton candidate $J_{\alpha\beta\gamma}$ is defined by
\[ (\star J)_{\alpha\beta\gamma} := \frac{1}{2} \epsilon_{\alpha\beta}{}^{\mu\nu} J_{\mu\nu\gamma} \,. \]
Then $(\star J)_{\alpha\beta\gamma}$ is also a Cotton candidate. Indeed,
\begin{align*}
    (\star J)_{\alpha\beta}{}^\beta &= \frac{1}{2} \epsilon_{\alpha\beta}{}^{\mu\nu} J_{\mu\nu}{}^\beta = \frac{1}{2} \epsilon_\alpha{}^{\beta\mu\nu} J_{[\beta\mu\nu]} = 0 \,, \\
    (\star J)_{[\rho\sigma\tau]} &= \frac{1}{3} \epsilon^\alpha{}_{\rho\sigma\tau} J_{\alpha\beta}{}^\beta = 0 \,.
\end{align*}
Moreover, for a second Bianchi identity type equation between a Weyl candidate $Q_{\alpha\beta\mu\nu}$ and a Cotton candidate $J_{\alpha\beta\gamma}$ with respect to a Weyl connection $\widehat{\nabla}$ for $(\M,[g])$, one has
\begin{align*}
    \widehat{\nabla}_\alpha Q^\alpha{}_{\beta\mu\nu} = J_{\mu\nu\beta} &\iff \widehat{\nabla}_\alpha (Q\star)^\alpha{}_{\beta\mu\nu} = (\star J)_{\mu\nu\beta} \\
    &\iff \widehat{\nabla}_{[\mu} Q^\alpha{}_{|\beta|\nu\xi]} = S_{\beta[\mu}{}^{\alpha\gamma} J_{\nu\xi]\gamma} \\
    &\iff \widehat{\nabla}_{[\mu} (Q\star)^\alpha{}_{|\beta|\nu\xi]} = S_{\beta[\mu}{}^{\alpha\gamma} (\star J)_{\nu\xi]\gamma} \,.
\end{align*}
More precisely,
\begin{equation}
    \label{div_Q_div_Q_star}
    \Big( \widehat{\nabla}_\alpha (Q\star)^\alpha{}_{\beta\lambda\xi} - (\star J)_{\lambda\xi\beta} \Big) = \frac{1}{2} \epsilon_{\lambda\xi}{}^{\mu\nu} \Big( \widehat{\nabla}_\alpha Q^\alpha{}_{\beta\mu\nu} - J_{\mu\nu\beta} \Big) \,,
\end{equation}
and
\begin{equation}
    \label{eq:aux_duality_4}
    - \frac{1}{2} \epsilon_\beta{}^{\alpha\lambda\xi} g_{\mu\rho} \Big( \widehat{\nabla}_{[\alpha} Q^\rho{}_{|\nu|\lambda\xi]} - S_{\nu[\alpha}{}^{\rho\gamma} J_{\lambda\xi]\gamma} \Big) = \Big( \widehat{\nabla}_\alpha (Q\star)^\alpha{}_{\beta\mu\nu} - (\star J)_{\mu\nu\beta} \Big)  \,. 
\end{equation}
The last equation rewrites under the form
\begin{equation}
    \label{curl_Q_div_Q_star}
    3 \Big( \widehat{\nabla}_{[\alpha} Q^\rho{}_{|\nu|\lambda\xi]} - S_{\nu[\alpha}{}^{\rho\gamma} J_{\lambda\xi]\gamma} \Big) = \epsilon^\beta{}_{\alpha\lambda\xi} g^{\rho\mu} \Big( \widehat{\nabla}_\sigma (Q\star)^\sigma{}_{\beta\mu\nu} - (\star J)_{\mu\nu\beta} \Big) \,.
\end{equation}

\begin{rem}
    Let $\omega_\alpha$ be a smooth covector field on $\M$. Then $J_{\mu\nu\beta} := \omega_\alpha Q^\alpha{}_{\beta\mu\nu}$ is a Cotton candidate and
    \begin{align*}
        (\star J)_{\mu\nu\beta} &= \omega_\alpha (Q\star)^\alpha{}_{\beta\mu\nu} = \omega_\alpha (\star Q)^\alpha{}_{\beta\mu\nu} = \frac{1}{2} \epsilon^\alpha{}_{\beta\sigma\tau} \omega_\alpha Q^{\sigma\tau}{}_{\mu\nu} \\
        &= -\frac{1}{2} \epsilon_\beta{}^{\alpha\sigma\tau} \omega_{[\alpha} Q_{\sigma\tau]\mu\nu} = - \frac{1}{2} \epsilon_\beta{}^{\alpha\lambda\xi} g_{\mu\rho} \left(Q^\rho{}_{\nu[\alpha\lambda} \omega_{\xi]}\right) \,.
    \end{align*} 
    Thus, using \eqref{eq:aux_duality_4} with $Q^\alpha{}_{\beta\mu\nu} = 0$, one has
    \begin{equation}
        \label{eq:SxJ}
        S_{\nu[\alpha}{}^{\rho\gamma} J_{\lambda\xi]\gamma} = Q^\rho{}_{\nu[\alpha\lambda} \omega_{\xi]} \,. \qedhere
    \end{equation}
\end{rem}

\subsubsection{Duality in dimension 3 and the Cotton-York tensor}
\label{sec:duality_dim_3}

Let $(\mathcal{S},h)$ be an oriented smooth Lorentzian manifold of dimension $3$, $D$ be its Levi-Civita connection and $J_{ijk}$ be a Cotton candidate. Define the left dual tensor of $J_{ijk}$ by
\[ (\star J)_{ij} := \frac{1}{2} \epsilon_i{}^{kl} J_{klj} \,. \]
Then $(\star J)_{ij}$ is symmetric and trace-free. Indeed,
\begin{align*}
    2 (\star J)_{[ij]} &= 2 \delta_{[i}{}^k \delta_{j]}{}^l (\star J)_{kl} = - \epsilon^m{}_{ij} \epsilon_m{}^{kl} (\star J)_{kl} = -\epsilon^m{}_{ij} J_{ml}{}^l = 0 \,, \\
    2 (\star J)^i{}_i &= \epsilon^{ijk} J_{ijk} = \epsilon^{ijk} J_{[ijk]} = 0 \,.
\end{align*}
Moreover, for any Weyl connection $\widehat{D}$ for $(\mathcal{S},[h])$,
\begin{align*}
    \widehat{D}_i (\star J)^i{}_j - 3 (\widehat{\kappa}_h)_i (\star J)^i{}_j &= \frac{1}{2} \epsilon^{ikl} \widehat{D}_{[i} J_{kl]j} \,, \\
    \widehat{D}_i (\star J)^i{}_j - 3 (\widehat{\kappa}_h)_i (\star J)^i{}_j &= \widehat{D}_i (\star J)_j{}^i - 3 (\widehat{\kappa}_h)_i (\star J)_j{}^i = \frac{1}{2} \epsilon_j{}^{kl} \left( \widehat{D}_i J_{kl}{}^i - 2(\widehat{\kappa}_h)_i J_{kl}{}^i \right) \,.
\end{align*}
In particular, for a fourth Bianchi identity type equation on the Cotton candidate $J_{ijk}$ with respect to a Weyl connection $\widehat{D}$,
\[ \widehat{D}_{[i} J_{jk]l} = 0 \iff \widehat{D}_i (\star J)^i{}_j = 3(\widehat{\kappa}_h)_i (\star J)^i{}_j \iff  h^{ij} \widehat{D}_i J_{klj} = 0 \,. \]

\begin{defi}
    \label{def:cotton_york}
    The \emph{Cotton-York tensor} $Y_{ij}$ of a $3$-dimensional oriented smooth Lorentzian manifold $(\mathcal{S},h)$ is the left dual tensor of the Cotton tensor $C_{ijk}$ of the Levi-Civita connection $D$.
\end{defi}

\begin{prop}
    \label{prop:cottonyork}
    The Cotton-York tensor of a $3$-dimensional oriented smooth Lorentzian manifold $(\mathcal{S},h)$ is symmetric, trace-free with respect to the metric $h$ and divergence-free for the Levi-Civita connection $D$ of $h$. It is furthermore conformally invariant of weight $1$.
\end{prop}

\begin{proof}
    From the above general discussion, it is clear that the Cotton-York is symmetric, trace-free and divergence-free. Since $C_{ijk}$ is conformally invariant of weight $0$ and $\epsilon_i{}^{jk}$ is conformally invariant of weight $1$, the Cotton-York tensor is also conformally invariant of weight $1$.
\end{proof}

\subsubsection{Duality on a distribution of dimension 2}
\label{sec:duality_dim_2}

To interpret geometrically certain analytic boundary conditions, it is convenient in \Cref{sec:tensorial_bc_distribution} to exploit duality on a 2-dimensional distribution arising from the gauge constructed in \Cref{prop:construction_gauge}.

To that end, consider an oriented $4$-dimensional smooth Lorentzian manifold $(\M,g)$ and let $(e_{\bf a})$ be a smooth $g$-orthonormal frame field on $\M$ with $e_{\bf 0}$ timelike and $\epsilon_{\bf 0123} = +1$ where $\epsilon$ is the volume form of $g$. Define
\begin{itemize}
    \item $\mathcal{D}^2 := \vectorspan(e_{\bf 1},e_{\bf 2})$ the distribution of dimension 2 spanned by $e_{\bf 1}$ and $e_{\bf 2}$, equipped with the induced scalar product $\delta_{\bf AB}$,
    \item $\slashed{\epsilon}_{\bf AB} := \epsilon_{\bf 0AB3}$,
    \item $\mathscr{S}(\mathcal{D}^2)$ be the space of trace-free symmetric smooth 2-tensor $q_{\bf AB}$ on $\mathcal{D}^2$,
    \item $\mathscr{B}(\mathcal{D}^2)$ be the space of smooth 4-tensors $L_{\bf AB}{}^{\bf CD}$ on $\mathcal{D}^2$ verifying
    \[ L_{\bf AB}{}^{\bf CD} = L_{\bf (AB)}{}^{\bf (CD)} \,, \qquad L_{\bf A}{}^{\bf ACD} = 0 \,, \qquad L_{\bf ABC}{}^{\bf C} = 0 \,. \]
\end{itemize}
Note that $\mathscr{B}(\mathcal{D}^2)$ can be seen as the space of endomorphisms of $\mathscr{S}(\mathcal{D}^2)$.

\begin{defi}
    The dual of any $q_{\bf AB} \in \mathscr{S}(\mathcal{D}^2)$ is the element of $\mathscr{S}(\mathcal{D}^2)$ defined by
    \begin{equation}
        (\star q)_{\bf AB} := \slashed{\epsilon}\vphantom{q}_{\bf (A}{}^{\bf C} q_{\bf B)C} \,.
    \end{equation}
    This defines an operator $\star \in \mathscr{B}(\mathcal{D}^2)$.
\end{defi}

\begin{lem}
    \label{lem:identities_dual_2}
    For all $q_{\bf AB}$, $(q')_{\bf AB} \in \mathscr{S}(\mathcal{D}^2)$,
    \begin{subequations}
        \begin{align*}
        (\star q)_{\bf AB} (q')^{\bf AB} &= -q_{\bf AB} (\star q')^{\bf AB} \,, \\
        (\star(\star q))_{\bf AB} &= -q_{\bf AB} \,.
        \end{align*}
    \end{subequations}
    In other words, $\star$ is anti self-adjoint and $\star\star = - \Id$. In particular,
    \begin{align*}
        (\star q)_{\bf AB} (\star q)^{\bf AB} &= q_{\bf AB} q^{\bf AB}  \,, \\
        (\star q)_{\bf AB} q^{\bf AB} &= 0 \,.
    \end{align*}
\end{lem}

\begin{proof}
    The identities come from \eqref{eq:identity_volume_form} and simple calculus.
\end{proof}

\begin{lem}
    \label{lem:decomp_duality_dim_2}
    For any $q_{\bf AB} \in \mathscr{S}(\mathcal{D}^2)$ and any $(u,v) \in \mathcal{C}^\infty(\M, \R^2\setminus\{(0,0)\})$, there exists a unique $(q')_{\bf AB} \in \mathscr{S}(\mathcal{D}^2)$ such that
    \[ q_{\bf AB} = u \, (q')_{\bf AB} + v \, (\star q')_{\bf AB} \,. \]
    In other words, for all $(u,v) \in \mathcal{C}^\infty(\M, \R^2\setminus\{(0,0)\})$, $u\Id + v\star \in \mathscr{B}(\mathcal{D}^2)$ is an automorphism of $\mathscr{S}(\mathcal{D}^2)$.
\end{lem}

\begin{proof}
     Let $q_{\bf AB} \in \mathscr{S}(\mathcal{D}^2)$ and $(u,v) \in \R^2\setminus\{(0,0)\}$. If $(q')_{\bf AB}$ exists then one deduces, by taking a linear combination of the decomposition and its dual, that it is given by
    \[ (q')_{\bf AB} = \frac{u}{u^2+v^2} \, q_{\bf AB} - \frac{v}{u^2+v^2} \, (\star q)_{\bf AB} \,. \]
    Furthermore, one can check that this choice of $(q')_{\bf AB}$ yields the decomposition.
\end{proof}

\begin{cor}
    \label{cor:duality_dim_2}
    Let $q_{\bf AB} \in \mathscr{S}(\mathcal{D}^2)$. The following statements are equivalent
    \begin{itemize}
        \item[i)] $q_{\bf AB} = 0$,
        \item[ii)] $\forall \, (u,v) \in \mathcal{C}^\infty(\M, \R^2)$, $u \, q_{\bf AB} + v \, (\star q)_{\bf AB} = 0$,
        \item[iii)] $\exists \, (u,v) \in \mathcal{C}^\infty(\M, \R^2\setminus\{(0,0)\})$, $u \, q_{\bf AB} + v \, (\star q)_{\bf AB} = 0$.
    \end{itemize}
\end{cor}

\begin{proof}
    The implications $i)\implies ii)$ and $ii)\implies iii)$ are trivial. The implication $iii) \implies i)$ follows from \Cref{lem:decomp_duality_dim_2}.
\end{proof}

\noindent Let us describe more precisely the spaces $\mathscr{S}(\mathcal{D}^2)$ and $\mathscr{B}(\mathcal{D}^2)$. Define $\omega \in \mathscr{S}(\mathcal{D}^2)$ by
\[ \omega_{\bf 11} = - \omega_{\bf 22} = 0 \,, \qquad \omega_{\bf 12} = \omega_{\bf 21} = \frac{1}{\sqrt{2}} \,. \]
Since $(e_{\bf 1},e_{\bf 2})$ is an orthonormal basis of $\mathcal{D}^2$, one deduces that $(\omega,\star\omega)$ is an orthonormal basis of $\mathscr{S}(\mathcal{D}^2)$. One has
\[ \Mat_{(e_{\bf 1},e_{\bf 2})} (\omega) = \frac{1}{\sqrt{2}} \begin{pmatrix}
    0 & 1 \\
    1 & 0
\end{pmatrix} \,, \qquad \Mat_{(e_{\bf 1},e_{\bf 2})} (\star\omega) = \frac{1}{\sqrt{2}} \begin{pmatrix}
    1 & 0 \\
    0 & -1
\end{pmatrix} \,. \]
Now let $\sigma \in \mathscr{B}(\mathcal{D}^2)$ be the orthogonal reflection symmetry with respect to the hyperplane $\vectorspan(\omega)$. Then
\begin{alignat*}{3}
    \Mat_{(\omega,\star\omega)}(\Id) &= \begin{pmatrix}
    1 & 0 \\
    0 & 1
    \end{pmatrix} \,, &\qquad
    \Mat_{(\omega,\star\omega)}(\star) &= \begin{pmatrix}
    0 & -1 \\
    1 & 0
    \end{pmatrix} \,, \\
    \Mat_{(\omega,\star\omega)}(\sigma) &= \begin{pmatrix}
    1 & 0 \\
    0 & -1
    \end{pmatrix} \,, &\qquad
    \Mat_{(\omega,\star\omega)}(\star\sigma) &= \begin{pmatrix}
    0 & 1 \\
    1 & 0
    \end{pmatrix} \,.
\end{alignat*}
and $(\Id/\sqrt{2},\star/\sqrt{2},\sigma/\sqrt{2},(\star\sigma)/\sqrt{2})$ is an orthonormal basis of $\mathscr{B}(\mathcal{D}^2)$. Moreover, one has $\{\star,\sigma\} := \star\sigma+\sigma\star = 0$. The following lemma is a generalisation of \Cref{lem:decomp_duality_dim_2}.

\begin{lem}
    \label{lem:invertibility_endo}
    An endomorphism $L = \alpha \Id + \beta \star + \gamma \sigma + \delta (\star\sigma) \in \mathscr{B}(\mathcal{D}^2)$ is everywhere invertible if and only if $\alpha^2+\beta^2-\delta^2-\gamma^2$ nowhere vanishes. If $\alpha^2+\beta^2=\gamma^2+\delta^2$ and $(\alpha,\beta,\gamma,\delta)\neq(0,0,0,0)$ at a point $p \in \mathcal{M}$ then $\ker L = \vectorspan\left((\alpha-\gamma)\omega-(\beta+\delta)(\star\omega)\right)$ and $\im L = \vectorspan\left((\alpha+\gamma)\omega+(\beta+\delta)(\star\omega)\right)$ at $p$.
\end{lem}

\begin{proof}
    If $\alpha^2+\beta^2-\delta^2-\gamma^2\neq 0$ at a point $p \in \M$, then the inverse of $L$ at $p$ is given by
    \[ L^{-1} = \frac{\alpha \Id - \beta \star - \gamma \sigma - \delta (\star\sigma)}{\alpha^2+\beta^2-\delta^2-\gamma^2} \,. \]
    If not and if $L\neq0$ at $p$, the expressions of the kernel and the image can be easily deduced using the matrix representation in the basis $(\omega,\star\omega)$.
\end{proof}

\subsection{Non-null hypersurfaces of Lorentzian conformal structures}

We detail here structure inheritance from a Lorentzian conformal structure on a non-null hypersurface. In particular, the induced conformal structure and induced Weyl connections in \Cref{sec:induced_conformal_structure} and extrinsic curvatures in \Cref{sec:extrinsic_curvature}. The links between curvature tensors are given by the Codazzi equations derived for Weyl connections in \Cref{sec:codazzi_weyl}. All decompositions, including the decomposition of Weyl candidates for all dimensions established in \Cref{sec:decomposition_weyl_candidate}, use frame fields adapted to the hypersurface as defined in \Cref{sec:adapted_frame_fields}.

Let $(\mathcal{M},[\widetilde{g}])$ be a smooth Lorentzian conformal structure of dimension $N=n+1$. Consider a non-null hypersurface $\mathcal{S} \subset \mathcal{M}$. The tangential part of a vector field or covector field with respect to $\mathcal{S}$ is denoted by $\Tan$.

\subsubsection{Adapted frame fields}
\label{sec:adapted_frame_fields}

\begin{defi}
    An adapted frame field to $\mathcal{S}$ with respect to a representative $\widetilde{g} \in [\widetilde{g}]$ is a frame field $(e_{\bf a}) = (e_\perp,(e_{\bf i}))$ on $\mathcal{M}$ such that the restriction of $e_\perp$ on $\mathcal{S}$ is a normal unit vector field to $\mathcal{S}$ with respect to $\widetilde{g}$ and the restriction of the vector fields $(e_{\bf i})$ on $\mathcal{S}$ are tangent to $\mathcal{S}$. The dual coframe field is denoted by $(\omega^{\bf a}) = (\omega^\perp,(\omega^{\bf i}))$. 
\end{defi}

\begin{rems} \, 
    \begin{itemize}
        \item The family $(e_{\bf i})$ is not necessarily orthogonal.
        \item Let $\varepsilon := \widetilde{g}(e_\perp,e_\perp) \in \{\pm1\}$ be the pseudo-norm of $e_\perp$. The sign depends on whether $\mathcal{S}$ is timelike or spacelike. One has $\omega^\perp = \varepsilon \widetilde{g}(e_\perp,.)$.
        \item If $g = \Omega^2 \widetilde{g} \in [\widetilde{g}]$ is another representative then $(\Omega^{-1} e_\perp,(e_{\bf i}))$ is an adapted frame with respect to $g$. \qedhere
    \end{itemize}
\end{rems}

\noindent In what follows, the tensorial indices $\perp$ and ${\bf i}, {\bf j}, {\bf k} \dots$ will refer to the decomposition in an adapted frame field to $\mathcal{S}$ with respect to a representative $\widetilde{g} \in [\widetilde{g}]$ and its dual coframe field. For example,
\begin{align*}
    \widetilde{g}_{\perp\perp} &:= \widetilde{g}(e_\perp,e_\perp) = \varepsilon \,, \\
    \widetilde{g}_{\perp \bf i} &:= \widetilde{g}(e_\perp, e_{\bf i}) = 0 \,. 
\end{align*}

\subsubsection{The induced conformal structure and Weyl connection}
\label{sec:induced_conformal_structure}

Let $\widetilde{g}$, $g$ be two representatives of $[\widetilde{g}]$ and $\widetilde{h}$, $h$ be the corresponding induced metrics on $\mathcal{S}$. Since $g = \Omega^2 \widetilde{g}$ for some $\Omega \in \mathcal{C}^\infty(\M,\R_+^\star)$, one has $h = \Psi^2 \widetilde{h}$ with $\Psi := \Omega|_\mathcal{S} \in \mathcal{C}^\infty(\mathcal{S},\R_+^\star)$. Consequently, the conformal structure $(\mathcal{M},[\widetilde{g}])$ induces the conformal structure $(\mathcal{S},[\widetilde{h}])$ on $\mathcal{S}$.

Recall that a Weyl connection $\widehat{\nabla}$ on $(\M,[\widetilde{g}])$ induces a connection $\widehat{D}$ on $\mathcal{S}$ through
\[ \widehat{D}_X Y := \Tan \widehat{\nabla}_V W \,, \]
for all smooth vector fields $X,Y$ on $\mathcal{S}$ and where $V,W$ are smooth vector fields on $\mathcal{M}$ such that $V|_\mathcal{S} = X$ and $W|_\mathcal{S} = Y$.

\begin{lem}
    \label{lem:induced_weyl_connection}
    Let $\widehat{\nabla}$ be a Weyl connection for $(\M,[\widetilde{g}])$ and $\widehat{D}$ be its induced connection on $\mathcal{S}$. Then $\widehat{D}$ is a Weyl connection for the induced conformal structure $(\mathcal{S},[\widetilde{h}])$. More precisely, if $\widehat{\kappa}_{\widetilde{g}}$ is the covector field associated to $\widehat{\nabla}$ with respect to a representative $\widetilde{g} \in [\widetilde{g}]$ then $\widehat{D}$ is the Weyl connection associated to $\Tan \widehat{\kappa}_{\widetilde{g}}$ with respect to $\widetilde{h}$, the induced metric of $\widetilde{g}$ on $\mathcal{S}$.
\end{lem}

\begin{proof}
    Since $\widehat{\nabla}$ is torsion-free, so is $\widehat{D}$. Let $\widetilde{g} \in [\widetilde{g}]$ be a representative, $\widetilde{h}$ be the induced metric on $\mathcal{S}$ by $\widetilde{g}$ and $\widehat{\kappa}_{\widetilde{g}}$ be the covector field associated to $\widehat{\nabla}$ with respect to $\widetilde{g}$. In an adapted frame field $(e_\perp,(e_{\bf i}))$ to $\mathcal{S}$ with respect to $\widetilde{g}$, one has
    \[ \widehat{D}_{\bf i} \widetilde{h}_{\bf jk} = \widehat{D}_{\bf i} \widetilde{g}_{\bf jk} = \widehat{\nabla}_{\bf i} \widetilde{g}_{\bf jk} + \widehat{\Gamma}_{\bf i}{}^{\bf \perp} {}_{\bf j} \widetilde{g}_{\bf \perp k} + \widehat{\Gamma}_{\bf i}{}^{\bf \perp} {}_{\bf k} \widetilde{g}_{\bf j\perp} = \widehat{\nabla}_{\bf i} \widetilde{g}_{\bf jk} = - 2(\widehat{\kappa}_{\widetilde{g}})_{\bf i} \widetilde{g}_{\bf jk} = - 2(\widehat{\kappa}_{\widetilde{g}})_{\bf i} \widetilde{h}_{\bf jk} \,. \qedhere \]
\end{proof}

\subsubsection{Extrinsic curvature of a Weyl connection}
\label{sec:extrinsic_curvature}

\begin{defi}
    \label{def:extrinsic_curvature}
    The \emph{extrinsic curvature} of a Weyl connection $\widehat{\nabla}$ on $\mathcal{S}$ with respect to a representative $\widetilde{g} \in [\widetilde{g}]$ is defined by
    \[ \widehat{K}_{\widetilde{g}}(X,Y) := \langle \omega^\perp, \widehat{\nabla}_V W \rangle \]
    for all smooth vector fields $X,Y$ on $\mathcal{S}$ and where $V,W$ are smooth vector fields on $\mathcal{M}$ such that $V|_\mathcal{S} = X$, $W|_\mathcal{S} = Y$, $\omega^\perp$ is a unit normal covector field to $\mathcal{S}$ with respect to $\widetilde{g}$. The extrinsic curvature is a symmetric 2-tensor field on $\mathcal{S}$.
\end{defi}

\begin{rems} \,
    \begin{itemize}
        \item If $g=\Omega^2\widetilde{g} \in [\widetilde{g}]$ is another representative then $(\widehat{K}_g)_{ij} = \Omega (\widehat{K}_{\widetilde{g}})_{ij}$.
        \item If $\widecheck{\nabla}$ is another Weyl connection then
        \begin{equation}
            \label{eq:extrinsic_curvature_change_Weyl_conn}
            (\widecheck{K}_{\widetilde{g}})_{ij} = (\widehat{K}_{\widetilde{g}})_{ij} - \varepsilon(\widecheck{\kappa}_{\widetilde{g}}-\widehat{\kappa}_{\widetilde{g}})_\perp \, \widetilde{h}_{ij} \,,
        \end{equation}
        where $\widetilde{h}$ is the induced metric by $\widetilde{g}$ on $\mathcal{S}$ and $\widecheck{\kappa}_{\widetilde{g}}$, $\widehat{\kappa}_{\widetilde{g}}$ are the covector fields associated to $\widecheck{\nabla}$ and $\widehat{\nabla}$ with respect to ${\widetilde{g}}$.
        \item By combining the two precedent remarks, if $\widetilde{g}$ and $g = \Omega^2 \widetilde{g}$ are two representatives then the extrinsic curvatures\footnote{In this particular instance, the metric is omitted in the notation of the extrinsic curvature for simplicity.} $K_{ij}$, $\widetilde{K}_{ij}$ of their Levi-Civita connections $\nabla$, $\widetilde{\nabla}$ with respect to respectively $g$, $\widetilde{g}$ are linked by 
        \begin{equation}
            \label{eq:aux_ec}
            K_{ij} = \Omega \widetilde{K}_{ij} - \varepsilon 
            \langle d\ln\Omega, e_\perp\rangle \, h_{ij} \,,
        \end{equation}
        where $e_\perp$ is unitary for the metric $g$. The respective traces $K$ and $\widetilde{K}$ are thus linked by
        \begin{equation}
            K = \Omega^{-1} \widetilde{K} - \varepsilon n \langle d\ln\Omega, e_\perp\rangle \,.
        \end{equation}
        \item The indices of $\widehat{K}_{\widetilde{g}}$ are raised or lowered with respect to $\widetilde{g}$. For instance,
        \[ (\widehat{K}_{\widetilde{g}})_i{}^j := \widetilde{g}^{jk} (\widehat{K}_{\widetilde{g}})_{ik} \,. \qedhere \]
    \end{itemize}
\end{rems}

\begin{lem}
    \label{lem:connection_coeff}
    In a frame field $(e_\perp,(e_{\bf i}))$ adapted to $\mathcal{S}$ with respect to a representative $\widetilde{g} \in [\widetilde{g}]$, one has on $\mathcal{S}$
    \begin{align*}
        \widehat{\Gamma}_{\bf i}{}^\perp{}_{\bf j} &= (\widehat{K}_{\widetilde{g}})_{\bf ij} \,, \\
        \widehat{\Gamma}_{\bf i}{}^{\bf j}{}_\perp &= - \varepsilon (\widehat{K}_{\widetilde{g}})_{\bf i}{}^{\bf j} \,, \\
        \widehat{\Gamma}_{\bf i}{}^\perp{}_\perp &= (\widehat{\kappa}_{\widetilde{g}})_{\bf i} \,,
    \end{align*}
    where $\widehat{\kappa}_{\widetilde{g}}$ is the covector field associated to the Weyl connection $\widehat{\nabla}$ with respect to $\widetilde{g}$ and $\varepsilon := \widetilde{g}(e_\perp,e_\perp) \in \{\pm1\}$.
\end{lem} 

\begin{proof}
    The first equation follows straightforwardly from the definition of the extrinsic curvature. Moreover,
    \[ \widetilde{g}_{\bf jk} \widehat{\Gamma}_{\bf i}{}^{\bf k}{}_\perp = \widetilde{g}(e_{\bf k},\widehat{\nabla}_{e_{\bf i}} e_\perp) =  e_{\bf i}(\widetilde{g}_{\bf k\perp}) + 2 (\widehat{\kappa}_{\widetilde{g}})_{\bf i} \widetilde{g}_{\bf k\perp} -\widetilde{g}(\widehat{\nabla}_{e_{\bf i}} e_{\bf k},e_\perp) = -\varepsilon \langle \omega^\perp,\widehat{\nabla}_{e_{\bf i}} e_{\bf k}) = -\varepsilon (\widehat{K}_{\widetilde{g}})_{\bf ik} \,, \]
    and
    \[ \widehat{\Gamma}_{\bf i}{}^\perp{}_\perp = \langle \omega^\perp, \widehat{\nabla}_{e_{\bf i}} e_\perp \rangle = \varepsilon \widetilde{g}(e_\perp,\widehat{\nabla}_{e_{\bf i}} e_\perp) = \varepsilon \left(\frac{e_{\bf i}(\widetilde{g}_{\perp\perp})}{2}  + (\widehat{\kappa}_{\widetilde{g}})_{\bf i} \widetilde{g}_{\perp\perp} \right) = (\widehat{\kappa}_{\widetilde{g}})_{\bf i} \,. \qedhere \]
\end{proof}

\begin{defi}
    \label{def:tllygeod_umbilical}
    The hypersurface $\mathcal{S}$ is said to
    \begin{itemize}
        \item be \emph{totally geodesic for a Weyl connection} $\widehat{\nabla}$ if the extrinsic curvature $\widehat{K}_g$ with respect to any (and thus all) representative $g \in [\widetilde{g}]$ vanishes identically,
        \item be \emph{umbilical} if the extrinsic curvature $\widehat{K}_g$ of any (and thus all) Weyl connection $\widehat{\nabla}$ with respect to any (and thus all) representative $g \in [\widetilde{g}]$ is pure trace,
        \item have \emph{zero mean curvature for a Weyl connection} $\widehat{\nabla}$ if the extrinsic curvature $\widehat{K}_g$ with respect to any (and thus all) representative $g \in [\widetilde{g}]$ is trace-free.
    \end{itemize}
\end{defi}

\begin{lem} \,
    \label{lem:link_umbilical_totally_geodesic}
    \begin{itemize}
        \item[i)] There exists a Weyl connection $\widehat{\nabla}$ such that $\mathcal{S}$ has zero mean curvature for $\widehat{\nabla}$.
        \item[ii)] $\mathcal{S}$ is umbilical if and only if there exists a Weyl connection $\widehat{\nabla}$ such that $\mathcal{S}$ is totally geodesic for $\widehat{\nabla}$.
    \end{itemize}
\end{lem}

\begin{proof}
    $i)$ Let $\widetilde{g} \in [\widetilde{g}]$ be any representative and $\widehat{\nabla}$ be any Weyl connection. Take a smooth covector field $\kappa$ on $\M$ such that $\kappa_\perp = (\widehat{\kappa}_{\widetilde{g}})_\perp + \widetilde{h}^{ij} (\widehat{K}_{\widetilde{g}})_{ij}/3$ on $\mathcal{S}$. Then the Weyl connection associated to $\kappa$ with respect to $\widetilde{g}$ by \Cref{prop:exist&uniq_weylconn} works.
    
    \bigbreak
    
    $ii)$ By $i)$, there exists a Weyl connection $\widehat{\nabla}$ such that $\mathcal{S}$ has zero mean curvature for $\widehat{\nabla}$. Since $\mathcal{S}$ is also umbilical, one deduces that it is totally geodesic for $\widehat{\nabla}$.
\end{proof}

\begin{lem}
    \label{lem_tot_geod_curve}
    Let $x : I \to \mathcal{S}$ be a smooth curve on $\mathcal{S}$ and $\widehat{\nabla}$ be a Weyl connection on $(\mathcal{M},[\widetilde{g}])$. Denote by $X$ the smooth vector field obtained by parallel transport along the curve $x$ of a vector $X_\star \in T_{x(\tau_\star)} \mathcal{M}$ for some $\tau_\star \in I$, that is
    \[ \widehat{\nabla}_{\dot{x}} X = 0 \,, \qquad X|_{\tau_\star} = X_\star \,. \]
    If $\mathcal{S}$ is totally geodesic for $\widehat{\nabla}$ and $X_\star$ is tangent to $\mathcal{S}$ then $X$ remains tangent to $\mathcal{S}$.
\end{lem}

\begin{proof}
    Take a representative $g \in [\widetilde{g}]$ and a frame field $(e_\perp,(e_{\bf i}))$ adapted to $\mathcal{S}$ with respect to $g$. Then
    \[ 0 = (\widehat{\nabla}_{\dot{x}} X)^\perp = \dot{x}(X^\perp) + \langle\widehat{\kappa}_g,\dot{x}\rangle X^\perp + (\widehat{K}_g)_{\bf ij} \dot{x}^{\bf i} X^{\bf j} \,. \]
    Thus, assuming that $\mathcal{S}$ is totally geodesic for $\widehat{\nabla}$, $X^\perp$ vanishes identically along $x$ if and only if $X_\star^\perp = 0$.
\end{proof}

\subsubsection{Decomposition of a Weyl candidate}
\label{sec:decomposition_weyl_candidate}

The following lemma gives a decomposition of Weyl candidates (see \Cref{def_Weyl_candidate}) at a non-null hypersurface into algebraically independent tensor fields. In the case of the Weyl tensor, this decomposition is closely related to the Codazzi equations as we will see in \Cref{cor:codazzi}. 

\begin{lem}
    \label{lem:decomposition_Weyl_cand}
    Let $g \in [\widetilde{g}]$ be a representative and $Q^\alpha{}_{\beta\mu\rho}$ be a Weyl candidate on $(\mathcal{M},g)$. Define the following tensor fields on $\mathcal{M}$, using an adapted frame field $(e_\perp,(e_{\bf i}))$  to $\mathcal{S}$ with respect to $g$ whose dual frame field is denoted by $(\omega^\perp,(\omega^{\bf i}))$,
    \begin{itemize}
        \item $E := Q^\perp{}_{\bf i\perp j} \, \omega^{\bf i} \otimes \omega^{\bf j}$ which is a symmetric trace-free 2-tensor field,
        \item $M := Q^\perp{}_{\bf kij} \, \omega^{\bf i} \otimes \omega^{\bf j} \otimes \omega^{\bf k}$ which is a Cotton candidate (see \Cref{def:cotton_candidate}),
        \item $P := \left(Q^{\bf i}{}_{\bf jkl} + \frac{2}{n-2} S_{\bf j[k}{}^{\bf im} E_{\bf l]m}\right) e_{\bf i} \otimes \omega^{\bf j} \otimes \omega^{\bf k} \otimes \omega^{\bf l}$ which is a Weyl candidate. 
    \end{itemize}    
    Then the following decomposition holds
    \begin{equation}
        \label{dec_Weyl_candidate}
        Q_{\alpha\beta\mu\rho} = P_{\alpha\beta\mu\rho} +2 \left( (\omega^\perp)_{[\alpha} M_{|\mu\rho|\beta]} + (\omega^\perp)_{[\mu} M_{|\alpha\beta|\rho]} \right) - \frac{2}{n-2} \left( X_{\alpha[\mu} E_{\rho]\beta} - X_{\beta[\mu} E_{\rho]\alpha} \right) \,, 
    \end{equation}
    where $X_{\alpha\beta} := g_{\alpha\beta} - (n-1) \varepsilon (\omega^\perp)_\alpha (\omega^\perp)_\beta$.
\end{lem}

\begin{rems} \,
    \begin{itemize}
        \item The restriction of $E$, $M$ and $P$ on $\mathcal{S}$ are tangent to $\mathcal{S}$ and do not depend on the choice of adapted frame field $(e_\perp,(e_{\bf i}))$. In addition, the restriction of $E$ and $P$ are conformally invariant, whereas the one of $M$ is conformally invariant of weight $1$.
        \item Consider the special case $n=3$. One has $P=0$ since Weyl candidates vanish identically in dimension 3. To recover the usual Bel decomposition, assume that $\mathcal{M}$ is oriented and denote the volume forms on $\mathcal{M}$ and $\mathcal{S}$ by $\epsilon$. Recall that the dual tensor field $(Q\star)^\alpha{}_{\beta\mu\nu}$ is also a Weyl candidate, see \Cref{sec:duality_dim_4}. Then one defines the symmetric and trace-free 2-tensor
        \[ H := (Q\star)^\perp{}_{\bf i\perp j} \, \omega^{\bf i} \otimes \omega^{\bf j} \,. \]
        The tensor fields $H$ and $M$ are duals of each other on $\mathcal{S}$
        \begin{equation}
            H_{\bf ij} = \frac{1}{2} \epsilon_{\bf i}{}^{\bf kl} M_{\bf klj} \, \quad \text{or equivalently} \quad M_{\bf ijk} = -\varepsilon \epsilon_{\bf ij}{}^{\bf l} H_{\bf lk} \,.
        \end{equation}
        By analogy with the decomposition of the electromagnetic tensor field, $E$ and $H$ are called the electric part and the magnetic part of the Weyl candidate $Q$ on $\mathcal{S}$. Furthermore, since
        \[ ((Q\star)\star)^\alpha{}_{\beta\lambda\mu} = - Q^\alpha{}_{\beta\lambda\mu} \,, \]
        the electric and magnetic parts of $Q$ and $Q\star$ are linked by
        \begin{equation}
            \label{sym_dual_EH}
            E(Q\star) = H(Q) \,, \qquad  H(Q\star) = - E(Q)\,. \qedhere
        \end{equation}
    \end{itemize}
\end{rems}

\subsubsection{The Codazzi equations for Weyl connections}
\label{sec:codazzi_weyl}

We detail here the Codazzi equations for a Weyl connection $\widehat{\nabla}$. In comparison with the Levi-Civita case, care must be taken since the connections do not preserve the metric and the Riemann tensors have fewer symmetries.

Let $\widehat{D}$ be the induced Weyl connection on $\mathcal{S}$ by a Weyl connection $\widehat{\nabla}$. The Riemann and Ricci tensors of $\widehat{\nabla}$ and $\widehat{D}$ are linked by the Codazzi equations stated below.

\begin{lem}
    The Gauss-Codazzi and Codazzi-Mainardi equations write as follows, in an frame field $(e_{\bf a}) = (e_\perp,(e_{\bf i}))$ adapted to $\mathcal{S}$ with respect to a representative $g \in [\widetilde{g}]$, 
    \begin{subequations}
        \label{codazzi}
        \begin{align}
            \label{gauss_codazzi}
            \widehat{R}^{\bf i}{}_{\bf jkl} & = \widehat{r}^{\bf i}{}_{\bf jkl} - 2\varepsilon (\widehat{K}_g)^{\bf i}{}_{\bf [k} (\widehat{K}_g)_{\bf l]j} \,, \\
            \label{codazzi_mainardi_1}
            \widehat{R}^\perp{}_{\bf jkl} & = 2 \widehat{D}_{\bf [k} (\widehat{K}_g)_{\bf l]j} + 2 (\widehat{\kappa}_g)_{\bf [k} (\widehat{K}_g)_{\bf l]j} \,, \\
            \label{codazzi_mainardi_2}
            -\varepsilon \widehat{R}^{\bf i}{}_{\bf \perp kl} &= 2 \widehat{D}_{\bf [k} (\widehat{K}_g)_{\bf l]}{}^{\bf i} - 2 (\widehat{\kappa}_g)_{\bf [k} (\widehat{K}_g)_{\bf l]}{}^{\bf i} \,,
        \end{align}
    \end{subequations}
    where $\widehat{R}^{\bf a}{}_{\bf bcd}$, $\widehat{r}^{\bf i}{}_{\bf jkl}$ are the frame components of the Riemann tensor of $\widehat{\nabla}$ and $\widehat{D}$ respectively, $\widehat{K}_g$ is the extrinsic curvature of $\widehat{\nabla}$ with respect to $g$, see \Cref{sec:extrinsic_curvature}.
    Their contractions give
    \begin{subequations}
        \label{contracted_codazzi}
        \begin{align}
            \widehat{R}_{\bf jl} - \widehat{r}_{\bf jl} &= \widehat{R}^\perp{}_{\bf j\perp l} - \varepsilon \left( (\widehat{K}_g)^{\bf i}{}_{\bf i} (\widehat{K}_g)_{\bf jl} - (\widehat{K}_g)^{\bf i}{}_{\bf l} (\widehat{K}_g)_{\bf ij}\right) \,, \\
            - \widehat{R}^\perp{}_{\bf l} = -\varepsilon \widehat{R}_{\bf \perp l} & = \widehat{D}_{\bf i} (\widehat{K}_g)_{\bf l}{}^{\bf i} - \widehat{D}_{\bf l} (\widehat{K}_g)_{\bf i}{}^{\bf i} - (\widehat{\kappa}_g)_{\bf i}(\widehat{K}_g)_{\bf l}{}^{\bf i}  + (\widehat{\kappa}_g)_{\bf l} (\widehat{K}_g)_{\bf i}{}^{\bf i} \,.
        \end{align}
    \end{subequations}
\end{lem}

\begin{proof}
By definition of the Riemann tensor,
    \begin{align*}
        \widehat{R}^{\bf i}{}_{\bf jkl} &:= \langle \omega^{\bf i}, \widehat{\nabla}_{e_{\bf k}} \widehat{\nabla}_{e_{\bf l}} e_{\bf j} - \widehat{\nabla}_{e_{\bf l}} \widehat{\nabla}_{e_{\bf k}} e_{\bf j} - \widehat{\nabla}_{[e_{\bf k},e_{\bf l}]} e_{\bf j} \rangle \,, \\
        \widehat{R}^\perp{}_{\bf jkl} &:= \langle \omega^\perp, \widehat{\nabla}_{e_{\bf k}} \widehat{\nabla}_{e_{\bf l}} e_{\bf j} - \widehat{\nabla}_{e_{\bf l}} \widehat{\nabla}_{e_{\bf k}} e_{\bf j} - \widehat{\nabla}_{[e_{\bf k},e_{\bf l}]} e_{\bf j} \rangle \,, \\
        \widehat{R}^{\bf i}{}_{\bf \perp kl} &:= \langle \omega^{\bf i}, \widehat{\nabla}_{e_{\bf k}} \widehat{\nabla}_{e_{\bf l}} e_\perp - \widehat{\nabla}_{e_{\bf l}} \widehat{\nabla}_{e_{\bf k}} e_\perp - \widehat{\nabla}_{[e_{\bf k},e_{\bf l}]} e_\perp \rangle \,. 
    \end{align*}
    Yet
    \begin{align*}
        \widehat{\nabla}_{e_{\bf k}} \widehat{\nabla}_{e_{\bf l}} e_{\bf j} &= \widehat{\nabla}_{e_{\bf k}} \left( \widehat{D}_{e_{\bf l}} e_{\bf j} + \widehat{\Gamma}_{\bf l}{}^\perp{}_{\bf j} e_\perp \right) \\
        &= \widehat{D}_{e_{\bf k}} \widehat{D}_{e_{\bf l}} e_{\bf j} + \widehat{\Gamma}_{\bf l}{}^{\bf m}{}_{\bf j} \widehat{\Gamma}_{\bf k}{}^\perp{}_{\bf m}  e_\perp + e_{\bf k}(\widehat{\Gamma}_{\bf l}{}^\perp{}_{\bf j}) e_\perp + \widehat{\Gamma}_{\bf l}{}^\perp{}_{\bf j} \widehat{\Gamma}_{\bf k}{}^{\bf a}{}_\perp e_{\bf a} \\
        &= \widehat{D}_{e_{\bf k}} \widehat{D}_{e_{\bf l}} e_{\bf j} + \widehat{\Gamma}_{\bf l}{}^{\bf m}{}_{\bf j} (\widehat{K}_g)_{\bf km} e_\perp + e_{\bf k}((\widehat{K}_g)_{\bf lj}) e_\perp  -\varepsilon (\widehat{K}_g)_{\bf lj} (\widehat{K}_g)_{\bf k}{}^{\bf m}e_{\bf m} + (\widehat{\kappa}_g)_{\bf k} (\widehat{K}_g)_{\bf lj} e_\perp \\
        &= \widehat{D}_{e_{\bf k}} \widehat{D}_{e_{\bf l}} e_{\bf j} -\varepsilon (\widehat{K}_g)_{\bf lj} (\widehat{K}_g)_{\bf k}{}^{\bf m}e_{\bf m} + \left( e_{\bf k}((\widehat{K}_g)_{\bf lj}) + \widehat{\Gamma}_{\bf l}{}^{\bf m}{}_{\bf j} (\widehat{K}_g)_{\bf km} + (\widehat{\kappa}_g)_{\bf k} (\widehat{K}_g)_{\bf lj} \right) e_\perp \,,
    \end{align*}
    and
    \begin{align*}
        \widehat{\nabla}_{[e_{\bf k},e_{\bf l}]} e_{\bf j} &= \widehat{D}_{[e_{\bf k},e_{\bf l}]} e_{\bf j} + \langle \omega^{\bf m}, [e_{\bf k},e_{\bf l}]\rangle \widehat{\Gamma}_{\bf m}{}^\perp{}_{\bf j} e_\perp \\
        &= \widehat{D}_{[e_{\bf k},e_{\bf l}]} e_{\bf j} + \left(\widehat{\Gamma}_{\bf k}{}^{\bf m}{}_{\bf l} - \widehat{\Gamma}_{\bf l}{}^{\bf m}{}_{\bf k}\right) \widehat{\Gamma}_{\bf m}{}^\perp{}_{\bf j} e_\perp \\
        &= \widehat{D}_{[e_{\bf k},e_{\bf l}]} e_{\bf j} + \left(\widehat{\Gamma}_{\bf k}{}^{\bf m}{}_{\bf l} - \widehat{\Gamma}_{\bf l}{}^{\bf m}{}_{\bf k}\right) (\widehat{K}_g)_{\bf mj} e_\perp \,.
    \end{align*}
    Equations~\eqref{gauss_codazzi} and \eqref{codazzi_mainardi_1} follow. Moreover,
    \begin{align*}
        \langle \omega^{\bf i}, \widehat{\nabla}_{e_{\bf k}} \widehat{\nabla}_{e_{\bf l}} e_\perp \rangle &= \langle \omega^{\bf i}, \widehat{\nabla}_{e_{\bf k}} \left( \widehat{\Gamma}_{\bf l}{}^{\bf j}{}_\perp e_{\bf j} + \widehat{\Gamma}_{\bf l}{}^\perp{}_\perp e_\perp \right) \rangle \\
        &= \langle \omega^{\bf i}, e_{\bf k}(\widehat{\Gamma}_{\bf l}{}^{\bf j}{}_\perp) e_{\bf j} + \widehat{\Gamma}_{\bf l}{}^{\bf j}{}_\perp \widehat{\Gamma}_{\bf k}{}^{\bf a}{}_{\bf j} e_{\bf a} + e_{\bf k}(\widehat{\Gamma}_{\bf l}{}^\perp{}_\perp) e_\perp + \widehat{\Gamma}_{\bf l}{}^\perp{}_\perp \widehat{\Gamma}_{\bf k}{}^{\bf a}{}_\perp e_{\bf a} \rangle \\
        &= e_{\bf k}(\widehat{\Gamma}_{\bf l}{}^{\bf i}{}_\perp) + \widehat{\Gamma}_{\bf l}{}^{\bf j}{}_\perp \widehat{\Gamma}_{\bf k}{}^{\bf i}{}_{\bf j} + \widehat{\Gamma}_{\bf l}{}^\perp{}_\perp \widehat{\Gamma}_{\bf k}{}^{\bf i}{}_\perp \\
        &= -\varepsilon \left( e_{\bf k}( (\widehat{K}_g)_{\bf l}{}^{\bf i}) + \widehat{\Gamma}_{\bf k}{}^{\bf i}{}_{\bf j} (\widehat{K}_g)_{\bf l}{}^{\bf j} + (\widehat{\kappa}_g)_{\bf l} (\widehat{K}_g)_{\bf k}{}^{\bf i} \right) \,,
    \end{align*}
    and
    \[ \langle \omega^{\bf i}, \widehat{\nabla}_{[e_{\bf k},e_{\bf l}]} e_\perp \rangle = \langle \omega^{\bf j}, [e_{\bf k},e_{\bf l}]\rangle \langle \omega^{\bf i}, \widehat{\nabla}_{e_{\bf j}} e_\perp \rangle = -\varepsilon \left( \widehat{\Gamma}_{\bf k}{}^{\bf j}{}_{\bf l} - \widehat{\Gamma}_{\bf l}{}^{\bf j}{}_{\bf k} \right) (\widehat{K}_g)_{\bf j}{}^{\bf i} \,. \]
    Hence equation~\eqref{codazzi_mainardi_2}.
\end{proof}

\noindent One can rewrite the Codazzi equations in terms of the Weyl and Schouten tensors instead of the Riemann and Ricci tensors.

\begin{cor}
    \label{cor:codazzi}
    Let $g \in [\widetilde{g}]$ be a representative, $(e_\perp,(e_{\bf i}))$ be a frame field adapted to $\mathcal{S}$ with respect to $g$ and $(E,M,P)$ be the decomposition defined in~\Cref{lem:decomposition_Weyl_cand} of the Weyl tensor $\widehat{W}^\alpha{}_{\beta\mu\nu}$ of $\widehat{\nabla}$. Then
    \begin{subequations}
        \label{codazzi_weyl_schouten}
        \begin{align}
            \label{eq:codazzi_E}
            \widehat{L}_{\bf lj} - \widehat{l}_{\bf lj} &= \frac{E_{\bf lj}}{n-2} - \varepsilon F_{\bf lj}(\widehat{K}_g) \,, \\
            \label{eq:codazzi_P}
            \widehat{w}^{\bf i}{}_{\bf jkl} &= P^{\bf i}{}_{\bf jkl} + \varepsilon \left( 2(\widehat{K}_g)^{\bf i}{}_{\bf [k} (\widehat{K}_g)_{\bf l]j} - 2 S_{\bf j[k}{}^{\bf im} F_{\bf l]m}(\widehat{K}_g) \right) \,, \\
            \label{eq:codazzi_schouten}
            \widehat{L}_{\bf k\perp} &= \frac{\varepsilon}{n-1} h^{\bf jl} G_{\bf klj}(\widehat{K}_g,\widehat{\kappa}_g) \,, \\
            \label{eq:codazzi_M}
            M_{\bf klj} &= G_{\bf klj}(\widehat{K}_g,\widehat{\kappa}_g) + \frac{2}{n-1} h_{\bf j[k} h^{\bf ip} G_{\bf l]ip}(\widehat{K}_g,\widehat{\kappa}_g) \,,
        \end{align}
    \end{subequations}
    where $\widehat{L}_{\bf ab}$, $\widehat{l}_{\bf ij}$ are the frame components of the Schouten tensor of $\widehat{\nabla}$ and $\widehat{D}$ respectively, $\widehat{K}_g$ is the extrinsic curvature of $\widehat{\nabla}$ with respect to $g$ and
    \begin{subequations}
    \label{eq:def_F_G}
    \begin{align}
        F_{\bf lj}(\widehat{K}_g) &:= \frac{(\widehat{K}_g)^{\bf i}{}_{\bf i} (\widehat{K}_g)_{\bf lj} - (\widehat{K}_g)^{\bf i}{}_{\bf l} (\widehat{K}_g)_{\bf ij}}{n-2} - \frac{\left((\widehat{K}_g)^{\bf i}{}_{\bf i}\right)^2 - (\widehat{K}_g)_{\bf ik} (\widehat{K}_g)^{\bf ik}}{2(n-1)(n-2)} h_{\bf lj} \,, \\
        G_{\bf klj}(\widehat{K}_g,\widehat{\kappa}_g) &:= 2 \widehat{D}_{\bf [k} (\widehat{K}_g)_{\bf l]j} + 2 (\widehat{\kappa}_g)_{\bf [k} (\widehat{K}_g)_{\bf l]j} \,.
    \end{align}        
    \end{subequations}
\end{cor}

\begin{rem}
    If $\mathcal{S}$ is umbilical, write $(\widehat{K}_g)_{\bf ij} = f h_{\bf ij}$. Then
    \begin{alignat*}{3}
        F_{\bf lj}(\widehat{K}_g) &= \frac{f^2}{2}  h_{\bf lj} \,, &\qquad\qquad G_{\bf klj}(\widehat{K}_g,\widehat{\kappa}_g) &= - 2 h_{\bf j[k} \left( (df)_{\bf l]} - (\widehat{\kappa}_g)_{\bf l]} f \right) \,, \\
        P^{\bf i}{}_{\bf jkl} &= \widehat{w}^{\bf i}{}_{\bf jkl} \,, &\qquad E_{\bf lj} &= (n-2) \left(\widehat{L}_{\bf lj}-\widehat{l}_{\bf lj} + \frac{\varepsilon f^2}{2} h_{\bf lj} \right) \,, \\
        M_{\bf klj} &= 0 &\qquad \widehat{L}_{\bf k\perp} &=  \varepsilon \left( (df)_{\bf k} - (\widehat{\kappa}_g)_{\bf k} f \right) \,. \tag*{\qedhere}
    \end{alignat*}
\end{rem}

\begin{proof}
    From the decomposition \eqref{eq:decomposition_Riemann_Weyl} of the Riemann tensor of a Weyl connection, one has
    \begin{subequations}
        \begin{align}
            \label{eq:decomp_aux_a}
            \widehat{R}^{\bf i}{}_{\bf jkl} & = \widehat{W}^{\bf i}{}_{\bf jkl} + 2 S_{\bf j[k}{}^{\bf ia} \widehat{L}_{\bf l]a} = \widehat{W}^{\bf i}{}_{\bf jkl} + 2 S_{\bf j[k}{}^{\bf im} \widehat{L}_{\bf l]m} \,, \\
            \label{eq:decomp_aux_b}
            \widehat{r}^{\bf i}{}_{\bf jkl} &= \widehat{w}^{\bf i}{}_{\bf jkl} + 2S_{\bf j[k}{}^{\bf ia} \widehat{l}_{\bf l]a} = \widehat{w}^{\bf i}{}_{\bf jkl} + 2S_{\bf j[k}{}^{\bf im} \widehat{l}_{\bf l]m} \,, \\
            \label{eq:decomp_aux_c}
            \widehat{R}^\perp{}_{\bf jkl} & = \widehat{W}^\perp{}_{\bf jkl} + 2 S_{\bf j[k}{}^{\bf \perp a} \widehat{L}_{\bf l]a} = \widehat{W}^\perp{}_{\bf jkl} - 2 h_{\bf j[k} \widehat{L}_{\bf l]}{}^\perp \,.
        \end{align}
    \end{subequations}
    By combining \eqref{gauss_codazzi} with \eqref{eq:decomp_aux_a}-\eqref{eq:decomp_aux_b}, one deduces
    \[ \widehat{W}^{\bf i}{}_{\bf jkl} - \widehat{w}^{\bf i}{}_{\bf jkl} + 2 S_{\bf j[k}{}^{\bf im} \left( \widehat{L}_{\bf l]m} - \widehat{l}_{\bf l]m} \right) = \widehat{R}^{\bf i}{}_{\bf jkl} - \widehat{r}^{\bf i}{}_{\bf jkl} = - 2\varepsilon (\widehat{K}_g)^{\bf i}{}_{\bf [k} (\widehat{K}_g)_{\bf l]j} \,. \]
    By contracting the indices $({\bf i},{\bf k})$ and denoting $\widehat{A}_{\bf jl} := \widehat{L}_{\bf jl} - \widehat{l}_{\bf jl}$, one has
    \[ - \widehat{W}^\perp{}_{\bf j\perp l} + n \widehat{A}_{\bf [lj]} + (n-2)\widehat{A}_{\bf (lj)} + h_{\bf lj} h^{\bf ik} \widehat{A}_{\bf ik} = -\varepsilon \left( (\widehat{K}_g)^{\bf i}{}_{\bf i} (\widehat{K}_g)_{\bf lj} - (\widehat{K}_g)^{\bf i}{}_{\bf l} (\widehat{K}_g)_{\bf ij} \right) \,. \]
    Taking the anti-symmetric part gives $\widehat{A}_{\bf [jl]} = 0$ (this can also be deduced from \eqref{eq:antisym_part_schouten} and \Cref{lem:induced_weyl_connection}). By taking the trace with respect to $h$,
    \[ 2(n-1) h^{\bf jl} \widehat{A}_{\bf jl} = -\varepsilon \left( \left((\widehat{K}_g)^{\bf i}{}_{\bf i}\right)^2 - (\widehat{K}_g)_{\bf jl} (\widehat{K}_g)^{\bf jl} \right) \,. \]
    Injecting the traces back leads to \eqref{eq:codazzi_E} and then \eqref{eq:codazzi_P}.
    
    Furthermore, by combining \eqref{codazzi_mainardi_1} with \eqref{eq:decomp_aux_c}, one has
    \[ \widehat{W}^\perp{}_{\bf jkl} - 2 h_{\bf j[k} \widehat{L}_{\bf l]}{}^\perp = \widehat{R}^\perp{}_{\bf jkl} = 2 \widehat{D}_{\bf [k} (\widehat{K}_g)_{\bf l]j} + 2 (\widehat{\kappa}_g)_{\bf [k} (\widehat{K}_g)_{\bf l]j} \,. \]
    Taking the trace with respect to $h$ on the pair $({\bf j},{\bf l})$ gives
    \[ \varepsilon (n-1) \widehat{L}_{\bf k\perp} = h^{\bf jl} \left( 2\widehat{D}_{\bf [k} (\widehat{K}_g)_{\bf l]j} + 2 (\widehat{\kappa}_g)_{\bf [k} (\widehat{K}_g)_{\bf l]j} \right) \,. \]
    Hence \eqref{eq:codazzi_schouten} and then \eqref{eq:codazzi_M} follows by injecting the trace back.
\end{proof}

\section{Asymptotically Anti-de Sitter spaces}
\label{sec:aAdS_spaces}

This section is dedicated to defining asymptotically Anti-de Sitter (aAdS) spaces and presenting their standard properties. We start by introducing the definition used in this article and comparing it to the literature in \Cref{sec:def_aAdS}.

We then recall in \Cref{sec:expansion_conf_bound} the notorious expansion of aAdS metrics satisfying the vacuum Einstein equations with a negative cosmological constant derived by Fefferman and Graham in \cite{FG85}. This result was obtained in order to solve the so-called Poincaré problem, equivalent to the ambient space problem, which can be interpreted as a certain boundary value problem for aAdS spaces. Their method gives a formal series, described in~\Cref{thm:FG_exp}, with two free data: the boundary metric and the trace-free and divergence-free part of the boundary stress-energy tensor. These are discussed in detail and a special effort is made to justify their geometric interpretation as an equivalence class for boundaries of odd dimensions. \Cref{sec:expansion_conf_bound} concludes with the asymptotic expansions of some curvature tensors obtained from the expansion of the metric. In particular, one important result of this section, equation \eqref{eq:lim_weyl_frakt}, is crucial to interpret geometrically the analytic boundary conditions in \Cref{sec:geometric_bc}.

For illustrative purposes, several examples of aAdS spaces are provided in~\Cref{sec:ex_aAdS}. Finally, we introduced potential boundary stress-energy tensors inferred from these examples and discuss the interactions between boundary stress-energy tensors and conformal Killing fields in \Cref{sec:BSET}.

\subsection{Definitions}
\label{sec:def_aAdS}

Let us state the definition of asymptotically Anti-de Sitter spaces using manifolds with corners introduced in \Cref{sec:manifolds_with_corners}.

\begin{defi}
    \label{def_aAdS}
    An \emph{asymptotically Anti-de Sitter (aAdS) space} is a triple $(\M,\mathfrak{I},\widetilde{g})$ where
    \begin{itemize}
        \item $\M$ is a (n+1)-dimensional smooth manifold with corners,
        \item $\mathfrak{I}$ is a non-empty set of disjoint boundary hypersurfaces of $\M$,\item $\widetilde{g}$ is a smooth Lorentzian metric  on $\M\setminus\mathfrak{I}$,
    \end{itemize}
    such that there exists a boundary defining function $x$ of $\mathfrak{I}$ verifying
    \begin{itemize}
        \item[i)] the associated \emph{rescaled metric} $x^2 \widetilde{g}$ extends as a $\mathcal{C}^2$-metric up to $\mathfrak{I}$,
        \item[ii)] one has $(x^2\widetilde{g})^{-1}(dx,dx) = 1$ on $\mathfrak{I}$.
    \end{itemize}
    Each boundary hypersurface $\mathfrak{S} \in \mathfrak{I}$ is equipped with the conformal class $[\mathfrak{h}]$, where $\mathfrak{h}$ is the metric on $\mathfrak{S}$ induced by $x^2\widetilde{g}$, to form a Lorentzian conformal structure. The set $\mathfrak{I}$ is called the \emph{conformal boundary} of $(\M,\mathfrak{I},\widetilde{g})$.
\end{defi}

\begin{rems} \,
    \begin{itemize}   
        \item In the definition and in the whole article, the union $\bigcup_{\mathfrak{S} \in \mathfrak{I}} \mathfrak{S} \subset \M$ is identified to $\mathfrak{I}$ and identically denoted for simplicity.
        
        \item \Cref{def_aAdS} is independent of the choice of a boundary defining function of $\mathfrak{I}$. Indeed, assume $i)-ii)$ hold for a boundary defining function $x$ of $\mathfrak{I}$ and let $x'$ be another boundary defining function of $\mathfrak{I}$. By \Cref{lemma:bdf}, one can write $x'=x\Theta$ with $\Theta \in \mathcal{C}^\infty(\M,\R_+^\star)$. Thus $(x')^2 \widetilde{g} = \Theta^2 x^2 \widetilde{g}$ also extends as a $\mathcal{C}^2$-metric on $\M$ and $((x')^2\widetilde{g})^{-1}(dx',dx') = (x^2 \widetilde{g})^{-1}(dx,dx) = 1$ on $\mathfrak{I}$. Moreover, the induced metric on $\mathfrak{S}$ verifies
        \begin{equation}
            \label{eq:frakh}
            \mathfrak{h}' = \Theta|_{\mathfrak{S}}^2 \, \mathfrak{h} \,.
        \end{equation}
        This explains why $\mathfrak{S}$ is equipped with the conformal class $[\mathfrak{h}]$ rather than the metric $\mathfrak{h}$ which is dependent on the boundary defining function.

        \item All rescaled metrics of an aAdS space have the same regularity thanks to \Cref{lemma:bdf}. The common regularity classes are $\mathcal{C}^m(\M)$ for $m\geq 2$ and $\mathcal{C}^2(\M) \cap \mathcal{C}^m_{phg}(\mathcal{M}\,|\, \mathfrak{I}, E)$ for $m\geq0$ and $E$ an index set, see \Cref{sec:polyhomogeneous_functions} for the definitions.

        \item If $(\M,\mathfrak{I},\widetilde{g})$ is an aAdS space then $(\M\setminus\mathfrak{I},\widetilde{g})$ is approximately Einstein near $\mathfrak{I}$ for the normalised negative cosmological constant $\Lambda = -n(n-1)/2$. Indeed, on any compact neighbourhood of a point $p \in \mathfrak{S}$,
        \begin{equation}
            \label{approximate_solution}
            \Ric(\widetilde{g}) + n \widetilde{g} = \grando{\frac{1}{x}} \,.
        \end{equation}
        This implies that $\widetilde{R} = -n(n+1) + \grando{x}$ where $\widetilde{R}$ is the scalar curvature of $\widetilde{g}$.
        
        \item \Cref{def_aAdS} allows for various topologies and geometries at the conformal boundary, not only that of AdS. In this sense, it is more general than the definition of asymptotically flat spacetimes. See~\cite{HT85} for a definition imposing the same topology and conformal class on the boundary as AdS.
        
        \item This definition is based on~\cite[Definition 3.1]{EK19}. Moreover, if the rescaled metrics are smooth on $\M$ then it is similar to \cite[Definition 1]{W13} with the weak decay rates.

        \item Former definitions of aAdS spaces such as Friedrich's \cite{F95} are based on the vacuum Einstein equations and conformal embeddings, in the lines of Penrose's asymptotic simplicity \cite{P60}. However, multiple non-equivalent conformal embeddings of a space exist since the differential structure at the conformal boundary thereby created is a priori not imposed. Furthermore, \Cref{def_aAdS} allows for coupling of the Einstein equations with matter fields as long as they are compatible with the existence of a smooth conformal boundary.
        
        \item Even though one cannot define the physical extrinsic curvature of a connected component $\mathfrak{S}$ of the conformal boundary $\mathfrak{I}$ since the metric $\widetilde{g}$ is singular on $\mathfrak{S}$, the transformation rule \eqref{eq:aux_ec} implies that the physical mean curvature of $\mathfrak{S}$ would be equal to $-n$ (for the outward-pointing normal choice). \qedhere
    \end{itemize}
\end{rems}

Anticipating the initial data problem, one can define asymptotically hyperbolic spaces in a similar manner. Most of the above remarks also apply in this context.

\begin{defi}
    \label{def:aH}
    An \emph{asymptotically Hyperbolic (aH) space} is a triple $(\mathcal{S},\slashed{\mathfrak{I}},\widetilde{h})$ where
    \begin{itemize}
        \item $\mathcal{S}$ is a n-dimensional smooth manifold with corners,
        \item $\slashed{\mathfrak{I}}$ is a non-empty set of disjoint boundary hypersurfaces of $\mathcal{S}$,\item $\widetilde{h}$ is a smooth Riemannian metric on $\mathcal{S}\setminus\slashed{\mathfrak{I}}$,
    \end{itemize}
    such that there exists a boundary defining function $x$ of $\slashed{\mathfrak{I}}$ verifying
    \begin{itemize}
        \item[i)] the associated \emph{rescaled metric} $x^2 \widetilde{h}$ extends as a $\mathcal{C}^2$-metric up to $\slashed{\mathfrak{I}}$,
        \item[ii)] one has $(x^2 \widetilde{h})^{-1}(dx,dx) = 1$ on $\slashed{\mathfrak{I}}$.
    \end{itemize}
    Each boundary hypersurface $\mathfrak{S} \in \slashed{\mathfrak{I}}$ is equipped with the conformal class $[\slashed{\mathfrak{h}}]$, where $\slashed{\mathfrak{h}}$ is the metric on $\mathfrak{S}$ induced by $x^2\widetilde{h}$, to form a Riemannian conformal structure. The set $\slashed{\mathfrak{I}}$ is called the \emph{conformal boundary} of $(\mathcal{S},\slashed{\mathfrak{I}},\widetilde{h})$.
\end{defi}

\subsection{Asymptotic expansions at the conformal boundary}
\label{sec:expansion_conf_bound}

The definition of aAdS spaces is designed so that the metric $\widetilde{g}$ is approximately Einstein near the conformal boundary, see \eqref{approximate_solution}. Fefferman and Graham \cite{FG85}, see also their comprehensive book \cite{FG12}, derived iteratively an asymptotic expansion of the rescaled metric $g := x^2\widetilde{g}$ ensuring that the Einstein equations with (the normalised) negative cosmological constant $\Lambda = -n(n-1)/2$ are verified by $\widetilde{g}$ nearby the conformal boundary at any given order. This expansion will be referred to as the Fefferman-Graham expansion (FG expansion).

We present first the gauge used to obtain the FG expansion in \Cref{sec:FG_gauge} before stating and commenting the FG expansion in \Cref{sec:FG_expansion}. The Einstein equations fully impose all the terms of the expansion apart from two which are therefore called the free data on the conformal boundary. We review in length how to geometrically interpret these free data from their gauge-dependent definition as the conformal class and the boundary stress-energy tensor class. Finally, asymptotic expansions for some curvature tensor fields implied by the FG expansion are derived  in \Cref{sec:expansions_curvature_tensors}. This enables us to identify the boundary stress-energy tensor in the asymptotic expansions of (part of) the Weyl tensor, see equation~\eqref{eq:lim_weyl_frakt}, which will help us to pass from analytic to geometric boundary conditions in \Cref{sec:geometric_bc}. 

In this whole section, we consider an aAdS space $(\M,\mathfrak{I},\widetilde{g})$ of dimension $n+1\geq3$.

\subsubsection{Fefferman-Graham gauges}
\label{sec:FG_gauge}

\begin{defi}
    A \emph{Fefferman-Graham gauge} (FG gauge) is a chart $(\phi,\mathcal{U},1)$ with $\mathcal{U} \cap \mathfrak{I} \neq \varnothing$ and whose coordinates are denoted by $\phi := (x,y^0,\dots,y^{n-1})$ such that
    \begin{itemize}
    	\item[i)] $x$ is a local boundary defining function of $\mathfrak{I}$ on $\mathcal{U}$,
    	\item[ii)] the rescaled metric has the following form on $\mathcal{U}$
    	\begin{equation}
    		\label{FG_gauge}
    		g := x^2 \widetilde{g} = dx^2 + \underline{\mathfrak{h}}(x)
    	\end{equation}
    	with $\underline{\mathfrak{h}} : x \mapsto \underline{\mathfrak{h}}(x)$ a $\mathcal{C}^2$-map with values into smooth Lorentzian metrics on the level sets of $x$.
    \end{itemize}
\end{defi}

\begin{rems} \,
	\begin{itemize}
		\item The FG gauge is a geodesic gauge or a $1+n$ decomposition with respect to the spacelike coordinate $x$ with zero shift and lapse $1$ for the rescaled metric $g = x^2 \widetilde{g}$.
        
        \item The frame field $(\partial_x,(\partial_{y^i}))$ is an adapted frame field to the level sets of $x$ on $\mathcal{U}$ with respect to the rescaled metric $g$, see \Cref{sec:adapted_frame_fields}.
        
		\item The construction of FG gauges for aAdS spaces can be derived from~\cite[Section 5]{GL91}. Let us recall the proof which decomposes into three steps.
        \begin{itemize}
            \item[a)] Take any local boundary defining function $x$ of $\mathfrak{I}$ and any smooth function $\omega$ on $\mathfrak{I}$.
            
            \item[b)] Let $\theta$ be a smooth solution on a neighbourhood of $\mathfrak{I}$ to the following non-characteristic first order hyperbolic partial differential equation, with $\omega$ as initial data on $\mathfrak{I}$,
            \[  2 (\grad_g x) \cdot \theta + x g^{-1}(d\theta,d\theta) = \frac{1-g^{-1}(dx,dx)}{x} \,, \]
            where $\grad_g x$ is the gradient of $x$ with respect to the rescaled metric $g := x^2 \widetilde{g}$. Note that the function on the right hand side is well-defined up to the conformal boundary by definition of aAdS spaces. If the rescaled metrics are smooth then $\theta$ always exist.
            
            \item[c)] Define $x':=\exp(\theta)x$. Take a local chart $(y^0,\dots,y^{n-1})$ in the neighbourhood of $p$ in $\mathfrak{I}$ and extend it away from the conformal boundary by taking the functions constant on the integral curves of $\grad_g x'$.
        \end{itemize}
        Then $(x',y^0,\dots,y^{n-1})$ is a FG gauge. Note that not all local boundary defining functions of $\mathfrak{I}$ can become the first coordinate of a FG gauge, only those verifying
        \begin{equation}
            \label{eq:FG_bdf}
            (x^2\widetilde{g})^{-1}(dx,dx) = 1
        \end{equation}
        in a neighbourhood of a point $p\in\mathfrak{I}$ in $\M$.
        
		\item Let $(x,y^0,\dots,y^{n-1})$ be a FG gauge on an open set $\mathcal{U}$. Then for all $\mathfrak{S} \in \mathfrak{I}$ such that $\mathfrak{S} \cap \mathcal{U} \neq \varnothing$,
        \[ \underline{\mathfrak{h}} = \mathfrak{h} \quad \text{on } \mathfrak{S} \cap \mathcal{U} \,, \]
        where $\mathfrak{h}$ is the Lorentzian metric on $\mathfrak{S}$ associated to $x$ by~\Cref{def_aAdS}. \qedhere
	\end{itemize}
\end{rems}

\begin{defi}
    A \emph{transformation preserving the FG gauge} is a diffeomorphism $\chi$ on an open set $\mathcal{U}$ such that there exists two FG gauges $\phi_1$ and $\phi_2$ on $\mathcal{U}$ verifying $\chi = \phi_2 \circ (\phi_1)^{-1}$.
\end{defi}

\begin{rem}
    Let $\phi_1 = (x,y^0,\dots,y^{n-1})$ and $\phi_2 = (x',z^0,\dots,z^{n-1})$ be two FG gauges on an open set $\mathcal{U}$. Then the map $\psi : (y^0,\dots,y^{n-1}) \in \Pi \circ \phi_1(\mathfrak{I} \cap \mathcal{U}) \mapsto (z^0,\dots,z^{n-1}) \in \Pi\circ\phi_2(\mathfrak{I} \cap \mathcal{U})$ is a diffeomorphism, where $\Pi : (x,y) \in \R^{n+1} \mapsto y \in \R^n$.
\end{rem}

\noindent The relation between the families of Lorentzian metrics $\underline{\mathfrak{h}}$ of two FG gauges is given by the following lemma.

\begin{lem}
    \label{lem_transfo_FG_gauge}
    Let $\phi_1 = (x,y^0,\dots,y^{n-1})$ and $\phi_2 = (x',z^0,\dots,z^{n-1})$ be two FG gauges defined on a common open set $\mathcal{U}$. Define $\Theta := x'/x\in \mathcal{C}^\infty(\mathcal{U},\R_+^\star)$ thanks to \Cref{lemma:bdf}. Then
    \begin{subequations}
        \label{cond}
	    \begin{align}
            \label{cond1}
	        \left( \underline{\mathfrak{h}}'_{kl} \circ \phi_2 \circ (\phi_1)^{-1} \right) \, \frac{\partial z^k}{\partial x} \frac{\partial z^l}{\partial x} &= \left( 1-\left(1+\frac{x}{\Theta} \frac{\partial \Theta}{\partial x}\right)^2 \right) \Theta^2  \,, \\
            \label{cond2} 
	        \left( \underline{\mathfrak{h}}'_{kl}\circ \phi_2 \circ (\phi_1)^{-1} \right) \, \frac{\partial z^k}{\partial x} \frac{\partial z^l}{\partial y^i} &= - \frac{x}{\Theta} \frac{\partial \Theta}{\partial y^i} \left(1+\frac{x}{\Theta} \frac{\partial \Theta}{\partial x}\right) \Theta^2 \,, \\
            \label{cond3}
	        \left( \underline{\mathfrak{h}}'_{kl}\circ \phi_2 \circ (\phi_1)^{-1} \right) \, \frac{\partial z^k}{\partial y^i} \frac{\partial z^l}{\partial y^j} &= \left( \underline{\mathfrak{h}}_{ij} - \frac{x^2}{\Theta^2} \frac{\partial \Theta}{\partial y^i}\frac{\partial \Theta}{\partial y^j} \right) \Theta^2 \,.
	    \end{align}
	\end{subequations}
\end{lem}

\begin{proof}
    One has
    \[ \frac{\partial x'}{\partial x} = \Theta \left( 1+\frac{x}{\Theta} \frac{\partial \Theta}{\partial x}\right) \,, \qquad \frac{\partial x'}{\partial y^i} = x \frac{\partial \Theta}{\partial y^i} \,. \]
    Equations~\eqref{cond} are found by comparing~\eqref{FG_gauge} in the two FG gauges.
\end{proof}

\begin{cor}
    \label{cor:FG_gauge}
    There exists an open neighbourhood $\mathcal{V}$ of $\mathfrak{I} \cap \mathcal{U}$ in $\mathcal{U}$ such that the functions
    \[ 1+\frac{x}{\Theta} \frac{\partial\Theta}{\partial x} \qquad \text{and} \qquad \det\left(\frac{\partial z}{\partial y}\right) \]
    nowhere vanish on $\mathcal{V}$. Furthermore, the following differential equations hold on $\mathcal{V}$
    \begin{subequations}
        \label{equadiff}
        \begin{align}
            \label{equadiff_z}
            \frac{\partial z^k}{\partial x} &= - \frac{\frac{x}{\Theta}}{\sqrt{1 - \frac{x^2}{\Theta^2} \underline{\mathfrak{h}}^{pq} \frac{\partial \Theta}{\partial y^p} \frac{\partial \Theta}{\partial y^q}}} \underline{\mathfrak{h}}^{ij} \frac{\partial \Theta}{\partial y^j} \frac{\partial z^k}{\partial y^i} \,, \\
            \label{equadiff_theta}
            \frac{\partial \Theta}{\partial x} &= \frac{\Theta}{x} \left( \sqrt{1 - \frac{x^2}{\Theta^2} \underline{\mathfrak{h}}^{ij} \frac{\partial \Theta}{\partial y^i} \frac{\partial \Theta}{\partial y^j}} - 1 \right) \,.
        \end{align}
    \end{subequations}
\end{cor}

\begin{rem}
    If $\Theta$ is independent of $x$ then \eqref{equadiff_theta} implies that the restriction of $\Theta$ on the level set $\{x=c >0\}$ is an eikonal function for $\underline{\mathfrak{h}}(c)$. By continuity, this also holds for $c=0$, that is on $\mathfrak{I}\cap\mathcal{U}$.
\end{rem}

\begin{proof}
    The two functions nowhere vanish on $\mathfrak{I} \cap \mathcal{U}$ thus the existence of $\mathcal{V}$ is given by continuity. By dimension counting, there exists a map $(\lambda,\mu^0,\dots,\mu^{n-1}) : \mathcal{V}\to\R^{n+1}\setminus\{0\}$ such that
	\begin{equation}
		\label{lincomb}
		\lambda u + \mu^i v_i = 0
	\end{equation}    
    where
	\[ z := (z^0,\dots,z^{n-1}) \,, \;\; u := \frac{\partial z}{\partial x} \,, \;\; v_i := \frac{\partial z}{\partial y^i} \in \mathcal{C}^\infty(\mathcal{V},\R^n) \,. \]
	Taking $\lambda$\eqref{cond1}$+\mu^i$\eqref{cond2}, one obtains
	\begin{equation}
		\label{cond4}
		 \frac{x}{\Theta} \left(1+\frac{x}{\Theta} \frac{\partial \Theta}{\partial x}\right) \, \mu^i \frac{\partial \Theta}{\partial y^i} = \lambda \left( 1 - \left(1+\frac{x}{\Theta} \frac{\partial \Theta}{\partial x}\right)^2 \right) \,.
	\end{equation}
	In the same way, $\lambda$\eqref{cond2}$+\mu^j$\eqref{cond3} gives
	\begin{equation} 
		\label{cond5}
		\lambda \frac{x}{\Theta} \left(1+\frac{x}{\Theta} \frac{\partial \Theta}{\partial x}\right) \frac{\partial \Theta}{\partial y^i} = \mu^j \underline{\mathfrak{h}}_{ij} - \frac{x^2}{\Theta^2} \, \mu^j \frac{\partial \Theta}{\partial y^i} \frac{\partial \Theta}{\partial y^j} \,.
	\end{equation}
	Injecting~\eqref{cond4} in~\eqref{cond5} gives
	\[ \left( 1+\frac{x}{\Theta} \frac{\partial \Theta}{\partial x} \right) \mu^i = \lambda  \, \frac{x}{\Theta} \underline{\mathfrak{h}}^{ij} \frac{\partial \Theta}{\partial y^j} \,.\]
	It follows that $\lambda \neq 0$. By injecting the above formula for $\mu^i$ and simplifying by $\lambda$, \eqref{lincomb} and~\eqref{cond4} rewrite as
	\begin{align*}
	    \frac{\partial z^k}{\partial x} = - \frac{\frac{x}{\Theta}}{1+\frac{x}{\Theta} \frac{\partial \Theta}{\partial x}} \underline{\mathfrak{h}}^{ij} \frac{\partial \Theta}{\partial y^j} \frac{\partial z^k}{\partial y^i} \,, \\
        \left( 1+\frac{x}{\Theta} \frac{\partial \Theta}{\partial x} \right)^2 = 1 - \frac{x^2}{\Theta^2} \underline{\mathfrak{h}}^{ij} \frac{\partial \Theta}{\partial y^i} \frac{\partial \Theta}{\partial y^j} \,.
	\end{align*}
    Hence the result.
\end{proof}

\noindent Recall that conformal rescalings of the conformal boundary are linked to changes of the boundary defining function, see equation  \eqref{eq:frakh}. A remarkable fact is that they are also linked to transformations preserving the FG gauge as stated in the next lemma, even though the boundary defining functions of FG gauges verify \eqref{eq:FG_bdf}.

\begin{lem}
    \label{lem:hfrak_change}
	Let $\phi_1 = (x,y^0,\dots,y^{n-1})$ and $\phi_2 = (x',z^0,\dots,z^{n-1})$ be two FG gauges on an open set $\mathcal{U}$. Define $\Theta := x'/x \in \mathcal{C}^\infty(\mathcal{U})$ and $\psi : (y^0,\dots,y^{n-1}) \mapsto (z^0,\dots,z^{n-1})$ the diffeomorphism on $\mathfrak{I}\cap\mathcal{U}$. Then, for all $\mathfrak{S} \in \mathfrak{I}$ such that $\mathfrak{S} \cap \mathcal{U} \neq \varnothing$, the diffeomorphism $\psi$ induces a confomorphism on $\mathfrak{S} \cap \mathcal{U}$ with
    \begin{equation}
        \label{conf_trans_conf_bound}
        \psi^\star \mathfrak{h}' = \Theta^2 \mathfrak{h} \quad \text{on } \mathfrak{S}\cap\mathcal{U}  \,,
    \end{equation}
    where $\mathfrak{h}$ and $\mathfrak{h}'$ are the metrics on $\mathfrak{S}$ associated to $x$ and $x'$ respectively. Reciprocally, all local confomorphisms of $\mathfrak{S} \in \mathfrak{I}$ can be induced this way.
\end{lem}

\begin{rems} \,
    \begin{itemize}
        \item The reciprocal can for instance be found expressed in infinitesimal form in~\cite{ISTY00}.
        \item The particular case $\Theta = 1$ is equivalent to a diffeomorphism on $\mathfrak{I}$ since then $x'=x$ and the functions $z^k$ are independent of $x$ by \eqref{equadiff_z}. \qedhere
    \end{itemize}
\end{rems}

\begin{proof}
    By evaluating~\eqref{cond3} on $\mathfrak{S} \cap \mathcal{U}$, one obtains \eqref{conf_trans_conf_bound} straightforwardly.
    
    Let us turn to the reciprocal. Let $p \in \mathfrak{S} \in \mathfrak{I}$ be a point on the conformal boundary, $(x,y^0,\dots,y^{n-1})$ be a FG gauge on an open neighbourhood $\mathcal{U}$ of $p$ in $\M$ and $\Omega \in \mathcal{C}^\infty(\mathfrak{S}\cap\mathcal{U},\R_+^\star)$ be a local conformal factor. One can construct another FG gauge $(x\Theta,z^0,\dots,z^{n-1})$ on a neighbourhood of $p$ in $\mathcal{U}$ as in the proof of the existence of FG gauges, with initial data on the boundary $\omega := \ln \Omega$ and a chart $(z^0,\dots,z^{n-1})$ of $\mathfrak{I}$ around $p$.
\end{proof}

\subsubsection{Fefferman-Graham expansions}
\label{sec:FG_expansion}

Let us now detail the Fefferman and Graham expansions derived in the celebrated work \cite{FG85}. A study of these expansions at a finite order can be found in \cite{S21}.

\begin{thm}[Fefferman and Graham \protect{\cite[Theorem 2.3]{FG85}} and \protect{\cite[Theorem 4.5]{FG12}}]
    \label{thm:FG_exp}
    Assume that the aAdS space $(\M,\mathfrak{I},\widetilde{g})$ is such that
    \begin{itemize}
        \item[i)] the rescaled metrics $g$ are of class $\mathcal{C}^2(\M) \cap \mathcal{C}^m_{\text{phg}}(\mathcal{M}\,|\,\mathfrak{I},E)$ with $m \geq n$ and some index set $E$ containing $(0,0)$ and $(n,1)$,
        \item[ii)] $(\M\setminus\mathfrak{I},\widetilde{g})$ is a solution to the \eqref{eq:VE} with $\Lambda=-n(n-1)/2$.
    \end{itemize}
    Let $(x,y^0,\dots,y^{n-1})$ be a FG gauge on an open set $\mathcal{U}$, $p$ be a point in $\mathfrak{I} \cap \mathcal{U}$ and $\mathcal{K}$ be a compact neighbourhood of $p$ in $\mathcal{U}$. Then, according to the parity of $n$, the following asymptotic expansion in $x$ -- called the \emph{Fefferman-Graham expansion} (FG expansion) -- holds on $\mathcal{K}$:
    \begin{itemize}
        \item if $n$ is odd,
        \begin{equation}
            \label{FG_exp_n_odd}
            \underline{\mathfrak{h}}(x) = \sum_{k=0}^{\lfloor m/2 \rfloor} \mathfrak{h}^{(2k)} \, x^{2k} + \sum_{k=(n-1)/2}^{\lfloor (m-1)/2 \rfloor}  \mathfrak{h}^{(2k+1)} \, x^{2k+1} + o(x^m)
        \end{equation}
        where $\mathfrak{h}^{(k)}$ are smooth symmetric 2-tensor fields on $\mathfrak{I}$. Let us emphasise that the second sum starts at $k=(n-1)/2$. Furthermore, introducing
        \begin{itemize}
            \item[\ding{228}] $\mathfrak{h} := \mathfrak{h}^{(0)}$ the induced metric of the rescaled metric $g := x^2 \widetilde{g}$ on $\mathfrak{I}$,
            \item[\ding{228}] $\mathfrak{t} := \mathfrak{h}^{(n)}$,
        \end{itemize}
        one has that $\mathfrak{t}$ is trace-free and divergence-free with respect to $\mathfrak{h}$ and that there exists universal functions $\mathcal{F}_{n,k}$ such that
        \begin{align*}
            \forall \; 1 \leq k \leq (n-1)/2, \quad \mathfrak{h}^{(2k)} &= \mathcal{F}_{n,2k}(\mathfrak{h},\dots,\partial^{2k} \mathfrak{h}) \,, \\
            \forall \; n+1 \leq k \leq m, \quad \mathfrak{h}^{(k)} &= \mathcal{F}_{n,k}(\mathfrak{h},\dots,\partial^{k} \mathfrak{h},\mathfrak{t},\dots,\partial^{k-n} \mathfrak{t}) \,.
        \end{align*}
        \item if $n$ is even,
            \begin{equation}
                \label{FG_exp_n_even}
                \underline{\mathfrak{h}}(x) = \sum_{k=0}^{\lfloor m/2 \rfloor} \mathfrak{h}^{(2k)} \, x^{2k} + \sum_{k=n/2}^{\lfloor m/2 \rfloor}  \mathfrak{h}^{(\star 2k)} \, x^{2k}\log x + o(x^m)
            \end{equation}
            where $\mathfrak{h}^{(k)}$ and $\mathfrak{h}^{(\star 2k)}$ are smooth symmetric 2-tensors on $\mathfrak{I}$. Let us emphasise that the second sum starts at $k=n/2$. Furthermore, introducing
            \begin{itemize}
                \item[\ding{228}] $\mathfrak{h} := \mathfrak{h}^{(0)}$ the induced metric of the rescaled metric $g := x^2 \widetilde{g}$ on $\mathfrak{I}$,
                \item[\ding{228}] $\mathfrak{t}$ the divergence-free and trace-free part of $\mathfrak{h}^{(n)}$ with respect to $\mathfrak{h}$,
        \end{itemize}
        there exists universal functions $\mathcal{F}_{n,2k}$, $\mathcal{F}_{n,\star 2k}$, $\mathcal{F}_{n,n}^{\text{tr}}$ and $\mathcal{F}_{n,n}^{\text{div}}$ such that
        \begin{align*}
            \forall \; 1 \leq k \leq (n-2)/2, \quad \mathfrak{h}^{(2k)} &= \mathcal{F}_{n,2k}(\mathfrak{h},\dots,\partial^{2k} \mathfrak{h}) \,, \\
            \mathfrak{h}^{(\star n)} &= \mathcal{F}_{n,\star n}(\mathfrak{h},\dots,\partial^{n} \mathfrak{h}) \,, \\
            \trace_{\mathfrak{h}} \mathfrak{h}^{(n)} &= \mathcal{F}_{n,n}^{\text{tr}}(\mathfrak{h},\dots,\partial^{n} \mathfrak{h}) \,, \\
            \diver_\mathfrak{h} \mathfrak{h}^{(n)} &= \mathcal{F}_{n,n}^{\text{div}}(\mathfrak{h},\dots,\partial^{n+1}\mathfrak{h}) \,, \\
            \forall \; (n+2)/2 \leq k \leq \lfloor m/2 \rfloor, \quad \mathfrak{h}^{(\star2k)} &= \mathcal{F}_{n,\star 2k}(\mathfrak{h},\dots,\partial^{2k} \mathfrak{h},\mathfrak{t},\dots,\partial^{2k-n} \mathfrak{t}) \,, \\
            \forall \; (n+2)/2 \leq k \leq \lfloor m/2 \rfloor, \quad \mathfrak{h}^{(2k)} &= \mathcal{F}_{n,2k}(\mathfrak{h},\dots,\partial^{2k} \mathfrak{h},\mathfrak{t},\dots,\partial^{2k-n} \mathfrak{t}) \,.
        \end{align*}
    \end{itemize} 
\end{thm}

\begin{rems} \,
    \begin{itemize}
        \item The pair $(\mathfrak{h},\mathfrak{t})$ is called the \emph{free data} on the boundary as its value is not imposed by the vacuum Einstein equations.
        
        \item The tensor field $\mathfrak{h}^{(n)}$ is called in the physical literature the \emph{boundary stress-energy tensor} or holographic stress-energy tensor. For more details, see for example~\cite{dHSS01,S01}. We will also refer to $\mathfrak{t}$ as the boundary stress-energy tensor for $n$ odd because it coincide with $\mathfrak{h}^{(n)}$ in this case.  
        
        \item The first coefficients $\mathfrak{h}^{(2k)}$ are given by
        \begin{subequations}
            \label{eq:coeff_FG_expansion}
            \begin{alignat}{2}
                \mathfrak{h}^{(2)}_{ij} &= - \mathfrak{l}_{ij} &\qquad \text{if } n \geq 3 \,, \\
                \mathfrak{h}^{(4)}_{ij} &= - \frac{1}{4(n-4)} \mathfrak{b}_{ij} + \frac{1}{4}\mathfrak{l}_{ik} \mathfrak{l}^k{}_j &\qquad \text{if } n \geq 5 \,,
            \end{alignat}
        \end{subequations}
        where $\mathfrak{l}_{ij}$ and $\mathfrak{b}_{ij}$ are the Schouten and Bach tensors of $\mathfrak{h}$.
        
        \item When $n$ is even, logarithmic terms may arise in the expansion~\eqref{FG_exp_n_even}. The first coefficient of the logarithmic terms denoted by $\mathfrak{h}^{(\star n)}$ is proportional to the $n$th obstruction tensor, see~\cite{FG85,GH04} for more details on the obstruction tensors. If it vanishes then all the other coefficients $\mathfrak{h}^{(\star 2k)}$ also vanish, leaving no logarithmic terms in the expansion. In particular, this is the case if $\mathfrak{h}$ is locally conformally flat or conformally Einstein (note that it is always true for $n=2$ by \Cref{Weyl_Schouten_theorem}).

        \item The absence of a first order term in \eqref{FG_exp_n_odd} and \eqref{FG_exp_n_even} implies that the boundary hypersurfaces $\mathfrak{S} \in \mathfrak{I}$ are totally geodesic for the Levi-Civita connection $\nabla$ of the rescaled metric $g := x^2 \widetilde{g}$, when $x$ is a boundary defining function verifying $(x^2\widetilde{g})^{-1}(dx,dx) = 1$ in a neighbourhood of the conformal boundary. In view of \Cref{lem:link_umbilical_totally_geodesic}, this implies that all $\mathfrak{S}\in\mathfrak{I}$ are umbilical in $(\M,[g])$. \qedhere
    \end{itemize}
\end{rems}

\noindent The fact that the coefficients $\mathfrak{h}^{(k)}$ and $\mathfrak{h}^{(\star k)}$ are tensor fields on the conformal boundary is not immediate. As noted before, diffeomorphisms of the conformal boundary are associated to transformations preserving the FG gauge with $x'=x$ or equivalently $\Theta=1$. Yet, in this case the other coordinates $z^k$ of the FG gauge $(x',z^0,\dots,z^{n-1})$ are independent of $x$ by \eqref{equadiff_z}. Thus the comparison between the FG expansions associated to two such FG gauges is direct and gives the tensorial covariance of the coefficients. Another important consequence is that these tensors $\mathfrak{h}^{(k)}$ and $\mathfrak{h}^{(\star k)}$ do not depend on the full FG gauge $(x,y^0,\dots,y^{n-1})$ but only on the boundary defining function $x$.

The behaviours of $\mathfrak{h}^{(k)}$ and $\mathfrak{h}^{(\star k)}$ under transformations preserving the FG gauge but inducing non-trivial conformal rescalings on the conformal boundary still are to be derived. Comparing the FG expansions is not as easy as above since now the function $\Theta$ depends a priori on $x$. It is nonetheless crucial to understand these behaviours, and especially these of the free data $(\mathfrak{h},\mathfrak{t})$, to obtain a gauge-independent meaning.

This was addressed in the physics literature under infinitesimal form \cite{ISTY00,S01}. For our purposes, it is sufficient to only look at the free data behaviour. That of $\mathfrak{h}$ is already given by~\Cref{lem:hfrak_change}. That of $\mathfrak{t}$ for odd $n$ is the object of the following lemma.

\begin{lem}
    \label{lem:tfrak_change}
    Assume $n$ is odd and the hypotheses of \Cref{thm:FG_exp} hold. Let $\phi_1 = (x,y^0,\dots,y^{n-1})$ and $\phi_2 = (x',z^0,\dots, z^{n-1})$ be two FG gauges on an open set $\mathcal{U}$. Pose $\Theta := x'/x \in \mathcal{C}^\infty(\mathcal{U},\R_+^\star)$ given by \Cref{lemma:bdf} and $\psi := (y^0,\dots,y^{n-1}) \mapsto (z^0,\dots,z^{n-1})$. For all $\mathfrak{S} \in \mathfrak{I}$ such that $\mathfrak{S} \cap \mathcal{U} \neq \varnothing$,
    \begin{equation}
        \psi^\star \mathfrak{t}' = \Theta^{2-n} \, \mathfrak{t} \quad \text{on } \mathfrak{S}\cap\mathcal{U} \,,
    \end{equation}
    where $\mathfrak{t}$ and $\mathfrak{t}'$ are the boundary stress-energy tensors associated to the two FG gauges by \Cref{thm:FG_exp}.
\end{lem}

\begin{proof}
   Take the open neighbourhood $\mathcal{V}$ of $\mathfrak{I}\cap\mathcal{U}$ in $\mathcal{U}$ given by \Cref{cor:FG_gauge}. From the FG expansion~\eqref{FG_exp_n_odd}, the functions $\underline{\mathfrak{h}}_{ij}$ verify locally
   \[ \underline{\mathfrak{h}}_{ij}(x) = P_{ij}(x^2) + \grando{x^n} \,, \]
   for some suitable polynomials $P_{ij}$ of degree at most $(n-1)/2$. Write locally
   \[ \Theta(x,.) = Q(x^2) + \grando{x^s} \,, \]
   where $s \in \N_0$ is odd and maximal, $Q$ is a polynomial of degree at most $(s-1)/2$ with values into $\mathcal{C}^\infty(\R^n,\R)$. From the differential equation~\eqref{equadiff_theta} on $\Theta$, one deduces that $s\geq2+\min(s,n)$ and thus $s\geq n+2$. Using \eqref{cond3} and the FG expansions \eqref{FG_exp_n_odd} for the two FG gauges, the result follows by identifying the term of order $n$ in $x$.
\end{proof}

\noindent For $n$ odd, the previous lemma enables to extend the definition of the boundary stress-energy tensor  for all boundary defining functions, not only those verifying \eqref{eq:FG_bdf}.

\begin{defi}
    \label{def:extension_frakt}
    Assume $n$ is odd and the hypotheses of \Cref{thm:FG_exp} hold. Let $x'$ be a boundary defining function of $\mathfrak{I}$. Define its boundary stress-energy tensor by
    \[ \mathfrak{t}' := \exp(\theta)^{2-n} \mathfrak{t} \,, \]
    where $\theta$ is a smooth function on a neighbourhood of $\mathfrak{I}$ such that $x := \exp(-\theta)x'$ is a boundary defining function of $\mathfrak{I}$ verifying \eqref{eq:FG_bdf} on this neighbourhood and $\mathfrak{t}$ is the boundary stress-energy tensor associated to $x$.
\end{defi}

\begin{rem}
    The existence of such a function $\theta$ is ensured in a neighbourhood of the conformal boundary, see the proof of the existence of FG gauges. Although $\theta$ is not unique, the definitions for two different functions $\theta$ coincide by \Cref{lem:tfrak_change}. 
\end{rem}

\noindent For $n$ odd, \Cref{lem:hfrak_change} and \Cref{lem:tfrak_change} lead to the following definitions.

\begin{defi}
	\label{def:equivalence_relation_pair}
    Let $n\geq3$ be an odd integer. Two triples $(\mathfrak{S},\mathfrak{h},\mathfrak{t})$ and $(\mathfrak{S}',\mathfrak{h}',\mathfrak{t}')$ where
    \begin{itemize}
        \item $(\mathfrak{S},\mathfrak{h})$ and $(\mathfrak{S}',\mathfrak{h}')$ are pseudo-Riemannian manifolds of dimension $n$,
        \item $\mathfrak{t}$ and $\mathfrak{t}'$ are trace-free and divergence-free symmetric 2-tensors on $(\mathfrak{S},\mathfrak{h})$ and $(\mathfrak{S}',\mathfrak{h}')$ respectively,
    \end{itemize}
    are said to be equivalent if there exists a confomorphism $\psi : (\mathfrak{S},\mathfrak{h}) \to (\mathfrak{S}',\mathfrak{h}')$ such that
    \[ \qquad \psi^\star\mathfrak{t}' = \Omega^{2-n} \mathfrak{t} \,, \]
    where $\Omega$ is the conformal factor associated to $\psi$. The classes for this equivalence relation are denoted by $[(\mathfrak{S},\mathfrak{h},\mathfrak{t})]$, or for simplicity $[(\mathfrak{h},\mathfrak{t})]$.
\end{defi}

\begin{defi}
    \label{def:free_data_class}
    Assume $n$ is odd. The free data class of a connected component $\mathfrak{S}$ of $\mathfrak{I}$ is the class $[(\mathfrak{h},\mathfrak{t})]$ where $(\mathfrak{h},\mathfrak{t})$ is the free data on $\mathfrak{S}$ associated to any boundary defining function $x$ of $\mathfrak{I}$. 
\end{defi}

\subsubsection{Asymptotics of curvature tensors}
\label{sec:expansions_curvature_tensors}

Fix a FG gauge $(x,y^0,\dots,y^{n-1})$ on an open subset $\mathcal{U}$. Recall that the rescaled metric is in this case given by
\[ g := x^2 \widetilde{g} = dx^2 + \underline{\mathfrak{h}}(x) \,. \]
The goal of this section is to derive the asymptotics of the curvature tensors of $g$ near the conformal boundary from the FG expansion presented in the previous section. Remind that since $g$ extends to a $\mathcal{C}^2$ metric on $\M$, its Weyl and Schouten tensors are continuous up to the conformal boundary $\mathfrak{I}$.

\begin{prop}
    \label{prop:expansion_curvature_tensors}
    Assume $n$ is odd and the hypotheses of \Cref{thm:FG_exp} hold. Then the Weyl and Schouten tensor of the rescaled metric $g$ verify the following asymptotic developments in $x$ on any compact neighbourhood of a point $p \in \mathfrak{I} \cap \mathcal{U}$
    \begin{subequations}
        \begin{align}
            \label{eq:asymp_Schouten_frakt}
            L_{ij} &= \mathfrak{l}_{ij} - \frac{n}{2} x^{n-2} \mathfrak{t}_{ij}  + \grando{x^2} \,, \\
            L_{ix} &= \grando{x^3} \,, \\
            L_{xx} &= \frac{\mathfrak{l}_i{}^i}{n} + \grando{x^2} \,, \\
            \label{eq:asymp_Weyl_frakt}
            W^x{}_{ixj} &=  - \frac{n(n-2)}{2} x^{n-2} \mathfrak{t}_{ij}  +\grando{x^2} \,, \\
            \label{eq:asymp_Weyl_frakc}
            W^x{}_{ijk} &= x \mathfrak{c}_{jki} - nx^{n-1} \mathfrak{D}_{[j} \mathfrak{t}_{k]i} + \grando{x^3} \,, \\
            \label{eq:asymp_Weyl_frakw}
            W^i{}_{jkl} &= \mathfrak{w}^i{}_{jkl} + nx^{n-2} \left(\delta^i{}_{[k} \mathfrak{t}_{l]j} - \mathfrak{h}_{j[k}\mathfrak{t}_{l]i}\right) + \grando{x^2} \,,
        \end{align}
    \end{subequations}
	where $(\mathfrak{h},\mathfrak{t})$ is the free data associated to the boundary defining function $x$ and $\mathfrak{l}_{ij}$, $\mathfrak{c}_{ijk}$ and $\mathfrak{w}^i{}_{jkl}$ are respectively the Schouten, Cotton and Weyl tensors of $\mathfrak{h}$.
\end{prop}

\begin{rem}
    The coefficients of order $n-2$ in \eqref{eq:asymp_Schouten_frakt}, \eqref{eq:asymp_Weyl_frakt}, \eqref{eq:asymp_Weyl_frakw} and order $n-1$ in \eqref{eq:asymp_Weyl_frakc} are exact, even for odd dimensions $n\geq5$ for which these terms can be included in the $\mathcal{O}$.
\end{rem}

\begin{proof}
    Denote the Levi-Civita connection, Christoffel symbols, Riemman tensor, Ricci tensor and scalar curvature of $\underline{\mathfrak{h}}(x)$ respectively by $\underline{\mathfrak{D}}$, $\underline{\mathfrak{G}}_i{}^k{}_j$, $\underline{\mathfrak{r}}^i{}_{jkl}$, $\underline{\mathfrak{r}}_{ij}$, $\underline{\mathfrak{r}}$. The Christoffel symbols of the rescaled metric $g$ are given by
    \begin{align*}
        \Gamma_x{}^x{}_x &= 0 \,, & \Gamma_x{}^x{}_i &= 0 \,, & \Gamma_x{}^i{}_x &= 0 \,, \\
        \Gamma_i{}^x{}_j &= -\frac{1}{2} \partial_x \underline{\mathfrak{h}}_{ij} \,, & \Gamma_i{}^j{}_x &= \frac{1}{2} \underline{\mathfrak{h}}^{jk} \partial_x \underline{\mathfrak{h}}_{ik} \,, & \Gamma_i{}^k{}_j &= \underline{\mathfrak{G}}_i{}^k{}_j \,.
    \end{align*}
    Therefore the Riemann tensor of $g$ is
    \begin{subequations}
    	\label{formules_Riemann}
    \begin{align}
        R^x{}_{ixj} &= - \frac{1}{2} \partial_x^2 \underline{\mathfrak{h}}_{ij} + \frac{1}{4} \underline{\mathfrak{h}}^{kl} \partial_x \underline{\mathfrak{h}}_{ik}  \partial_x \underline{\mathfrak{h}}_{jl} \,, \\
        R^x{}_{ijk} &= -\frac{1}{2} \left( \underline{\mathfrak{D}}_j (\partial_x \underline{\mathfrak{h}}_{ki}) - \underline{\mathfrak{D}}_k (\partial_x \underline{\mathfrak{h}}_{ji}) \right) \,, \\
        R^i{}_{jkl} &= \underline{\mathfrak{r}}^i{}_{jkl} - \frac{1}{4} \underline{\mathfrak{h}}^{ip} \left( \partial_x \underline{\mathfrak{h}}_{pk} \partial_x \underline{\mathfrak{h}}_{jl} - \partial_x \underline{\mathfrak{h}}_{jk} \partial_x \underline{\mathfrak{h}}_{pl} \right) \,,
    \end{align}
	\end{subequations}
    where one defines
    \[ \underline{\mathfrak{D}}_j(\partial_x \underline{\mathfrak{h}}_{ki}) := \partial_j \partial_x \underline{\mathfrak{h}}_{ki} - \underline{\mathfrak{G}}_j{}^l{}_k \partial_x \underline{\mathfrak{h}}_{li} - \underline{\mathfrak{G}}_j{}^l{}_i \partial_x \underline{\mathfrak{h}}_{kl} \,. \]
    It follows that
    \begin{subequations}
    	\label{formules_Ricci}
    \begin{align}
        R_{xx} &= - \frac{1}{2} \underline{\mathfrak{h}}^{ij} \partial_x^2 \underline{\mathfrak{h}}_{ij} + \frac{1}{4} \underline{\mathfrak{h}}^{ij} \underline{\mathfrak{h}}^{kl} \partial_x \underline{\mathfrak{h}}_{ik} \partial_x \underline{\mathfrak{h}}_{jl} \,, \\
        R_{xi} &= \frac{1}{2} \left(\underline{\mathfrak{D}}^j \partial_x \underline{\mathfrak{h}}_{ji} - \underline{\mathfrak{D}}_i \left(\underline{\mathfrak{h}}^{jk}\partial_x \underline{\mathfrak{h}}_{jk}\right) \right) \,, \\
        R_{ij} &= \underline{\mathfrak{r}}_{ij} - \frac{1}{4} \partial_x \underline{\mathfrak{h}}_{ij} \left( \underline{\mathfrak{h}}^{kl} \partial_x \underline{\mathfrak{h}}_{kl} \right) + \frac{1}{2} \underline{\mathfrak{h}}^{kl} \partial_x \underline{\mathfrak{h}}_{ik} \partial_x \underline{\mathfrak{h}}_{jl} - \frac{1}{2} \partial_x^2 \underline{\mathfrak{h}}_{ij} \,, \\
        R &= \underline{\mathfrak{r}} - \frac{1}{4} \left(\underline{\mathfrak{h}}^{kl} \partial_x \underline{\mathfrak{h}}_{kl}\right)^2 + \frac{3}{4} \underline{\mathfrak{h}}^{ij} \underline{\mathfrak{h}}^{kl} \partial_x \underline{\mathfrak{h}}_{ik} \partial_x \underline{\mathfrak{h}}_{jl} - \underline{\mathfrak{h}}^{ij} \partial_x^2 \underline{\mathfrak{h}}_{ij} \,.
    \end{align}
	\end{subequations}
    From \eqref{formules_Ricci} and the definition of the Schouten tensor \eqref{def:schouten}, one deduces
    \begin{subequations}
        \label{formules_Schouten}
        \begin{align}
            L_{ij} &= \frac{1}{n-1} \Bigg( \underline{\mathfrak{r}}_{ij} - \frac{1}{2} \partial_x^2 \underline{\mathfrak{h}}_{ij} - \frac{1}{4} \partial_x \underline{\mathfrak{h}}_{ij} \left( \underline{\mathfrak{h}}^{kl} \partial_x \underline{\mathfrak{h}}_{kl} \right) + \frac{1}{2} \underline{\mathfrak{h}}^{kl} \partial_x \underline{\mathfrak{h}}_{ik} \partial_x \underline{\mathfrak{h}}_{jl} \nonumber \\
            &\quad - \frac{\underline{\mathfrak{h}}_{ij}}{2n}  \left( \underline{\mathfrak{r}} - \underline{\mathfrak{h}}^{kl} \partial_x^2 \underline{\mathfrak{h}}_{kl} - \frac{1}{4} \left( \underline{\mathfrak{h}}^{kl} \partial_x \underline{\mathfrak{h}}_{kl} \right)^2 + \frac{3}{4} \underline{\mathfrak{h}}^{pq} \underline{\mathfrak{h}}_{kl} \partial_x \underline{\mathfrak{h}}_{pk} \partial_x \underline{\mathfrak{h}}_{ql} \right) \Bigg) \,, \\
            L_{ix} &= \frac{1}{2(n-1)} \left(\underline{\mathfrak{D}}^j \partial_x \underline{\mathfrak{h}}_{ji} - \underline{\mathfrak{D}}_i \left(\underline{\mathfrak{h}}^{jk}\partial_x \underline{\mathfrak{h}}_{jk}\right) \right) \,, \\
            L_{xx} &= \frac{1}{2n(n-1)} \Big( - \underline{\mathfrak{r}} + \frac{1}{4} \left(\underline{\mathfrak{h}}^{kl} \partial_x \underline{\mathfrak{h}}_{kl}\right)^2 + \frac{2n-3}{4} \underline{\mathfrak{h}}^{ij} \underline{\mathfrak{h}}^{kl} \partial_x \underline{\mathfrak{h}}_{ik} \partial_x \underline{\mathfrak{h}}_{jl} \nonumber \\
            &\quad - (2n-1) \underline{\mathfrak{h}}^{ij} \partial_x^2 \underline{\mathfrak{h}}_{ij} \Big) \,.
        \end{align}
    \end{subequations}
    By combining the decomposition of the Riemann tensor \eqref{eq:decomposition_Riemann} with equations \eqref{formules_Riemann} and \eqref{formules_Schouten}, one has
    \begin{subequations}
        \label{asymp_Weyl}
        \begin{align}
            W^x{}_{ixj} &= \frac{1}{n-1} \tf_{\underline{\mathfrak{h}}(x)} \Big( - \underline{\mathfrak{r}}_{ij} - \frac{n-2}{2} \partial_x^2 \underline{\mathfrak{h}}_{ij} + \frac{(n-3)}{4} \underline{\mathfrak{h}}^{kl} \partial_x \underline{\mathfrak{h}}_{ik} \partial_x \underline{\mathfrak{h}}_{jl} + \frac{1}{4} \partial_x \underline{\mathfrak{h}}_{ij} \left(\underline{\mathfrak{h}}^{kl}\partial_x\underline{\mathfrak{h}}_{kl}\right)  \Big) \,, \\
            W^x{}_{ijk} &= - \underline{\mathfrak{D}}_{[j} \partial_x \underline{\mathfrak{h}}_{k]i} + \frac{1}{n-1} \underline{\mathfrak{h}}_{i[j} \left( \underline{\mathfrak{D}}^l \partial_x \underline{\mathfrak{h}}_{k]l} - \underline{\mathfrak{D}}_{k]} \left( \underline{\mathfrak{h}}^{lp} \partial_x \underline{\mathfrak{h}}_{lp} \right) \right) \,, \\
            W^i{}_{jkl} &= \underline{\mathfrak{r}}^i{}_{jkl} - \frac{2}{n-1} \Big( \delta^i{}_{[k} \left( \underline{\mathfrak{r}}_{l]j} - \frac{1}{2} \partial_x^2 \underline{\mathfrak{h}}_{l]j} \right) - \underline{\mathfrak{h}}^{ip} \underline{\mathfrak{h}}_{j[k} \left( \underline{\mathfrak{r}}_{l]p} - \frac{1}{2} \partial_x^2 \underline{\mathfrak{h}}_{l]p} \right) \nonumber \\
            &- \frac{1}{n} \delta^i{}_{[k} \underline{\mathfrak{h}}_{l]j} \left( \underline{\mathfrak{r}} - \underline{\mathfrak{h}}^{pq} \partial_x^2 \underline{\mathfrak{h}}_{pq} \right) \Big) + \mathcal{F}(\partial_x \underline{\mathfrak{h}}) \,,
        \end{align}
    \end{subequations}
    where $\tf_{\underline{\mathfrak{h}}(x)}$ is the trace-free part with respect to $\underline{\mathfrak{h}}(x)$ and $\mathcal{F}(\partial_x\underline{\mathfrak{h}})$ represents terms quadratic in $\partial_x \underline{\mathfrak{h}}$. By assumption, the FG expansion~\eqref{FG_exp_n_odd} holds. Using \eqref{eq:coeff_FG_expansion}, one has
    \begin{align*}
        \underline{\mathfrak{h}}_{ij} &= \mathfrak{h}^{(0)}_{ij} + \mathfrak{h}^{(2)}_{ij} x^2 + \mathfrak{h}^{(n)}_{ij} x^n + \grando{x^4} \\
        &= \mathfrak{h}_{ij} - \mathfrak{l}_{ij} x^2 + \mathfrak{t}_{ij} x^n + \grando{x^4} \,,
    \end{align*}
    and
    \begin{align*}
        \underline{\mathfrak{h}}^{ij} &= \mathfrak{h}^{(0)ij} - \mathfrak{h}^{(2)ij} x^2 - \mathfrak{h}^{(n)ij} x^n + \grando{x^4} \\
        &= \mathfrak{h}^{ij} + \mathfrak{l}^{ij} x^2 - \mathfrak{t}^{ij} x^n + \grando{x^4} \,,
    \end{align*}
    where the indices are raised with the metric $\mathfrak{h}$. Plugging these expansions in \eqref{formules_Schouten} gives
    \begin{align*}
        L_{ij} &= \frac{1}{n-1} \left( \mathfrak{r}_{ij} + \mathfrak{l}_{ij} - \frac{n(n-1)}{2} \mathfrak{t}_{ij} x^{n-2} - \frac{\mathfrak{r}+2\mathfrak{l}_k{}^k}{2n} \mathfrak{h}_{ij} \right) + \grando{x^2} \\
        &= \mathfrak{l}_{ij} - \frac{n}{2} \mathfrak{t}_{ij} x^{n-2} + \grando{x^2} \,, \\
        L_{ix} &= - \frac{x}{n-1} \mathfrak{c}_{ji}{}^j + \frac{nx^{n-1}}{2(n-1)} \left( \mathfrak{D}^j \mathfrak{t}_{ij} - \mathfrak{D}_i(\mathfrak{h}^{kl}\mathfrak{t}_{kl}) \right) + \grando{x^3} = \grando{x^3} \,, \\
            L_{xx} &= \frac{1}{2n(n-1)} \left(-\mathfrak{r}+4(n-1)\mathfrak{l}_i{}^i\right) + \grando{x^2} = \frac{\mathfrak{l}_i{}^i}{n}  + \grando{x^2}. 
    \end{align*}
    In the same way with \eqref{asymp_Weyl}, one finds
    \begin{subequations}
        \begin{align*}
            W^x{}_{ixj} &= \frac{1}{n-1} \tf_\mathfrak{h} \left( - \mathfrak{r}_{ij} + (n-2) \mathfrak{l}_{ij} - \frac{(n-2)n(n-1)}{2} \mathfrak{t}_{ij} x^{n-2} \right) + \grando{x^2} \\
            &= - \frac{n(n-2)}{2} x^{n-2} \mathfrak{t}_{ij} + \grando{x^2} \,, \\
            W^x{}_{ijk} &= x \mathfrak{c}_{jki} - nx^{n-1} \mathfrak{D}_{[j} \mathfrak{t}_{k]i} + \grando{x^3} \,, \\
            W^i{}_{jkl} &= \mathfrak{r}^i{}_{jkl} - 2 S_{j[k}{}^{ip} \mathfrak{l}_{l]p} + nx^{n-2} S_{j[k}{}^{ip} \mathfrak{t}_{l]p} + \grando{x^2} \\
            &= \mathfrak{w}^i{}_{jkl} + nx^{n-2} \left(\delta^i{}_{[k} \mathfrak{t}_{l]j} - \mathfrak{h}_{j[k}\mathfrak{t}_{l]i}\right) + \grando{x^2} \,.
        \end{align*}
    \end{subequations}
    Hence the results.
\end{proof}

\noindent One can generalise the results of \Cref{prop:expansion_curvature_tensors} by considering any boundary defining function (and thus leaving the FG gauge) up to the cost of restricting to the leading term of the expansions. This is particularly interesting for the Weyl components which is the object of the next corollary.

\begin{cor}
    Assume $n$ is odd and the hypotheses of \Cref{thm:FG_exp} hold. Let $x$ be a local boundary defining function of the conformal boundary $\mathfrak{I}$ on an open set $\mathcal{U}$ and $(\mathfrak{h},\mathfrak{t})$ be its associated free data. Then
    \begin{subequations}
        \label{eq:lim_weyl_frakwc}
        \begin{align}
            \left. \left( \lim_{x\to0^+} W \right) \right|_{T^\star \mathfrak{I} \times (T\mathfrak{I})^3} &= \mathfrak{w} \,, \\
            \left. \left( \lim_{x\to0^+} x^{-1} W(dx,.,.,.) \right) \right|_{(T\mathfrak{I})^3} &= \mathfrak{c} \,,
        \end{align}
    \end{subequations}
    where $W$ is the Weyl tensor of the rescaled metric $g := x^2\widetilde{g}$, $\mathfrak{w}$ and $\mathfrak{c}$ are the Weyl and Cotton tensors of $\mathfrak{h}$. Furthermore, if $n=3$ then one also has
    \begin{equation}
        \label{eq:lim_weyl_frakt}
        - \frac{2}{3} \left. \left(\lim_{x\to0^+} x^{-1} W(dx,.,\grad_g x,.) \right) \right|_{(T\mathfrak{I})^2} = \mathfrak{t} \,,
    \end{equation}
    where $\grad_g$ is the gradient with respect to the rescaled metric $g$.
\end{cor}

\begin{rems} \,
    \begin{itemize}
        \item In fact, the proofs of the limits \eqref{eq:lim_weyl_frakwc} hold for all dimensions $n \geq 2$.
        \item For $n\geq5$, $\mathfrak{t}$ is screened by lower order terms in \eqref{eq:asymp_Weyl_frakt} and thus cannot be retrieved as in \eqref{eq:lim_weyl_frakt}.
        \item Equation \eqref{eq:lim_weyl_frakt} is key to interpret geometrically certain boundary conditions for the local existence of aAdS spaces. It will be used in \Cref{sec:replacing_tensors}. On the other hand, equations \eqref{eq:lim_weyl_frakwc} are somehow already encoded in the extended conformal Einstein equations, see \Cref{sec:CVE_conf_boundary}. \qedhere
    \end{itemize}
\end{rems}

\begin{proof}
    Let $(x',y^0,\dots,y^{n-1})$ be a FG gauge on $\mathcal{U}$ and pose $\Theta := x/x'$. One has $\Theta\in \mathcal{C}^\infty(\mathcal{U},\R_+^\star)$ by \Cref{lemma:bdf}. Since the Weyl tensor is conformally invariant, the Weyl tensors $W$ and $W'$ of the rescaled metrics $g:=x^2\widetilde{g}$ and $g':=(x')^2\widetilde{g}$ coincide, and so do the Weyl tensors $\mathfrak{w}$ and $\mathfrak{w}'$ of $\mathfrak{h}$ and $\mathfrak{h}'$. Using \eqref{eq:asymp_Weyl_frakw}, one deduces
    \[ \left. \left( \lim_{x\to0^+} W \right) \right|_{T^\star \mathfrak{I} \times (T\mathfrak{I})^3} = \left. \left( \lim_{x'\to0^+} W' \right) \right|_{T^\star \mathfrak{I} \times (T\mathfrak{I})^3} = \mathfrak{w}' = \mathfrak{w} \,. \]
    Moreover, since $dx = x'd\Theta + \Theta dx'$,
    \begin{align*}
        x^{-1} W(dx,.,.,.) &= (x')^{-1}W'(dx',.,.,.) + \Theta^{-1}W'(d\Theta,.,.,.) \\
        &= \left(1+\frac{x'}{\Theta} \frac{\partial \Theta}{\partial x'}\right) (x')^{-1}W'(dx',.,.,.) + \frac{1}{\Theta}\frac{\partial \Theta}{\partial y^i} W'(dy^i,.,.,.) \,.
    \end{align*}
    With \eqref{eq:asymp_Weyl_frakc} and \eqref{eq:asymp_Weyl_frakw}, one deduces
    \[ \left. \left( \lim_{x\to0^+} x^{-1}W(dx,.,.,.) \right) \right|_{(T\mathfrak{I})^3} = \mathfrak{c}' + \mathfrak{w}'(d\ln\Omega,.,.,.) \,, \]
    where $\Omega := \Theta |_{\mathfrak{I}\cap\mathcal{U}}$. The transformation law of the Cotton tensor \eqref{transfo_cotton} gives in particular
    \[ \mathfrak{c} = \mathfrak{c}' + \mathfrak{w}'(d\ln\Omega,.,.,.) \,, \]
    since the Levi-Civita connections of $\mathfrak{h}'$ and $\mathfrak{h}$ are the Weyl connections associated to the covector fields $0$ and $d\ln\Omega$ with respect to $\mathfrak{h}'$. Hence the second limit.
    
    The limit \eqref{eq:lim_weyl_frakt} when $n=3$ is derived analogously. First, remark that
    \[ \grad_g x = \Theta^{-2} \grad_{g'} x = \Theta^{-1} \grad_{g'} x' + x' \Theta^{-2} \grad_{g'} \Theta \,. \]
    Thus
    \begin{align*}
        x^{-1} W(dx,.,\grad_g x,.) &= \left(1+\frac{x'}{\Theta} \frac{\partial\Theta}{\partial x'}\right) \Theta^{-1} (x')^{-1} W'(dx',.,\grad_{g'}x',.) \\
        &\quad + \left(1+\frac{x'}{\Theta} \frac{\partial\Theta}{\partial x'}\right) \Theta^{-2} W'(dx',.,\grad_{g'}\Theta,.) \\
        &\quad + \Theta^{-2} \frac{\partial\Theta}{\partial y^i}W'(dy^i,.,\grad_{g'} x',.) \\
        &\quad + \Theta^{-3} \frac{\partial\Theta}{\partial y^i} x' W'(dy^i,.,\grad_{g'}\Theta,.) \,.
    \end{align*}
    Thanks to the symmetries of the Weyl tensor and the asymptotic developments \eqref{eq:asymp_Weyl_frakt}-\eqref{eq:asymp_Weyl_frakw}, one deduces
    \[ \left. \left(\lim_{x\to0^+} x^{-1} W(dx,.,\grad_g x,.) \right) \right|_{(T\mathfrak{I})^2} = - \frac{3}{2} \Omega^{-1} \mathfrak{t}' + 0 + 0 + 0 \,. \]
    Yet $\mathfrak{t} := \Omega^{-1} \mathfrak{t}'$ by \Cref{def:extension_frakt} with $n=3$. Hence the result.
\end{proof}

\subsection{Examples of aAdS spaces}
\label{sec:ex_aAdS}

In this section, we illustrate the previous definitions by looking at known aAds spaces and describe their free data class on their conformal boundary. More particularly, apart from the pure AdS space, a few aAdS black holes and aAdS spaces arising from simple free data classes are presented.

This illustrative endeavour is especially helpful to gain certain insight on potential boundary stress-energy tensor and to provide non-trivial spaces satisfying our boundary conditions \eqref{eq:bc_robin}.

\subsubsection{Anti-de Sitter}
\label{sec:AdS}

Let us start, naturally, with the AdS space.

\begin{defi}
    The \emph{Anti-de Sitter (AdS) space} of dimension $n+1$ and radius $l>0$ is the smooth manifold without boundary $\M_{AdS} := \R^{n+1}$ equipped with the Lorentzian metric given in the spherical coordinate system $(t,r,\phi) \in \R \times \R_+^\star \times \mathbb{S}^{n-1}$ by
    \[ \widetilde{g}_{AdS}(l) := - f(r) dt^2 + f(r)^{-1} dr^2 + r^2 \dsphere{n-1}^2 \,, \]
    where
    \begin{equation*}
         f(r) := 1+\frac{r^2}{l^2} \,.
    \end{equation*}
\end{defi}

\begin{rems} \,
    \begin{itemize}
        \item $(\M_{AdS},\widetilde{g}_{AdS}(l))$ is solution to the (VE) with $\Lambda = -n(n-1)/(2l^2)$.
        \item One has $\widetilde{g}_{AdS}(l) = l^2 \widetilde{g}_{AdS}(1)$ by the change of variables $(t,r) \in \R \times \R_+^\star \mapsto (t/l,r/l) \in \R \times \R_+^\star$. In what follows, it will be assumed that $l=1$. \qedhere
    \end{itemize}
\end{rems}

\noindent With the change of variables $r \in \R_+^\star \mapsto \arctan r \in (0,\pi/2)$, of inverse $\psi \in (0,\pi/2) \mapsto \tan \psi \in \R_+^\star$, one obtains
\begin{equation}
    \label{eq_gAdS_gEC}
    \widetilde{g}_{AdS} = \frac{1}{\cos^2 \psi} \left( -dt^2 +\dsphere{n}^2 \right)
\end{equation}
where $\psi$ is the polar angle in the hyperspherical coordinates on the sphere $\mathbb{S}^n$, the north pole being identified to the point $r=0$. Thus, the AdS space is, up to a confomorphism, the open hemisphere $\{ \psi < \pi/2 \}$ of the Einstein cylinder $(\M_{EC},g_{EC})$ defined by \eqref{eq:EC}.

This confomorphism highlights the existence of a timelike conformal boundary, namely the equator $\{\psi = \pi/2\} \simeq \R \times \mathbb{S}^{n-1}$ of the Einstein cylinder, since the conformal factor $\Omega = \cos\psi$ of the confomorphism extends naturally by zero on it. In order to attach the conformal boundary, one extends the above coordinate system on the equator. Introducing the function $1-\cos\psi$ as a radius, this defines a smooth manifold with boundary $\mathcal{N}_{AdS} \simeq \R \times \overline{B}\vphantom{B}^n$, where $\overline{B}\vphantom{B}^n$ is the closed unit ball of $\R^n$. It follows that $(\mathcal{N}_{AdS},\partial\mathcal{N}_{AdS},\widetilde{g}_{AdS})$ is an aAdS space in the sense of \Cref{def_aAdS}. The conformal class $[\mathfrak{h}]$ on the $n$-dimensional conformal boundary $\partial\mathcal{N}_{AdS} \simeq \R \times \mathbb{S}^{n-1}$ is $[g_{EC}]$.

\bigbreak

One can easily construct a FG gauge for the AdS space. Define
\[ x(r) := \exp\left(-\int_0^r f(s)^{-1/2} ds\right) = \exp(-\sinh^{-1}(r)) = \frac{1}{r+\sqrt{f(r)}} \,. \]
The map $r \in \R_+^\star \mapsto x(r) \in (0,1)$ is a change of variables of inverse $x \in (0,1) \mapsto (1-x^2)/(2x) \in \R_+^\star$. Furthermore, $x$ is a boundary defining function of the conformal boundary since
\[ x = \frac{\cos\psi}{1+\sin\psi} \,. \]
It follows that
\[ x^2 \widetilde{g}_{AdS} = dx^2 - \frac{(1+x^2)^2}{4} dt^2 + \frac{(1-x^2)^2}{4} \dsphere{n-1}^2 \,. \]
Thus the chart $(x,t,\phi)$ is a FG gauge of associated free data $\mathfrak{h} = g_{EC}/4$ and $\mathfrak{t} = 0$. Consequently, the free data class of AdS is given by
\begin{equation}
    [(\mathfrak{h},\mathfrak{t})] = [(g_{EC},0)] \,.
\end{equation}

\subsubsection{Schwarzschild-Anti-de Sitter}
\label{sec:SAdS}

Let us turn to aAdS black holes, the simplest of which is the Schwarzschild-Anti-de Sitter space. As we are interested in the free data on the conformal boundary, only the exterior region is presented.

\begin{defi}
    The exterior of the \emph{Schwarzschild-Anti-de Sitter space} of dimension $n+1\geq 3$, radius $l>0$ and mass $m>0$ is the smooth manifold without boundary $\M_{SAdS} := \R \times (r_\star,+\infty) \times \mathbb{S}^{n-1}$ equipped with the Lorentzian metric given in the spherical coordinate system $(t,r,\phi) \in \R \times (r_\star,+\infty) \times \mathbb{S}^{n-1}$ by
    \[ \widetilde{g}_{SAdS}(l,m) := - f_m(r) dt^2 + f_m(r)^{-1} dr^2 + r^2 \dsphere{n-1}^2 \,, \]
    where
    \begin{equation*}
        f_m(r) := 1+\frac{r^2}{l^2} - \frac{2m}{r^{n-2}}
    \end{equation*}
    and $r_\star = r_\star(n,l,m) > 0$ is the only real positive root of the polynomial $X^n + l^2 X^{n-2} - 2m l^2$ (assuming that $m>1/2$ if $n=2$).
\end{defi}

\begin{rems} \,
    \begin{itemize}
        \item $(\M_{SAdS},\widetilde{g}_{SAdS}(l,m))$ is solution to the (VE) with $\Lambda = -n(n-1)/(2l^2)$.
        \item One has $\widetilde{g}_{SAdS}(l,m) = l^2 \widetilde{g}_{SAdS}(1,m/l^{n-2})$ by the change of variables $(t,r) \in \R \times \R_+^\star \mapsto (t/l,r/l) \in \R \times \R_+^\star$. In what follows, it will be assumed that $l=1$. \qedhere
    \end{itemize}
\end{rems}

Contrary to AdS, no natural confomorphism arises. Let us imitate the construction of a FG gauge for AdS conducted above by introducing the new variable
\begin{equation*}
    x_m(r) := \exp \left(-\int_{r_\star}^r f_m(s)^{-1/2} ds \right) \,.
\end{equation*}
This is a well defined decreasing diffeomorphism from $(r_\star,+\infty)$ to $(0,1)$.

\begin{ex}
    One have explicit formulae when
    \begin{itemize}
        \item $n=2$
        \[ x_m(r) = \exp\left(-\left[\ln\left(\sqrt{f_m(r)} + r\right)\right]_{r_\star}^r\right) = \frac{r_\star}{\sqrt{f_m(r)}+r} \,, \]
        of inverse function
        \[ r(x_m) = r_\star \frac{1+x_m^2}{2x_m} \,, \]
        \item $n=4$
        \[ x_m(r) = \exp\left(-\left[\frac{1}{2} \ln\left(2r\sqrt{f_m(r)} + 2r^2+1\right)\right]_{r_\star}^r\right) = \sqrt{\frac{2r_\star^2+1}{2r\sqrt{f_m(r)}+2r^2+1}} \,, \]
        of inverse function
        \[ r(x_m) = \frac{2mx_m^4+\sqrt{(2mx_m^4)^2+(2r_\star^2+1)x_m^2(2r_\star^2+1-x_m^2)^2}}{2(2r_\star^2+1)x_m^2} \,. \qedhere \]
    \end{itemize}
\end{ex}

\noindent Since
\[ \frac{x_m(r)}{x(r)} = \exp\left( - \int_0^{r_\star} f(s)^{-1/2} ds \right) \exp\left( - \int_{r_\star}^r f_m(s)^{-1/2} - f(s)^{-1/2} ds \right) \,, \]
and
\[ 0 \leq f_m(s)^{-1/2} - f(s)^{-1/2} \; \underset{+\infty}{\sim} \; m s^{-(n+1)} \,, \]
one deduces that
\[ \lim_{r\to+\infty} \frac{x_m(r)}{x(r)} \in \R_+^\star \,. \]
Thus $x_m$ and $x$ have the same behaviour at infinity, which is not surprising since Schwarzschild-Anti-de Sitter can be seen as a fast decaying perturbation of AdS.

This leads us to extend $\M_{SAdS} \simeq \R \times (0,1) \times \mathbb{S}^{n-1}$ into the smooth manifold with boundary $\mathcal{N}_{SAdS} := \R \times [0,1) \times \mathbb{S}^{n-1}$ whose differential structure at the boundary $\partial\mathcal{N}_{SAdS} \simeq \R \times \mathbb{S}^{n-1}$ is imposed by requiring that $x_m$ is a smooth function. Then $(\mathcal{N}_{SAdS},\partial\mathcal{N}_{SAdS},\widetilde{g}_{SAdS})$ is an aAdS space in the sense of \Cref{def_aAdS}. The conformal class on the boundary is the same as the one of AdS, that is $[\mathfrak{h}] = [g_{EC}]$. Furthermore, the global chart $(x_m,t,\phi)$ is a FG gauge.

\bigbreak

Let us now determine the boundary stress-energy tensor $\mathfrak{t}$. Identifying it from the definition of the metric $\widetilde{g}_{SAdS}$ in the FG gauge $(x_m,t,\phi)$ would require an expansion of $r$ and $f_m(r)$ in term of $x_m$. We will instead use the more convenient asymptotic expansion \eqref{eq:asymp_Weyl_frakt} which holds for $n$ odd. The relevant components of the Weyl tensor $\widetilde{W}$ of $\widetilde{g}_{SAdS}$ in the spherical coordinates $(t,r,\phi)$ are given by
\begin{align*}
    \widetilde{W}^r{}_{trt} &= \frac{-(n-1)(n-2)m}{r^n} f_m(r) \,, \\
    \widetilde{W}^r{}_{trA} &= 0 \,, \\
    \widetilde{W}^r{}_{ArB} &= \frac{-(n-2)m}{r^n} r^2 (\dsphere{n-1}^2)_{AB} \,,
\end{align*}
where $A,B$ correspond to the sphere indices. Note that the smooth function $1/r = 2x/(1-x^2)$ is a boundary defining function of the conformal boundary. However it does not verify \eqref{eq:FG_bdf} and thus \eqref{eq:asymp_Weyl_frakt} does not directly apply. The key point to overcome this difficulty is that the leading term in the above expressions for the Weyl tensor are of order $n-2$ in $1/r$. Hence the term of order $n-2$ must be the leading term in the expansion \eqref{eq:asymp_Weyl_frakt} for all odd $n$. It follows that \eqref{eq:lim_weyl_frakt} holds for all odd $n$ up to replacing the factor $2/3$ by $2/(n(n-2))$. Therefore, the boundary stress-energy tensor associated to $1/r$ for $n$ odd is
\begin{equation}
    \label{eq:frakt_SAdS}
    \mathfrak{t}= \frac{2m}{n} \Big( (n-1)dt^2 +\dsphere{n-1}^2 \Big) \,.
\end{equation}
Consequently, for $n$ odd, the free data class of the Schwarzschild-Anti-de Sitter space is given by
\begin{equation}
    [(\mathfrak{h},\mathfrak{t})] = \left[\left(g_{EC},\frac{2m}{n} \Big( (n-1)dt^2 +\dsphere{n-1}^2 \Big) \right)\right] \,.
\end{equation}
Let us mention that the first $n$ terms of the FG expansion are explicitly given by \cite[Equation (3.93)]{S21} for all dimensions. In particular, it shows that \eqref{eq:frakt_SAdS} remains valid for all dimensions $n\geq3$. 

\subsubsection{Birmingham-Anti-de Sitter}
\label{sec:BAdS}

The Birmingham spaces\footnote{I would like to thanks P. Chru\'sciel for pointing out the existence of these spaces to my advisor.} form a generalisation of the Schwarzschild spaces. They have been introduced by Birmingham in \cite{B99}, see also \cite[Section 5.5]{C20}. 

\begin{defi}
    Let $(\mathcal{P},\slashed{h})$ be an Einstein Riemannian manifold of dimension $n-1\geq2$, whose (constant) scalar curvature is denoted by $\slashed{r}$. The exterior of the associated \emph{Birmingham-Anti-de Sitter space} of dimension $n+1$, radius $l>0$ and mass $m\geq0$ is the smooth manifold without boundary $\M_{BAdS} := \R \times (r_\star,+\infty) \times \mathcal{P}$ equipped with the Lorentzian metric given by
    \[ \widetilde{g}_{BAdS}(l,m,\slashed{h}) := - f_{m,\slashed{h}}(r) dt^2 + f_{m,\slashed{h}}(r)^{-1} dr^2 + r^2 \slashed{h} \,, \]
    where
    \begin{equation*}
        f_{m,\slashed{h}}(r) := \frac{\slashed{r}}{(n-2)(n-1)}+\frac{r^2}{l^2} - \frac{2m}{r^{n-2}}
    \end{equation*}
    and $r_\star$ is the only real positive root of the polynomial
    \[ X^n + \frac{\slashed{r} l^2}{(n-2)(n-1)} X^{n-2} - 2m l^2 \,. \]
\end{defi}

\begin{rems} \,
    \begin{itemize}
        \item $(\M_{BAdS},\widetilde{g}_{BAdS}(l,m,\slashed{h}))$ is solution to the (VE) with $\Lambda = -n(n-1)/(2l^2)$.
        \item For all $k >0$ and $p,q \in \R$, one has
        \[ \widetilde{g}_{BAdS}(l,m,\slashed{h}) = k^{2p} \widetilde{g}_{BAdS}(k^{-p} l, k^{2q-(p+q)(n-2)} m, k^{2q} \slashed{h}) \]
        with the change of variables $(t,r) \in \R \times \R_+^\star \mapsto (t/k^{p+q},r/k^{p+q}) \in \R \times \R_+^\star$. In what follows, it will be assumed that $l=1$ and $\slashed{r} = \text{sign}(\slashed{r})(n-2)(n-1)$. \qedhere
    \end{itemize}
\end{rems}

In a similar way to Schwarzschild-Anti-de Sitter, one defines a function $x_{m,\slashed{h}}$ with values into $(0,1)$ and extend the manifold into the smooth manifold with boundary $\mathcal{N}_{BAdS} := \R \times [0,1) \times \mathcal{P}$. Then one obtains an aAdS space in the sense of \Cref{def_aAdS}. Using the boundary defining function $1/r$ such as for Schwarzschild-Anti-de Sitter, one finds
\begin{subequations}
    \begin{align}
        \mathfrak{h} &= -dt^2 + \slashed{h} \,, \\
        \label{eq:frakt_BAdS}
        \mathfrak{t} &= \frac{2m}{n} \Big( (n-1) dt^2 + \slashed{h} \Big) \,.
    \end{align}
\end{subequations}
Consequently, for $n$ odd, the free data class of Birmingham-Anti-de Sitter on its conformal boundary $\mathfrak{I} \simeq \R \times \mathcal{P}$ is given by
\begin{equation}
    [(\mathfrak{h},\mathfrak{t})] = \left[\left(-dt^2+\slashed{h},\frac{2m}{n} \Big( (n-1)dt^2 +\slashed{h} \Big) \right)\right] \,.
\end{equation}

\subsubsection{Kerr-Anti-de Sitter}
\label{sec:KAdS}

The Kerr black holes are a generalisation of the Schwarzschild black holes by allowing them to rotate. We will only consider the case $n=3$.

\begin{defi}
    The exterior of the \emph{Kerr-Anti-de Sitter space} of dimension $4$, radius $l>0$, mass $m>0$ and angular momentum per mass $a \in \R$ with $|a| < l$ is the smooth manifold without boundary $\M_{KAdS} := \R \times (r_\star,+\infty) \times \mathbb{S}^2$ equipped with the Lorentzian metric given in the Boyer-Lindquist chart $(t,r,\theta,\varphi) \in \R \times (r_\star,+\infty) \times (0,\pi) \times (0,2\pi)$ by
    \begin{align*}
        \widetilde{g}_{KAdS}(l,m,a) &:= \rho^2 \left( \frac{dr^2}{\Delta_r} + \frac{d\theta^2}{\Delta_\theta} \right) + \frac{\Delta_\theta \sin^2\theta}{\rho^2 \Xi^2} \left( adt - (r^2+a^2) d\varphi\right)^2 \\
        &\quad - \frac{\Delta_r}{\rho^2 \Xi^2} \left( dt - a\sin^2\theta d\varphi \right)^2 \,,
    \end{align*}
    where
    \begin{align*}
        \Delta_r &:= (r^2+a^2)\left(1+\frac{r^2}{l^2}\right) - 2mr \,, & \qquad \Delta_\theta &:= 1-\frac{a^2}{l^2} \cos^2\theta \,, \\
        \rho^2 &:= r^2 + a^2 \cos^2\theta \,, & \qquad \Xi &:= 1-\frac{a^2}{l^2} \,.
    \end{align*}
    and $r_\star$ is the greatest real positive root of the polynomial $\Delta_r$ if it admits at least one or $0$ otherwise.
\end{defi}

\begin{rems} \,
    \begin{itemize}
        \item $(\M_{KAdS},\widetilde{g}_{KAdS}(l,m,a))$ is solution to the (VE) with $\Lambda = -3/l^2$.
        \item One has $\widetilde{g}_{KAdS}(l,m,a) = l^2 \widetilde{g}_{KAdS}(1,m/l,a/l)$ by the change of variables $(t,r) \in \R \times \R_+^\star \mapsto (t/l,r/l) \in \R \times \R_+^\star$. In what follows, it will be assumed that $l=1$. \qedhere
    \end{itemize}
\end{rems}

\noindent With the benefit of hindsight, let us use $1/r$ as a boundary defining function and attach to the manifold the cylinder at infinity $\mathfrak{I} \simeq \R \times \mathbb{S}^2$. The free data associated to $1/r$ is given by
\begin{subequations}
    \begin{align}
        \label{eq:frakh_KAdS}
        \mathfrak{h} &= - \frac{1}{\Xi^2} \left(dt-a \sin^2\theta d\varphi \right)^2 + \frac{d\theta^2}{\Delta_\theta} + \frac{\sin^2\theta \Delta_\theta}{\Xi^2} d\varphi^2  \,, \\
        \label{eq:frakt_KAdS}
        \mathfrak{t} &= \frac{2m}{3} \left( \frac{2}{\Xi^2} \left( dt -a\sin^2\theta d\varphi \right)^2 + \frac{d\theta^2}{\Delta_\theta}  + \frac{\sin^2\theta\Delta_\theta}{\Xi^2} d\varphi^2 \right) \,.
    \end{align}
\end{subequations}
Equation~\eqref{eq:frakh_KAdS} follows straightforwardly from the expression of the metric $\widetilde{g}_{KAdS}$. Note that $\mathfrak{h}$ is locally conformally flat but not Einstein. We used Mathematica to compute the Weyl tensor of Kerr-Anti-de Sitter and take the limit \eqref{eq:lim_weyl_frakt} in order to obtain \eqref{eq:frakt_KAdS}.

\subsubsection{Real analytic free data}
\label{sec:real_analytic_free_data}

Another way to construct aAdS spaces is to determine cases for which the FG expansion given in \Cref{thm:FG_exp} converges. The simplest case is when the free data $(\mathfrak{h},\mathfrak{t})$ is real analytic. For $n$ odd, this implies that the formal series converges and that the resulting metric is real analytic as well in a neighbourhood of the boundary. One can refer to \cite[Theorem 2.3]{FG85} in the case $\mathfrak{t}=0$ and to \cite[Theorem 4.8 and the following remarks]{FG12} in the general case. For $n$ even, the convergence of the formal series, which contains logarithmic terms, for real analytic free data can be derived from \cite{K04}.

Let us also mention the work of LeBrun \cite{LB82}, previous to Fefferman and Graham \cite{FG85}, in which the holomorphic case for $n=3$ and with the additional assumption of self-duality is investigated. This last assumption amounts to set the boundary stress-energy tensor proportional to the Cotton-York tensor by a fixed imaginary number by adapting the discussion in \cite[Section 5]{FG12} to the Lorentzian case. Therefore, it can be seen as a special case of our family of boundary conditions \eqref{eq:bc_robin} in a complex setting.

\subsubsection{Conformally Einstein or locally conformally flat boundaries}
\label{sec:special_boundaries}

Two remarkable cases, discussed for example in \cite[Section 7]{FG12}, give rise to a finite FG expansion. First, consider a smooth Lorentzian manifold $(\mathfrak{I},\mathfrak{h})$ of dimension $n\geq3$ which is solution to the (VE) for a cosmological constant $\Lambda \in \R$. Define $\mathcal{M} := [0,x_{\max}) \times \mathfrak{I}$ where
\[ x_{\max} = \begin{cases}
                \sqrt{\frac{2(n-1)(n-2)}{\Lambda}} & \text{if } \Lambda > 0 \,, \\
                \quad \quad \; +\infty & \text{otherwise} \,,
            \end{cases} \]
and equip the interior $\mathring{\mathcal{M}} = (0,x_{\max}) \times \mathfrak{I}$ with the smooth Lorentzian metric
\[ \widetilde{g} := \frac{dx^2 + \underline{\mathfrak{h}}(x)}{x^2} \quad \text{where} \quad \underline{\mathfrak{h}}(x) = \left(1-\frac{\Lambda x^2}{2(n-1)(n-2)} \right)^2 \mathfrak{h} \,. \]
By construction, $(\mathcal{M},\mathfrak{I},\widetilde{g})$ is an aAdS space. Its free data class is $[(\mathfrak{h},0)]$.

\bigbreak

Secondly, consider a smooth Lorentzian conformal structure $(\mathfrak{I},[\mathfrak{h}])$ of dimension $n\geq3$ which is locally conformally flat. Fix a representative $\mathfrak{h}$ of the conformal class whose Schouten tensor is denoted by $\mathfrak{l}_{ij}$. Define the following Lorentzian metric
\[ \widetilde{g} := \frac{dx^2 + \underline{\mathfrak{h}}(x)}{x^2} \quad \text{where} \quad \underline{\mathfrak{h}}\vphantom{h}_{ij}(x) = \mathfrak{h}_{ij}-\mathfrak{l}_{ij}x^2 + \frac{1}{4} \mathfrak{l}_i{}^k \mathfrak{l}_{kj} x^4  \,, \]
on a neighbourhood $\M$ of $\{0\}\times \mathfrak{I}$ in $\R_+ \times \mathfrak{I}$. Then $(\mathcal{M},\mathfrak{I},\widetilde{g})$ is an aAdS space and its free data class is $[(\mathfrak{h},0)]$. It was shown in \cite{SS00} that, in this case, $\widetilde{g}$ is locally conformally flat as well.

\begin{rem}
    Since $\widetilde{g}$ is locally conformally flat, the Weyl tensor of $\widetilde{g}$ (and thus of any rescaled metric $g$) vanishes identically. This implies that if one introduces a non-zero boundary stress-energy tensor $\mathfrak{t}$, the leading term of the asymptotic expansion \eqref{eq:asymp_Weyl_frakt} of $W^x{}_{ixj}$ is only due to $\mathfrak{t}$. Hence, \eqref{eq:lim_weyl_frakt} holds for all odd dimensions $n\geq3$ if the factor $2/3$ is replaced by $2/(n(n-2))$. This apply to the Schwarzschild-Anti-de Sitter and Birmingham-Anti-de Sitter spaces, already presented in \Cref{sec:SAdS} and \Cref{sec:BAdS}. 
\end{rem}

\subsubsection{Conformally pp-wave boundaries}
\label{sec:ppwave}

We recall here results obtained by Anderson, Leistner and Nurowski in \cite{LN10,ALN20}. Consider a smooth Lorentzian conformal structure $(\mathfrak{I},[\mathfrak{h}])$ of dimension $n\geq3$ which contains a pp-wave metric, that is to say such that there exists a representative $\mathfrak{h}$ and a smooth null vector field $\mathfrak{K}$ on $\mathfrak{I}$ such that
\[ \mathfrak{D}_i \mathfrak{K}_j = 0 \,. \]
Note that this is equivalent to say that $\mathfrak{K}$ is a Killing field for $\mathfrak{h}$ (and thus a conformal Killing field for the conformal structure) which is locally a gradient. The pp-wave representative $\mathfrak{h}$ can be written in a local coordinate system $(u,y,z^1,\dots,z^{n-2})$ under the form
\begin{equation}
    \mathfrak{h} = 2\mathfrak{f}(u,z) du^2 + du \otimes dy + dy \otimes du + dz^2 \,,
\end{equation}
where $\mathfrak{f}$ is a smooth function of $u$ and $z$. With these coordinates, one has $\mathfrak{K} := \partial_y$. The inverse metric is given by
\[ \mathfrak{h}^{-1} = \partial_u \otimes \partial_y + \partial_y \otimes \partial_u - 2\mathfrak{f}(u,z) \partial_y \otimes \partial_y + \partial_z \otimes \partial_z \,. \]
By \cite[Section 3]{LN10}, the only non-vanishing components of the Ricci, Cotton, Weyl and Bach tensors of $\mathfrak{h}$ are given by
\begin{alignat*}{5}
    \mathfrak{r}_{uu} &= - \Delta_z \mathfrak{f} \,, &\quad\quad\quad& \mathfrak{c}_{Auu} &= -\frac{\partial_{z^A} \Delta_z \mathfrak{f}}{n-2} \,, \\
    \mathfrak{w}_{AuBu} &= \partial_{z^A} \partial_{z^B} \mathfrak{f} - \frac{\Delta_z \mathfrak{f}}{n-2} \delta_{AB} \,, && \mathfrak{b}_{uu} &= -\frac{\Delta^2_z \mathfrak{f}}{n-2} \,.
\end{alignat*}
where $\Delta_z$ is the Laplacian in the $z = (z^1,\dots,z^{n-2})$ coordinate and the indices $A$, $B$ ranges from $1$ to $n-2$.

In this setting, the FG expansions simplify greatly. For $n$ odd, the formal series \eqref{FG_exp_n_odd} reduces to
\begin{equation}
    \underline{\mathfrak{h}}(x) := 2\mathfrak{F}(u,z,x) du^2 + du \otimes dy + dy \otimes du + dz^2 \,,
\end{equation}
where
\begin{equation}
    \label{eq:ppwave}
    \mathfrak{F}(u,z,x) := \sum_{k=0}^\infty \frac{\Delta_z^{2k} \mathfrak{f}}{k! \, 2^k \prod_{j=1}^k (n-2j)} x^{2k} + \sum_{k=0}^\infty \frac{(-1)^k \Delta_z^{2k} \mathfrak{a}}{k!2^k\prod_{j=1}^k (n+2j)} x^{n+2k}
\end{equation}
with $\mathfrak{a}$ a smooth function of $u$ and $z$, see \cite[Theorem 3.1]{ALN20}. This function $\mathfrak{a}$ determines entirely the boundary stress-energy tensor through
\begin{equation}
    \mathfrak{t} = 2\mathfrak{a} \, du^2 \,,
\end{equation}
or in covariant form
\begin{equation}
    \label{eq:frakt_ppwave}
    \mathfrak{t} = 2\mathfrak{a} \; \mathfrak{K} \otimes \mathfrak{K} \,.
\end{equation}
In particular, if $\mathfrak{f}$ and $\mathfrak{a}$ are real analytic functions then the formal series \eqref{eq:ppwave} converges and the resulting function $\mathfrak{F}$ is real analytic as well. This is a special case of \Cref{sec:real_analytic_free_data}. One also retrieves subcases of \Cref{sec:special_boundaries} for special choices of $\mathfrak{f}$. For example, $\mathfrak{h}$ Einstein is equivalent to $\Delta_z \mathfrak{f} = 0$ and $\mathfrak{h}$ locally conformally flat is equivalent for $n=3$ to $\partial_z^3 \mathfrak{f} = 0$.

Outside the real analytic class, one can also construct various aAdS spaces with a pp-wave conformal boundary. For instance, consider the functions $\mathfrak{f}$ of the form $\mathfrak{f}(u,z) = \mathfrak{P}(z)\mathfrak{e}(u)$ with $\mathfrak{P}$ a polynomial function and $\mathfrak{e}$ a smooth function. Then, taking $\mathfrak{P}$ of order 3 and $\mathfrak{a}=0$ yields an aAds space with a conformally Bach-flat but not locally conformally flat boundary. Imposing $\mathfrak{a}(u,z) = \mu' (\partial_z^3 \mathfrak{P})(z) \mathfrak{e}(u)$ with $\mu'\in\R$ yields aAdS spaces verifying our boundary conditions \eqref{eq:bc_robin}.

\subsection{Boundary stress-energy tensors}
\label{sec:BSET}

This section focuses on boundary stress-energy tensors. First, a few potential boundary stress-energy tensor classes are inferred in \Cref{sec:ex_BSET} from the examples of aAdS spaces presented in \Cref{sec:ex_aAdS}. This could represent a first step towards treating the geometric \emph{inhomogeneous Robin} boundary conditions. In particular, some of these potential boundary stress-energy tensors arise from conformal Killing fields of the boundary. The interactions between boundary stress-energy tensors and conformal Killing fields are then highlighted in the general case in \Cref{sec:symmetries_currents_bdry}. 
 
\subsubsection{Examples of boundary stress-energy tensor classes}
\label{sec:ex_BSET}

Let $(\mathfrak{S},[\mathfrak{h}])$ be a smooth Lorentzian conformal structure of dimension $n\geq3$. The goal of this section is to construct boundary stress-energy tensor candidates for the conformal structure $(\mathfrak{S},[\mathfrak{h}])$, that is to say to construct trace-free and divergence-free symmetric 2-tensor fields which are conformally invariant of weight $n-2$ (in contravariant form). Let us begin with our family of boundary conditions \eqref{eq:bc_robin}. 

\begin{lem}
    \label{lem:BSET_robin}
    Assume that $n=3$ and $(\mathfrak{S},[\mathfrak{h}])$ is oriented. Then for all $\mu \in \R$, one can construct a boundary stress-energy tensor class by defining for any representative $\mathfrak{h} \in [\mathfrak{h}]$
    \begin{equation}
        \label{eq:BSET_cotton_york}
        \mathfrak{t}_{ij} := \mu \, \mathfrak{y}_{ij} \,,
    \end{equation}
    where $\mathfrak{y}_{ij}$ is the Cotton-York tensor of $\mathfrak{h}$.
\end{lem}

\begin{rem}
    If $\mu$ is not constant then one has to require that
    \[ \mathfrak{y}_{ij} \mathfrak{D}^j \mu = 0 \]
    for \eqref{eq:BSET_cotton_york} to define a suitable boundary stress-energy tensor candidate. This is a non-trivial condition on the Cotton-York tensor, and thus on the conformal class $[\mathfrak{h}]$. In particular, we expect that the Cauchy problem with the geometric boundary condition \eqref{eq:BSET_cotton_york} for a non-constant $\mu$ is ill-posed, since it would both impose the boundary stress-energy tensor and restrict the conformal class. This is coherent with the fact that \eqref{eq:BSET_cotton_york} with a non-constant $\mu$ cannot be propagated on the conformal boundary as in \Cref{sec:homogeneous_robin}.
\end{rem}

\begin{proof}
    This is a suitable choice thanks to  \Cref{prop:cottonyork}.
\end{proof}

\noindent The two following classes give rise to fluid-like and null-dust-like boundary stress-energy tensors.

\begin{lem}
    Let $\mathfrak{X}$ be a smooth vector field on $\mathfrak{S}$. If the equation
    \begin{equation}
        \label{eq:diff_frakX}
        \mathfrak{D}_{\mathfrak{X}} \mathfrak{X} + \frac{1}{2} \grad_{\mathfrak{h}} \mathfrak{h}(\mathfrak{X},\mathfrak{X}) = \frac{2\diver \mathfrak{X}}{n} \mathfrak{X}
    \end{equation}
    holds for a representative $\mathfrak{h} \in [\mathfrak{h}]$ then it holds for any representative. Furthermore, if equation~\eqref{eq:diff_frakX} is verified and $\mathfrak{X}$ is nowhere null then one can construct a boundary stress-energy tensor class by defining for any representative $\mathfrak{h} \in [\mathfrak{h}]$
    \begin{equation}
        \label{eq:BSET_fluid}
        \mathfrak{t}^{ij} := \frac{\pm1}{|\mathfrak{h}(\mathfrak{X},\mathfrak{X})|^{n/2+1}} \left( \mathfrak{X}^i \mathfrak{X}^j - \frac{\mathfrak{h}(\mathfrak{X},\mathfrak{X})}{n} \mathfrak{h}^{ij} \right) \,.
    \end{equation}
\end{lem}

\begin{rems} \,
    \begin{itemize}
        \item Equation~\eqref{eq:diff_frakX} implies
        \begin{equation}
            \label{eq:diff_frakX_aux}
            \mathfrak{X} \cdot \mathfrak{h}(\mathfrak{X},\mathfrak{X}) = \frac{2 \diver \mathfrak{X}}{n} \mathfrak{h}(\mathfrak{X},\mathfrak{X}) \,.
        \end{equation}
        \item In particular, all conformal Killing fields of $(\mathfrak{S},[\mathfrak{h}])$ verify \eqref{eq:diff_frakX}.
        \item Since $\mathfrak{X}$ is nowhere null, a representative $\mathfrak{h}_{ij}$ and its associated boundary stress-energy tensor $\mathfrak{t}_{ij}$ defined by \eqref{eq:BSET_fluid} are co-orthogonalisable in the sense of symmetric bilinear forms.
        \item The boundary stress-energy tensors of Schwarzschild-Anti-de Sitter \eqref{eq:frakt_SAdS}, Bir-\\mingham-Anti-de Sitter \eqref{eq:frakt_BAdS} and Kerr-Anti-de Sitter \eqref{eq:frakt_KAdS} are of the form \eqref{eq:BSET_fluid} with $\mathfrak{X}=\sqrt{\frac{2m(n-1)}{n}}\partial_t$. \qedhere
    \end{itemize}
\end{rems}

\begin{proof}
    Equation~\eqref{eq:diff_frakX} can be rewritten as
    \[ Z_i{}^{jk}{}_l \mathfrak{X}^i \mathfrak{D}_k \mathfrak{X}^l = 0 \,, \]
    where $Z_i{}^{jk}{}_l$ is defined by \eqref{def_Z}. The conformal invariance of the equation is thus similar to that of the conformal Killing equation, see \Cref{sec:conformal_Killing_fields}.

    By definition, the tensor field $\mathfrak{t}$ is symmetric, trace-free and is conformally invariant of weight $n-2$. Let us compute its divergence using \eqref{eq:diff_frakX} and \eqref{eq:diff_frakX_aux}
    \begin{align*}
        \pm |\mathfrak{h}(\mathfrak{X},\mathfrak{X})|^{n/2+1} \mathfrak{D}_i \mathfrak{t}^{ij} &= \mathfrak{D}_\mathfrak{X} \mathfrak{X}^j + \frac{1}{2} \mathfrak{D}^j \mathfrak{h}(\mathfrak{X},\mathfrak{X}) + \left(-\frac{n+2}{2} \frac{\mathfrak{X} \cdot \mathfrak{h}(\mathfrak{X},\mathfrak{X})}{\mathfrak{h}(\mathfrak{X},\mathfrak{X})} + \diver \mathfrak{X} \right) \mathfrak{X}^j \\
        &= - \frac{n+2}{2} \left( \frac{\mathfrak{X} \cdot \mathfrak{h}(\mathfrak{X},\mathfrak{X})}{\mathfrak{h}(\mathfrak{X},\mathfrak{X})} - \frac{\diver \mathfrak{X}}{n} \right) \mathfrak{X}^j \\
        &= 0 \,. \qedhere
    \end{align*}
\end{proof}

\begin{lem}
    Let $(\mathfrak{X},\mathfrak{a})$ be a pair constituted of a smooth null vector field $\mathfrak{X}$ and a smooth scalar function $\mathfrak{a}$ on $\mathfrak{S}$ verifying
    \begin{equation}
        \label{eq:diff_frakX_fraka}
        \mathfrak{a} \, \mathfrak{D}_\mathfrak{X} \mathfrak{X} + \diver(\mathfrak{a}\mathfrak{X}) \mathfrak{X} = 0 \,,
    \end{equation}
    for the Levi-Civita connection $\mathfrak{D}$ of a representative $\mathfrak{h} \in [h]$. Equation \eqref{eq:diff_frakX_fraka} is verified for all representatives of the conformal class $[\mathfrak{h}]$ if one assumes that $\mathfrak{X}$ and $\mathfrak{a}$ are conformally invariant of weight $p$ and $q \in \R$ with $2p+q=n+2$. Furthermore, one can construct a boundary stress-energy tensor class by
    \begin{equation}
        \label{eq:BSET_null_dust}
        \mathfrak{t}^{ij} = \mathfrak{a} \, \mathfrak{X}^i \mathfrak{X}^j \,.
    \end{equation}
\end{lem}

\begin{rem}
    The boundary stress-energy tensor of a conformally pp-wave boundary \eqref{eq:frakt_ppwave} is of the form \eqref{eq:BSET_null_dust} with the pair $(\mathfrak{K},2\mathfrak{a})$.
\end{rem}

\noindent Finally, note that the boundary stress-energy tensor classes above can be added together to create other classes if the hypotheses permit it.

\subsubsection{Conformal Killing fields on the boundary}
\label{sec:symmetries_currents_bdry}

Let $(\M,\mathfrak{I},\widetilde{g})$ be an aAdS space. Recall that from the FG expansion, one deduces that all $\mathfrak{S} \in \mathfrak{I}$ are umbilical in the conformal structure $(\M,[g])$, where $[g]$ is the conformal class of rescaled metrics. This implies that the tangent part on $\mathfrak{S}\in \mathfrak{I}$ of conformal Killing fields of $(\M,[g])$ are conformal Killing fields of the induced conformal structure $(\mathfrak{S},[\mathfrak{h}])$. This fact leads to the following known result.

\begin{lem}
    Let $K$ be a Killing field of $(\M\setminus\mathfrak{I},\widetilde{g})$. If $K$ extends smoothly on $\mathfrak{I}$ then it is tangent to $\mathfrak{I}$ and its restriction on $\mathfrak{S} \in \mathfrak{I}$ is a conformal Killing field of $(\mathfrak{S},[\mathfrak{h}])$.
\end{lem}

\begin{proof}
    Let $K$ be a Killing field of $(\M\setminus\mathfrak{I},\widetilde{g})$ extending smoothly on $\mathfrak{I}$. Then for any boundary defining function $x$ of $\mathfrak{I}$, one has
    \[ x \mathcal{L}_K g = 2 (K \cdot x) g \,, \]
    where $g:=x^2\widetilde{g}$ is the rescaled metric associated to $x$. Hence $K \cdot x = 0$ on $\mathfrak{I}$ and thus $K$ is tangent to the conformal boundary. Moreover, $K$ is a conformal Killing field of $(\M,[g])$ by continuity. It follows from the above discussion that its restriction on $\mathfrak{S} \in \mathfrak{I}$ is a conformal Killing field of $(\mathfrak{S},[\mathfrak{h}])$.
\end{proof}

\noindent Thus symmetries in the bulk imply symmetries on the conformal boundary. The converse was studied for stationary Killing fields by Anderson \cite{A06} and for holographic Killing fields of a domain of the conformal boundary by Holzegel and Shao \cite{HS23}.

Now that we know how to derive conformal Killing fields on the conformal boundary, let us remind that they lead to conserved currents on the boundary by association with the boundary stress-energy tensor. Indeed, if $\mathfrak{K}$ is a conformal Killing field of $(\mathfrak{S},[\mathfrak{h}])$ then for any representative $(\mathfrak{h},\mathfrak{t})$ of the free data class, the covector field
\[ \mathfrak{J}_i := \mathfrak{t}_{ij} \mathfrak{K}^j \]
is divergence-free for the Levi-Civita connection $\mathfrak{D}$ of $\mathfrak{h}$. Note that this covector field is conformally invariant of weight $n-2$.

\section{Local existence and uniqueness of aAdS spaces in dimension 4}
\label{sec:local_existence}

This section contains the precise statement of our main theorem, \Cref{thm:main_precise}, and its proof. For this, we rely on the framework established by Friedrich in \cite{F95}. Recall that a local existence and uniqueness theory for the vacuum Einstein equations \eqref{eq:VE} with a negative cosmological constant $\Lambda<0$ requires to consider an initial boundary value problem due to the presence of a timelike conformal boundary for generic solutions. Applying the standard method for solving the \eqref{eq:VE} by means of generalised wave coordinates would, in the context of aAdS spaces, lead to two main difficulties:
\begin{enumerate}
    \item[a)] The Ricci tensor is not regular up to the conformal boundary $\mathscr{I}$. Thus one has to deal with a system of partial differential equations with singular coefficients for the metric components.
    \item[b)] No information is a priori known on the localisation of the conformal boundary. Consequently, this is a free boundary problem.
\end{enumerate}
As established in \cite{F95}, Friedrich's strategy to overcome these two problems is as follows:
\begin{enumerate}
    \item Recast the \eqref{eq:VE} into an enlarged tensorial system which is regular up to the conformal boundary and includes a boundary defining function of the conformal boundary.
    \item Add some gauge invariance to the previous system by considering Weyl connections instead of only the Levi-Civita connection. This enables more freedom in the gauge choice. The resulting system is called the extended conformal vacuum Einstein equations, see~\eqref{eq:CVE}.
    \item Choose a gauge adapted to the conformal boundary for which the tensorial system reduces to a symmetric hyperbolic system of partial differential equations and for which the conformal boundary is straightened out.
\end{enumerate}
The only limitations of this method are that it constructs Einstein manifolds with a regular conformal extension and that it requires the spacetime dimension to be equal to 4. The former is not restrictive given our definition of aAdS spaces, see \Cref{def_aAdS}. However, as a consequence of the latter, we will always assume that $n=3$ in this section.

\bigbreak

In \Cref{sec:CVE}, the general framework of the extended conformal vacuum Einstein equations is introduced. A particular emphasis is put on describing the gauge invariances as they determine whether an object is analytic (gauge-dependent) or geometric (gauge-independent). The presentation includes a noticeable extension of the Friedrich scalar definition to all Weyl connections, see \Cref{sec:friedrich_scalar}. Let us also point out that the relevance of the (CVE) in the context of 4-dimensional aAdS spaces verifying the vacuum Einstein equations is detailed in \Cref{sec:application_to_4dim_aAdS_spaces}.

The core of the analytic part of the proof is contained in \Cref{sec:geometric_existence}. In particular, the gauge is constructed in \Cref{sec:gauge_construction} and the evolution problem is addressed in \Cref{sec:solving_evolution_system}. The analytic initial boundary value problem is then solved in \Cref{sec:solving_evolution_system} using maximal dissipative boundary conditions for symmetric hyperbolic systems, see \Cref{evol_solve_near}, and we prove the propagation of the constraints in \Cref{sec:propagation_constraints}.

\Cref{sec:geometric_bc} is dedicated to the geometric interpretation of a subset of the analytic boundary conditions. The first step, conducted in \Cref{sec:tensorial_bc_distribution}, consists in rewriting them as tensorial boundary conditions on a 2-dimensional distribution. Thanks to the asymptotic expansions of the curvature tensors based on the FG expansion found in \Cref{sec:expansions_curvature_tensors}, the analytic boundary conditions with constant coefficients are expressed in terms of the boundary stress-energy tensor and the Cotton-York tensor of the conformal boundary in \Cref{sec:replacing_tensors}. Most of the reflective analytic boundary conditions are then interpreted as the geometric Dirichlet boundary conditions in \Cref{sec:dirichlet} or as the geometric homogeneous Robin boundary conditions in \Cref{sec:homogeneous_robin}. This last step relies on auxiliary evolution systems on the conformal boundary. In the case of the Dirichlet boundary conditions, this was done by Friedrich in \cite{F95} while the treatment of the homogeneous Robin boundary conditions is new, see \Cref{prop:partial_to_total_robin_bc}.

The geometric initial data problem is discussed in \Cref{sec:geometric_initial_data}. Its precise formulation is given in \Cref{sec:formulation_gidp}. An adapted version of the conformal method established by Andersson and Chru\'sciel \cite{AC94,AC96} allows for a partial resolution of the geometric initial data problem in certain cases. This method is outlined in \Cref{sec:conformal_method}. Finally, \Cref{sec:unphysical_fields} focuses on the smoothness of the unphysical fields, a requirement of the geometric initial data problem. In the case of initial data constructed by the conformal method and under certain hypotheses, K\'ann\'ar derived necessary and sufficient condition for this requirement to be satisfied in \cite{K96}. We extend this result by deriving the necessary and sufficient conditions in all generality, see \Cref{thm:smooth_unphysical_fields} and its \Cref{cor:smooth_unphysical_fields}.

The last section, \Cref{sec:main_thm}, is comprised of the precise statement of the main theorem, \Cref{thm:main_precise}, as well as further remarks.

\subsection{The extended conformal vacuum Einstein equations}
\label{sec:CVE}

This section introduces the general framework of the extended Conformal Vacuum Einstein equations (CVE). We start by stating the equations and its gauge invariances in \Cref{sec:CVE_properties}. We then introduce the extended Friedrich scalar in \Cref{sec:friedrich_scalar}. It generalises the Friedrich scalar, used by Friedrich in \cite{F95}, to all Weyl connections.

In \Cref{sec:CVE_constraints}, we derive the constraints implied by the (CVE) on any non-null hypersurface. These are then rewritten into an intrinsic framework in \Cref{sec:CVC}. In particular, the constraints greatly simplify on a particular subset of a solution to the (CVE), called the conformal boundary. As its name suggests, this subset is closely related to the conformal boundary of Einstein spacetimes. More generally, the explicit relations between the (VE) and the (CVE) are given in \Cref{sec:CVE&VE}, as well as the equivalent relations at the level of the constraints. The specific case of 4-dimensional aAdS spaces satisfying the (VE) is further discussed in \Cref{sec:application_to_4dim_aAdS_spaces}, justifying the use of the (CVE) to solve our problem.

\subsubsection{The equations and gauge invariances}
\label{sec:CVE_properties}

\begin{defi}
    Denote by $\mathscr{E}$ the set of all quintuples $(\M,g,V,\Theta,\kappa)$ where
    \begin{itemize}
        \item $(\M,g)$ is a 4-dimensional smooth Lorentzian manifold with corners (see \Cref{sec:manifolds_with_corners}),
        \item $V$ is a Weyl candidate on $\M$ (see \Cref{def_Weyl_candidate}),
        \item $\Theta$ is a smooth function on $\M$,
        \item $\kappa$ is a smooth covector field on $\M$.
    \end{itemize}
\end{defi}

\begin{defi}
    A quintuple $(\M,g,V,\Theta,\kappa) \in \mathscr{E}$ is said to be a solution to the \emph{extended Conformal Vacuum Einstein equations (CVE)} if the following tensorial equations are satisfied on $\M$
    \begin{subequations}
    \makeatletter
    \def\@currentlabel{CVE}
    \makeatother
    \label{eq:CVE}
    \renewcommand{\theequation}{CVE.\arabic{equation}}
        \begin{align}
            \label{eq:CVE_V}
            \widehat{\nabla}_\alpha V^\alpha {}_{\beta\mu\nu} &= \kappa_\alpha V^\alpha {}_{\beta\mu\nu} \,, \\
            \label{eq:CVE_Gamma} 
            \widehat{R}^\alpha {}_{\beta\mu\nu} &= \Theta V^\alpha {}_{\beta\mu\nu} + 2 S_{\beta[\mu} {}^{\alpha\xi} \widehat{L}_{\nu]\xi} \,, \\
            \label{eq:CVE_L}
            2 \widehat{\nabla}_{[\alpha} \widehat{L}_{\beta]\gamma} &= (\widehat{\zeta}_g)_\mu V^\mu {}_{\gamma\alpha\beta} \,, \\
            \label{eq:CVE_const}
            \Theta^2 \widehat{L}_{\alpha\beta} &= - \Theta \widehat{\nabla}_{\alpha} (\widehat{\zeta}_g)_{\beta} + (\widehat{\zeta}_g)_{\beta} (d\Theta)_{\alpha} - \frac{1}{2} S_{\alpha\beta}{}^{\mu\nu} (\widehat{\zeta}_g)_\mu (\widehat{\zeta}_g)_\nu + \frac{\Lambda}{6}  g_{\alpha\beta} \,.
        \end{align}
    \end{subequations}
    where
    \begin{itemize}
        \item $\widehat{\nabla}$ is the Weyl connection on the conformal structure $(\M,[g])$ associated to $\kappa$ with respect to $g$ given by~\Cref{prop:exist&uniq_weylconn},
        \item $\widehat{R}^\alpha {}_{\beta\mu\nu}$ and $\widehat{L}_{\mu\nu}$ are the Riemann and Schouten tensors of $\widehat{\nabla}$ (see \Cref{def_riemann_hat} and \Cref{def:schouten_hat}),
        \item $S_{\alpha\beta}{}^{\mu\nu}$ is the tensor field defined by \eqref{eq:def_S},
        \item $\widehat{\zeta}_g$ is the smooth covector field on $\M$ defined by
        \begin{equation}
            \label{eq:def_zeta}
            (\widehat{\zeta}_g)_\alpha := \Theta \, \kappa_\alpha + (d\Theta)_\alpha \,,
        \end{equation}
        \item $\Lambda \in \R$ is the cosmological constant.
    \end{itemize}
\end{defi}

\begin{rem}
    By definition of the Cotton tensor \eqref{eq:def_cotton} and the decomposition of the Riemann tensor \eqref{eq:decomposition_Riemann_Weyl} of a Weyl connection, equations \eqref{eq:CVE_Gamma} and \eqref{eq:CVE_L} mean that the Cotton and Weyl tensors of $\widehat{\nabla}$ are given by
    \[ \widehat{W}^\alpha{}_{\beta\mu\nu} = \Theta V^\alpha{}_{\beta\mu\nu} \,, \qquad \qquad  \widehat{C}_{\alpha\beta\gamma} = (\widehat{\zeta}_g)_\mu V^\mu{}_{\gamma\alpha\beta} \,. \qedhere \]
\end{rem}

\noindent To describe the gauge invariances of the (CVE), let us introduce the following transformations on quintuples.

\begin{defi}
    \label{def:transfo_dim4}
    Let $(\M,g,V,\Theta,\kappa) \in \mathscr{E}$. Define
    \begin{itemize}
        \item for any confomorphism $\phi : (\M,g) \to (\overline{\M},\overline{g})$ (see \Cref{def:confomorphism}),
        \[ \Conf_\phi \left[(\M,g,V,\Theta,\kappa)\right] := \left(\overline{\M},\overline{g}, \phi_\star(\Omega^{-1} V), \phi_\star(\Omega \Theta), \phi_\star(\kappa-d\ln\Omega) \right) \in \mathscr{E} \,, \]
        where $\Omega \in \mathcal{C}^\infty(\M,\R_+^\star)$ is the conformal factor of $\phi$, that is $\phi^\star \overline{g} = \Omega^2 g$,
        \item for any smooth covector field $\omega$ on $\M$,
        \[ \Weylchg_\omega \left[(\M,g,V,\Theta,\kappa)\right] := \left(\M,g,V,\Theta,\kappa+\omega\right) \in \mathscr{E} \,, \]
        \item the sign change operator by
        \[ \Signchg \left[(\M,g,V,\Theta,\kappa)\right] := \left(\M,g,-V,-\Theta,\kappa\right) \in \mathscr{E} \,. \]
    \end{itemize}
\end{defi}

\begin{rem}
    Under the above transformations, the covector field $\widehat{\zeta}_g$ transforms as follows
    \begin{equation}
        \label{eq:transfo_zeta}
        \widehat{\zeta}_g \overset{\Conf_\phi}{\longmapsto} \phi_\star(\Omega \widehat{\zeta}_g) \,, \qquad \widehat{\zeta}_g \overset{\Weylchg_\omega}{\longmapsto} \widehat{\zeta}_g + \Theta \omega \,, \qquad \widehat{\zeta}_g \overset{\Signchg}{\longmapsto} -\widehat{\zeta}_g \,. \qedhere
    \end{equation}
\end{rem}

\begin{prop}
    \label{prop:invariance_CVE}
    The \eqref{eq:CVE} enjoy the following gauge invariances
    \begin{itemize}
        \item[i)] \emph{invariance under confomorphisms}: for any quintuple $(\M,g,V,\Theta,\kappa) \in \mathscr{E}$ and any confomorphism $\phi$ of $(\M,g)$, $(\M,g,V,\Theta,\kappa)$ is a solution if and only if $\Conf_\phi\big[(\M,g,V,\Theta,\kappa)\big]$ is a solution,
        \item[ii)] \emph{invariance under Weyl connection changes}: for any quintuple $(\M,g,V,\Theta,\kappa) \in \mathscr{E}$ and any smooth covector field $\omega$ on $\M$, $(\M,g,V,\Theta,\kappa)$ is a solution if and only if $\Weylchg_\omega \left[(\M,g,V,\Theta,\kappa)\right]$ is a solution,
        \item[iii)] \emph{invariance under sign change}: a quintuple $(\M,g,V,\Theta,\kappa) \in \mathscr{E}$ is a solution if and only if $\Signchg [(\M,g,V,\allowbreak \Theta,\kappa)]$ is a solution.
    \end{itemize}
\end{prop}

\begin{rem}
    The local statement of invariance under confomorphisms includes, in particular, invariance under changes of coordinates which expresses the tensorial nature of the equations.
\end{rem}

\begin{proof}
    The invariance of the (CVE) under confomorphisms and under sign change is trivial. The invariance under Weyl connection changes can be proven using the transformation rules of the Weyl, Schouten and Cotton tensors given by \Cref{prop:transformation_rules}.
\end{proof}

\begin{defi}
    Two quintuples in $\mathscr{E}$ are said to be \emph{gauge-equivalent} if they are related by a composition of confomorphisms $\Conf_\phi$, Weyl connection changes $\Weylchg_\omega$ and sign changes $\Signchg$.
\end{defi}

\begin{rem}
    From the properties
    \begin{alignat*}{3}
        S^2 &= \text{Id} \,, &\qquad\quad T_{\omega_1} T_{\omega_2} &= T_{\omega_1+\omega_2} \,, & \qquad\quad C_{\phi_1} C_{\phi_2} &= C_{\phi_1 \circ \phi_2} \,, \\
        [S,T_\omega] &= 0 \,, & \qquad [S,C_\phi] &= 0 \,, & \qquad  C_\phi T_\omega &= T_{\phi_\star \omega} C_\phi \,,
    \end{alignat*}
    one deduces that if $(\M,g,V,\Theta,\kappa)$, $\left(\overline{\M},\overline{g},\overline{V},\overline{\Theta},\overline{\kappa}\right) \in \mathscr{E}$ are two gauge-equivalent quintuples then there exists a confomorphism $\phi : (\M,g) \to \left(\overline{\M},\overline{g}\right)$, a smooth covector field $\omega$ on $\M$ and $k\in\{0,1\}$ such that
    \[ \left(\overline{\M},\overline{g},\overline{V},\overline{\Theta},\overline{\kappa}\right) = C_\phi T_\omega S^k \left[(\M,g,V,\Theta,\kappa)\right] \,. \qedhere \]
\end{rem}

\subsubsection{The extended Friedrich scalar}
\label{sec:friedrich_scalar}

In this section, we introduce a new quantity, called the extended Friedrich scalar, which extends the original Friedrich scalar to all Weyl connections. This scalar field enables us to split equation~\eqref{eq:CVE_const} into two equations that behave better where $\Theta$ vanishes, see \Cref{lem:split_CVE_regular} below. This modification enables us to obtain a propagation system for the constraints free of singularities contrary to Friedrich's work. The (extended) Friedrich scalar is also linked to the properties of the conformal boundary as we will see in~\Cref{sec:CVE_conf_boundary}.

\begin{defi}
    The \emph{extended Friedrich scalar} of a quintuple $(\M,g,V,\Theta,\kappa) \in \mathscr{E}$ is defined by
    \begin{equation}
        \label{eq:extended_Friedrich_scalar}
        \widehat{s}_g := \frac{g^{\alpha\beta}}{4}  \left( \widehat{\nabla}_\alpha (\widehat{\zeta}_g)_\beta + \kappa_\alpha (\widehat{\zeta}_g)_\beta + \Theta \widehat{L}_{\alpha\beta} \right) \,,
    \end{equation}
    where $\widehat{\zeta}_g$ is the covector field given by \eqref{eq:def_zeta} and $\widehat{\nabla}$ is the Weyl connection associated to $\kappa$ with respect to $g$.
\end{defi}

\begin{rem}
    When $\kappa = 0$, that is for the Levi-Civita connection $\nabla$ of $g$, the extended Friedrich scalar reduces to the Friedrich scalar given by
    \[ s_g = \frac{1}{4} \left( \Box_g \Theta + L \Theta \right) \,, \]
    where $L := L_\alpha{}^\alpha$ is the trace of the Schouten tensor of $\nabla$. Note that $\Box_g+L$ is not the linear scalar conformal wave operator which writes $\Box_g - L$ in dimension $4$.
\end{rem}

\begin{ex}
    On the Einstein cylinder $(\M_{EC},g_{EC})$ defined by \eqref{eq:EC} and with $\Theta = \cos\psi$, one has 
    \[ s_{g_{EC}} = \frac{1}{4} \left( \Box_{g_{EC}} \cos\psi + \cos\psi \right) = \frac{1}{4} \left( -3\cos\psi +\cos\psi\right) = - \frac{\Theta}{2} \,. \qedhere \]
\end{ex}

\begin{prop}
    \label{prop:transfo_s}
    Under the transformations on quintuples defined in \Cref{def:transfo_dim4}, the extended Friedrich scalar transforms as follows
        \begin{equation}
            \label{eq:transfo_s}
            \widehat{s}_g \overset{\Conf_\phi}{\longmapsto} \phi_\star(\Omega^{-1} \widehat{s}_g) \,, \qquad \widehat{s}_g \overset{\Weylchg_\omega}{\longmapsto} \widehat{s}_g + g^{-1}(\omega,\widehat{\zeta}_g) + \frac{\Theta}{2} g^{-1}(\omega,\omega) \,, \qquad \widehat{s}_g \overset{\Signchg}{\longmapsto} -\widehat{s}_g \,.
        \end{equation}
\end{prop}

\begin{lem}
    \label{lem:split_CVE_regular}
    Let $(\M,g,V,\Theta,\kappa) \in \mathscr{E}$. Equation \eqref{eq:CVE_const} is equivalent to 
    \begin{subequations}
        \begin{align}
            \label{eq:CVE_zeta}
            \widehat{\nabla}_\alpha (\widehat{\zeta}_g)_\beta + \kappa_\alpha (\widehat{\zeta}_g)_\beta + \Theta \widehat{L}_{\alpha\beta} &= \widehat{s}_g \,  g_{\alpha\beta} \,, \\
            \label{eq:CVE_const_bis}
            \frac{1}{2} \left( g^{-1}(\widehat{\zeta}_g,\widehat{\zeta}_g) + \frac{\Lambda}{3} \right) &= \widehat{s}_g \Theta \,.
        \end{align}
    \end{subequations}
    Furthermore, if \eqref{eq:CVE_zeta} and \eqref{eq:CVE_const_bis} hold then
    \begin{equation}
        \label{eq:grad_s}
        (d\widehat{s}_g)_\alpha - \kappa_\alpha \widehat{s}_g = - \widehat{L}_{\alpha\beta} g^{\beta\gamma} (\widehat{\zeta}_g)_\gamma \,.
    \end{equation}
\end{lem}

\begin{proof}
    By definition of $\widehat{\zeta}_g$, see equation~\eqref{eq:def_zeta}, \eqref{eq:CVE_const} rewrites as
    \[ \Theta \left( \widehat{\nabla}_\alpha (\widehat{\zeta}_g)_\beta + \kappa_\alpha (\widehat{\zeta}_g)_\beta + \Theta \widehat{L}_{\alpha\beta} \right) = \frac{g_{\alpha\beta}}{2}  \left( g^{-1}(\widehat{\zeta}_g,\widehat{\zeta}_g) + \frac{\Lambda}{3} \right) \,. \]
    Taking the trace with respect to the metric $g$ gives \eqref{eq:CVE_const_bis}. Then
    \[ \Theta \left( \widehat{\nabla}_\alpha (\widehat{\zeta}_g)_\beta + \kappa_\alpha (\widehat{\zeta}_g)_\beta + \Theta \widehat{L}_{\alpha\beta} - \widehat{s}_g \,  g_{\alpha\beta} \right) = 0 \,. \]
    Thus, \eqref{eq:CVE_zeta} holds on $\{p \in \M \mid \Theta(p) \neq 0\}$, and by continuity on
    \[ \supp \Theta := \overline{\{p\in \M \mid \Theta(p) \neq 0\}} \,. \]
    If $p \notin \supp \Theta$ then there exists an open neighbourhood $\mathcal{U}$ of $p$ on which $\Theta$ vanishes identically. Then $\widehat{\zeta}_g$ also vanishes on $\mathcal{U}$ and \eqref{eq:CVE_zeta} becomes trivial. Hence, equation \eqref{eq:CVE_zeta} is verified everywhere on $\M$.

    Furthermore, by taking the covariant derivative of \eqref{eq:CVE_const_bis} and using \eqref{eq:CVE_zeta},
    \begin{align*}
        \Theta (d\widehat{s}_g)_\alpha + \widehat{s}_g (d\Theta)_\alpha &= \kappa_\alpha g^{-1}(\widehat{\zeta}_g,\widehat{\zeta}_g) + g^{-1}(\widehat{\nabla}_\alpha \widehat{\zeta}_g,\widehat{\zeta}_g) \\
        &= - \Theta \widehat{L}_{\alpha\beta} g^{\beta\gamma} (\widehat{\zeta}_g)_\gamma + \widehat{s}_g (\widehat{\zeta}_g)_\alpha \,.
    \end{align*}
    Thus
    \[ \Theta \left( (d\widehat{s}_g)_\alpha - \kappa_\alpha \widehat{s}_g + \widehat{L}_{\alpha\beta} g^{\beta\gamma} (\widehat{\zeta}_g)_\gamma \right) = 0 \,. \]
    Consequently, \eqref{eq:grad_s} is verified on $\supp\Theta$. As above, one can show that it also holds on the complementary set. Hence the result.
\end{proof}

\subsubsection{The constraint equations on non-null hypersurfaces}
\label{sec:CVE_constraints}

Let $(\M,g,V,\Theta,\kappa) \in \mathscr{E}$ be a solution to the (CVE) and $\mathcal{S} \subset \M$ be a non-null hypersurface. The (CVE) induce some tensorial equations on $\mathcal{S}$ called the constraint equations. To state them, introduce
\begin{itemize}
    \item $e_\perp$ a unit-normal vector field on $\mathcal{S}$ with respect to $g$,
    \item $\varepsilon := g(e_\perp,e_\perp) \in \{\pm 1\}$ the $g$ pseudo-norm of $e_\perp$,
    \item $h$ the metric induced by $g$ on $\mathcal{S}$,
    \item $\widehat{D}$ the Weyl connection induced by $\widehat{\nabla}$ on $\mathcal{S}$ (see \Cref{sec:induced_conformal_structure}),
    \item $\widehat{K}_g$ the extrinsic curvature of $\widehat{\nabla}$ with respect to $g$ on $\mathcal{S}$ (see \Cref{sec:extrinsic_curvature}).
\end{itemize}

\begin{prop}
    The constraint equations induced by the \eqref{eq:CVE} on the non-null hypersurface $\mathcal{S}$ are
    \begin{subequations}
        \label{eq:constraints_hypersurface}
        \begin{align}
            \label{eq:constraint_a}
            \widehat{D}_i M_{jk}{}^i &=  2 (\widehat{K}_g)^i{}_{[j} E_{k]i} + 2\kappa_i M_{jk}{}^i \,, \\
            \label{eq:constraint_b}
            \widehat{D}_i E^i{}_j &= \varepsilon (\widehat{K}_g)^{ik} M_{jki} + 3\kappa_i E^i{}_j \,, \\
            \label{eq:constraint_c}
            \widehat{L}_{ij} - \widehat{l}_{ij} &= \Theta E_{ij} -\varepsilon F_{ij}(\widehat{K}_g) \,, \\
            \label{eq:constraint_d}
            2\widehat{L}_{k\perp} &= \varepsilon h^{ij} G_{kij}(\widehat{K}_g,\kappa) \,, \\
            \label{eq:constraint_e}
            \Theta M_{klj} &= G_{klj}(\widehat{K}_g,\kappa) + h_{j[k} h^{ip} G_{l]ip}(\widehat{K}_g,\kappa) \,, \\
            \label{eq:constraint_f}
            2 \widehat{D}_{[i} \widehat{L}_{j]k} &= 2 (\widehat{K}_g)_{k[i} \widehat{L}_{j]\perp} + z M_{ijk} - 2(\widehat{\zeta}_g)_l S_{k[i}{}^{lm} E_{j]m}  \,, \\
            \label{eq:constraint_g}
            2 \widehat{D}_{[i} \widehat{L}_{j]\perp} &= - 2\varepsilon (\widehat{K}_g)^k{}_{[i} \widehat{L}_{j]k} + 2 \kappa_{[i} \widehat{L}_{j]\perp} -\varepsilon  M_{ij}{}^k (\widehat{\zeta}_g)_k \,, \\
            \label{eq:constraint_h}
            \widehat{D}_i (\widehat{\zeta}_g)_j + \kappa_i (\widehat{\zeta}_g)_j & = z (\widehat{K}_g)_{ij} - \Theta \widehat{L}_{ij} + \widehat{s}_g h_{ij} \,, \\
            \label{eq:constraint_i}
            (dz)_i & = - \varepsilon (\widehat{K}_g)_i{}^j (\widehat{\zeta}_g)_j - \Theta \widehat{L}_{i\perp} \,, \\
            \label{eq:constraint_j}
            (d\widehat{s}_g)_i  - \kappa_i \widehat{s}_g & = - \varepsilon z \widehat{L}_{i\perp} - \widehat{L}_{ij} h^{jk} (\widehat{\zeta}_g)_k \,, \\
            \label{eq:constraint_k}
            \widehat{s}_g \Theta & = \frac{1}{2} \left( h^{ij} (\widehat{\zeta}_g)_i (\widehat{\zeta}_g)_j + \varepsilon z^2 + \frac{\Lambda}{3} \right) \,,
        \end{align}
    \end{subequations}
    where
    \begin{itemize}
        \item $(E_{ij},M_{ijk})$ is the decomposition of the Weyl candidate $V$ with respect to $\mathcal{S}$ given by~\Cref{lem:decomposition_Weyl_cand} (the indices of $E$ and $M$ are raised with the metric $h$),
        \item $z := (\widehat{\zeta}_g)_\perp$ is seen as a smooth scalar field on $\mathcal{S}$,
        \item $\widehat{L}_{i\perp}$ is seen as smooth vector fields on $\mathcal{S}$,
        \item $F_{ij}(\widehat{K}_g)$ and $G_{ijk}(\widehat{K}_g,\kappa)$ are defined by equations~\eqref{eq:def_F_G}.
    \end{itemize}
\end{prop}

\begin{rem}
    The fields $M$, $z$, $\widehat{L}_{i\perp}$ and $\widehat{K}_g$ are defined up to a sign unless an orientation on $\M$ is specified so that $e_\perp$ can be defined uniquely.
\end{rem}

\begin{proof}
    Let us use a frame field $(e_\perp,(e_{\bf i}))$ adapted to $\mathcal{S}$ with respect to $g$, see \Cref{sec:adapted_frame_fields}. The constraints on $\mathcal{S}$ arise from the equations in the (CVE) with no normal covariant derivatives. Namely,
    \begin{itemize}
        \item Equation~\eqref{eq:CVE_V}. One has
        \begin{align*}
            \widehat{\nabla}_{\bf i} V^{\bf i}{}_{\perp \bf jk} &= \widehat{D}_{\bf i} V^{\bf i}{}_{\perp \bf jk} - \widehat{\Gamma}_{\bf i}{}^\perp{}_\perp V^{\bf i}{}_{\bf \perp jk} - \widehat{\Gamma}_{\bf i}{}^{\bf l}{}_\perp V^{\bf i}{}_{\bf ljk}  - \widehat{\Gamma}_{\bf i}{}^\perp{}_{\bf j} V^{\bf i}{}_{\perp \bf \perp k} -  \widehat{\Gamma}_{\bf i}{}^\perp{}_{\bf k} V^{\bf i}{}_{\perp \bf j\perp} \\
            &= -\varepsilon \widehat{D}_{\bf i} M_{\bf jk}{}^{\bf i}  +\varepsilon \kappa_{\bf i} M_{\bf jk}{}^{\bf i} -2 \varepsilon (\widehat{K}_g)_{\bf i}{}^{\bf l} S_{\bf l[j}{}^{\bf im} E_{\bf k]m} + 2\varepsilon (\widehat{K}_g)^{\bf i}{}_{\bf [j} E_{\bf k]i} \\
            &= -\varepsilon \left(\widehat{D}_{\bf i} M_{\bf jk}{}^{\bf i} -\kappa_{\bf i} M_{\bf jk}{}^{\bf i} \right)+ 2\varepsilon (\widehat{K}_g)^{\bf i}{}_{\bf [j} E_{\bf k]i} \,, \\
        \widehat{\nabla}_{\bf i} V^{\bf i}{}_{\perp \bf j \perp} &= \widehat{D}_{\bf i} V^{\bf i}{}_{\perp \bf j\perp} - 2\widehat{\Gamma}_{\bf i}{}^\perp{}_\perp V^{\bf i}{}_{\bf \perp j\perp} - \widehat{\Gamma}_{\bf i}{}^{\bf k}{}_\perp V^{\bf i}{}_{\bf kj\perp} - \widehat{\Gamma}_{\bf i}{}^{\bf k}{}_\perp V^{\bf i}{}_{\bf \perp jk} \\
            &= \varepsilon \widehat{D}_{\bf i} E^{\bf i}{}_{\bf j} - 2\varepsilon \kappa_{\bf i} E^{\bf i}{}_{\bf j} - (\widehat{K}_g)^{\bf ik} M_{\bf ikj} - (\widehat{K}_g)^{\bf ik} M_{\bf jki} \\
            &= \varepsilon \left(\widehat{D}_{\bf i} E^{\bf i}{}_{\bf j} - 2\kappa_{\bf i} E^{\bf i}{}_{\bf j}\right) - (\widehat{K}_g)^{\bf ik} M_{\bf jki} \,,
        \end{align*}
        and
        \begin{align*}
            \kappa_{\bf a} V^{\bf a}{}_{\bf \perp jk} &= -\varepsilon \kappa_{\bf i} M_{\bf jk}{}^{\bf i} \,, \\
            \kappa_{\bf a} V^{\bf a}{}_{\bf \perp j\perp} &= \varepsilon \kappa_{\bf i} E^{\bf i}{}_{\bf j} \,.
        \end{align*}
        Hence \eqref{eq:constraint_a} and \eqref{eq:constraint_b}.
        
        \item Equation~\eqref{eq:CVE_Gamma}. This leads to a slightly modified version of the Codazzi equations \eqref{codazzi_weyl_schouten} where the Weyl tensor $\widehat{W}^\alpha{}_{\beta\mu\rho}$ is replaced by $\Theta V^\alpha{}_{\beta\mu\rho}$. Since $\dim \mathcal{S} = n = 3$, Weyl candidates vanish identically. Therefore, equation~\eqref{eq:codazzi_P} trivially holds and is discarded. Hence \eqref{eq:constraint_c}-\eqref{eq:constraint_e}.

        \item Equation~\eqref{eq:CVE_L}. From
            \begin{align*}
                \widehat{\nabla}_{\bf i} \widehat{L}_{\bf jk} &= \widehat{D}_{\bf i} \widehat{L}_{\bf jk} - \widehat{\Gamma}_{\bf i}{}^\perp{}_{\bf j} \widehat{L}_{\perp \bf k} - \widehat{\Gamma}_{\bf i}{}^\perp{}_{\bf k} \widehat{L}_{\bf j \perp} \\
                &= \widehat{D}_{\bf i} \widehat{L}_{\bf jk} - (\widehat{K}_g)_{\bf ij} \widehat{L}_{\bf k\perp} - (\widehat{K}_g)_{\bf ik} \widehat{L}_{\bf j\perp} \,, \\
                \widehat{\nabla}_{\bf i} \widehat{L}_{\bf j\perp} &= \widehat{D}_{\bf i} \widehat{L}_{\bf j\perp} - \widehat{\Gamma}_{\bf i}{}^\perp{}_{\bf j} \widehat{L}_{\perp\perp} - \widehat{\Gamma}_{\bf i}{}^{\bf k}{}_\perp \widehat{L}_{\bf jk} - \widehat{\Gamma}_{\bf i}{}^\perp{}_\perp \widehat{L}_{\bf j\perp} \\
                &= \widehat{D}_{\bf i} \widehat{L}_{\bf j\perp} - (\widehat{K}_g)_{\bf ij} \widehat{L}_{\perp\perp} + \varepsilon (\widehat{K})_{\bf i}{}^{\bf k} \widehat{L}_{\bf jk} - \kappa_{\bf i} \widehat{L}_{\bf j\perp} \,,
            \end{align*}
            one deduces
            \begin{align*}
                2 \widehat{\nabla}_{\bf [i} \widehat{L}_{\bf j]k} &= 2 \widehat{D}_{\bf [i} \widehat{L}_{\bf j]k} - 2 (\widehat{K}_g)_{\bf k[i} \widehat{L}_{\bf j]\perp} \,, \\
                2 \widehat{\nabla}_{\bf [i} \widehat{L}_{\bf j]\perp} &= 2 \widehat{D}_{\bf [i} \widehat{L}_{\bf j]\perp} + 2\varepsilon (\widehat{K}_g)^{\bf k}{}_{\bf [i} \widehat{L}_{\bf j]k} - 2 \kappa_{\bf [i} \widehat{L}_{\bf j]\perp} \,.
            \end{align*}
            Moreover
            \begin{align*}
                (\widehat{\zeta}_g)_{\bf d} V^{\bf d}{}_{\bf kij} &= (\widehat{\zeta}_g)_\perp V^\perp{}_{\bf kij} + (\widehat{\zeta}_g)_{\bf l} V^{\bf l}{}_{\bf kij} \\
                &= (\widehat{\zeta}_g)_\perp M_{\bf ijk} - 2(\widehat{\zeta}_g)_{\bf l} S_{\bf k[i}{}^{\bf lm} E_{\bf j]m} \,, \\
                (\widehat{\zeta}_g)_{\bf d} V^{\bf d}{}_{\bf \perp ij} &= (\widehat{\zeta}_g)_{\bf l} V^{\bf l}{}_{\bf \perp ij} \\
                &= -\varepsilon M_{\bf ij}{}^{\bf k} (\widehat{\zeta}_g)_{\bf k} \,.
            \end{align*}
            Hence \eqref{eq:constraint_f} and \eqref{eq:constraint_g}.

        \item Equation~\eqref{eq:CVE_zeta}. From
        \begin{align*}
            \widehat{\nabla}_{\bf i} (\widehat{\zeta}_g)_{\bf j} &= \widehat{D}_{\bf i} (\widehat{\zeta}_g)_{\bf j}  - \widehat{\Gamma}_{\bf i}{}^\perp{}_{\bf j} (\widehat{\zeta}_g)_\perp \\
            &= \widehat{D}_{\bf i} (\widehat{\zeta}_g)_{\bf j} - (\widehat{K}_g)_{\bf ij} (\widehat{\zeta}_g)_\perp \,, \\
            \widehat{\nabla}_{\bf i} (\widehat{\zeta}_g)_\perp  &= (d(\widehat{\zeta}_g)_\perp)_{\bf i} - \widehat{\Gamma}_{\bf i}{}^\perp{}_\perp (\widehat{\zeta}_g)_\perp - \widehat{\Gamma}_{\bf i}{}^{\bf k}{}_\perp (\widehat{\zeta}_g)_{\bf k}  \\
            &= (d(\widehat{\zeta}_g)_\perp)_{\bf i} - \kappa_{\bf i}{} (\widehat{\zeta}_g)_\perp + \varepsilon (\widehat{K}_g)_{\bf i}{}^{\bf k} (\widehat{\zeta}_g)_{\bf k} \,,
        \end{align*}
        one obtains \eqref{eq:constraint_h} and \eqref{eq:constraint_i}.
        
        \item Equation~\eqref{eq:grad_s} gives \eqref{eq:constraint_j}.

        \item Equation~\eqref{eq:CVE_const_bis} gives directly \eqref{eq:constraint_k}. \qedhere
    \end{itemize}
\end{proof}

\subsubsection{The extended conformal vacuum constraint equations}
\label{sec:CVC}

An intrinsic framework for the constraint equations \eqref{eq:constraints_hypersurface}, that is one independent of the (CVE), can be introduced as done below. 

\begin{defi}
    Denote by $\mathscr{D}$ the set of all 10-tuples $(\mathcal{S},h,E,M,\Psi,\kappa,K,z,T, Q)$ where
    \begin{itemize}
        \item $(\mathcal{S},h)$ is a 3-dimensional smooth Riemannian or Lorentzian manifold with corners (see \Cref{sec:manifolds_with_corners}),
        \item $E$ is a symmetric trace-free 2-tensor and $M$ is a Cotton candidate (see \Cref{def:cotton_candidate}) on $(\mathcal{S},h)$,
        \item $\Psi$ and $z$ are smooth functions on $\mathcal{S}$,
        \item $\kappa$ and $T$ are smooth covector fields on $\mathcal{S}$,
        \item $K$ is a smooth symmetric 2-tensor on $\mathcal{S}$,
        \item $Q$ is a smooth 2-tensor on $\mathcal{S}$.
    \end{itemize}
\end{defi}

\begin{defi}
    A 10-tuple $(\mathcal{S},h,E,M,\Psi,\kappa,K,z,T,Q) \in \mathscr{D}$ is said to be a solution to the \emph{extended Conformal Vacuum Constraint equations (CVC)} if the following tensorial equations are satisfied on $\mathcal{S}$
    \begin{subequations}
        \makeatletter
        \def\@currentlabel{CVC}
        \makeatother
        \label{eq:CVC}
        \renewcommand{\theequation}{CVC.\arabic{equation}}
        \begin{align}
            \label{eq:CVC_a}
            \widehat{D}_i M_{jk}{}^i &=  2 K^i{}_{[j} E_{k]i} + 2\kappa_i M_{jk}{}^i \,, \\
            \label{eq:CVC_b}
            \widehat{D}_i E^i{}_j &= \varepsilon K^{ik} M_{jki} + 3\kappa_i E^i{}_j \,, \\
            \label{eq:CVC_c}
            Q_{ij} - \widehat{l}_{ij} &= \Psi E_{ij} - \varepsilon F_{ij}(K) \,, \\
            \label{eq:CVC_d}
            2T_k &= \varepsilon h^{ij} G_{kij}(K,\kappa) \,, \\
            \label{eq:CVC_e}
            \Psi M_{klj} &= G_{klj}(K,\kappa) + h_{j[k} h^{ip} G_{l]ip}(K,\kappa) \,, \\
            \label{eq:CVC_f}
            2 \widehat{D}_{[i} Q_{j]k} &= 2 K_{k[i} T_{j]} + z M_{ijk} - 2(\widehat{\zeta}_h)_l S_{k[i}{}^{lm} E_{j]m}  \,, \\
            \label{eq:CVC_g}
            2 \widehat{D}_{[i} T_{j]} &= - 2\varepsilon K^k{}_{[i} Q_{j]k} + 2 \kappa_{[i} T_{j]} -\varepsilon M_{ij}{}^k (\widehat{\zeta}_h)_k  \,, \\
            \label{eq:CVC_h}
            \widehat{D}_i (\widehat{\zeta}_h)_j + \kappa_i (\widehat{\zeta}_h)_j & = z K_{ij} - \Psi Q_{ij} + \widehat{\sigma} h_{ij} \,, \\
            \label{eq:CVC_i}
            (dz)_i & = - \varepsilon K_i{}^j (\widehat{\zeta}_h)_j - \Psi T_i \,, \\
            \label{eq:CVC_j}
            (d\widehat{\sigma})_i  - \kappa_i \widehat{\sigma} & = - \varepsilon z T_i - Q_{ij} h^{jk} (\widehat{\zeta}_h)_k \,, \\
            \label{eq:CVC_k}
            \widehat{\sigma} \Psi & = \frac{1}{2} \left( h^{ij} (\widehat{\zeta}_h)_i (\widehat{\zeta}_h)_j + \varepsilon z^2 + \frac{\Lambda}{3} \right) \,,
        \end{align}    
    \end{subequations}
    where
    \begin{itemize}
        \item $\varepsilon = -1$ (respectively $+1$) if $(\mathcal{S},h)$ is Riemannian (respectively Lorentzian), 
        \item $\widehat{D}$ is the Weyl connection associated to the covector field $\kappa$ with respect to $h$,
        \item $\widehat{l}_{ij}$ is the Schouten tensor of $\widehat{D}$,
        \item $F_{ij}(K)$ and $G_{ijk}(K,\kappa)$ are defined by equations~\eqref{eq:def_F_G},
        \item one defines the following fields
            \begin{subequations}
                \label{eq:dim3}
                \begin{align}
                    \label{eq:zeta_dim3}
                    (\widehat{\zeta}_h)_i &:= \Psi \kappa_i + (d\Psi)_i \,, \\
                    \label{eq:Friedrich_scalar_dim3}
                    \widehat{s}_h &:= \frac{h^{ij}}{3} \left( \widehat{D}_i (\widehat{\zeta}_h)_j + \kappa_i (\widehat{\zeta}_h)_j + \Psi \widehat{l}_{ij} \right) \,, \\
                    \label{eq:def_sigma_hat}
                    \widehat{\sigma} &:= \widehat{s}_h - \frac{z}{3} K_k{}^k - \varepsilon \frac{\Psi}{12} \left( \left(K_k{}^k\right)^2 - K^{kl}K_{kl} \right) \,.
                \end{align}
            \end{subequations}
    \end{itemize}
\end{defi}

\begin{rems} \,
	\begin{itemize}
		\item The correspondence between the intrinsic fields of the (CVC) and the fields of \eqref{eq:constraints_hypersurface} are as follows: $h$ is the metric induced by $g$ on $\mathcal{S}$, $(E,M)$ is still the restriction of the electromangetic decomposition of $V$ with respect to $\mathcal{S}$, $\Psi$ is the restriction of $\Theta$ on $\mathcal{S}$, $\kappa_i$ is the tangent part on $\mathcal{S}$ of the 4-dimensional covector field $\kappa_\mu$, $K_{ij}$ is the extrinsic curvature of the connection $\widehat{\nabla}$ with respect to $h$, $z$ is the scalar field $(\widehat{\zeta}_h)_\perp$, $T_i$ is the covector field $\widehat{L}_{i\perp}$ and $Q_{ij}$ is the field $\widehat{L}_{ij}$ on $\mathcal{S}$.
		
		\item Equations \eqref{eq:zeta_dim3} and \eqref{eq:Friedrich_scalar_dim3} are the 3-dimensional analogue of equations \eqref{eq:def_zeta} and \eqref{eq:extended_Friedrich_scalar}. Hence, $\widehat{s}_h$ is called the 3-dimensional (extended) Friedrich scalar. Its non-extended version, that is when $\kappa = 0$, is given by
		\begin{equation}
			\label{eq:def_Friedrich_scalar_dim3_nonextended}
			s_h := \frac{1}{3} \left( \Delta_h \Psi + \Psi h^{ij} l_{ij} \right) = \frac{1}{3} \left( \Delta_h \Psi + \frac{r}{4} \Psi \right) \,. \qedhere
		\end{equation}
	\end{itemize}
\end{rems}

\noindent Let us detail below the gauges invariances of the (CVC), inherited from these of the (CVE) which have been presented in \Cref{sec:CVE_properties}. The invariance under Weyl connection changes split into two gauge invariances by distinguishing tangent and normal directions.

\begin{defi}
    \label{def:transfo_dim3}
    Let $(\mathcal{S},h,E,M,\Psi,\kappa,K,z,T,Q) \in \mathscr{D}$. Define
    \begin{itemize}
        \item for any confomorphism $\phi : (\mathcal{S},h) \to (\overline{\mathcal{S}},\overline{h})$ (see \Cref{def:confomorphism}),
        \begin{align*}
            \Conf_\phi \left[(\mathcal{S},h,E,M,\Psi,\kappa,K,z,T,Q) \right] := \big(\overline{\mathcal{S}},\overline{h},\phi_\star(\Omega^{-1} E),\phi_\star M,\phi_\star(\Omega\Psi),\phi_\star(\kappa-d\ln\Omega), \\
             \phi_\star(\Omega K),\phi_\star z, \phi_\star(\Omega^{-1} T),\phi_\star Q \big) \in \mathscr{D} \,,
        \end{align*}
        where $\Omega \in \mathcal{C}^\infty(\mathcal{S},\R_+^\star)$ is the conformal factor of $\phi$, that is $\phi^\star \overline{h} = \Omega^2 h$,
        
        \item for any smooth covector field $\omega \in T^\star \mathcal{S}$,
        \begin{align*}
            \Weylchg_\omega^\parallel \left[(\mathcal{S},h,E,M,\Psi,\kappa,K,z,T,Q)\right] := \Big(\mathcal{S},h,E,M,\Psi,\kappa+\omega,K,z,T_i-\varepsilon K_i{}^j\omega_j, \\
            Q_{ij}-\widehat{D}_i\omega_j+\frac{1}{2}S_{ij}{}^{kl}\omega_k\omega_l \Big) \in \mathscr{D} \,,
        \end{align*}

        \item for any smooth function $\chi \in \mathcal{C}^\infty(\mathcal{S},\R)$,
        \begin{align*}
            \Weylchg_\chi^\perp \left[(\mathcal{S},h,E,M,\Psi,\kappa,K,z,T,Q)\right] := \Big(\mathcal{S},h,E,M,\Psi,\kappa,K-\varepsilon \chi h,z+\chi\Psi,T-d\chi+\chi\kappa,\\
            Q+\chi K - \frac{\varepsilon\chi^2}{2} h \Big) \in \mathscr{D} \,,
        \end{align*}
        
        \item the sign change operator by
        \begin{align*}
            \Signchg \left[(\mathcal{S},h,E,M,\Psi,\kappa,K,z,T,Q)\right] := (\mathcal{S},h,-E,-M,-\Psi,\kappa,K,-z,T,Q) \in \mathscr{D} \,,
        \end{align*}
    \end{itemize}
\end{defi}

\begin{rem}
    Under the transformations defined above, the scalar field $\widehat{s}_h$ transforms as follows
    \begin{subequations}
        \label{eq:transfo_s_dim3}
        \begin{alignat}{3}
            &\widehat{s}_h \overset{\Conf_\phi}{\longmapsto} \phi_\star(\Omega^{-1} \widehat{s}_h) \,, \qquad & \widehat{s}_h & \overset{\Weylchg_\omega^\parallel}{\longmapsto} \widehat{s}_h + h^{-1}(\omega,\widehat{\zeta}_h) + \frac{\Psi}{2} h^{-1}(\omega,\omega) \,, \\
            &\widehat{s}_h \overset{\Signchg}{\longmapsto} -\widehat{s}_h \,. \qquad & \widehat{s}_h & \overset{\Weylchg_\chi^\perp}{\longmapsto} \widehat{s}_h \,,
    \end{alignat}
    \end{subequations}
    and the scalar field $\widehat{\sigma}$ as follows
    \begin{subequations}
    \begin{alignat}{3}
        &\widehat{\sigma} \overset{\Conf_\phi}{\longmapsto} \phi_\star(\Omega^{-1} \widehat{\sigma}) \,, \qquad &\widehat{\sigma}& \overset{\Weylchg_\omega^\parallel}{\longmapsto} \widehat{\sigma} + h^{-1}(\omega,\widehat{\zeta}_h) + \frac{\Psi}{2} h^{-1}(\omega,\omega) \,, \\
        &\widehat{\sigma} \overset{\Signchg}{\longmapsto} -\widehat{\sigma} \,, \qquad &\widehat{\sigma}& \overset{\Weylchg_\chi^\perp}{\longmapsto} \widehat{\sigma} + \varepsilon \left( \chi z +\frac{\Psi}{2} \chi^2 \right) \,.
    \end{alignat} 
    \end{subequations}
\end{rem}

\begin{defi}
    Two 10-tuples in $\mathscr{D}$ are said to be \emph{gauge-equivalent} if they are related by a composition of confomorphisms $\Conf_\phi$, tangent Weyl connection changes $\Weylchg^\parallel_\omega$, normal Weyl connection changes $\Weylchg^\perp_\chi$ and sign changes $\Signchg$.
\end{defi}

\begin{prop}
    The (CVC) are invariant under the transformations of \Cref{def:transfo_dim3}.
\end{prop}

\subsubsection{The conformal boundary}
\label{sec:CVE_conf_boundary}

One main advantage of the (CVE) compared with the (VE) is that they allow to directly include the conformal boundary of spacetimes within the manifold by removing all singularities on it. In fact, the conformal boundary is simply encoded through the following definition.

\begin{defi}
    \label{def:conf_bound_VCE}
    The \emph{conformal boundary} of a solution $(\M,g,V,\Theta,\kappa) \in \mathscr{E}$ of the (CVE) is defined as the gauge-invariant subset
    \begin{equation}
        \mathscr{I} := \{ p \in \M \mid \Theta(p) = 0 \} \,.
    \end{equation}
\end{defi}

\noindent The geometry and causal nature of $\mathscr{I}$ is given by the following lemma.

\begin{lem}
    \label{lem:pseudonorm_dTheta}
    Let $(\M,g,V,\Theta,\kappa) \in \mathscr{E}$ be a solution to the (CVE) with $\mathscr{I} \neq \varnothing$. Then
    \begin{equation}
        \label{eq:pseudonorm_dTheta}
        g^{-1}(d\Theta,d\Theta) = -\frac{\Lambda}{3} \quad \text{on } \mathscr{I} \,.
    \end{equation}
    In particular, if $\Lambda<0$ (respectively $\Lambda>0$) then $\mathscr{I}$ is a union of disjoint timelike (respectively spacelike) hypersurfaces.
\end{lem}

\begin{proof}
    The result is straightforward from \eqref{eq:def_zeta} and \eqref{eq:CVE_const_bis}.
\end{proof}

\begin{cor}
	\label{cor:I}
	Let $(\M,g,V,\Theta,\kappa) \in \mathscr{E}$ be a solution to the (CVE) for a cosmological constant $\Lambda \neq 0$ such that $\mathscr{I} \neq \varnothing$. Then for any connected component $\mathfrak{S}$ of $\mathscr{I}$, the 10-tuple $(\mathfrak{S},\mathfrak{h},\mathfrak{E},\mathfrak{M},\Psi,\kappa,\mathfrak{K},\mathfrak{z},\mathfrak{T},\mathfrak{Q})$ induced by $(\M,g,V,\Theta,\kappa)$ on $\mathfrak{S}$ satisfies
	\begin{subequations}
		\label{eq:constraints_I}
		\begin{align}
			\mathfrak{z}^2 &= \frac{|\Lambda|}{3} \,, \\
			\label{eq:extrinsic_curvature_conf_boundary}
			\mathfrak{K}_{ij} &= -\frac{\widehat{\sigma}}{\mathfrak{z}} \mathfrak{h}_{ij} \,, \\
			\widehat{\mathfrak{D}}_i \mathfrak{E}^i{}_j &= 3\kappa_i \mathfrak{E}^i{}_j \,, \\
            \label{eq:conf_bound_cotton_magnetic}
			\mathfrak{z} \mathfrak{M}_{ijk} &= \widehat{\mathfrak{c}}_{ijk} = \mathfrak{c}_{ijk} \,, \\
			\mathfrak{T}_i &= \frac{\sign(\Lambda)}{\mathfrak{z}} \left( (d\widehat{\sigma})_i - \kappa_i \widehat{\sigma} \right) \,, \\
			\mathfrak{Q}_{ij} &=  \widehat{\mathfrak{l}}_{ij} + \frac{3\widehat{\sigma}^2}{2\Lambda} \mathfrak{h}_{ij} \,,
		\end{align}
	\end{subequations}
	where $\widehat{\sigma}$ is the scalar field defined by \eqref{eq:def_sigma_hat}, $\widehat{\mathfrak{D}}$ is the Weyl connection associated to $\kappa$ with respect to $\mathfrak{h}$ (it is also the connection induced by $\widehat{\nabla}$ on $\mathfrak{S}$ by \Cref{lem:induced_weyl_connection}), $\widehat{\mathfrak{l}}_{ij}$ and $\widehat{\mathfrak{c}}_{ijk}$ are respectively the Schouten and Cotton tensors of $\widehat{\mathfrak{D}}$ and $\mathfrak{c}_{ijk}$ is the Cotton tensor of the Levi-Civita connection $\mathfrak{D}$ of $\mathfrak{h}$. Note that since $\mathfrak{S}$ is 3-dimensional, the Cotton tensors of any two Weyl connection coincide, see the transformation law \eqref{transfo_cotton}.
\end{cor}

\begin{rem}
	Under the hypotheses of the corollary, $\mathscr{I}$ is a union of disjoint non-null hypersurfaces by \Cref{lem:pseudonorm_dTheta}. Thus, a connected component $\mathfrak{S}$ of $\mathscr{I}$ is a hypersurface of $(\M,g)$.
\end{rem}

\begin{proof}
	Let $\mathfrak{S}$ be a connected component of $\mathscr{I}$. Since $\Psi = \Omega|_\mathfrak{S} = 0$, the (CVC) on $\mathfrak{S}$ simplify into
	\begin{subequations}
		\begin{align}
			\label{eq:constraint_a_bis}
			\widehat{\mathfrak{D}}_i \mathfrak{M}_{jk}{}^i &=  2 \mathfrak{K}^i{}_{[j} \mathfrak{E}_{k]i} + 2\kappa_i \mathfrak{M}_{jk}{}^i \,, \\
			\label{eq:constraint_b_bis}
			\widehat{\mathfrak{D}}_i \mathfrak{E}^i{}_j &= \varepsilon \mathfrak{K}^{ik} \mathfrak{M}_{jki} + 3\kappa_i \mathfrak{E}^i{}_j \,, \\
			\label{eq:constraint_c_bis}
			\mathfrak{Q}_{ij} - \widehat{\mathfrak{l}}_{ij} &= -\varepsilon F_{ij}(\mathfrak{K}) \,, \\
			\label{eq:constraint_d_bis}
			2\mathfrak{T}_k &= \varepsilon \mathfrak{h}^{ij} G_{kij}(\mathfrak{K},\kappa) \,, \\
			\label{eq:constraint_e_bis}
			0 &= G_{klj}(\mathfrak{K},\kappa) + \mathfrak{h}_{j[k} \mathfrak{h}^{ip} G_{l]ip}(\mathfrak{K},\kappa) \,, \\
			\label{eq:constraint_f_bis}
			2 \widehat{\mathfrak{D}}_{[i} \mathfrak{Q}_{j]k} &= 2 \mathfrak{K}_{k[i} \mathfrak{T}_{j]} + \mathfrak{z} \mathfrak{M}_{ijk}  \,, \\
			\label{eq:constraint_g_bis}
			2 \widehat{\mathfrak{D}}_{[i} \mathfrak{T}_{j]} &= - 2\varepsilon \mathfrak{K}^k{}_{[i} \mathfrak{Q}_{j]k} + 2 \kappa_{[i} \mathfrak{T}_{j]} \,, \\
			\label{eq:constraint_h_bis}
			0 & = \mathfrak{z} \mathfrak{K}_{ij} + \widehat{\sigma} \mathfrak{h}_{ij} \,, \\
			\label{eq:constraint_i_bis}
			(d\mathfrak{z})_i & = 0 \,, \\
			\label{eq:constraint_j_bis}
			(d\widehat{\sigma})_i  - \kappa_i \widehat{\sigma} & = - \varepsilon \mathfrak{z} \mathfrak{T}_i \,, \\
			\label{eq:constraint_k_bis}
			0 & = \varepsilon \mathfrak{z}^2 + \frac{\Lambda}{3} \,.
		\end{align}    
	\end{subequations}
	Equation~\eqref{eq:constraint_i_bis} implies that $\mathfrak{z}$ is constant. Then \eqref{eq:constraint_k_bis} gives the value of $\mathfrak{z}^2$ and $\varepsilon = -\sign(\lambda)$. Then equation \eqref{eq:constraint_h_bis} gives \eqref{eq:extrinsic_curvature_conf_boundary}. It follows from \eqref{eq:constraint_c_bis} and \eqref{eq:constraint_d_bis} or \eqref{eq:constraint_j_bis} that
	\begin{align*}
		\mathfrak{Q}_{ij} &=  \widehat{\mathfrak{l}}_{ij} - \frac{\varepsilon \widehat{\sigma}^2}{2\mathfrak{z}^2} \mathfrak{h}_{ij} = \widehat{\mathfrak{l}}_{ij} + \frac{3\widehat{\sigma}^2}{2\Lambda} \mathfrak{h}_{ij} \,, \\
		\mathfrak{T}_i &= - \frac{\varepsilon}{\mathfrak{z}} \left( (d\widehat{\sigma})_i - \kappa_i \widehat{\sigma} \right) = \frac{\sign(\Lambda)}{\mathfrak{z}} \left( (d\widehat{\sigma})_i - \kappa_i \widehat{\sigma} \right) \,.
	\end{align*}
	Then \eqref{eq:constraint_f_bis} and \eqref{eq:constraint_b_bis} resume to
	\begin{align*}
		\widehat{\mathfrak{c}}_{ijk} &= \mathfrak{z} \mathfrak{M}_{ijk} \,, \\
		\widehat{\mathfrak{D}}_i \mathfrak{E}^i{}_j &= 3\kappa_i \mathfrak{E}^i{}_j \,.
	\end{align*}
	Finally, \eqref{eq:constraint_a_bis} is just the contracted fourth Bianchi identity~\eqref{eq:fourth_Bianchi_id_Weyl_con}, and equations \eqref{eq:constraint_e_bis}, \eqref{eq:constraint_g_bis} hold.
\end{proof}

\noindent Equation \eqref{eq:extrinsic_curvature_conf_boundary} means that the extrinsic curvature $\widehat{\mathfrak{K}}_g$ of $\widehat{\nabla}$ with respect to $g$ is given by
\[ (\widehat{\mathfrak{K}}_g)_{ij} = - \frac{\widehat{s}_g |_\mathfrak{S}}{\mathfrak{z}} \mathfrak{h}_{ij} \,, \]
where $\widehat{s}_g$ is the extended Friedrich scalar defined by \eqref{eq:extended_Friedrich_scalar}. This implies that all connected components $\mathfrak{S}$ of the conformal boundary $\mathscr{I}$ are umbilical in the conformal structure $(\M,[g])$, see \Cref{def:tllygeod_umbilical}. Furthermore, the extended Friedrich scalar controls the extrinsic curvature and thus plays an important role to determine whether a conformal geodesic starting on the conformal boundary stay on it or not, as stated in the lemma below. This lemma is given here for comparison purpose since it was used in Friedrich's proof but will not be necessary in our work.

\begin{lem}[Extended version of Friedrich \protect{\cite[Lemma 4.1]{F95}}]
    \label{lem:conf_geod_stay_conf_bound}
    Let $(\M,g,V,\Theta,\kappa)\in\mathscr{E}$ be a solution to the (CVE), $\widehat{\nabla}$ be the Weyl connection associated to $\kappa$ with respect to $g$ given by \Cref{prop:exist&uniq_weylconn} and $(x,\widehat{\beta})$ be a $\widehat{\nabla}$-conformal geodesic of initial data $(x_\star,\dot{x}_\star,\widehat{\beta}_\star)$ with $x_\star \in \mathfrak{S}$, where $\mathfrak{S}$ is a connected component of $\mathscr{I}$. If the initial data satisfy
    \begin{equation}
        \label{eq:hypotheses_lem_conf_geo_stay_conf_bound}
         \langle d\Theta, \dot{x}_\star\rangle = 0 \quad \text{and} \quad g^{-1}(\widehat{\beta}_\star, d\Theta) = - \widehat{s}_g(x_\star) \,,
    \end{equation}
    then
    \begin{itemize}
        \item[i)] the curve $x$ stays on $\mathfrak{S}$,
        \item[ii)] along the curve $x$, the normal part of $\widehat{\beta}$ is fully imposed by
        \[ g^{-1}(\widehat{\beta},d\Theta) = -\widehat{s}_g \,, \]
        \item[iii)] $(x,\Tan \widehat{\beta})$ is a $\widehat{\mathfrak{D}}$-conformal geodesic where $\widehat{\mathfrak{D}}$ is the Weyl connection induced by $\widehat{\nabla}$ on $\mathfrak{S}$, see \Cref{sec:induced_conformal_structure}.
    \end{itemize}
\end{lem}

\begin{proof}
    We adapt the proof provided in \cite[Section 4.4]{F95} to Weyl connections. Let us divide it into three steps.
    \begin{itemize}
        \item \underline{Step 1}: Reduction to the case $\widehat{s}_g = 0$ on $\mathfrak{S}$

        Take a smooth covector field $\omega$ on $\M$ such that
        \[ \omega = \frac{3\widehat{s}_g}{\Lambda} d\Theta \quad \text{on } \mathfrak{S} \,. \]
        Then, using \eqref{eq:pseudonorm_dTheta},
        \[ g^{-1}(\omega,d\Theta) = -\widehat{s}_g \quad \text{on } \mathfrak{S} \,. \]
        Let $\widecheck{\nabla}$ be the Weyl connection associated to $\widecheck{\kappa}_g :=\kappa+\omega$ with respect to $g$. By the transformation laws \eqref{eq:transfo_s} of the extended Friedrich scalar, one has $\widecheck{s}_g = 0$ on $\mathfrak{S}$. By~\Cref{prop:confgeo_changweylconn}, the $\widehat{\nabla}$-conformal geodesic $(x,\widehat{\beta})$ corresponds to the $\widecheck{\nabla}$-conformal geodesic $(x,\widecheck{\beta})$ where $\widecheck{\beta} := \widehat{\beta} -\omega$ along $x$. 
        
        \item \underline{Step 2}: Gauge construction
        
        Fix $p \in \mathfrak{S}$. Let $(y^i)$ be a local chart on an open neighbourhood $\mathfrak{U}$ of $p$ in $\mathfrak{S}$ and $(e_{\bf i})$ be a smooth frame field on $\mathfrak{U}$.
        
        Consider $\widecheck{\nabla}$-geodesic curves starting from $\mathfrak{U}$ with initial vector a unit normal vector $e_\perp$ to $\mathfrak{S}$ with respect to $\mathfrak{h}$, the metric induced by $g$ on $\mathfrak{S}$. The geodesics form a congruence on a neighbourhood $\mathcal{U}$ of $p$ in $\M$. Let us denote the tangent vector field to the geodesics by $e_\perp$ and the parameter of the geodesics initialised at $0$ on $\mathfrak{U}$ by $y^\perp$. Thus
        \[ \widecheck{\nabla}_{e_\perp} e_\perp = 0 \,, \quad e_\perp{}^\perp := \langle dy^\perp,e_\perp\rangle =1 \,, \quad \mathfrak{S} \cap \mathcal{U} = \{ q \in \mathcal{U} \mid y^\perp(q) = 0\} \,. \]
        
        Extend the functions $(y^i)$ from $\mathfrak{U}$ to $\mathcal{U}$ by taking them constant along the geodesics and the vector fields $(e_{\bf i})$ by parallel transport along the geodesics. That is to say
        \[ e_\perp{}^i := \langle dy^i,e_\perp\rangle = 0 \,, \qquad \widecheck{\nabla}_{e_\perp} e_{\bf i} = 0 \,. \]
        Then $(y^\alpha) := (y^\perp,(y^i))$ are Gaussian coordinates associated to the congruence of geodesics and $(e_{\bf a}) := (e_\perp,(e_{\bf i}))$ is a frame field adapted to $\mathfrak{S}$ with respect to $g$ on $\mathcal{U}$ with $g(e_\perp,e_{\bf i}) = 0$ on $\mathcal{U}$.

        \item \underline{Step 3}: Analysis

        Denote by $I$ the (maximal) open interval (containing the initial parameter $\tau_\star$) on which the conformal geodesic $(x,\widecheck{\beta})$ is defined. Define the subinterval
        \[ I_\mathfrak{S} := \left\{ \tau \in I \, \big| \, x(\tau) \in \mathfrak{S}, \dot{x}(\tau) \text{ tangent to } \mathfrak{S}, \widecheck{\beta}(\tau) \text{ tangent to } \mathfrak{S} \right\} \,. \]
        Thanks to hypotheses \eqref{eq:hypotheses_lem_conf_geo_stay_conf_bound} and step 1, one has $\tau_\star \in I_\mathfrak{S}$. By smoothness of the conformal geodesic $(x,\widecheck{\beta})$, $I_\mathfrak{S}$ is a closed subset of $I$.

        Let $\tau \in I_\mathfrak{S}$ and $((y^\alpha),(e_{\bf a}))$ be the gauge constructed on a neighbourhood $\mathcal{U}$ of $x(\tau) \in \mathfrak{S}$ in $\M$ by step 2. As $\widecheck{s}_g = 0$ on $\mathfrak{S}$, the constraint equations~\eqref{eq:constraints_I} on $\mathfrak{S}$ give in particular
        \[ (\widecheck{\mathfrak{K}}_g)_{\bf ij} = 0 \,, \qquad \widecheck{L}_{\bf i\perp} = 0 \,, \qquad \widecheck{L}_{\perp\bf i} = 0 \,, \qquad \widecheck{L}_{\bf ij} = \widecheck{\mathfrak{l}}_{\bf ij} \,. \]
        By writing the $\widecheck{\nabla}$-conformal geodesic equations in the frame field $(e_{\bf a})$ and coordinates $(y^\alpha)$, one finds that the normal components verify
        \begin{align*}
            \dot{x}(y^\perp \circ x) &= \dot{x}^\perp \,, \\
            \dot{x}(\dot{x}^\perp) &= -2 \langle \widecheck{\beta}, \dot{x}\rangle \dot{x}^\perp + g(\dot{x},\dot{x}) g(e_\perp,e_\perp) \widecheck{\beta}_\perp \,, \\
            \dot{x}(\widecheck{\beta}_\perp) &= \langle \widecheck{\beta},\dot{x}\rangle \widecheck{\beta}_\perp - \frac{1}{2} g^{-1}(\widecheck{\beta},\widecheck{\beta}) g(e_\perp,e_\perp) \dot{x}^\perp + \widecheck{L}_{\perp\perp} \dot{x}^\perp \,,
        \end{align*}
        This is a homogeneous system in $(y^\perp \circ x,\dot{x}^\perp,\widecheck{\beta}_\perp)$. Since the latter triple vanishes at $\tau \in I_\mathfrak{S}$, it vanishes on an open neighbourhood of $\tau$ in $I$. Thus $I_\mathfrak{S}$ is an open subset of $I$. Since $I$ is connected, one has $I_\mathfrak{S} = I$. Hence the curve $x$ stays on $\mathfrak{S}$ and the covector field $\widecheck{\beta}$ remains tangent to $\mathfrak{S}$.
        
        Finally, the tangential components of the $\widecheck{\nabla}$-conformal geodesic equations, in the gauge constructed by step 2, reduce to the $\widecheck{\mathfrak{D}}$-conformal geodesic equations. \qedhere
    \end{itemize}
\end{proof}

Similarly, one can define the conformal boundary of a solution $(\mathcal{S},h,E,M,\Psi,\kappa,K,z,T,Q) \in \mathscr{D}$ to the (CVC) as the gauge-invariant subset
\[ \slashed{\mathscr{I}} := \{ p \in \mathcal{S} \mid \Psi(p) = 0 \} \,. \]
However, contrary to the (CVE), the conformal boundary can cover the whole of $\mathcal{S}$, just as considered above in \Cref{cor:I}.

\begin{lem}
    Let $(\mathcal{S},h,E,M,\Psi,\kappa,K,z,T,Q) \in \mathscr{D}$ be a solution to the (CVC) with $\slashed{\mathscr{I}} \neq \varnothing$. Then
    \begin{equation}
        \label{eq:pseudonorm_dPsi}
        h^{-1}(d\Psi,d\Psi) = -\varepsilon z^2 - \frac{\Lambda}{3} \quad \text{on } \slashed{\mathscr{I}} \,.
    \end{equation}
\end{lem}

\begin{proof}
    The result is straightforward from \eqref{eq:zeta_dim3} and \eqref{eq:CVC_k}.
\end{proof}

\subsubsection{Relation with the vacuum Einstein equations}
\label{sec:CVE&VE}

The (CVE) are closely related to the Vacuum Einstein equations (VE) away from the conformal boundary as shown in the following proposition.

\begin{prop} \,
    \label{prop:CVE_VE}
    \begin{itemize}
        \item[i)] If $(\M,\widetilde{g})$ is a solution to the \eqref{eq:VE} then for any smooth nowhere vanishing function $\Theta \in \mathcal{C}^\infty(\M,\R^\star)$ and any smooth covector field $\kappa$ on $\M$, the quintuple $(\M,\Theta^2 \widetilde{g},\Theta^{-1} \widetilde{W}, \Theta,\kappa) \in \mathscr{E}$, where $\widetilde{W}$ is the Weyl tensor of $\widetilde{g}$, is a solution to the \eqref{eq:CVE}.
        \item[ii)] Let $(\M,g,V,\Theta,\kappa) \in \mathscr{E}$ be a solution to the \eqref{eq:CVE}. Then, on any open set $\mathcal{U}$ where $\Theta$ nowhere vanishes, $(\mathcal{U},|\Theta|^{-2} g)$ is a solution to the \eqref{eq:VE}.
    \end{itemize}
\end{prop}

\begin{proof}
    $i)$ Assume that $(\M,\widetilde{g})$ is a 4-dimensional smooth Lorentzian manifold solution to the (VE), that is
    \begin{equation}
        \label{eq:einstein_proof}
        \widetilde{L}_{\alpha\beta} = \frac{\Lambda}{6} \widetilde{g}_{\alpha\beta} \,.
    \end{equation}
    Let $\Theta$ be a smooth nowhere vanishing function and $\kappa$ be a smooth covector field on $\M$. Take $\widehat{\nabla}$ the Weyl connection associated to $\kappa + d\ln|\Theta|$ with respect to $\widetilde{g}$ by \Cref{prop:exist&uniq_weylconn}. It follows that it is associated to $\kappa$ with respect to $g := \Theta^2 \widetilde{g}$.
    
    Equation~\eqref{eq:CVE_Gamma} is the decomposition~\eqref{eq:decomposition_Riemann_Weyl} of the Riemann tensor of a Weyl connection with $V := \Theta^{-1} \widehat{W}$. Since $\widehat{W} = \widetilde{W}$ by \eqref{transfo_weyl}, one has $V = \Theta^{-1} \widetilde{W}$. Equation~\eqref{eq:einstein_proof} implies that the Cotton tensor $\widetilde{C}_{\alpha\beta\gamma}$ of $\widetilde{g}$ vanishes identically. Then the transformation law~\eqref{transfo_cotton} of the Cotton tensor gives equation~\eqref{eq:CVE_L}. Moreover, equation~\eqref{eq:CVE_V} comes from the once contracted second Bianchi identity~\eqref{Bianchi2_Weyl}. Finally, equation~\eqref{eq:CVE_const} is derived from the transformation law~\eqref{transfo_schouten} of the Schouten tensor and \eqref{eq:einstein_proof}.

    \bigbreak

    \noindent $ii)$ Let $(\M,g,V,\Theta,\kappa) \in \mathscr{E}$ be a solution to the (CVE) and $\mathcal{U}$ be an open set where $\Theta$ nowhere vanishes. Then $\widetilde{g} := \Theta^{-2} g|_{\mathcal{U}}$ is a well-defined Lorentzian metric on $\mathcal{U}$ in the conformal class $[g]$ on $\mathcal{U}$. Using \eqref{eq:CVE_const} and \eqref{eq:def_zeta}, one finds
    \[ \widehat{L}_{\alpha\beta} = -\widehat{\nabla}_\alpha (\widehat{\kappa}_{\widetilde{g}})_\beta - \frac{1}{2} S_{\alpha\beta}{}^{\mu\nu} (\widehat{\kappa}_{\widetilde{g}})_\mu (\widehat{\kappa}_{\widetilde{g}})_\nu + \frac{\Lambda}{6}  \frac{g_{\alpha\beta}}{\Theta^2} \,, \]
    where $\widehat{\kappa}_{\widetilde{g}} = \kappa + d\ln|\Theta|$. Thus, the transformation law~\eqref{transfo_schouten} of the Schouten tensor gives that the Schouten tensor of $\widetilde{g}$ verifies
    \[ \widetilde{L}_{\alpha\beta} = \frac{\Lambda}{6} \widetilde{g}_{\alpha\beta} \,. \]
    Consequently, $(\mathcal{U},\widetilde{g})$ is a solution to the (VE).
\end{proof}

A similar situation occurs for the constraint equations. Recall that a triple $(\widetilde{\mathcal{S}},\widetilde{h},\widetilde{K})$, where $(\widetilde{\mathcal{S}},\widetilde{h})$ is a 3-dimensional Riemannian or Lorentzian manifold and $\widetilde{K}$ is a symmetric 2-tensor on $\widetilde{\mathcal{S}}$, is said to be a solution to the Vacuum Constraint equations (VC) if
\begin{subequations}
    \makeatletter
    \def\@currentlabel{VC}
    \makeatother
    \label{eq:VC}
    \renewcommand{\theequation}{VC.\arabic{equation}}
    \begin{align}
        \label{eq:VC_hamiltonian}
        \widetilde{r} -\varepsilon \left( (\widetilde{K}^i{}_i)^2 - \widetilde{K}_{ij} \widetilde{K}^{ij} \right) &= 2\Lambda \,, \\
        \label{eq:VC_momentum}
        \widetilde{D}_i \widetilde{K}^i{}_j - \widetilde{D}_j \widetilde{K}^i{}_i &= 0 \,,
    \end{align}
\end{subequations}
where $\widetilde{D}$ is the Levi-Civita connection of $\widetilde{h}$, $\widetilde{r}$ is its scalar curvature, $\Lambda$ is the cosmological constant and $\varepsilon = \pm 1$ according to the signature. The counterpart of \Cref{prop:CVE_VE} is the following proposition.

\begin{prop} \,
    \label{prop:CVC_VC}
    \begin{itemize}
        \item[i)] If $(\widetilde{\mathcal{S}},\widetilde{h},\widetilde{K})$ is a solution to the \eqref{eq:VC} then, for any smooth nowhere vanishing scalar function $\Psi \in \mathcal{C}^\infty(\widetilde{\mathcal{S}},\R^\star)$, any smooth covector field $\kappa$ on $\widetilde{\mathcal{S}}$ and any smooth scalar field $z$ on $\widetilde{\mathcal{S}}$, the 10-tuple $(\widetilde{\mathcal{S}},h:=\Psi^2 \widetilde{h},E,M,\Psi,\kappa,K,z,T,Q) \in \mathscr{D}$ is a solution to the \eqref{eq:CVC} where one defines in this order
        \begin{subequations}
            \label{CVC_def}
            \begin{align}
                \label{CVC_def_K}
                K_{ij} &:= |\Psi| \widetilde{K}_{ij} - \varepsilon z\Psi^{-1} h_{ij} \,, \\
                \label{CVC_def_Q}
                Q_{ij} &:= \Psi^{-1} \left( -\widehat{D}_i (\widehat{\zeta}_h)_j - \kappa_i (\widehat{\zeta}_h)_j + z K_{ij} + \widehat{\sigma} h_{ij} \right) \,, \\
                \label{CVC_def_E}
                E_{ij} &:= \Psi^{-1} \left( Q_{ij} - \widehat{l}_{ij} + \varepsilon F_{ij}(K) \right) \,, \\
                \label{CVC_def_T}
                T_i &:= \Psi^{-1} \left( -(dz)_i - \varepsilon K_i{}^j (\widehat{\zeta}_h)_j \right)  \,, \\
                \label{CVC_def_M}
                M_{klj} &:= \Psi^{-1} \left( G_{klj}(K,\kappa) + h_{j[k} h^{ip} G_{l]ip}(K,\kappa) \right) \,,
            \end{align} 
        \end{subequations}
        $\widehat{D}$ being the Weyl connection associated to $\kappa$ with respect to $h$, $\widehat{l}_{ij}$ being its Schouten tensor, $\widehat{\sigma}$ and $\widehat{\zeta}_h$ being defined by \eqref{eq:dim3} and $F_{ij}(K)$, $G_{ijk}(K,\kappa)$ being defined by \eqref{eq:def_F_G}.
        
        \item[ii)] Let $(\mathcal{S},h,E,M,\Psi,\kappa,K,z,T,Q) \in \mathscr{D}$ be a solution to the \eqref{eq:CVC}. Then, on any open set $\mathcal{U}$ where $\Psi$ nowhere vanishes, $(\mathcal{U},\Psi^{-2} h,|\Psi|^{-1}(K+\varepsilon z \Psi^{-1} h))$ is a solution to the \eqref{eq:VC}.
    \end{itemize}
\end{prop}

\begin{rem}
    In the case of $ii)$, note that the fields
    \[ \widetilde{h}_{ij} := \Psi^{-2} h_{ij} \,, \qquad \widetilde{K}_{ij} := \frac{1}{|\Psi|} \left(K_{ij} +\varepsilon \frac{z}{\Psi} h_{ij} \right) \,, \]
    are gauge invariants, meaning that they are left invariant by the gauge transformations of the (CVC) described in \Cref{def:transfo_dim3}.
\end{rem}

\begin{proof}
    $i)$ Let $(\widetilde{\mathcal{S}},\widetilde{h},\widetilde{K})$ be a solution to the (VC), $\Psi$ be a smooth nowhere vanishing function, $\kappa$ be a smooth covector field and $z$ be a smooth scalar field on $\widetilde{\mathcal{S}}$. By definition of the fields $K$, $Q$, $E$, $M$ and $T$ through equations \eqref{CVC_def}, the 10-tuple written in the proposition belongs in $\mathscr{D}$ and \eqref{eq:CVC_c}, \eqref{eq:CVC_e}, \eqref{eq:CVC_h}, \eqref{eq:CVC_i} are directly verified. Let us prove that the other equations of the (CVC) also hold.
    \begin{itemize}
        \item Equation \eqref{eq:CVC_k}: one deduces from \eqref{CVC_def_K} that
        \[ (\widetilde{K}_i{}^i)^2 - \widetilde{K}_{ij} \widetilde{K}^{ij} = \Psi^2 \left( (K_i{}^i)^2 - K_{ij} K^{ij}\right)  + 4\varepsilon z \Psi K_i{}^i + 6z^2 \,. \]
        Furthermore, one can derive from \eqref{eq:transfo_s_dim3} that
        \[ \widetilde{r} = 4 \widetilde{l} = 12 \left( \widehat{s}_h \Psi - \frac{1}{2} h^{-1}(\widehat{\zeta}_h,\widehat{\zeta}_h) \right) \,. \]
        The result follows from \eqref{eq:VC_hamiltonian}.
        \item Equation \eqref{eq:CVC_j}: by differentiating \eqref{eq:CVC_k}, one has
        \[ \widehat{\sigma} (d\Psi)_i + \Psi (d\widehat{\sigma})_i = \varepsilon z (dz)_i + \left( \widehat{D}_i (\widehat{\zeta}_h)_k + \kappa_i (\widehat{\zeta}_h)_k \right) h^{kl} (\widehat{\zeta}_h)_l \,. \]
        Using \eqref{CVC_def_Q}, \eqref{CVC_def_T} and by definition of $\widehat{\zeta}_h$, one finds
        \[ \Psi \left( (d\widehat{\sigma})_i  - \kappa_i \widehat{\sigma} + \varepsilon z T_i + Q_{ij} h^{jk} (\widehat{\zeta}_h)_k \right) = 0 \,. \]
        \item Equation \eqref{eq:CVC_d}: one can write \eqref{eq:VC_momentum} under the form
        \[ 2 \widetilde{h}^{ik} \widetilde{D}_{[j} \widetilde{K}_{k]i} = 0 \,. \]
        Yet using \eqref{CVC_def_K},
        \begin{align*}
            \sign(\Psi) \widetilde{D}_{[j} \widetilde{K}_{k]i} &= \widetilde{D}_{[j} \left( \Psi^{-1} K_{k]i} + \varepsilon z \widetilde{h}_{k]i} \right) \\
            &= \Psi^{-1} \widetilde{D}_{[j}  K_{k]i} - \Psi^{-2} (d\Psi)_{[j} K_{k]i} - \varepsilon \widetilde{h}_{i[j} (dz)_{k]} \\
            &= \Psi^{-1} \left( \widehat{D}_{[j} K_{k]i} + S_{i[j}{}^{lm} K_{k]l} (\widehat{\zeta}_h)_m \Psi^{-1} \right) - \Psi^{-2} (d\Psi)_{[j} K_{k]i} - \varepsilon \widetilde{h}_{i[j} (dz)_{k]}  \\
            &= \Psi^{-1} \widehat{D}_{[j} K_{k]i} - \Psi^{-2} \left( - (\widehat{\zeta}_h)_{[j} K_{k]i} + h_{i[j} K_{k]}{}^l (\widehat{\zeta}_h)_l + (d\Psi)_{[j} K_{k]i} \right) \\
            &\quad - \varepsilon \widetilde{h}_{i[j} (dz)_{k]}  \\
            &= \Psi^{-1} \left( \widehat{D}_{[j} K_{k]i} + \kappa_{[j} K_{k]i} \right) + \varepsilon \widetilde{h}_{i[j} \left( - (dz)_{k]} + \varepsilon K_{k]}{}^l (\widehat{\zeta}_h)_l \right) \,.
        \end{align*}
        By taking the trace and with equations \eqref{CVC_def_T} and \eqref{eq:def_F_G}, one deduces that
        \[ 0 = 2 \widetilde{h}^{ik} \widetilde{D}_{[j} \widetilde{K}_{k]i} = |\Psi| \varepsilon \left( \varepsilon h^{ik} G_{jki}(K,\kappa) - 2T_j \right) \,.  \]
        Hence the result. Note that by combining \eqref{CVC_def_M} and \eqref{eq:CVC_d}, one has now
        \begin{equation}
            \label{eq:aux_G}
            G_{ijk}(K,\kappa) = \Psi M_{ijk} - 2\varepsilon h_{k[i} T_{j]} \,. 
        \end{equation}
        
        \item Equation \eqref{eq:CVC_g}: from \eqref{CVC_def_T} and \eqref{CVC_def_Q}, one has
        \begin{align*}
            \widehat{D}_i T_j &= \Psi^{-1} \Big( -(d\Psi)_i T_j - \widehat{D}_i \widehat{D}_j z - \varepsilon \left(\widehat{D}_i K_{jk} + \kappa_i K_{jk} \right) h^{kl} (\widehat{\zeta}_h)_l + \varepsilon K_j{}^l \left(-\widehat{D}_i (\widehat{\zeta}_h)_l - \kappa_i  (\widehat{\zeta}_h)_l \right) \Big) \\
            &= \Psi^{-1} \Big( -(d\Psi)_i T_j - \widehat{D}_i \widehat{D}_j z - \varepsilon \left(\widehat{D}_i K_{jk} + \kappa_i K_{jk} \right) h^{kl} (\widehat{\zeta}_h)_l + \varepsilon K_j{}^l \left(\Psi Q_{il} -z K_{il} - \widehat{\sigma} h_{il} \right) \Big) \,.
        \end{align*}
        Hence
        \[ 2 \widehat{D}_{[i} T_{j]} = \Psi^{-1} \left( -2(d\Psi)_{[i} T_{j]} - \varepsilon G_{ijk} h^{kl} (\widehat{\zeta}_h)_l - 2\varepsilon \Psi K^l{}_{[i} Q_{j]l} \right) \,. \]
        By injecting \eqref{eq:aux_G}, one has
        \begin{align*}
            2 \widehat{D}_{[i} T_{j]} &= \Psi^{-1} \left( -2(d\Psi)_{[i} T_{j]} - \varepsilon \Psi M_{ijk} h^{kl} (\widehat{\zeta}_h)_l + 2 (\widehat{\zeta}_h)_{[i} T_{j]} - 2\varepsilon \Psi K^l{}_{[i} Q_{j]l} \right) \\
            &= \Psi^{-1} \left(- \varepsilon \Psi M_{ijk} h^{kl} (\widehat{\zeta}_h)_l + 2 \Psi \kappa_{[i} T_{j]} - 2\varepsilon \Psi K^l{}_{[i} Q_{j]l} \right) \\
            &= - \varepsilon M_{ijk} h^{kl} (\widehat{\zeta}_h)_l + 2\kappa_{[i} T_{j]} - 2\varepsilon K^l{}_{[i} Q_{j]l} \,.
        \end{align*}
        
        \item Equation \eqref{eq:CVC_f}: from \eqref{CVC_def_Q}, one has
        \begin{align*}
            \widehat{D}_i Q_{jk} &= \Psi^{-1} \Big( - (d\Psi)_i Q_{jk} - \widehat{D}_i \widehat{D}_j (\widehat{\zeta}_h)_k - (\widehat{\zeta}_h)_k \widehat{D}_i \kappa_j - \kappa_j \widehat{D}_i (\widehat{\zeta}_h)_k \\
            &\qquad\quad + (dz)_i K_{jk} + z \widehat{D}_i K_{jk} + \left((d\widehat{\sigma})_i - 2\kappa_i \widehat{\sigma}\right) h_{jk} \Big) \,.
        \end{align*}
        By using \eqref{CVC_def_Q}, \eqref{CVC_def_T} and \eqref{eq:CVC_j}, one deduces that
        \begin{align*}
            \widehat{D}_i Q_{jk} &= \Psi^{-1} \Big( - (d\Psi)_i Q_{jk} - \widehat{D}_i \widehat{D}_j (\widehat{\zeta}_h)_k - (\widehat{\zeta}_h)_k \widehat{D}_i \kappa_j \\
            &\qquad\quad + \kappa_j \left( \kappa_i (\widehat{\zeta}_h)_k - z K_{ik} - \widehat{\sigma} h_{ik} + \Psi Q_{ik} \right) - \Psi T_i K_{jk} \\
            &\qquad\quad - \varepsilon K_i{}^l (\widehat{\zeta}_h)_l K_{jk} + z \widehat{D}_i K_{jk} - \widehat{\sigma} \kappa_i h_{jk} -\varepsilon z T_i h_{jk} - h_{jk} Q_{il} h^{lm} (\widehat{\zeta}_h)_m \Big) \,.
        \end{align*}
        Hence
        \begin{align*}
            2 \widehat{D}_{[i} Q_{j]k} &= \Psi^{-1} \Big( - 2(d\Psi)_{[i} Q_{j]k} - 2 \widehat{D}_{[i} \widehat{D}_{j]} (\widehat{\zeta}_h)_k - 2(\widehat{\zeta}_h)_k \widehat{D}_{[i} \kappa_{j]} \\
            &\qquad\quad + 2z \kappa_{[i} K_{j]k} - 2 \Psi \kappa_{[i} Q_{j]k} - 2 \Psi T_{[i} K_{j]k} - 2\varepsilon K_{[i}{}^l K_{j]k} (\widehat{\zeta}_h)_l \\
            &\qquad\quad + 2z \widehat{D}_{[i} K_{j]k} - 2\varepsilon z T_{[i} h_{j]k} + 2h_{k[i} Q_{j]l} h^{lm} (\widehat{\zeta}_h)_m \Big) \\
            &= \Psi^{-1} \Big( - 2 \widehat{D}_{[i} \widehat{D}_{j]} (\widehat{\zeta}_h)_k - 2(\widehat{\zeta}_h)_{[i} Q_{j]k} - 2(\widehat{\zeta}_h)_k \widehat{D}_{[i} \kappa_{j]} \\
            &\qquad\quad + 2h_{k[i} Q_{j]l} h^{lm} (\widehat{\zeta}_h)_m - 2\varepsilon K_{[i}{}^l K_{j]k} (\widehat{\zeta}_h)_l - 2\Psi T_{[i} K_{j]k} \\
            &\qquad\quad + z \left( 2\widehat{D}_{[i} K_{j]k} + 2\kappa_{[i} K_{j]k} + 2\varepsilon h_{k[i} T_{j]} \right) \Big) \,.
        \end{align*}
        With \eqref{CVC_def_Q} and \eqref{eq:aux_G}, it follows that
        \begin{align*}
            2 \widehat{D}_{[i} Q_{j]k} &= \Psi^{-1} \left( -2 \widehat{D}_{[i} \widehat{D}_{j]} (\widehat{\zeta}_h)_k - 2 S_{k[i}{}^{lm} Q_{j]m} (\widehat{\zeta}_h)_l - 2\varepsilon K_{[i}{}^l K_{j]k} (\widehat{\zeta}_h)_l  \right) + 2 K_{k[i} T_{j]} + zM_{ijk} \,.
        \end{align*}
        One has
        \[ -2 \widehat{D}_{[i} \widehat{D}_{j]} (\widehat{\zeta}_h)_k = \widehat{r}^l{}_{kij} (\widehat{\zeta}_h)_l = 2S_{k[i}{}^{lm} \widehat{l}_{j]m} (\widehat{\zeta}_h)_l \,, \]
        and, since Weyl candidates vanish identically in dimension $3$ (see also \eqref{eq:codazzi_P}),
        \[ 2 K_{[i}{}^l K_{j]k} = 2 S_{k[i}{}^{lm} F_{ j]m}(K) \,. \]
        One concludes with \eqref{CVC_def_E}.
        
        \item Equation \eqref{eq:CVC_a}: from \eqref{eq:aux_G}, one has
        \[ \widehat{D}_i M_{jk}{}^i = \Psi^{-1} \left( -(d\Psi)_i M_{jk}{}^i + \widehat{D}_i G_{jk}{}^i(K,\kappa) + 2\varepsilon \widehat{D}_{[j} T_{k]} \right) \,. \]
        By definition of $G_{ijk}(K,\kappa)$, see \eqref{eq:def_F_G} and using \eqref{eq:CVC_d}, one has
        \begin{align*}
            \widehat{D}_i G_{jk}{}^i(K,\kappa) &= \widehat{D}_i \left( h^{il} \left( 2\widehat{D}_{[j} K_{k]l} + 2 \kappa_{[j} K_{k]l} \right) \right) \\
            &= \widehat{D}_i \left( 2\widehat{D}_{[j} K_{k]}{}^i - 2 \kappa_{[j} K_{k]}{}^i \right) \\
            &= \widehat{D}_i \widehat{D}_j K_k{}^i - \widehat{D}_i \widehat{D}_k K_j{}^i - \widehat{D}_i \left( 2 \kappa_{[j} K_{k]}{}^i \right) \\
            &= \widehat{D}_j \widehat{D}_i K_k{}^i - \widehat{r}^l{}_{kij} K_l{}^i + \widehat{r}^i{}_{lij} K_k{}^l - \widehat{D}_k \widehat{D}_i K_j{}^i + \widehat{r}^l{}_{jik} K_l{}^i - \widehat{r}^i{}_{lik} K_j{}^l - \widehat{D}_i \left( 2 \kappa_{[j} K_{k]}{}^i \right) \\
            &= \widehat{D}_j \widehat{D}_i K_k{}^i - \widehat{D}_k \widehat{D}_i K_j{}^i + \widehat{r}^l{}_{ijk} K_l{}^i + 2 \widehat{r}^i{}_{li[j} K_{k]}{}^l - \widehat{D}_i \left( 2 \kappa_{[j} K_{k]}{}^i \right) \\
            &= \widehat{D}_j \left( \widehat{D}_i K_k{}^i - \kappa_i K_k{}^i - \widehat{D}_k K + \kappa_k K \right) + \widehat{D}_j \left( \kappa_i K_k{}^i - \kappa_k K \right) \\
            &\quad  - \widehat{D}_k \left( \widehat{D}_i K_j{}^i - \kappa_i K_j{}^i - \widehat{D}_j K + \kappa_j K \right) - \widehat{D}_k \left( \kappa_i K_j{}^i - \kappa_j K \right) \\
            &\quad + \widehat{r}^l{}_{ijk} K_l{}^i + 2 \widehat{r}^i{}_{li[j} K_{k]}{}^l - \widehat{D}_i \left( 2 \kappa_{[j} K_{k]}{}^i \right) \\
            &= - 2\widehat{D}_{[j} \left( h^{il} G_{k]il}(K,\kappa) \right) + \widehat{r}_{(il)jk} K^{il} + 2 \widehat{r}^i{}_{li[j} K_{k]}{}^l + \kappa_i \left( 2\widehat{D}_{[j} K_{k]}{}^i - 2 \kappa_{[j} K_{k]}{}^i \right)  \\
            &\quad - 2 K^i{}_{[j} \widehat{D}_{k]} \kappa_i + 2 K^i{}_{[j} \widehat{D}_{|i|} \kappa_{k]} - 2K \widehat{D}_{[j} \kappa_{k]} - 2 \kappa_{[j} \left( \widehat{D}_{|i|} K_{k]}{}^i - K_{k]}{}^i \kappa_i -  \widehat{D}_{k]} K \right) \\
            &= -4\varepsilon \widehat{D}_{[j} T_{k]} + \widehat{r}_{(il)jk} K^{il} + 2 \widehat{r}^i{}_{li[j} K_{k]}{}^l + \kappa_i G_{jk}{}^i(K,\kappa) \\
            &\quad - 2 K^i{}_{[j} \widehat{D}_{k]} \kappa_i + 2 K^i{}_{[j} \widehat{D}_{|i|} \kappa_{k]} - 2K \widehat{D}_{[j} \kappa_{k]} + 4\varepsilon \kappa_{[j} T_{k]} \,.
        \end{align*}
        One deduces from \eqref{eq:symetric_part_Riemann} and \eqref{eq:antisymetric_part_ricci} that
        \[ \widehat{r}_{(il)jk} K^{il} = 2 g_{il} K^{il} \widehat{D}_{[j} \kappa_{k]}  = 2 K \widehat{D}_{[j} \kappa_{k]} \,, \]
        and from \eqref{eq:decomposition_Riemann_Weyl} and \eqref{eq:antisym_part_schouten} that
        \[ \widehat{r}^i{}_{lij} = 2 S_{l[i}{}^{im} \widehat{l}_{j]m} = \widehat{l}_{jl} - \widehat{D}_{[j} \kappa_{l]} + g_{jl} \widehat{l}_i{}^i \,. \]
        It follows that
        \[ \widehat{D}_i G_{jk}{}^i(K,\kappa) = -4\varepsilon \widehat{D}_{[j} T_{k]} + \kappa_i G_{jk}{}^i(K,\kappa) - 2 K^i{}_{[j} \widehat{l}_{k]i} + 4\varepsilon \kappa_{[j} T_{k]} \,, \]
        and thus
        \[ \widehat{D}_i M_{jk}{}^i = \Psi^{-1} \left( -(d\Psi)_i M_{jk}{}^i + \kappa_i G_{jk}{}^i(K,\kappa) - 2 K^i{}_{[j} \widehat{l}_{k]i} - 2\varepsilon \widehat{D}_{[j} T_{k]} + 4\varepsilon \kappa_{[j} T_{k]} \right) \,. \]
        Using \eqref{eq:aux_G} and \eqref{eq:CVC_g}, one obtains
        \begin{align*}
            \widehat{D}_i M_{jk}{}^i &= \Psi^{-1} \left( -(d\Psi)_i M_{jk}{}^i + \Psi \kappa_i M_{jk}{}^i - 2 K^i{}_{[j} \widehat{l}_{k]i} - \varepsilon \left( 2\widehat{D}_{[j} T_{k]} - 2 \kappa_{[j} T_{k]} \right) \right) \\
            &= \Psi^{-1} \left( -(d\Psi)_i M_{jk}{}^i + \Psi \kappa_i M_{jk}{}^i - 2 K^i{}_{[j} \widehat{l}_{k]i} + K^i{}_{[j} Q_{k]i} + M_{jk}{}^i (\widehat{\zeta}_h)_i \right) \\
            &= 2 \kappa_i M_{jk}{}^i + \Psi^{-1} 2 K^i{}_{[j} \left( Q_{k]i} - \widehat{l}_{k]i} \right) \,.
        \end{align*}
        From the definition of $F_{ij}(K)$, see \eqref{eq:def_F_G}, one deduces that $K^i{}_{[j} F_{k]i}(K) = 0$. One concludes using \eqref{CVC_def_E}.
        
        \item Equation \eqref{eq:CVC_b}: from \eqref{CVC_def_E} one has
        \begin{align*}
            \widehat{D}_i E_j{}^i &= h^{ik} \widehat{D}_i E_{jk} + 2\kappa_i E_j{}^i \\
            &= \Psi^{-1} h^{ik} \left( -(d\Psi)_i E_{jk} + \widehat{D}_i Q_{jk} - \widehat{D}_i \widehat{l}_{jk} + \varepsilon \widehat{D}_i F_{jk}(K) \right) + 2\kappa_i E_j{}^i \\
            &= \Psi^{-1} h^{ik} \left( -(d\Psi)_i E_{jk} + 2\widehat{D}_{[i} Q_{j]k} - \widehat{c}_{ijk} + \widehat{D}_j \left( \Psi E_{ik} - \varepsilon F_{ik}(K) \right) + \varepsilon \widehat{D}_i F_{jk}(K) \right) + 2\kappa_i E_j{}^i \\
            &= \Psi^{-1} h^{ik} \left( -(d\Psi)_i E_{jk} + 2\widehat{D}_{[i} Q_{j]k} - \widehat{c}_{ijk} + 2\varepsilon \widehat{D}_{[i} F_{j]k}(K) \right) + 2\kappa_i E_j{}^i \,.
        \end{align*}
        Using \eqref{eq:CVC_f} and  \eqref{Bianchi2_Weyl}, one has
        \begin{align*}
            h^{ik} \left( 2\widehat{D}_{[i} Q_{j]k} - \widehat{c}_{ijk} \right) &= h^{ik} \left( 2 K_{k[i} T_{j]} + z M_{ijk} - 2 (\widehat{\zeta}_h)_l S_{k[i}{}^{lm} E_{j]m} \right) \\
            &= 2 h^{ik} K_{k[i} T_{j]} + (\widehat{\zeta}_h)_i E_j{}^i \,.
        \end{align*}
        Thus, it remains to prove that
        \[ K^{ik} M_{jki} = \Psi^{-1} h^{ik} \left( 2\varepsilon K_{k[i} T_{j]} + 2 \widehat{D}_{[i} F_{j]k}(K) \right) \,. \]
        With \eqref{eq:aux_G}, this is equivalent to
        \[ K^{ik} G_{jki}(K,\kappa) + 4\varepsilon K^k{}_{[j} T_{k]} = 2 h^{ik} \widehat{D}_{[i} F_{j]k}(K) \,. \]
        To prove the above equality, define
        \[ P_{ij} := K K_{ij} - K_{ik} K^k{}_j \,, \qquad P := h^{ij} P_{ij} = K^2 - K_{ij} K^{ij} \,. \]
        By definition of $F_{ij}(K)$, see \eqref{eq:def_F_G}, one has $F_{ij}(K) = P_{ij} - P h_{ij}/4$. Thus
        \begin{align*}
            2 h^{ik} \widehat{D}_{[i} F_{j]k}(K) &= \widehat{D}_i F^i{}_j(K) - 2\kappa_i F^i{}_j(K) - \widehat{D}_j F^i{}_i(K) + 2\kappa_j F^i{}_i(K) \\
            &= \widehat{D}_i P^i{}_j - 2\kappa_i P^i{}_j - \frac{1}{2} \left( \widehat{D}_j P - 2\kappa_j P \right) \\
            &= \left( K \widehat{D}_i K^i{}_j + K_j{}^i \widehat{D}_i K - K^{ik} \widehat{D}_i K_{kj} - K_j{}^k \widehat{D}_i K^i{}_k - 2\kappa_i K_{kj} K^{ki} \right) \\
            &\quad - 2\kappa_i \left( K K^i{}_j - K^{ik} K_{kj} \right) - \left( K \widehat{D}_j K - K^{ik} \widehat{D}_j K_{ik} - 2 \kappa_j K^{ik} K_{ik} \right) \\
            &\quad + \kappa_j \left( K^2 - K_{ik} K^{ik} \right) \\
            &= \left(K_j{}^i - K \delta_j{}^i\right) \left( \widehat{D}_i K -\widehat{D}_k K^k{}_i \right) + 2 K^{ik} \widehat{D}_{[j} K_{i]k} \\
            &\quad - 2 K \kappa_i K^i{}_j + \kappa_j K^{ik}K_{ik} + \kappa_j K^2   \\
            &= \left(K_j{}^i - K \delta_j{}^i\right) \left( \widehat{D}_i K -\widehat{D}_k K^k{}_i \right) + K^{ik} G_{jik}(K,\kappa) \\
            &\quad - 2 K \kappa_i K^i{}_j + \kappa_i K^{ik} K_{jk} + \kappa_j K^2   \\
            &= \left(K_j{}^i - K \delta_j{}^i\right) \left( \widehat{D}_i K - \kappa_i K - \widehat{D}_k K^k{}_i + \kappa_k K^k{}_i \right) + K^{ik} G_{jik}(K,\kappa) \\
            &= \left(K_j{}^i - K \delta_j{}^i\right) h^{kl} G_{ikl}(K,\kappa) + K^{ik} G_{jik}(K,\kappa) \\
            &= 2\varepsilon \left(K_j{}^i - K \delta_j{}^i\right) T_i + K^{ik} G_{jik}(K,\kappa) \\
            &= 4\varepsilon K^i{}_{[j} T_{i]} + K^{ik} G_{jik}(K,\kappa) \,,
        \end{align*}
        where \eqref{eq:CVC_d} has been used to obtain the penultimate equality.
    \end{itemize} 

    \bigbreak

    $ii)$ Let $(\mathcal{S},h,E,M,\Psi,\kappa,K,z,T,Q) \in \mathscr{D}$ be a solution to the (CVC) and let $\mathcal{U}$ be a set where $\Psi$ nowhere vanishes. The 10-tuple $(\mathcal{U},h,E,M,\Psi,\kappa,K,z,T,Q) \in \mathscr{D}$ is a fortiori solution to the (CVC). Let use the gauge freedoms of \Cref{def:transfo_dim3} to modify it. By applying first $\Weylchg_\chi^\perp$ with $\chi := -(z/\Psi)|_\mathcal{U}$, then $\Weylchg_\omega^\parallel$ with $\omega := -(\kappa+d\ln|\Psi|)|_\mathcal{U}$ and finally $\Conf_\phi$ with the conformal factor $\Omega := |\Psi|^{-1}|_\mathcal{U}$, one obtains a gauge-equivalent 10-tuple $(\mathcal{U},\widetilde{h},\widetilde{E},\widetilde{M},\widetilde{\Psi},\widetilde{\kappa},\widetilde{K},\widetilde{z},\widetilde{T},\widetilde{Q}) \in \mathscr{D}$, thus also solution to the (CVC), such that $\widetilde{h} = \Psi^{-2} h|_\mathcal{U}$, $\widetilde{\Psi} = \pm 1$, $\widetilde{\kappa} = 0$, $\widetilde{K} = |\Psi|^{-1}(K+\varepsilon z \Psi^{-1} h)|_\mathcal{U}$ and $\widetilde{z}=0$. One deduces that $\widetilde{T} = 0$ from \eqref{eq:CVC_i}. Finally, equations \eqref{eq:CVC_k} and \eqref{eq:CVC_d} resume to the (VC) on the triple $(\mathcal{U},\widetilde{h},\widetilde{K})$.
\end{proof}

\subsubsection{Application to 4-dimensional aAdS spaces}
\label{sec:application_to_4dim_aAdS_spaces}

We would like to end this presentation of the general framework of the (CVE) by emphasizing its application for 4-dimensional aAdS spaces and the links between the conformal point of view and the FG-type asymptotic expansions near the conformal boundary described in \Cref{sec:expansion_conf_bound}.

A straightforward extension of \Cref{prop:CVE_VE} gives the following proposition.

\begin{prop} \,
    \label{prop:aAdS_VE}
    \begin{itemize}
        \item[i)] Let $(\M,\mathfrak{I},\widetilde{g})$ be a 4-dimensional aAdS space, in the sense of \Cref{def_aAdS}, with smooth rescaled metrics. If $(\M\setminus\mathfrak{I},\widetilde{g})$ is solution to the \eqref{eq:VE} -- for $\Lambda = -3$ necessarily -- then for any boundary defining function $x$ of $\mathfrak{I}$ and any smooth covector field $\kappa $ on $\M$, the rescaled Weyl tensor $x^{-1}\widetilde{W}$ extends smoothly to $\mathfrak{I}$ and the quintuple $(\M,g:=x^2\widetilde{g},x^{-1}W,x,\kappa)\in\mathscr{E}$ is solution to the \eqref{eq:CVE} with $\Lambda=-3$.

        Moreover, the conformal boundary $\mathscr{I}$, as defined in \Cref{def:conf_bound_VCE}, coincide with the conformal boundary $\mathfrak{I}$ of the aAdS space, hence the common name.
        
        \item[ii)] Conversely, let $(\M,g,V,\Theta,\kappa)\in \mathscr{E}$ be a quintuple with non-empty conformal boundary $\mathscr{I}$. If it is solution to the \eqref{eq:CVE} for $\Lambda = -3$ then for any closure $\mathcal{U}$ of a connected component of $\M\setminus\mathscr{I}$, the triple $(\mathcal{U},\mathfrak{I} := \mathscr{I}\cap\mathcal{U},\widetilde{g}:=|\Theta|^{-2} g)$ is a 4-dimensional aAdS space with smooth rescaled metrics and $(\mathcal{U}\setminus\mathfrak{I},\widetilde{g})$ is solution to the \eqref{eq:VE} -- for $\Lambda = -3$ necessarily.
    \end{itemize}
\end{prop}

\begin{proof} \,
    \begin{itemize}
        \item[i)] By $i)$ of \Cref{prop:CVE_VE}, the quintuplet $(\M\setminus\mathfrak{I},g,x^{-1}W,x,\kappa)$ is solution to the (CVE) with $\Lambda=-3$. The rescaled metric $g := x^2\widetilde{g}$ extends smoothly to $\mathfrak{I}$ by assumption. The rescaled Weyl tensor $x^{-1}W$ extends smoothly on $\mathfrak{I}$ thanks to \Cref{lem:dividing_bdf} and the asymptotic expansions of \Cref{prop:expansion_curvature_tensors} (with $n=3$). Hence the quintuplet $(\M,g,x^{-1}W,x,\kappa)$ is in $\mathscr{E}$ and, by a continuity argument, is solution to the (CVE) with $\Lambda=-3$. Since $x$ is a boundary defining function of $\mathfrak{I}$, one has $\mathscr{I}=\mathfrak{I}$.
        
        \item[ii)] By $ii)$ of \Cref{prop:CVE_VE}, $(\mathcal{U}\setminus\mathfrak{I},\widetilde{g})$ is solution to the (VE) for $\Lambda = -3$. By definition of $\mathcal{U}$, $|\Theta|$ is positive on $\mathcal{U}$. With \Cref{lem:pseudonorm_dTheta}, one deduces that $|\Theta|$ is a boundary defining function of $\mathscr{I}\cap\mathcal{U}$ and that $(\mathcal{U},\mathscr{I}\cap\mathcal{U},|\Theta|^{-2}g)$ is aAdS. \qedhere
    \end{itemize}
\end{proof}

Note that the umbilical character (see \Cref{def:tllygeod_umbilical}) of connected components of the conformal boundary has been proven by two different methods in dimension 4. First, \Cref{cor:I} states that the extrinsic curvature of a connected component of the conformal boundary for the Weyl connection $\widehat{\nabla}$ is given by
\begin{equation}
    \tag{\ref{eq:extrinsic_curvature_conf_boundary}}
    (\widehat{\mathfrak{K}}_g)_{ij} = -\frac{\widehat{s}_g |_\mathfrak{S}}{\mathfrak{z}} \mathfrak{h}_{ij} \,.
\end{equation}
This was obtained using the framework of the (CVE), which is only applicable in dimension 4. Secondly, the connected components of the conformal boundary are totally geodesic in all dimensions for the Levi-Civita connection $\nabla$ of a rescaled metric $g = x^2 \widetilde{g}$ with $x$ a boundary defining function verifying \eqref{eq:FG_bdf}, as has been deduced from the FG expansion in \Cref{sec:FG_expansion}. In particular, this implies that the connected components are umbilical by \Cref{lem:link_umbilical_totally_geodesic}.

A comparison of the two results indicates that the Friedrich scalar $s_g$, associated to a rescaled metric $g=x^2\widetilde{g}$ where $x$ is a boundary defining function verifying \eqref{eq:FG_bdf}, vanishes on the conformal boundary in dimension 4. We propose to prove this result more simply and in fact for all dimensions using the framework of \Cref{sec:aAdS_spaces}. Let $(\M,\mathfrak{I},\widetilde{g})$ be an aAdS space of dimension $n+1$. The transformation law of the scalar curvature under conformal rescalings implies that for any boundary defining function $x$ of $\mathfrak{I}$, one has
\begin{equation}
    \label{eq:scalar_curvatures}
    R - \frac{\widetilde{R}}{x^2} = n \left( -2 \frac{\Box_g x}{x} + (n+1) \frac{g^{-1}(dx,dx)}{x^2} \right) \,,
\end{equation}
where $g := x^2\widetilde{g}$ is the associated rescaled metric, $R$ is its scalar curvature and $\widetilde{R}$ is the scalar curvature of $\widetilde{g}$. Define the Friedrich scalar of $g$ by
\[ s_g := \frac{1}{n+1} \left( \Box_g x + Lx \right) \,, \]
where $L = R/2n$ is the trace of the Schouten tensor of the Levi-Civita connection $\nabla$ of $g$. From \eqref{eq:scalar_curvatures}, one deduces that
\[ s_g = \frac{1}{2x} \left( g^{-1}(dx,dx) + \frac{\widetilde{R}}{n(n+1)} \right) \,. \]
In particular, if $(\M\setminus\mathfrak{I},\widetilde{g})$ is solution to the (VE) for the normalised negative cosmological constant $\Lambda=-n(n-1)/2$ then
\[ s_g = \frac{g^{-1}(dx,dx)-1}{2x} \,. \]
Hence, $s_g$ vanishes on the conformal boundary if and only if $g^{-1}(dx,dx) = 1+\grando{x^2}$, which is a weaker condition than \eqref{eq:FG_bdf}.

\bigbreak

Finally, let us state the equivalent statement of \Cref{prop:aAdS_VE} for asymptotically hyperbolic spaces. Contrary to the aAdS case, this is not a direct application of \Cref{prop:CVC_VC}. Its proof actually relies on a new theorem which will be stated later in the article, see \Cref{thm:smooth_unphysical_fields}.

\bigbreak

\begin{prop} \,
    \label{prop:aH_VC}
    \begin{itemize}
        \item[i)] Let $(\mathcal{S},\slashed{\mathfrak{I}},\widetilde{h})$ be a 3-dimensional aH space, in the sense of \Cref{def:aH}, with smooth rescaled metrics and let $\widetilde{K}_{ij}$ be a smooth symmetric 2-tensor field on $\mathcal{S}\setminus\slashed{\mathfrak{I}}$. If $(\mathcal{S}\setminus\slashed{\mathfrak{I}},\widetilde{h},\widetilde{K})$ is solution to the \eqref{eq:VC} for $\Lambda = -3$ and if, for any boundary defining function $\Psi$ of $\slashed{\mathfrak{I}}$,
        \begin{itemize}
            \item[a)] the field $P_{ij} := \Psi \widetilde{K}_{ij}$ extends smoothly to $\slashed{\mathfrak{I}}$ and verifies
            \begin{equation}
                \tag{\ref{eq:normal_derivative_Pdagger}}
                \left( D_\perp P\right)^\dagger_{AB} = s_h P_{AB}^\dagger \quad \text{on } \slashed{\mathfrak{I}} \,,
            \end{equation}
        
        \item[b)] all connected components $\mathfrak{S}$ of $\slashed{\mathfrak{I}}$ are umbilical in $(\mathcal{S},[h])$ with
        \begin{equation}
            \tag{\ref{eq:umbilical_condition}}
            \slashed{\mathfrak{K}}_{AB} = s_h \;  \slashed{\mathfrak{h}}_{AB} \,,
        \end{equation}
    \end{itemize}
    where
    \begin{itemize}
        \item[$\bullet$] $h := \Psi^2 \widetilde{h}$ is the rescaled metric associated to $\Psi$, $D$ is the Levi-Civita connection of $h$ and $s_h$ is the non-extended 3-dimensional Friedrich scalar field of $h$ defined by \eqref{eq:def_Friedrich_scalar_dim3_nonextended},
        \item[$\bullet$] $\slashed{\mathfrak{h}}$ is the metric induced by $h$ on $\mathfrak{S}$, $\slashed{\mathfrak{K}}$ is the extrinsic curvature of $D$ on $\mathfrak{S}$,
        \item[$\bullet$] $\perp$ denotes the outward-pointing normal of $\slashed{\mathfrak{I}}$ with respect to $h$ while $A,B$ are tangential indices,
        \item[$\bullet$] $X^\dagger_{AB}$ is the trace-free part of the restriction of a 2-tensor $X_{ij}$ on $\slashed{\mathfrak{I}}$,
        \end{itemize}
        then for any boundary defining function $\Psi$ of $\slashed{\mathfrak{I}}$, any smooth covector field $\kappa$ on $\mathcal{S}$ and any smooth scalar field $z$ on $\mathcal{S}$ which vanishes on $\slashed{\mathfrak{I}}$,
        \begin{itemize}
            \item[\ding{70}] the fields $K_{ij}$, $Q_{ij}$, $E_{ij}$, $T_i$, $M_{ijk}$ defined by \eqref{CVC_def} extend smoothly to the conformal boundary $\slashed{\mathfrak{I}}$,
            \item[\ding{70}] the 10-tuple $(\mathcal{S},h,E,M,\Psi,\kappa,K,z,T,Q) \in \mathscr{D}$ is a Riemannian solution of the \eqref{eq:CVC} for $\Lambda=-3$ and verifies the geometric conditions
            \begin{subequations}
                \begin{align}
                \tag{\ref{eq:z_vanish}}
                z &= 0 \quad \text{on } \slashed{\mathscr{I}} \,, \\
                \tag{\ref{eq:K_vanish}}
                K_{\perp i} + \varepsilon \frac{z}{\Psi} h_{\perp i} &= 0 \quad \text{on } \slashed{\mathscr{I}} \,,
                \end{align}
            \end{subequations}
        \end{itemize}
        Moreover, the conformal boundary $\slashed{\mathscr{I}}$ of the 10-tuple coincide with the conformal boundary $\slashed{\mathfrak{I}}$ of the aH space, hence the common name.
        
        \item[ii)] Conversely, let $(\mathcal{S},h,E,M,\Psi,\kappa,K,z,T,Q) \in \mathscr{D}$ be a 10-tuple with non-empty conformal boundary $\slashed{\mathscr{I}}$. If it is a Riemannian solution to the \eqref{eq:CVC} for $\Lambda = -3$ and if it satisfies the geometric conditions \eqref{eq:gauge_independent_conditions_CVC} then for any closure $\mathcal{U}$ of a connected component of $\mathcal{S}\setminus\slashed{\mathfrak{I}}$, the triple $(\mathcal{U},\slashed{\mathfrak{I}} := \slashed{\mathscr{I}}\cap\mathcal{U},\widetilde{h}:=|\Psi|^{-2}h)$ is a 3-dimensional aH space with smooth rescaled metrics, $(\mathcal{U}\setminus\slashed{\mathfrak{I}},\widetilde{h}:=|\Psi|^{-2}h,\widetilde{K}:=|\Psi|^{-1}(K-z\Psi^{-1}h))$ is solution to the \eqref{eq:VC} for $\Lambda = -3$ and, for any boundary defining function $\Psi$ of $\slashed{\mathfrak{I}}$, the conditions $a)$ and $b)$ stated in $i)$ hold.
    \end{itemize}
\end{prop}

\begin{proof} \,
    \begin{itemize}
        \item[i)] By $i)$ of \Cref{prop:CVC_VC}, the 10-tuple $(\mathcal{S}\setminus\slashed{\mathfrak{I}},h,E,M,\Psi,\kappa,K,z,T,Q)$ is solution to the (CVC) for $\Lambda=-3$. By \Cref{thm:smooth_unphysical_fields}, the hypotheses imply that the unphysical fields extend smoothly to $\slashed{\mathfrak{I}}$ and verifies the geometric conditions \eqref{eq:gauge_independent_conditions_CVC}. Hence the 10-tuple $(\mathcal{S},h,E,M,\Psi,\kappa,K,z,T,Q)$ is in $\mathscr{D}$ and, by a continuity argument, is solution to the (CVC) with $\Lambda=-3$. Since $\Psi$ is a boundary defining function of $\slashed{\mathfrak{I}}$, one has $\slashed{\mathscr{I}} = \slashed{\mathfrak{I}}$.

        \item[ii)] By $ii)$ of \Cref{prop:CVC_VC}, $(\mathcal{U}\setminus\slashed{\mathscr{I}},\widetilde{h},\widetilde{K})$ is solution to the (VC) with $\Lambda=-3$. By definition of $\mathcal{U}$, $|\Psi|$ is positive on $\mathcal{U}$. Thanks to \eqref{eq:z_vanish}, one deduces from \Cref{eq:pseudonorm_dPsi} that $|\Psi|$ is a boundary defining function of $\slashed{\mathscr{I}}\cap\mathcal{U}$ and that $(\mathcal{U},\slashed{\mathscr{I}}\cap\mathcal{U},|\Psi|^{-2}h)$ is aH. Then, by \Cref{thm:smooth_unphysical_fields}, the conditions $a)$ and $b)$ hold. \qedhere
    \end{itemize}
\end{proof}

\subsection{Evolution system and propagation of the constraints}
\label{sec:geometric_existence}

This section contains the analytic part of the proof and follows, for the most part, the work of Friedrich \cite{F95}. We first introduce convenient quantities encoding the (CVE), called \emph{zero quantities}, in \Cref{sec:zero_quantities}. In a specific gauge adapted to the conformal boundary and constructed in \Cref{sec:gauge_construction}, a suitable evolution system can be derived from some of the zero quantities as shown in \Cref{sec:derivation_evolution_system}. Particular care has to be taken concerning this step so that the constraints, constituted by the remaining zero quantities, propagate tangentially to the boundary at a later step. If not, the propagation of the constraints cannot be ensured as it would require additional first-order boundary conditions. Friedrich's method consists in modifying the evolution system by adding some multiple of the constraints. Another strategy is followed in this paper: reducing the number of unknowns by using algebraic symmetries only and choosing judiciously the evolution equations. This was inspired by the work \cite{HLSW20} studying linear fields on AdS which can be seen as a mix of the two methods. The evolution system is then solved in \Cref{sec:solving_evolution_system} under certain analytic boundary conditions, see \eqref{raw_bc}. Finally, the propagation of the constraints is proven in \Cref{sec:propagation_constraints}.

For the sake of simplicity, from now on we will take the normalised negative cosmological constant, that is
\[ \Lambda = -3 \,. \]

\subsubsection{Zero quantities}
\label{sec:zero_quantities}

Let us rewrite the \eqref{eq:CVE} with the so-called \emph{zero quantities}. These vanish identically if and only if the (CVE) hold and are particularly convenient to write the propagation of the constraints, see \Cref{sec:propagation_constraints}. 

\begin{defi}
    \label{def:zero_quantities}
    Let $(\M,g,V,\Theta,\kappa) \in \mathscr{E}$ be a quintuple, $\acute{\nabla}$ be a connection on $\M$, $U_{\alpha\beta}$ be a smooth 2-tensor field on $\M$, $\zeta_\alpha$ be a smooth covector field on $\M$, $s$ be a smooth scalar function on $\M$ and $(e_{\bf a})$ be a smooth frame field on $\M$. Define the corresponding \emph{zero quantities} as follows
    \begin{subequations}
        \label{eq:zero_quantities}
        \begin{align}
            \label{eq:ZQ_varphi}
            \varphi_{\bf cdb} &:= \acute{\nabla}_{\bf a} V^{\bf a} {}_{\bf bcd} - \kappa_{\bf a} V^{\bf a} {}_{\bf bcd} \,, \\
            \label{eq:ZQ_varpi}
            \varpi^{\bf a}{}_{\bf bcd} &:= \acute{R}^{\bf a}{}_{\bf bcd} - \Theta V^{\bf a} {}_{\bf bcd} - 2 S_{\bf b[c} {}^{\bf ae} U_{\bf d]e} \,, \\
            \label{eq:ZQ_Pi}
            \Pi_{\bf ab} &:= \acute{\nabla}_{\bf a} \kappa_{\bf b} - \acute{\nabla}_{\bf b} \kappa_{\bf a} + 2 U_{\bf [ab]} \,, \\
            \label{eq:ZQ_varrho}
            \varrho_{\bf abc} &:= \acute{\nabla}_{\bf a} U_{\bf bc} - \acute{\nabla}_{\bf b} U_{\bf ac} - \zeta_{\bf d} V^{\bf d}{}_{\bf cab} \,, \\
            \label{eq:ZQ_varkappa}
            \varkappa_{\bf ab} &:= \acute{\nabla}_{\bf a} \zeta_{\bf b} + \kappa_{\bf a} \zeta_{\bf b} + \Theta U_{\bf ab} - s g_{\bf ab} \,, \\
            \label{eq:ZQ_mho}
            \mho_{\bf a} &:= e_{\bf a}( s) -\kappa_{\bf a} s + U_{\bf ab}  g^{\bf bc} \zeta_{\bf c} \,, \\
            \label{eq:ZQ_varsigma}
            \varsigma_{\bf a} &:= \zeta_{\bf a} - \Theta \, \kappa_{\bf a} - (d\Theta)_{\bf a}  \,, \\
            \label{eq:ZQ_Sigma}
            \acute{\Sigma}_{\bf a}{}^{\bf c}{}_{\bf b} &:= \acute{\Gamma}_{\bf a}{}^{\bf c}{}_{\bf b} - \acute{\Gamma}_{\bf b}{}^{\bf c}{}_{\bf a} - \langle \omega^{\bf c}, [e_{\bf a},e_{\bf b}] \rangle \,, \\
            \label{eq:ZQ_vartheta}
            \acute{\vartheta}_{\bf abc} &:= \acute{\nabla}_{\bf a} g_{\bf bc} + 2\kappa_{\bf a} g_{\bf bc} \,, \\
            \label{eq:ZQ_aleph}
            \aleph &:= \frac{1}{2} \left( g^{\bf ab}\zeta_{\bf a}\zeta_{\bf b} -1 \right) - s \Theta \,,
        \end{align}
    \end{subequations}
    where $(\omega^{\bf a})$ is the dual coframe field of $(e_{\bf a})$, $\acute{R}^{\bf a}{}_{\bf bcd}$ are the frame components of the Riemann tensor of $\acute{\nabla}$ and
    \[ \acute{\Gamma}_{\bf a}{}^{\bf c}{}_{\bf b} := \langle \omega^{\bf c}, \acute{\nabla}_{e_{\bf a}} e_{\bf b} \rangle \,. \]
\end{defi}

\begin{rems} \,
    \begin{itemize}
    \item The (CVE) are linked to the zero quantities in the following way
    \begin{alignat*}{12}
        \eqref{eq:ZQ_varphi} &\leftrightarrow \eqref{eq:CVE_V} \,, & \qquad \eqref{eq:ZQ_varpi} &\leftrightarrow \eqref{eq:CVE_Gamma} \,, & \qquad \eqref{eq:ZQ_varrho} &\leftrightarrow \eqref{eq:CVE_L} \,, &\qquad \eqref{eq:ZQ_varkappa} &\leftrightarrow \eqref{eq:CVE_zeta} \,, \\
        \eqref{eq:ZQ_varsigma} &\leftrightarrow \eqref{eq:def_zeta} \,, &\qquad
        \eqref{eq:ZQ_Sigma} &\leftrightarrow \eqref{eq_torsion_free} \,, &\qquad \eqref{eq:ZQ_vartheta} &\leftrightarrow \eqref{weyl_conn} \,, & \qquad  \eqref{eq:ZQ_aleph} &\leftrightarrow \eqref{eq:CVE_const_bis} \,.
    \end{alignat*}
    Furthermore, the remaining zero quantities come from the two equations indicated below, which are redundant to the (CVE) but must be added to form a closed system,
    \[ \eqref{eq:ZQ_Pi} \leftrightarrow \eqref{eq:antisym_part_schouten} \,, \qquad \eqref{eq:ZQ_mho} \leftrightarrow \eqref{eq:grad_s} \,. \]
    
    \item A general connection $\acute{\nabla}$ is introduced because in the propagation of the constraints we do not yet know that it is a Weyl connection. This is indeed encoded by the vanishing of $\acute{\Sigma}_{\bf a}{}^{\bf c}{}_{\bf b}$ and $\acute{\vartheta}_{\bf abc}$. 
    
    \item By definition of the Riemann tensor, one has
    \begin{align*}
        \acute{R}^{\bf a}{}_{\bf bcd} &= \langle \omega^{\bf a}, \acute{\nabla}_{\bf c} \acute{\nabla}_{\bf d} e_{\bf b} - \acute{\nabla}_{\bf d} \acute{\nabla}_{\bf c} e_{\bf b} - \acute{\nabla}_{[e_{\bf c},e_{\bf d}]} e_{\bf b} \rangle \\
        &= e_{\bf c}(\acute{\Gamma}_{\bf d}{}^{\bf a}{}_{\bf b}) - e_{\bf d}(\acute{\Gamma}_{\bf c}{}^{\bf a}{}_{\bf b}) + \acute{\Gamma}_{\bf d}{}^{\bf e}{}_{\bf b} \acute{\Gamma}_{\bf c}{}^{\bf a}{}_{\bf e} - \acute{\Gamma}_{\bf c}{}^{\bf e}{}_{\bf b} \acute{\Gamma}_{\bf d}{}^{\bf a}{}_{\bf e} + \left( \acute{\Gamma}_{\bf d}{}^{\bf e}{}_{\bf c} - \acute{\Gamma}_{\bf c}{}^{\bf e}{}_{\bf d} \right) \acute{\Gamma}_{\bf e}{}^{\bf a}{}_{\bf b} + \acute{\Sigma}_{\bf c}{}^{\bf e}{}_{\bf d} \acute{\Gamma}_{\bf e}{}^{\bf a}{}_{\bf b} \,. \qedhere
    \end{align*}
    \end{itemize}
\end{rems}

\noindent The following lemma holds as a direct consequence of \Cref{def:zero_quantities}.

\begin{lem} \,
    \label{lem:link_zero_quantities_CVE}
    \begin{itemize}
        \item[i)] Let $(\M,g,V,\Theta,\kappa) \in \mathscr{E}$ be a solution to the (CVE), $\acute{\nabla}$ be the Weyl connection $\widehat{\nabla}$ associated to $\kappa$ with respect to $g$, $U_{\alpha\beta}$ be the Schouten tensor $\widehat{L}_{\alpha\beta}$ of $\widehat{\nabla}$, $\zeta$ be the covector field $\widehat{\zeta}_g$ defined by \eqref{eq:def_zeta}, $s$ be the extended Friedrich scalar $\widehat{s}_g$ and $(e_{\bf a})$ be a smooth field on $\M$. Then all the zero quantities defined by \eqref{eq:zero_quantities} vanish identically on $\M$.
        
        \item[ii)] Reciprocally, let $(\M,g,V,\Theta,\kappa) \in \mathscr{E}$, $\acute{\nabla}$ be a connection on $\M$, $U_{\alpha\beta}$ be a smooth 2-tensor field on $\M$, $\zeta_\alpha$ be a smooth covector field on $\M$, $s$ be a smooth scalar field on $\M$ and $(e_{\bf a})$ be a smooth frame field on $\M$ such that all zero quantities defined by \eqref{eq:zero_quantities} vanish identically on $\M$. Then $\acute{\nabla}$ is the Weyl connection $\widehat{\nabla}$ associated to $\kappa$ with respect to $g$, $U_{\alpha\beta}$ is its Schouten tensor $\widehat{L}_{\alpha\beta}$, $\zeta_\alpha$ is the covector field $(\widehat{\zeta}_g)_\alpha$, $s$ is the extended Friedrich scalar $\widehat{s}_g$ and $(\M,g,V,\Theta,\kappa)$ is solution to the (CVE) with $\Lambda=-3$.
    \end{itemize}
\end{lem}

\begin{proof}
    The proof of $i)$ is immediate. For $ii)$, one has to deduce in this order that $\acute{\nabla}$ is the Weyl connection $\widehat{\nabla}$ associated to $\kappa$ with respect to $g$ from $\acute{\Sigma}_{\bf a}{}^{\bf c}{}_{\bf b} = 0$ and $\acute{\vartheta}_{\bf abc} = 0$, that $U_{\alpha\beta}$ is the Schouten tensor of $\widehat{\nabla}$ from $\varpi^{\bf a}{}_{\bf bcd} = 0$ (uniqueness of the decomposition), that $\zeta_\alpha$ is the covector field $(\widehat{\zeta}_g)_\alpha$ from $\varsigma_{\bf a} = 0$, that $s$ is the extended Friedrich scalar $\widehat{s}_g$ from $g^{\bf ab}\varkappa_{\bf ab} = 0$. It follows that $(\M,g,V,\Theta,\kappa)$ is solution to the (CVE).
\end{proof}

\noindent In order to propagate the constraints, one needs a system of differential equations on the zero quantities. This is given by the next proposition.

\begin{prop}
    Let $(\M,g,V,\Theta,\kappa) \in \mathscr{E}$ be a quintuple and $\acute{\nabla}$, $U_{\alpha\beta}$, $\zeta_\alpha$, $s$, $(e_{\bf a})$ be as in \Cref{def:zero_quantities}. Assume furthermore that $(\M,g)$ is oriented. Then the zero quantities of \Cref{def:zero_quantities} verify the following differential equations
    \begin{subequations}
        \label{eq:diff_zero_quantities}
        \begin{align}
            \label{eq:div_varphi}
            g^{\bf be} \acute{\nabla}_{\bf e} \varphi_{\bf cdb} &= -\frac{1}{2} \acute{\Sigma}_{\bf a}{}^{\bf b}{}_{\bf e} \acute{\nabla}_{\bf b} V^{\bf ae}{}_{\bf cd} - \varpi^{\bf a}{}_{\bf bae} V^{\bf be}{}_{\bf cd} - V^{\bf ae}{}_{\bf b[c} \varpi^{\bf b}{}_{\bf d]ae} \nonumber \\
            &\quad - g^{\bf ep} g^{\bf fq} \acute{\vartheta}_{\bf epq}\acute{\vartheta}_{\bf abf} V^{\bf ab}{}_{\bf cd} + g^{\bf ab} g^{\bf ef} \acute{\vartheta}_{\bf eaf} \varphi_{\bf cdb} \nonumber \\
            &\quad + g^{\bf ef} \acute{\vartheta}_{\bf abf} \acute{\nabla}_{\bf e} V^{\bf ab}{}_{\bf cd} + \kappa_{\bf a} g^{\bf ab} \left( \varphi_{\bf cdb} - \acute{\vartheta}_{\bf epb} V^{\bf ep}{}_{\bf cd} \right) \nonumber \\
            &\quad+ \frac{3}{2} V^{\bf ae}{}_{\bf cd} \Pi_{\bf ae} + g^{\bf ef} V^{\bf ab}{}_{\bf cd} \acute{\nabla}_{\bf e} \acute{\vartheta}_{\bf abf} \,, \\
            \label{eq:curl_varpi}
            \acute{\nabla}_{\bf [a} \varpi^{\bf d}{}_{\bf |e|bc]} &= - \acute{\Sigma}_{\bf [a}{}^{\bf f}{}_{\bf b} \acute{R}^{\bf d}{}_{\bf |e|c]f} - V^{\bf d}{}_{\bf e[ab} \varsigma_{\bf c]} + S_{\bf e[a}{}^{\bf df} \varrho_{\bf bc]f} \nonumber \\
            &\quad + 2 \left( g^{\bf df} \acute{\vartheta}_{\bf [ab|e|} - g^{\bf dp}g^{\bf fq} \acute{\vartheta}_{\bf [a|pq|} g_{\bf b|e|} \right) U_{\bf c]f} \nonumber \\
            &\quad - \frac{\Theta}{6} \epsilon^{\bf f}{}_{\bf abc} \epsilon{}^{\bf d}{}_{\bf e}{}^{\bf qr} \varphi_{\bf qrf} - \frac{\Theta}{6} \acute{\vartheta}_{\bf ruv} V^{\bf r}{}_{\bf fst} F_{\bf abce}{}^{\bf stuvfd} \,, \\
            \label{eq:curl_Pi}
            \acute{\nabla}_{\bf [a} \Pi_{\bf bc]} &= \frac{1}{2} \acute{\Sigma}_{\bf [a}{}^{\bf d}{}_{\bf b} \acute{\nabla}_{\bf |d|} \kappa_{\bf c]} - \frac{1}{2} \varpi^{\bf d}{}_{\bf [cab]} \kappa_{\bf d} +  \varrho_{\bf [abc]} \,, \\
            \label{eq:curl_varrho}
            \acute{\nabla}_{[\bf a} \varrho_{\bf bc]d} &= \acute{\Sigma}_{\bf [a}{}^{\bf e}{}_{\bf b} \acute{\nabla}_{\bf |e|} U_{\bf c]d} - \varpi^{\bf e}{}_{\bf [cab]} U_{\bf ed} - \varpi^{\bf e}{}_{\bf d[ab} U_{\bf c]e} - V^{\bf e}{}_{\bf d[ab} \varkappa_{\bf c]e} \nonumber \\
            &\quad - \frac{1}{6} \zeta_{\bf e} \epsilon^{\bf f}{}_{\bf abc} \epsilon{}^{\bf e}{}_{\bf d}{}^{\bf qr} \varphi_{\bf qrf} - \frac{1}{6} \zeta_{\bf e} \acute{\vartheta}_{\bf ruv} V^{\bf r}{}_{\bf fst} F_{\bf abcd}{}^{\bf stuvfe} \,, \\
            \label{eq:curl_varkappa}
            \acute{\nabla}_{\bf [a} \varkappa_{\bf b]c} &= \frac{1}{2} \acute{\Sigma}_{\bf a}{}^{\bf d}{}_{\bf b} \acute{\nabla}_{\bf d} \zeta_{\bf c} - \frac{1}{2} \varpi^{\bf d}{}_{\bf cab} \zeta_{\bf d} - \varsigma_{\bf [a} U_{\bf b]c} + \frac{1}{2} \zeta_{\bf c} \Pi_{\bf ab} \nonumber \\
            &\quad - \kappa_{\bf [a} \varkappa_{\bf b]c} + \frac{\Theta}{2} \varrho_{\bf abc} - s \acute{\vartheta}_{\bf [ab]c} + g_{\bf c[a} \mho_{\bf b]} \,, \\
            \label{eq:curl_mho}
            \acute{\nabla}_{\bf [a} \mho_{\bf b]} &= \frac{1}{2} \acute{\Sigma}_{\bf a}{}^{\bf c}{}_{\bf b} e_{\bf c}(s) -  \frac{s}{2} \Pi_{\bf ab} + \kappa_{\bf [a} \mho_{\bf b]} + \frac{1}{2} g^{\bf cd} \zeta_{\bf d} \varrho_{\bf abc} \nonumber \\
            &\quad - g^{\bf ce} g^{\bf df} \zeta_{\bf d} \acute{\vartheta}_{\bf [a|ef|} U_{\bf b]c} + g^{\bf cd} \varkappa_{\bf [a|d|} U_{\bf b]c} \,, \\
            \label{eq:curl_varsigma}
            \acute{\nabla}_{\bf [a} \varsigma_{\bf b]} &= \varkappa_{\bf [ab]} - \kappa_{\bf [a} \varsigma_{\bf b]} - \frac{\Theta}{2} \Pi_{\bf ab} - \frac{1}{2} \acute{\Sigma}_{\bf a}{}^{\bf c}{}_{\bf b} (d\Theta)_{\bf c} \,, \\
            \label{eq:curl_Sigma}
            \acute{\nabla}_{\bf [a} \acute{\Sigma}_{\bf b}{}^{\bf d}{}_{\bf c]} &= - \varpi^{\bf d}{}_{\bf [cab]} - \acute{\Sigma}_{\bf [a}{}^{\bf e}{}_{\bf b} \acute{\Sigma}_{\bf c]}{}^{\bf d}{}_{\bf e} \,, \\
            \label{eq:curl_vartheta}
            \acute{\nabla}_{\bf [a} \acute{\vartheta}_{\bf b]cd} &= \frac{1}{2} \acute{\Sigma}_{\bf a}{}^{\bf e}{}_{\bf b} \acute{\vartheta}_{\bf ecd} - g_{\bf e(c} \varpi^{\bf e}{}_{\bf d)ab} + g_{\bf cd} \Pi_{\bf ab} - 2\kappa_{\bf [a} \acute{\vartheta}_{\bf b]cd} \,, \\
            \label{eq:curl_aleph}
            \acute{\nabla}_{\bf a} \aleph &= g^{\bf bc} \zeta_{\bf c} \varkappa_{\bf ab} + s \varsigma_{\bf a} - \frac{1}{2} \zeta_{\bf b} \zeta_{\bf c} g^{\bf bd} g^{\bf ce} \acute{\vartheta}_{\bf ade} - \Theta \mho_{\bf a} \,,
        \end{align}
    \end{subequations}
    where
    \begin{subequations}
        \begin{align}
            G^{\bf uvabcd} &:= \frac{1}{2} g^{\bf uv} g^{\bf ab} g^{\bf cd} - g^{\bf au} g^{\bf bv} g^{\bf cd}  - g^{\bf cu} g^{\bf dv} g^{\bf ab} \,, \\
            F_{\bf abce}{}^{\bf stuvfd} &:= \epsilon_{\bf pe}{}^{\bf st}   \epsilon_{\bf xabc} G^{\bf uvxfdp} + \epsilon^{\bf f}{}_{\bf abc} \epsilon{}^{\bf d}{}_{\bf exy} G^{\bf uvsxty} \,.
        \end{align}
    \end{subequations}
\end{prop}

\begin{rem}
    Note that the right hand sides of equations~\eqref{eq:diff_zero_quantities} are homogeneous terms of the zero quantities except the last term in the right hand side of \eqref{eq:div_varphi}. This can be addressed by adding $\acute{\nabla}_{\bf b} \acute{\vartheta}_{\bf cde}$ as a zero quantity and expressing $\acute{\nabla}_{\bf [a} \acute{\nabla}_{\bf b]} \acute{\vartheta}_{\bf cde}$ with homogeneous terms. However, this will not be necessary in our case thanks to a triangular structure in the propagation of the constraints, see the proof of \Cref{lem:propagation_constraints}.
\end{rem}

\begin{proof}
    First, let us derive some useful identities.
    \begin{itemize}
        \item Equations \eqref{div_Q_div_Q_star} and \eqref{curl_Q_div_Q_star} can be generalised for an arbitrary connection $\acute{\nabla}$ as follows
        \[ \acute{\nabla}_\alpha (Q\star)^\alpha{}_{\beta\lambda\xi} - (\star J)_{\lambda\xi\beta} = \frac{1}{2} \epsilon_{\lambda\xi}{}^{\mu\nu} \left( \acute{\nabla}_\alpha Q^\alpha{}_{\beta\mu\nu} - J_{\mu\nu\beta} \right)
        + \frac{1}{2}Q^\alpha{}_{\beta\mu\nu} \acute{\nabla}_\alpha \epsilon_{\lambda\xi}{}^{\mu\nu} \,, \]
        \begin{align*}
            3 \left( \acute{\nabla}_{[\alpha} Q^\rho{}_{|\nu|\lambda\xi]} - S_{\nu[\alpha}{}^{\rho\gamma} J_{\lambda\xi]\gamma} \right) &= \epsilon^\beta{}_{\alpha\lambda\xi} g^{\rho\mu} \left( \acute{\nabla}_\sigma (Q\star)^\sigma{}_{\beta\mu\nu} - (\star J)_{\mu\nu\beta} \right) + (Q\star)^\sigma{}_{\beta\mu\nu} \acute{\nabla}_\sigma \left( \epsilon^\beta{}_{\alpha\lambda\xi} g^{\rho\mu} \right) \,.
        \end{align*}
    
        \item By definition of the volume form $\epsilon_{\alpha\beta\mu\nu}$ of $g$ 
        \begin{equation*}
            \acute{\nabla}_\sigma \epsilon_{\alpha\beta\mu\nu} = \left( \frac{g^{\xi\rho}}{2} \acute{\vartheta}_{\sigma\xi\rho} - 4\kappa_\sigma \right) \epsilon_{\alpha\beta\mu\nu} \,,
        \end{equation*}
        Thus
        \[ \acute{\nabla}_\sigma \left(\epsilon_{\alpha\beta\mu\nu} g^{\mu\xi} g^{\nu\rho} \right) = \acute{\vartheta}_{\bf \sigma\chi\psi} \epsilon_{\alpha\beta\mu\nu}  G^{\chi\psi\mu\xi\nu\rho} \,. \]
        In particular,
        \begin{align*}
            \acute{\nabla}_\sigma \left( \epsilon^\beta{}_{\alpha\lambda\xi} g^{\rho\mu} \right) &= \acute{\nabla}_\sigma \left( \epsilon_{\nu\alpha\lambda\xi} g^{\bf \nu\beta} g^{\rho\mu} \right) = \acute{\vartheta}_{\sigma\chi\psi} \epsilon_{\nu\alpha\lambda\xi}  G^{\bf \chi\psi\nu\beta\rho\mu} \,, \\
            \acute{\nabla}_\alpha \epsilon_{\lambda\xi}{}^{\mu\nu} &= \acute{\nabla}_\alpha \left( \epsilon_{\lambda\xi\rho\sigma} g^{\mu\rho} g^{\nu\sigma} \right) = \acute{\vartheta}_{\alpha\chi\psi} \epsilon_{\lambda\xi\rho\sigma}  G^{\chi\psi\mu\rho\nu\sigma} \,.
        \end{align*}

        \item By equation \eqref{eq:SxJ}, for all 1-form $\omega_\alpha$
        \begin{equation}
            \label{eq:id_SxV}
            S_{\lambda[\alpha}{}^{\mu\nu} V^\sigma{}_{|\nu|\beta\gamma]} \omega_\sigma = V^\mu{}_{\lambda[\alpha\beta} \omega_{\gamma]} \,.
        \end{equation}
    \end{itemize}
    Combining the three points above, one deduces that
    \begin{align}
        \acute{\nabla}_{\bf [a} V^{\bf d}{}_{\bf |e|bc]} - V^{\bf d}{}_{\bf e[ab} \kappa_{\bf c]} &= \frac{1}{6} \epsilon^{\bf f}{}_{\bf abc} \epsilon{}^{\bf d}{}_{\bf e}{}^{\bf qr} \varphi_{\bf qrf} + \frac{1}{3} (V\star)^{\bf r}{}_{\bf fpe} \acute{\vartheta}_{\bf ruv} \epsilon_{\bf sabc} G^{\bf uvsfdp} \nonumber \\
        &\quad + \frac{1}{6} \epsilon^{\bf f}{}_{\bf abc} g^{\bf dp} V^{\bf r}{}_{\bf fst} \acute{\vartheta}_{\bf ruv} \epsilon_{\bf pexy} G^{\bf uvsxty} \nonumber \\
        &= \frac{1}{6} \epsilon^{\bf f}{}_{\bf abc} \epsilon{}^{\bf d}{}_{\bf e}{}^{\bf qr} \varphi_{\bf qrf} + \frac{1}{6} \epsilon_{\bf pe}{}^{\bf xy} V^{\bf r}{}_{\bf fxy} \acute{\vartheta}_{\bf ruv} \epsilon_{\bf sabc} G^{\bf uvsfdp} \nonumber \\
        &\quad + \frac{1}{6} \epsilon^{\bf f}{}_{\bf abc} g^{\bf dp} V^{\bf r}{}_{\bf fst} \acute{\vartheta}_{\bf ruv} \epsilon_{\bf pexy} G^{\bf uvsxty} \nonumber \\
        &= \frac{1}{6} \epsilon^{\bf f}{}_{\bf abc} \epsilon{}^{\bf d}{}_{\bf e}{}^{\bf qr} \varphi_{\bf qrf} + \frac{1}{6} \acute{\vartheta}_{\bf ruv} V^{\bf r}{}_{\bf fst} F_{\bf abce}{}^{\bf stuvfd} \,.
        \label{eq:id_curl_V_varphi}
    \end{align}
    Now let us prove equations \eqref{eq:diff_zero_quantities}.
    \begin{itemize}

        \item Equation \eqref{eq:div_varphi}
        \begin{align*}
            g^{\bf be} \acute{\nabla}_{\bf e} \varphi_{\bf cdb} &=  \acute{\nabla}_{\bf e} \left(g^{\bf be} \varphi_{\bf cdb}\right) + g^{\bf ab} g^{\bf ef} \acute{\vartheta}_{\bf eaf} \varphi_{\bf cdb} - 2\kappa_{\bf e} g^{\bf be} \varphi_{\bf cdb} \\
            &= \acute{\nabla}_{\bf e} \left(g^{\bf be} \acute{\nabla}_{\bf a} V^{\bf a}{}_{\bf bcd} - \kappa_{\bf a} V^{\bf ae}{}_{\bf cd} \right) + g^{\bf ab} g^{\bf ef} \acute{\vartheta}_{\bf eaf} \varphi_{\bf cdb} - 2\kappa_{\bf e} g^{\bf be} \varphi_{\bf cdb} \\
            &= \acute{\nabla}_{\bf e} \acute{\nabla}_{\bf a}  V^{\bf ae}{}_{\bf cd} -\acute{\nabla}_{\bf e} \left( -g^{\bf ef}\acute{\vartheta}_{\bf abf} V^{\bf ab}{}_{\bf cd} + 3\kappa_{\bf a} V^{\bf ae}{}_{\bf cd} \right) + g^{\bf ab} g^{\bf ef} \acute{\vartheta}_{\bf eaf} \varphi_{\bf cdb} - 2\kappa_{\bf e} g^{\bf be} \varphi_{\bf cdb} \\
            &= - \acute{\nabla}_{\bf [a} \acute{\nabla}_{\bf e]}  V^{\bf ae}{}_{\bf cd} + \left( - g^{\bf ep} g^{\bf fq} \acute{\vartheta}_{\bf epq} + 2\kappa_{\bf e} g^{\bf ef} \right) \acute{\vartheta}_{\bf abf} V^{\bf ab}{}_{\bf cd} + g^{\bf ef} V^{\bf ab}{}_{\bf cd} \acute{\nabla}_{\bf e} \acute{\vartheta}_{\bf abf}  \\
            &\quad + g^{\bf ef} \acute{\vartheta}_{\bf abf} \acute{\nabla}_{\bf e} V^{\bf ab}{}_{\bf cd} + 3V^{\bf ae}{}_{\bf cd} \acute{\nabla}_{\bf [a} \kappa_{\bf e]} + 3\kappa_{\bf a} \acute{\nabla}_{\bf e} V^{\bf ea}{}_{\bf cd} + g^{\bf ab} g^{\bf ef} \acute{\vartheta}_{\bf eaf} \varphi_{\bf cdb}  \\
            &\quad - 2\kappa_{\bf e} g^{\bf be} \varphi_{\bf cdb} \\
            &= - \acute{\nabla}_{\bf [a} \acute{\nabla}_{\bf e]}  V^{\bf ae}{}_{\bf cd} + \left( - g^{\bf ep} g^{\bf fq} \acute{\vartheta}_{\bf epq} + 2\kappa_{\bf e} g^{\bf ef} \right) \acute{\vartheta}_{\bf abf} V^{\bf ab}{}_{\bf cd} + g^{\bf ef} V^{\bf ab}{}_{\bf cd} \acute{\nabla}_{\bf e} \acute{\vartheta}_{\bf abf}  \\
            &\quad + g^{\bf ef} \acute{\vartheta}_{\bf abf} \acute{\nabla}_{\bf e} V^{\bf ab}{}_{\bf cd} + 3V^{\bf ae}{}_{\bf cd} \left( \frac{1}{2} \Pi_{\bf ae} - U_{\bf [ae]} \right) + 3\kappa_{\bf a} g^{\bf ab} \left( \varphi_{\bf cdb} - \acute{\vartheta}_{\bf epb} V^{\bf ep}{}_{\bf cd} \right) \\
            &\quad + g^{\bf ab} g^{\bf ef} \acute{\vartheta}_{\bf eaf} \varphi_{\bf cdb} - 2\kappa_{\bf e} g^{\bf be} \varphi_{\bf cdb} \\
            &= - \acute{\nabla}_{\bf [a} \acute{\nabla}_{\bf e]}  V^{\bf ae}{}_{\bf cd} - g^{\bf ep} g^{\bf fq} \acute{\vartheta}_{\bf epq}\acute{\vartheta}_{\bf abf} V^{\bf ab}{}_{\bf cd} + g^{\bf ab} g^{\bf ef} \acute{\vartheta}_{\bf eaf} \varphi_{\bf cdb} + g^{\bf ef} \acute{\vartheta}_{\bf abf} \acute{\nabla}_{\bf e} V^{\bf ab}{}_{\bf cd} \\
            &\quad + 3V^{\bf ae}{}_{\bf cd} \left( \frac{1}{2} \Pi_{\bf ae} - U_{\bf ae} \right) + \kappa_{\bf a} g^{\bf ab} \left( \varphi_{\bf cdb} - \acute{\vartheta}_{\bf epb} V^{\bf ep}{}_{\bf cd} \right) + g^{\bf ef} V^{\bf ab}{}_{\bf cd} \acute{\nabla}_{\bf e} \acute{\vartheta}_{\bf abf} \,.
        \end{align*}
        Moreover,
        \begin{align*}
            - \acute{\nabla}_{\bf [a} \acute{\nabla}_{\bf e]}  V^{\bf ae}{}_{\bf cd} &= - \frac{1}{2} \acute{\Sigma}_{\bf a}{}^{\bf b}{}_{\bf e} \acute{\nabla}_{\bf b} V^{\bf ae}{}_{\bf cd} -\frac{1}{2} \acute{R}^{\bf a}{}_{\bf bae} V^{\bf be}{}_{\bf cd} - \frac{1}{2} \acute{R}^{\bf e}{}_{\bf bae} V^{\bf ab}{}_{\bf cd} \\
            &\quad + \frac{1}{2} \acute{R}^{\bf b}{}_{\bf cae} V^{\bf ae}{}_{\bf bd} + \frac{1}{2} \acute{R}^{\bf b}{}_{\bf dae} V^{\bf ae}{}_{\bf cb} \\
            &= -\frac{1}{2} \acute{\Sigma}_{\bf a}{}^{\bf b}{}_{\bf e} \acute{\nabla}_{\bf b} V^{\bf ae}{}_{\bf cd} - \acute{R}^{\bf a}{}_{\bf bae} V^{\bf be}{}_{\bf cd} - V^{\bf ae}{}_{\bf b[c} \acute{R}^{\bf b}{}_{\bf d]ae} \\
            &= -\frac{1}{2} \acute{\Sigma}_{\bf a}{}^{\bf b}{}_{\bf e} \acute{\nabla}_{\bf b} V^{\bf ae}{}_{\bf cd} - \left(\varpi^{\bf a}{}_{\bf bae} + 2S_{\bf b[a}{}^{\bf af} U_{\bf e]f} \right) V^{\bf be}{}_{\bf cd} \\
            &\quad - V^{\bf ae}{}_{\bf b[c} \left( \varpi^{\bf b}{}_{\bf d]ae} + 2 S_{\bf d][a}{}^{\bf bf}U_{\bf e]f} \right) \\
            &= -\frac{1}{2} \acute{\Sigma}_{\bf a}{}^{\bf b}{}_{\bf e} \acute{\nabla}_{\bf b} V^{\bf ae}{}_{\bf cd} - \left(\varpi^{\bf a}{}_{\bf bae} + S_{\bf ba}{}^{\bf af} U_{\bf ef} \right) V^{\bf be}{}_{\bf cd} \\
            &\quad - V^{\bf ae}{}_{\bf b[c} \left( \varpi^{\bf b}{}_{\bf d]ae} + 2 S_{\bf d]a}{}^{\bf bf}U_{\bf ef} \right) \\
            &= -\frac{1}{2} \acute{\Sigma}_{\bf a}{}^{\bf b}{}_{\bf e} \acute{\nabla}_{\bf b} V^{\bf ae}{}_{\bf cd} - \varpi^{\bf a}{}_{\bf bae} V^{\bf be}{}_{\bf cd} - V^{\bf ae}{}_{\bf b[c} \varpi^{\bf b}{}_{\bf d]ae} + 3 V^{\bf ae}{}_{\bf cd} U_{\bf ae} \,.
        \end{align*}
        Hence the result.
        
        \item Equation \eqref{eq:curl_varpi}. Note that by definition \eqref{eq:def_S} of the tensor $S_{\alpha\beta}{}^{\mu\nu}$,
        \begin{align*}
            \acute{\nabla}_{\bf a} S_{\bf bc}{}^{\bf de} &= - g^{\bf de} \acute{\nabla}_{\bf a} g_{\bf bc} - g_{\bf bc} \acute{\nabla}_{\bf a} g^{\bf de} \\
            &= - g^{\bf de} \left( \acute{\vartheta}_{\bf abc} - 2 \kappa_{\bf a} g_{\bf bc} \right) - g_{\bf bc} \left( -g^{\bf dp} g^{\bf eq} \acute{\vartheta}_{\bf apq} + 2\kappa_{\bf a} g^{\bf de} \right) \\
            &= - g^{\bf de} \acute{\vartheta}_{\bf abc} + g_{\bf bc} g^{\bf dp} g^{\bf eq} \acute{\vartheta}_{\bf apq} \,.
        \end{align*}
        Now, with \eqref{eq:id_SxV},
        \begin{align*}
            \acute{\nabla}_{\bf [a} \varpi^{\bf d}{}_{\bf |e|bc]} &= \acute{\nabla}_{\bf [a} \acute{R}^{\bf d}{}_{\bf |e|bc]} - \Theta \acute{\nabla}_{\bf [a} V^{\bf d}{}_{\bf |e|bc]} - V^{\bf d}{}_{\bf e[ab} (d\Theta)_{\bf c]} \\
            &\quad + 2 S_{\bf e[a}{}^{\bf df} \acute{\nabla}_{\bf b} U_{\bf c]f} - 2 \acute{\nabla}_{\bf [a} \left( S_{\bf |e|b}{}^{\bf df} \right) U_{\bf c]f} \\
            &= \acute{\nabla}_{\bf [a} \acute{R}^{\bf d}{}_{\bf |e|bc]} - \Theta \left( \acute{\nabla}_{\bf [a} V^{\bf d}{}_{\bf |e|bc]} -V^{\bf d}{}_{\bf e[ab} \kappa_{\bf c]} \right) - V^{\bf d}{}_{\bf e[ab}  \left( \varsigma_{\bf c]} - \zeta_{\bf c]} \right)  \\
            &\quad + S_{\bf e[a}{}^{\bf df} \left(\varrho_{\bf bc]f} - V^{\bf p}{}_{\bf |f|bc]} \zeta_{\bf p} \right) + 2 \left( g^{\bf df} \acute{\vartheta}_{\bf [ab|e|} - g^{\bf dp}g^{\bf fq} \acute{\vartheta}_{\bf [a|pq|} g_{\bf b|e|} \right) U_{\bf c]f} \\
            &= \acute{\nabla}_{\bf [a} \acute{R}^{\bf d}{}_{\bf |e|bc]} - \Theta \left( \acute{\nabla}_{\bf [a} V^{\bf d}{}_{\bf |e|bc]} - V^{\bf d}{}_{\bf e[ab} \kappa_{\bf c]} \right)- V^{\bf d}{}_{\bf e[ab}  \varsigma_{\bf c]} \\
            &\quad + S_{\bf e[a}{}^{\bf df} \varrho_{\bf bc]f} + 2 \left( g^{\bf df} \acute{\vartheta}_{\bf [ab|e|} - g^{\bf dp}g^{\bf fq} \acute{\vartheta}_{\bf [a|pq|} g_{\bf b|e|} \right) U_{\bf c]f} \,.
        \end{align*}
        The second Bianchi identity for $\acute{\nabla}$ gives
        \[ \acute{\nabla}_{\bf [a} \acute{R}^{\bf d}{}_{\bf |e|bc]} = - \acute{\Sigma}_{\bf [a}{}^{\bf f}{}_{\bf b} \acute{R}^{\bf d}{}_{\bf |e|c]f} \,, \]
        and one concludes with \eqref{eq:id_curl_V_varphi}.

        \item Equation \eqref{eq:curl_Pi}
        \begin{align*}
            \acute{\nabla}_{\bf [a} \Pi_{\bf bc]} &= \acute{\nabla}_{\bf [a} \acute{\nabla}_{\bf b} \kappa_{\bf c]} + 2 \acute{\nabla}_{\bf [a} U_{\bf bc]} \\
            &= \frac{1}{2} \acute{\Sigma}_{\bf [a}{}^{\bf d}{}_{\bf b} \acute{\nabla}_{\bf |d|} \kappa_{\bf c]} - \frac{1}{2} \acute{R}^{\bf d}{}_{\bf [cab]} \kappa_{\bf d} + 2 \acute{\nabla}_{\bf [a} U_{\bf bc]} \\
            &= \frac{1}{2} \acute{\Sigma}_{\bf [a}{}^{\bf d}{}_{\bf b} \acute{\nabla}_{\bf |d|} \kappa_{\bf c]} - \frac{1}{2} \varpi^{\bf d}{}_{\bf [cab]} \kappa_{\bf d} +  \varrho_{\bf [abc]} \,.
        \end{align*}

        \item Equation \eqref{eq:curl_varrho}
        \begin{align*}
            \acute{\nabla}_{[\bf a} \varrho_{\bf bc]d} &= 2 \acute{\nabla}_{\bf [a} \acute{\nabla}_{\vphantom{[}\bf b} U_{\bf c]d} - \zeta_{\bf e} \acute{\nabla}_{\bf [a} V^{\bf e}{}_{\bf |d|bc]} - V^{\bf e}{}_{\bf d[ab} \acute{\nabla}_{\bf c]} \zeta_{\bf e}  \\
            &= \acute{\Sigma}_{\bf [a}{}^{\bf e}{}_{\bf b} \acute{\nabla}_{\bf |e|} U_{\bf c]d} -\acute{R}^{\bf e}{}_{\bf [cab]} U_{\bf ed} - \acute{R}^{\bf e}{}_{\bf d[ab} U_{\bf c]e} \\
            &\quad - \zeta_{\bf e} \acute{\nabla}_{\bf [a} V^{\bf e}{}_{\bf |d|bc]} - V^{\bf e}{}_{\bf d[ab} \acute{\nabla}_{\bf c]} \zeta_{\bf e} \\
            &= \acute{\Sigma}_{\bf [a}{}^{\bf e}{}_{\bf b} \acute{\nabla}_{\bf |e|} U_{\bf c]d} - \varpi^{\bf e}{}_{\bf [cab]} U_{\bf ed} - \left(\varpi^{\bf e}{}_{\bf d[ab} + \Theta V^{\bf e}{}_{\bf d[ab} \right) U_{\bf c]e} \\
            &\quad - \zeta_{\bf e} \acute{\nabla}_{\bf [a} V^{\bf e}{}_{\bf |d|bc]} - V^{\bf e}{}_{\bf d[ab} \left( \varkappa_{\bf c]e} - \kappa_{\bf c]}\zeta_{\bf e} - \Theta U_{\bf c]e} \right) \\
            &= \acute{\Sigma}_{\bf [a}{}^{\bf e}{}_{\bf b} \acute{\nabla}_{\bf |e|} U_{\bf c]d} - \varpi^{\bf e}{}_{\bf [cab]} U_{\bf ed} - \varpi^{\bf e}{}_{\bf d[ab} U_{\bf c]e} - V^{\bf e}{}_{\bf d[ab} \varkappa_{\bf c]e} \\
            &\quad - \zeta_{\bf e} \left(\acute{\nabla}_{\bf [a} V^{\bf e}{}_{\bf |d|bc]} - V^{\bf e}{}_{\bf d[ab} \kappa_{\bf c]} \right) \,,
        \end{align*}
        and one concludes using \eqref{eq:id_curl_V_varphi}.
        
        \item Equation \eqref{eq:curl_varkappa}
        \begin{align*}
            \acute{\nabla}_{\bf [a} \varkappa_{\bf b]c} &= \acute{\nabla}_{\bf [a} \acute{\nabla}_{\bf b]} \zeta_{\bf c} + \zeta_{\bf c} \acute{\nabla}_{\bf [a} \kappa_{\bf b]} - \kappa_{\bf [a} \acute{\nabla}_{\bf b]} \zeta_{\bf c} + (d\Theta)_{\bf [a} U_{\bf b]c} \\
            &\quad + \Theta \acute{\nabla}_{\bf [a} U_{\bf b]c} - s \acute{\nabla}_{\bf [a} g_{\bf b]c} + g_{\bf c[a} \acute{\nabla}_{\bf b]} s  \\
            &= \frac{1}{2} \acute{\Sigma}_{\bf a}{}^{\bf d}{}_{\bf b} \acute{\nabla}_{\bf d} \zeta_{\bf c} - \frac{1}{2} \acute{R}^{\bf d}{}_{\bf cab} \zeta_{\bf d} + \zeta_{\bf c} \acute{\nabla}_{\bf [a} \kappa_{\bf b]} - \kappa_{\bf [a} \acute{\nabla}_{\bf b]} \zeta_{\bf c} \\
            &\quad + (d\Theta)_{\bf [a} U_{\bf b]c} + \Theta \acute{\nabla}_{\bf [a} U_{\bf b]c} - s \acute{\nabla}_{\bf [a} g_{\bf b]c} + g_{\bf c[a} \acute{\nabla}_{\bf b]} s \\
            &= \frac{1}{2} \acute{\Sigma}_{\bf a}{}^{\bf d}{}_{\bf b} \acute{\nabla}_{\bf d} \zeta_{\bf c} - \frac{1}{2} \left( \varpi^{\bf d}{}_{\bf cab} + \Theta V^{\bf d}{}_{\bf cab} + 2S_{\bf c[a}{}^{\bf de} U_{\bf b]e} \right) \zeta_{\bf d} \\
            &\quad + \zeta_{\bf c} \left( \frac{1}{2} \Pi_{\bf ab} - U_{\bf [ab]} \right) - \kappa_{\bf [a} \left( \varkappa_{\bf b]c} - \Theta U_{\bf b]c} + sg_{\bf b]c} \right) \\
            &\quad + (d\Theta)_{\bf [a} U_{\bf b]c} + \frac{\Theta}{2}  \left( \varrho_{\bf abc} + \zeta_{\bf d} V^{\bf d}{}_{\bf cab} \right) \\
            &\quad - s \left( \acute{\vartheta}_{\bf [ab]c} -2\kappa_{\bf [a}g_{\bf b]c} \right) + g_{\bf c[a} \left( \mho_{\bf b]} + \kappa_{\bf b]} s - U_{\bf b]d} g^{\bf de} \zeta_{\bf e} \right) \\
            &= \frac{1}{2} \acute{\Sigma}_{\bf a}{}^{\bf d}{}_{\bf b} \acute{\nabla}_{\bf d} \zeta_{\bf c} - \frac{1}{2} \varpi^{\bf d}{}_{\bf cab} \zeta_{\bf d} - \zeta_{\bf [a} U_{\bf b]c} + \frac{1}{2} \zeta_{\bf c} \Pi_{\bf ab} \\
            &\quad - \kappa_{\bf [a} \left( \varkappa_{\bf b]c} - \Theta U_{\bf b]c} \right) + (d\Theta)_{\bf [a} U_{\bf b]c} + \frac{\Theta}{2} \varrho_{\bf abc} - s \acute{\vartheta}_{\bf [ab]c} + g_{\bf c[a} \mho_{\bf b]} \\
            &= \frac{1}{2} \acute{\Sigma}_{\bf a}{}^{\bf d}{}_{\bf b} \acute{\nabla}_{\bf d} \zeta_{\bf c} - \frac{1}{2} \varpi^{\bf d}{}_{\bf cab} \zeta_{\bf d} - \varsigma_{\bf [a} U_{\bf b]c} + \frac{1}{2} \zeta_{\bf c} \Pi_{\bf ab} \\
            &\quad - \kappa_{\bf [a} \varkappa_{\bf b]c} + \frac{\Theta}{2} \varrho_{\bf abc} - s \acute{\vartheta}_{\bf [ab]c} + g_{\bf c[a} \mho_{\bf b]} \,.
        \end{align*}

        \item Equation \eqref{eq:curl_mho}
        \begin{align*}
            \acute{\nabla}_{\bf [a} \mho_{\bf b]} &= \acute{\nabla}_{\bf [a} \acute{\nabla}_{\bf b]} s - s \acute{\nabla}_{\bf [a} \kappa_{\bf b]} + \kappa_{\bf [a} \acute{\nabla}_{\bf b]} s + g^{\bf cd} \zeta_{\bf d} \acute{\nabla}_{\bf [a} U_{\bf b]c} \\
            &\quad - g^{\bf ce} g^{\bf df} \zeta_{\bf d} \acute{\nabla}_{\bf [a} g_{\bf |ef|} U_{\bf b]c} + g^{\bf cd} \acute{\nabla}_{\bf [a} \zeta_{\bf |d|} U_{\bf b]c} \\
            &= \frac{1}{2} \acute{\Sigma}_{\bf a}{}^{\bf c}{}_{\bf b} e_{\bf c}(s) - s \left( \frac{1}{2} \Pi_{\bf ab} - U_{\bf [ab]} \right) + \kappa_{\bf [a} \left( \mho_{\bf b]} - U_{\bf b]c} g^{\bf cd} \zeta_{\bf d} \right) \\
            &\quad + \frac{1}{2} g^{\bf cd} \zeta_{\bf d} \varrho_{\bf abc} - g^{\bf ce} g^{\bf df} \zeta_{\bf d} \left(\acute{\vartheta}_{\bf [a|ef|} -2g_{\bf ef} \kappa_{\bf [a} \right) U_{\bf b]c} \\
            &\quad + g^{\bf cd} \left( \varkappa_{\bf [a|d|} - \zeta_{\bf d} \kappa_{\bf [a} + s g_{\bf d[a} \right) U_{\bf b]c} \\
            &= \frac{1}{2} \acute{\Sigma}_{\bf a}{}^{\bf c}{}_{\bf b} e_{\bf c}(s) -  \frac{s}{2} \Pi_{\bf ab} + \kappa_{\bf [a} \mho_{\bf b]} + \frac{1}{2} g^{\bf cd} \zeta_{\bf d} \varrho_{\bf abc} \\
            &\quad - g^{\bf ce} g^{\bf df} \zeta_{\bf d} \acute{\vartheta}_{\bf [a|ef|} U_{\bf b]c} + g^{\bf cd} \varkappa_{\bf [a|d|} U_{\bf b]c} \,.
        \end{align*}
        
        \item Equation \eqref{eq:curl_varsigma}
        \begin{align*}
            \acute{\nabla}_{\bf [a} \varsigma_{\bf b]} &= \acute{\nabla}_{\bf [a} \zeta_{\bf b]} - (d\Theta)_{\bf [a} \kappa_{\bf b]} - \Theta \acute{\nabla}_{\bf [a} \kappa_{\bf b]} - \acute{\nabla}_{[\bf a} \acute{\nabla}_{\bf b]} \Theta \\
            &= \left( \varkappa_{\bf [ab]} - \kappa_{\bf [a}\zeta_{\bf b]} - \Theta U_{\bf [ab]} \right) + \kappa_{\bf [a} (d\Theta)_{\bf b]} - \Theta \left( \frac{1}{2} \Pi_{\bf ab} - U_{\bf [ab]} \right) - \frac{1}{2} \acute{\Sigma}_{\bf a}{}^{\bf c}{}_{\bf b} (d\Theta)_{\bf c} \\
            &= \varkappa_{\bf [ab]} - \kappa_{\bf [a}\varsigma_{\bf b]} - \frac{\Theta}{2} \Pi_{\bf ab} - \frac{1}{2} \acute{\Sigma}_{\bf a}{}^{\bf c}{}_{\bf b} (d\Theta)_{\bf c} \,.
        \end{align*}

        \item Equation \eqref{eq:curl_Sigma}. With the first Bianchi identity for $\acute{\nabla}$,
        \[ \acute{\nabla}_{\bf [a} \acute{\Sigma}_{\bf b}{}^{\bf d}{}_{\bf c]} = - \acute{R}^{\bf d}{}_{\bf [cab]} - \acute{\Sigma}_{\bf [a}{}^{\bf e}{}_{\bf b} \acute{\Sigma}_{\bf c]}{}^{\bf d}{}_{\bf e} = - \varpi^{\bf d}{}_{\bf [cab]} - \acute{\Sigma}_{\bf [a}{}^{\bf e}{}_{\bf b} \acute{\Sigma}_{\bf c]}{}^{\bf d}{}_{\bf e} \,. \]

        \item Equation \eqref{eq:curl_vartheta}
        \begin{align*}
            \acute{\nabla}_{\bf [a} \acute{\vartheta}_{\bf b]cd} &= \acute{\nabla}_{\bf [a} \acute{\nabla}_{\bf b]} g_{\bf cd} + 2 g_{\bf cd} \acute{\nabla}_{\bf [a} \kappa_{\bf b]} - 2\kappa_{\bf [a} \acute{\nabla}_{\bf b]} g_{\bf cd} \\
            &= \frac{1}{2} \acute{\Sigma}_{\bf a}{}^{\bf e}{}_{\bf b} \acute{\nabla}_{\bf e} g_{\bf cd} - \frac{1}{2} \left( \acute{R}^{\bf e}{}_{\bf cab} g_{\bf ed} + \acute{R}^{\bf e}{}_{\bf dab} g_{\bf ce} \right) + 2 g_{\bf cd} \acute{\nabla}_{\bf [a} \kappa_{\bf b]} - 2\kappa_{\bf [a} \acute{\vartheta}_{\bf b]cd} \\
            &= \frac{1}{2} \acute{\Sigma}_{\bf a}{}^{\bf e}{}_{\bf b} \acute{\vartheta}_{\bf ecd} - \frac{1}{2} \Big( \varpi^{\bf e}{}_{\bf cab} g_{\bf ed} + 2 S_{\bf c[a}{}^{\bf ef} U_{\bf b]f} g_{\bf ed} + \varpi^{\bf e}{}_{\bf dab} g_{\bf ce} \\
            &\quad + 2 S_{\bf d[a}{}^{\bf ef} U_{\bf b]f} g_{\bf ce} \Big) + g_{\bf cd} \left( \Pi_{\bf ab} - 2U_{\bf [ab]} \right) - 2\kappa_{\bf [a} \acute{\vartheta}_{\bf b]cd} \\
            &= \frac{1}{2} \acute{\Sigma}_{\bf a}{}^{\bf e}{}_{\bf b} \acute{\vartheta}_{\bf ecd} - g_{\bf e(c} \varpi^{\bf e}{}_{\bf d)ab} + g_{\bf cd} \Pi_{\bf ab} - 2\kappa_{\bf [a} \acute{\vartheta}_{\bf b]cd} \,.
        \end{align*}

        \item Equation \eqref{eq:curl_aleph}
        \begin{align*}
            \acute{\nabla}_{\bf a} \aleph &= g^{\bf bc} \zeta_{\bf c} \acute{\nabla}_{\bf a} \zeta_{\bf b} - \frac{1}{2} \zeta_{\bf b} \zeta_{\bf c} g^{\bf bd} g^{\bf ce} \acute{\nabla}_{\bf a} g_{\bf de} - \Theta \acute{\nabla}_{\bf a} s - s \acute{\nabla}_{\bf a} \Theta \\
            &= g^{\bf bc} \zeta_{\bf c} \left( \varkappa_{\bf ab} - \kappa_{\bf a} \zeta_{\bf b} - \Theta U_{\bf ab} + sg_{\bf ab} \right) - \frac{1}{2} \zeta_{\bf b} \zeta_{\bf c} g^{\bf bd} g^{\bf ce} \left( \acute{\vartheta}_{\bf ade} - 2 \kappa_{\bf a} g_{\bf de} \right) \\
            &\quad - \Theta \left(\mho_{\bf a} + \kappa_{\bf a} s - U_{\bf ab} g^{\bf bc} \zeta_{\bf c} \right) - s (d\Theta)_{\bf a} \\
            &= g^{\bf bc} \zeta_{\bf c} \varkappa_{\bf ab} + s \varsigma_{\bf a} - \frac{1}{2} \zeta_{\bf b} \zeta_{\bf c} g^{\bf bd} g^{\bf ce} \acute{\vartheta}_{\bf ade} - \Theta \mho_{\bf a} \,. \qedhere
        \end{align*}
    \end{itemize}
\end{proof}

\subsubsection{Gauge construction}
\label{sec:gauge_construction}

In order to solve the (CVE), one needs to fix the gauge invariances presented in \Cref{prop:invariance_CVE} to obtain a suitable hyperbolic system. A gauge choice for the (CVE) consists of choosing a conformal factor $\Theta$, a Weyl connection $\widehat{\nabla}$, an orthonormal frame field $(e_{\bf a})$ and a coordinate system $(y^\mu)$. The next proposition states the existence of a gauge choice which is adapted to the conformal boundary.

\begin{prop}[Adaptation of Friedrich \protect{\cite[Section 5.1]{F95}}]
    \label{prop:construction_gauge}
    Let $(\M,g,V,\Theta,\kappa)\in\mathscr{E}$ be a time-oriented and oriented solution to the \eqref{eq:CVE} for $\Lambda=-3$. Assume that there exists a compact spacelike hypersurface $\mathcal{S}_\star$ with (non-empty) boundary $\partial\mathcal{S}_\star$ such that
    \begin{subequations}
        \begin{align}
            \label{eq:hyp_boundary}
            \partial \mathcal{S}_\star &= \mathcal{S}_\star \cap \mathscr{I} \,, \\
            \label{eq:hyp_Sigma}
            n_\star(\Theta) &= 0 \quad \text{on } \partial \mathcal{S}_\star \,,
        \end{align}
    \end{subequations}
    where $n_\star$ is the future-directed unit normal vector field of $\mathcal{S}_\star$ with respect to $g$. The second assumption means that the normal $n_\star$ is tangent to $\mathscr{I}$ on $\partial \mathcal{S}_\star$.
    
    Then for all $p \in \mathcal{S}_\star$ there exists a neighbourhood $\mathcal{U}$ of $p$ in $\M$ and a gauge choice $\left(\breve{\Theta},\widecheck{\nabla},(e_{\bf a}),(y^\mu)\right)$ on $\mathcal{U}$ -- giving rise to a solution $\left(\mathcal{U},\breve{g},\breve{V},\breve{\Theta},\widecheck{\kappa}_{\breve{g}}\right)$ of the (CVE) which is gauge-equivalent to $\left(\mathcal{U},g,V,\Theta,\kappa\right)$ -- verifying the following core properties
    \begin{subequations}
        \label{eq:gauge_core_properties}
        \begin{align}
            \label{eq:core_a}
            \mathcal{U}_\star := \mathcal{S_\star} \cap \mathcal{U} &= \{q \in \mathcal{U} \mid y^0(q) = 0\}  \,, \\
            \label{eq:core_b}
            \breve{\Theta}_\star := \breve{\Theta} |_{\mathcal{U}_\star} &\geq 0 \,, \\
            \label{eq:core_c}
            e_{\bf i}{}^0 |_{\mathcal{U}_\star} &= 0 \,, \\
            \label{eq:core_d}
            y^3 &\geq 0 \,, \\
            \label{eq:core_e}
            e_{\bf 0}{}^\mu &= \delta_{\bf 0}{}^\mu \,, \\
            \label{eq:core_f}
            \widecheck{\Gamma}_{\bf 0}{}^{\bf b}{}_{\bf a} &= 0 \,, \\
            \label{eq:core_g}
            \widecheck{L}_{\bf 0a} &= 0 \,, \\
            \label{eq:core_h}
            (\widecheck{\kappa}_{\breve{g}})_{\bf 0} &= 0 \,, \\
            \label{eq:core_i}
            \breve{g}_{\bf ab} &= \eta_{\bf ab} \,, \\
            \label{eq:core_j}
            \breve{\epsilon}_{\bf 0123} &= +1 \,,
        \end{align}
    \end{subequations}
    and the following additional properties
    \begin{subequations}
        \label{gauge_additional_properties}
        \begin{align}
            \label{gauge_s}
            \widecheck{s}_{\breve{g}}(y^0,\dots,y^3) &= (\widecheck{s}_{\breve{g}})_\star(y^1,y^2,y^3) \,, \\
            \label{gauge_zeta_0}
            (\widecheck{\zeta}_{\breve{g}})_{\bf 0}(y^0,\dots,y^3) &= - (\widecheck{s}_{\breve{g}})_\star(y^1,y^2,y^3) \, y^0 \,, \\
            \label{gauge_Theta}
            \breve{\Theta}(y^0,\dots,y^3) &= \breve{\Theta}_\star(y^1,y^2,y^3) \left( 1 - f_\star(y^1,y^2,y^3) \frac{(y^0)^2}{2} \right) \,, \\
            \label{gauge_IcapU}
            \mathscr{I} \cap \mathcal{U} &= \{ q \in \mathcal{U} \mid y^3(q) = 0\} \,,
        \end{align}
    \end{subequations}
    where $(\widecheck{s}_{\breve{g}})_\star := \widecheck{s}_{\breve{g}}|_{\mathcal{U}_\star}$ and $f_\star := \breve{\Theta}_\star^{-1} (\widecheck{s}_{\breve{g}})_\star$ are smooth functions on $\mathcal{U}_\star$.
\end{prop}

\begin{rems} \,
    \begin{itemize}
        \item Only the additional properties rely on the hypothesis that $(\M,g,V,\Theta,\kappa)$ is a solution to the (CVE).
        
        \item The proof has been improved from Friedrich's original version which is based on \Cref{lem:conf_geod_stay_conf_bound}, $ii)$ of \Cref{prop:CVE_VE} and \Cref{prop_conf_factor}. The modifications enable to treat all conformal geodesics at once and only refers to the (CVE), instead of discriminating the conformal geodesics by whether they start on the conformal boundary or not and referring to the (VE).
        
        \item As pointed out by Friedrich, one has $e_{\bf i}{}^0=0$ on $\mathcal{U}_\star$ but not necessarily for all times $y^0$. Thus the vector fields $(e_{\bf i})$ are in general not tangent to the spacelike hypersurfaces of constant time coordinate $y^0$. They only provide a foliation of $\mathcal{U}$ by 3-dimensional spacelike distributions. \qedhere
    \end{itemize}
\end{rems}

\begin{proof}
    Apart from the modifications mentioned in the above remarks, the proof follows \cite[Section 5.1]{F95}. We divide the proof into five steps.
    \begin{itemize}
        \item \underline{Step 1}: Modifying the quintuple $(\M,g,V,\Theta,\kappa)$ via the gauge invariances
        
        Taking advantage of the gauge invariance of the (CVE) described in \Cref{prop:invariance_CVE}, one can assume that
        \begin{equation}
            \label{conditions}
            \Theta \geq 0 \quad \text{on } \mathcal{S}_\star \,, \qquad \langle\widehat{\zeta}_g,n_\star\rangle = 0 \quad \text{on } \mathcal{S}_\star \,, \qquad \widehat{s}_g = 0 \quad \text{on } \partial\mathcal{S}_\star\,. 
        \end{equation}
        Indeed,
        \begin{itemize}
            \item[\ding{70}] By hypothesis \eqref{eq:hyp_boundary}, the function $\Theta$ does not change sign on $\mathcal{S}_\star$. Thus, one can use the invariance by sign change $\Signchg$ to ensure that the gauge-equivalent quintuple thereby obtained satisfies the inequality in \eqref{conditions}.
            
            \item[\ding{70}] One can find a smooth covector field $\omega$ on $\M$ such that
            \begin{alignat*}{3}
                g^{-1}(\omega,d\Theta) &= -\widehat{s}_g \quad && \text{on } \partial\mathcal{S}_\star \,, \\
                \langle\omega,n_\star\rangle &= - \Theta^{-1} \langle\widehat{\zeta}_g,n_\star\rangle \quad && \text{on } \mathcal{S}_\star \,.
            \end{alignat*}
            Note that the right hand side of the second equation is smooth up to $\partial\mathcal{S}_\star$ thanks to hypotheses \eqref{eq:hyp_boundary}-\eqref{eq:hyp_Sigma}. Applying the Weyl connection change $\Weylchg_\omega$ guarantees that the gauge-equivalent quintuple thereby obtained satisfies the last two equations of \eqref{conditions} thanks to \eqref{eq:transfo_zeta} and \eqref{eq:transfo_s}.
        \end{itemize}
        The two operations described above can be performed in any order and do not conflict by \eqref{eq:transfo_zeta} and \Cref{prop:transfo_s}.
        
        \item \underline{Step 2}: Choice of conformal geodesics and construction of $\mathcal{U}$
        
        Let us assume that \eqref{conditions} hold thanks to step 1. Let $\widehat{\nabla}$ be the Weyl connection associated to $\kappa$ with respect to $g$. From each point $q$ of $\mathcal{S}_\star$, start a future-directed $\widehat{\nabla}$-conformal geodesic $(x,\widehat{\beta})$ with initial parameter $\tau_\star = 0$ and initial data
        \begin{equation}
            \label{initial_cdt_conf_geo}
            x_\star = q \,, \quad \dot{x}_\star = n_\star|_q \,, \quad \widehat{\beta}_\star = 0 \,.
        \end{equation}
        Fix $p \in \mathcal{S}_\star$ and let $\mathcal{U}_\star$ be an open neighbourhood of $p$ in $\mathcal{S}_\star$. Take $\overline{\tau} > 0$ such that all the conformal geodesics starting on $\mathcal{U}_\star$ are defined on $(-\overline{\tau},+\overline{\tau})$ and do not intersect (that is to say no caustic forms). Define $\mathcal{U}$ as the union of the images of the interval $(-\overline{\tau},+\overline{\tau})$ under the curves $x$ of the conformal geodesics described above. Then $\mathcal{U}$ is an open neighbourhood of $p$ in $\M$ such that
        \begin{itemize}
            \item[i)] it is covered by the conformal geodesics $(x,\widehat{\beta})$ defined above,
            \item[ii)] these conformal geodesics do not intersect on $\mathcal{U}$,
            \item[iii)] the intersection of $\mathcal{U}$ with each conformal geodesic $(x,\widehat{\beta})$ is connected.
        \end{itemize}
        As a consequence, one has a smooth  congruence of timelike conformal geodesics on $\mathcal{U}$ .
        
        \item \underline{Step 3}: Constructing $\widecheck{\nabla}$, $\breve{g}$, $(y^\mu)$ and $(e_{\bf a})$
        
        Let $\widecheck{\nabla}$ and $\breve{g}$ be respectively the Weyl connection on $\mathcal{U}$ and the metric on $\mathcal{U}$ conformally related to $g|_{\mathcal{U}}$ given by \Cref{prop_canonical_congruence}. Modifying accordingly the quintuple is done by a confomorphism of conformal factor $\Psi = \breve{\Theta}/\Theta$ and a Weyl connection change of covector field $\omega = \widehat{\beta}$. Note that $\widecheck{\nabla}$ and $\breve{g}$ coincide respectively with $\widehat{\nabla}$ and $g$ on $\mathcal{U}_\star$ due to the particular choice of initial data \eqref{initial_cdt_conf_geo}. 
        
        Up to reducing $\mathcal{U}_\star$, one can take
        \begin{itemize}
            \item[\ding{70}] $(y^i) = (y^1,y^2,y^3)$ a coordinate system on $\mathcal{U}_\star$ such that
            \begin{subequations}
                \label{gauge_choice_coord}
                \begin{align}
                    y^3 &\geq 0 \,, \\
                    \mathcal{U}_\star \cap \partial\mathcal{S}_\star &= \{ q \in \mathcal{U}_\star \mid y^3(q) = 0\} \,,
                \end{align}
            \end{subequations}
            \item[\ding{70}] $(e_{\bf i}) = (e_{\bf 1},e_{\bf 2},e_{\bf 3})$ a frame field on $\mathcal{U}_\star$ such that
            \begin{equation}
                \label{gauge_choice_frame}
                g(e_{\bf i},e_{\bf j}) = \delta_{\bf i,j} \,.
            \end{equation}
        \end{itemize}
        Construct now
        \begin{itemize}
            \item[\ding{70}] a conformal Gaussian coordinate system $(y^\mu)$ on $\mathcal{U}$ by
            \[ y^0(x(\tau)) := \tau \,, \qquad y^i(x(\tau)) := y^i(x(0)) \,, \]
            \item[\ding{70}] a frame field $(e_{\bf a})$ on $\mathcal{U}$ by
            \begin{equation}
                \label{eq_inter}
                e_{\bf 0} := \dot{x} \,, \qquad \widecheck{\nabla}_{\dot{x}} e_{\bf i} = 0 \,.
            \end{equation}
        \end{itemize} 
        
        \item \underline{Step 4}: Verifying the core properties \eqref{eq:gauge_core_properties}
        
        One derives directly \eqref{eq:core_a}, \eqref{eq:core_c}, \eqref{eq:core_d} and \eqref{eq:core_e} while \eqref{eq:core_b} comes from \eqref{conditions} ensured in step 1. Moreover, \Cref{prop_canonical_congruence} and \eqref{eq_inter} give
        \begin{align*}
            \widecheck{\nabla}_{e_{\bf 0}} e_{\bf a} &= 0 \,, \\
            \widecheck{L}(e_{\bf 0},.) &= 0 \,, \\
            \breve{g}(e_{\bf 0},e_{\bf 0}) &= -1 \,,
        \end{align*}
        or in frame components
        \begin{align*}
            \widecheck{\Gamma}_{\bf 0}{}^{\bf b} {}_{\bf a} &= 0 \,, \\
            \widecheck{L}_{\bf 0a} &= 0 \,, \\
            \breve{g}_{\bf 00} &= -1 \,.
        \end{align*}        
        Hence \eqref{eq:core_f} and \eqref{eq:core_g}. Additionally,
        \[ e_{\bf 0}(\breve{g}_{\bf ab}) = \widecheck{\nabla}_{e_{\bf 0}} \breve{g}_{\bf ab} = -2 (\widecheck{\kappa}_{\breve{g}})_{\bf 0} \breve{g}_{\bf ab} \,, \]
        where $\widecheck{\kappa}_{\breve{g}}$ is the covector field associated to $\widecheck{\nabla}$ with respect to $\breve{g}$. Since $\breve{g}_{\bf 00} = -1$, one deduces that $(\widecheck{\kappa}_{\breve{g}})_{\bf 0} = 0$ and that the functions $\breve{g}_{\bf ab}$ are constant along the conformal geodesics. By \eqref{initial_cdt_conf_geo} and \eqref{gauge_choice_frame}, it follows that $(e_{\bf a})$ is a $\breve{g}$-orthonormal frame field, that is \eqref{eq:core_i}. Finally, up to switching $e_{\bf 1}$ and $e_{\bf 2}$, one can assume that
        \[ \breve{\epsilon}_{\bf 0123} = +1 \,. \]

        \item \underline{Step 5}: Verifying the additional properties \eqref{gauge_additional_properties}

        Assume that $(\M,g,V,\Theta,\kappa)$ is solution to the (CVE). Then so it is of the quintuple $\left(\mathcal{U},\breve{g},\breve{V},\breve{\Theta},\widecheck{\kappa}_{\breve{g}}\right)$ because it is gauge-equivalent to $\left(\mathcal{U},g,V,\Theta,\kappa\right)$.

        By the transformation laws \eqref{eq:transfo_zeta},
        \[ \widecheck{\zeta}_{\breve{g}} = \widehat{\zeta}_{\breve{g}} + \breve{\Theta} \widehat{\beta} = \frac{\breve{\Theta}}{\Theta} \widehat{\zeta}_g + \breve{\Theta} \widehat{\beta} \,. \]
        In particular, since $\breve{\Theta}_\star = \Theta_\star := \Theta|_{\mathcal{U}_\star}$ and $\widehat{\beta}_\star = 0$, one has $(\widecheck{\zeta}_{\breve{g}})_\star = (\widehat{\zeta}_g)_\star$. It follows from \eqref{eq:CVE_const_bis}, derived from the (CVE), that $(\widehat{s}_g)_\star = (\widecheck{s}_{\breve{g}})_\star$.

        By \Cref{lem:link_zero_quantities_CVE} $i)$, all zero quantities vanish on $\mathcal{U}$. Using the core properties \eqref{eq:gauge_core_properties}, one deduces that
        \begin{alignat*}{6}
            0 &= \mho_{\bf 0} &&= e_{\bf 0}(\widecheck{s}_{\breve{g}}) \,, \\
            0 &= \varkappa_{\bf 00} && = e_{\bf 0}((\widecheck{\zeta}_{\breve{g}})_{\bf 0}) + \widecheck{s}_{\breve{g}} \,, \\
            0 &= \varsigma_{\bf 0} && =  (\widecheck{\zeta}_{\breve{g}})_{\bf 0} - (d\breve{\Theta})_{\bf 0} \,.
        \end{alignat*}
        Yet with \eqref{conditions} ensured in step 1, one has
        \begin{align*}
            (\widecheck{\zeta}_{\breve{g}})_{\bf 0} = (\widehat{\zeta}_g)_{\bf 0} = \langle \widehat{\zeta}_g,n_\star\rangle = 0 \quad \text{on } \mathcal{U}_\star\,, \\
            f_\star := \breve{\Theta}_\star^{-1} (\widecheck{s}_{\breve{g}})_\star = \Theta_\star^{-1} (\widehat{s}_g)_\star \in \mathcal{C}^\infty(\mathcal{U}_\star,\R) \,.
        \end{align*}
        Hence \eqref{gauge_s}-\eqref{gauge_Theta}. 

        Finally, let us show \eqref{gauge_IcapU}. If $\sup_{\overline{\mathcal{U}_\star}} f_\star > 0$, let us reduce $\overline{\tau}$ such that
        \[ \overline{\tau} < \sqrt{\frac{2}{\sup_{\overline{\mathcal{U}_\star}} f_\star}} \,, \]
        so that the inequality
        \begin{equation}
            \label{gauge_ineq}
            1 - f_\star(y^1,y^2,y^3) \frac{(y^0)^2}{2} > 0
        \end{equation}
        is enforced everywhere on $\mathcal{U}$. For $q \in \mathcal{U}$, one has
        \begin{alignat*}{7}
            y^3(q) = 0 & \iff && \exists \, (x,\widehat{\beta}) \text{ in the congruence and} &\quad & \text{[by Gaussian coord.]} \\
            & && \tau \in (-\overline{\tau},+\overline{\tau}) \text{ such that} && \\
            &&& q = x(\tau) \text{ and } y^3(x_\star) = 0 && \\
            & \iff && \exists \, (x,\widehat{\beta}) \text{ in the congruence and} && [\text{by } \eqref{gauge_choice_coord}] \\
            &&& \tau \in (-\overline{\tau},+\overline{\tau}) \text{ such that} && \\
            &&& q = x(\tau) \text{ and } x_\star \in \partial \mathcal{S}_\star \cap \mathcal{U}_\star && \\
            & \iff && \exists \, (x,\widehat{\beta}) \text{ in the congruence and} && [\text{by } \eqref{eq:hyp_boundary}] \\
            &&& \tau \in (-\overline{\tau},+\overline{\tau}) \text{ such that} && \\
            &&& q = x(\tau) \text{ and } x_\star \in \partial \mathscr{I} \cap \mathcal{U}_\star && \\
            & \iff && \exists \, (x,\widehat{\beta}) \text{ in the congruence and} && [\text{by def. of } \mathscr{I}]  \\
            &&& \tau \in (-\overline{\tau},+\overline{\tau}) \text{ such that } q = x(\tau) && \\
            &&&  \text{and } \breve{\Theta}_\star(y^1(x_\star),y^2(x_\star),y^3(x_\star)) = 0 && \\
            & \iff && \breve{\Theta}_\star(y^1(q),y^2(q),y^3(q)) && [\text{by Gaussian coord.}] \\
            & \iff && \breve{\Theta}(y^0(q),y^1(q),y^2(q),y^3(q)) = 0 && [\text{by } \eqref{gauge_ineq},\eqref{gauge_IcapU}] \\
            & \iff && q \in \mathscr{I} \cap \mathcal{U} && [\text{by def. of } \mathscr{I}] \,.
        \end{alignat*}
        Consequently, \eqref{gauge_IcapU} holds. \qedhere
    \end{itemize}
\end{proof}

\subsubsection{Derivation of the evolution system}
\label{sec:derivation_evolution_system}

The gauge constructed in the section above leads to a convenient evolution system based on some of the frame components of the zero quantities defined by \eqref{eq:zero_quantities}. Near the boundary, particular care must be paid to the treatment of the zero quantities $\varphi_{\bf abc}$, which lead to the evolution on the rescaled Weyl tensor, in order to prove the local existence and uniqueness of solutions for both the evolution system and the propagation of the constraints which will be studied in~\Cref{sec:propagation_constraints}. More precisely, the evolution system must be designed so that the constraints propagate tangentially to the boundary. This ensures that no additional first-order boundary conditions are required. For a detailed explanation of this difficulty, one can refer to \cite[Example 32]{ST12}, which discusses a model equation, or \cite{HLSW20}, which is set in the simpler framework of AdS space.

\begin{lem}
    \label{lem:weyl_system}
    Under the hypotheses of \Cref{prop:construction_gauge} and in the gauge constructed therein, the following differential equations on some frame components of the electromagnetic decomposition $(\breve{E},\breve{H})$ of the Weyl candidate $\breve{V}$ in the adapted frame $(e_{\bf 0},(e_{\bf i}))$ hold on $\mathcal{U}$
    \begin{subequations}
        \label{weyl_system}
        \begin{align}
            \label{evol_EAB}
            e_{\bf 0}\left(\breve{E}_{\bf AB}\right) + \slashed{\epsilon}_{\bf (A|}{}^{\bf C} e_{\bf C}\left(\breve{H}_{\bf 3|B)}\right) - \slashed{\epsilon}_{\bf (A|}{}^{\bf C}  e_{\bf 3} \left(\breve{H}_{\bf |B)C}\right) + N_{\bf AB}\left(\breve{E},\breve{H}\right) &= 0 \,, \\
            \label{evol_E3A}
            e_{\bf 0}\left(\breve{E}_{\bf 3A}\right)
            + \slashed{\epsilon}^{\bf BC} e_{\bf B}\left(\breve{H}_{\bf CA}\right) + N_{\bf A}\left(\breve{E},\breve{H}\right) &= 0 \,, \\
            \label{evol_HAB}
            e_{\bf 0}\left(\breve{H}_{\bf AB}\right) - \slashed{\epsilon}_{\bf (A|}{}^{\bf C} e_{\bf C}\left(\breve{E}_{\bf 3|B)}\right) + \slashed{\epsilon}_{\bf (A|}{}^{\bf C}  e_{\bf 3} \left(\breve{E}_{\bf |B)C}\right) + N_{\bf AB}\left(\breve{H},-\breve{E}\right) &= 0 \,, \\
            \label{evol_H3A}
            e_{\bf 0}\left(\breve{H}_{\bf 3A}\right)
            - \slashed{\epsilon}^{\bf BC} e_{\bf B}\left(\breve{E}_{\bf CA}\right) + N_{\bf A}\left(\breve{H},-\breve{E}\right) &= 0 \,,
        \end{align}
    \end{subequations}
    where
    \[ \slashed{\epsilon}_{\bf AB} := \breve{\epsilon}_{\bf 0AB3} \,, \]
    and $N_{\bf AB}(X,Y)$, $N_{\bf A}(X,Y)$ are linear terms in $(X,Y)$ of the form
    \[ \widecheck{\Gamma} \times X + \widecheck{\kappa}_{\breve{g}} \times X + \widecheck{\Gamma} \times Y + \widecheck{\kappa}_{\breve{g}} \times Y \,, \]
    with $N_{\bf AB}$ symmetric.
\end{lem}

\begin{proof}
    By \Cref{lem:link_zero_quantities_CVE} $i)$, all zero quantities with $\acute{\nabla}=\widecheck{\nabla}$, $U_{\alpha\beta} = \widecheck{L}_{\alpha\beta}$, $\zeta_\alpha = (\widecheck{\zeta}_{\breve{g}})_\alpha$ and $s=\widecheck{s}_{\breve{g}}$ vanish identically. We will used here that $\varphi_{\bf abc} = 0$  for some frame indices and the electromagnetic decomposition $(\breve{E},\breve{H})$ of the Weyl candidate $\breve{V}$.
    
    \begin{itemize}
    
        \item Consider $\varphi_{\bf 0(AB)} = 0$. Let us compute first $\varphi_{\bf 0AB}$. One has
        \begin{align*}
            \widecheck{\nabla}_{\bf 0} \breve{V}^{\bf 0}{}_{\bf B0A} &= e_{\bf 0}(\breve{V}^{\bf 0}{}_{\bf B0A}) = e_{\bf 0}(\breve{E}_{\bf AB}) \,, \\
            \widecheck{\nabla}_{\bf C} \breve{V}^{\bf C}{}_{\bf B0A} &= e_{\bf C}(\breve{V}^{\bf C}{}_{\bf B0A}) + \widecheck{\Gamma} \times \breve{V} = \slashed{\epsilon}_{\bf B}{}^{\bf C} e_{\bf C}(\breve{H}_{\bf 3A}) + \widecheck{\Gamma} \times \breve{E} + \widecheck{\Gamma} \times \breve{H} \,, \\
            \widecheck{\nabla}_{\bf 3} \breve{V}^{\bf 3}{}_{\bf B0A} &= e_{\bf 3}(\breve{V}^{\bf 3}{}_{\bf B0A}) + \widecheck{\Gamma} \times \breve{V} = - \slashed{\epsilon}_{\bf B}{}^{\bf C} e_{\bf 3}(\breve{H}_{\bf AC}) + \widecheck{\Gamma} \times \breve{E} + \widecheck{\Gamma} \times \breve{H} \,, \\
            (\widecheck{\kappa}_{\breve{g}})_{\bf a} \breve{V}^{\bf a}{}_{\bf A0B} &= \widecheck{\kappa}_{\breve{g}} \times \breve{V} = \widecheck{\kappa}_{\breve{g}} \times \breve{E} + \widecheck{\kappa}_{\breve{g}} \times \breve{H} \,.
        \end{align*}
        Taking the symmetric part in $({\bf A},{\bf B})$ gives \eqref{evol_EAB}.
        
        \item Consider $\varphi_{\bf 0A3} = 0$. One has
        \begin{align*}
            \widecheck{\nabla}_{\bf 0} \breve{V}^{\bf 0}{}_{\bf 30A} &= e_{\bf 0}(\breve{V}^{\bf 0}{}_{\bf 30A}) = e_{\bf 0}(\breve{E}_{\bf 3A}) \,, \\
            \widecheck{\nabla}_{\bf B} \breve{V}^{\bf B}{}_{\bf 30A} &= e_{\bf B}(\breve{V}^{\bf B}{}_{\bf 30A}) + \widecheck{\Gamma} \times \breve{V} = \slashed{\epsilon}^{\bf BC} e_{\bf B}(\breve{H}_{\bf CA}) + \widecheck{\Gamma} \times \breve{E} + \widecheck{\Gamma} \times \breve{H} \,, \\
            \widecheck{\nabla}_{\bf 3} \breve{V}^{\bf 3}{}_{\bf 30A} &= e_{\bf 3}(\breve{V}^{\bf 3}{}_{\bf 30A}) + \widecheck{\Gamma} \times \breve{V} = + \widecheck{\Gamma} \times \breve{E} + \widecheck{\Gamma} \times \breve{H} \,, \\
            (\widecheck{\kappa}_{\breve{g}})_{\bf a} \breve{V}^{\bf a}{}_{\bf 30A} &= \widecheck{\kappa}_{\breve{g}} \times \breve{V} =  \widecheck{\kappa}_{\breve{g}} \times \breve{E} + \widecheck{\kappa}_{\breve{g}} \times \breve{H} \,.
        \end{align*}
        This gives \eqref{evol_E3A}.
    \end{itemize}
    The other two equations are deduced by symmetry thanks to \eqref{div_Q_div_Q_star} and \eqref{sym_dual_EH}.
\end{proof}

\begin{lem}
   Under the hypotheses of \Cref{prop:construction_gauge} and in the gauge constructed therein, the following transport equations hold on $\mathcal{U}$
    \begin{subequations}
        \label{transport_system}
        \begin{align}
            \label{transport_a}
            e_{\bf 0}(\widecheck{\Gamma}_{\bf i} {}^{\bf k} {}_{\bf j}) + \widecheck{\Gamma}_{\bf i} {^{\bf l}} {}_{\bf 0} \widecheck{\Gamma}_{\bf l} {}^{\bf k} {}_{\bf j} &=  \breve{\Theta} \breve{\epsilon}_{\bf j}{}^{\bf kl}  \breve{H}_{\bf li} + \delta_{\bf j}{}^{\bf k} \widecheck{L}_{\bf i0} \,, \\
            \label{transport_b}
            e_{\bf 0}(\widecheck{\Gamma}_{\bf i} {}^{\bf j} {}_{\bf 0}) + \widecheck{\Gamma}_{\bf i} {^{\bf l}} {}_{\bf 0} \widecheck{\Gamma}_{\bf l} {}^{\bf j} {}_{\bf 0}  &=  \breve{\Theta} \, \delta^{\bf jk} \breve{E}_{\bf ik} + \delta^{\bf jk} \widecheck{L}_{\bf ik} \,, \\
            \label{transport_c}
            e_{\bf 0}(\widecheck{\Gamma}_{\bf i} {}^{\bf 0} {}_{\bf j}) + \widecheck{\Gamma}_{\bf i} {^{\bf l}} {}_{\bf 0} \widecheck{\Gamma}_{\bf l} {}^{\bf 0} {}_{\bf j}  &=  \breve{\Theta} \, \breve{E}_{\bf ij} + \widecheck{L}_{\bf ij} \,, \\
            \label{transport_d}
            e_{\bf 0}(\widecheck{\Gamma}_{\bf i}{}^{\bf 0}{}_{\bf 0}) + \widecheck{\Gamma}_{\bf i}{}^{\bf j}{}_{\bf 0} \widecheck{\Gamma}_{\bf j}{}^{\bf 0}{}_{\bf 0} &= \widecheck{L}_{\bf i0} \,, \\
            \label{transport_e}
            e_{\bf 0}((\widecheck{\kappa}_{\breve{g}})_{\bf i}) + \widecheck{\Gamma}_{\bf i}{}^{\bf j}{}_{\bf 0} (\widecheck{\kappa}_{\breve{g}})_{\bf j} &= \widecheck{L}_{\bf i0} \,, \\
            \label{transport_f}
            e_{\bf 0}(\widecheck{L}_{\bf ij}) + \widecheck{\Gamma}_{\bf i}{}^{\bf k}{}_{\bf 0} \widecheck{L}_{\bf kj} &= \breve{\epsilon}_{\bf j}{}^{\bf kl}(\widecheck{\zeta}_{\breve{g}})_{\bf k}  \breve{H}_{\bf li} + (\widecheck{\zeta}_{\breve{g}})_{\bf 0} \breve{E}_{\bf ij} \,, \\
            \label{transport_g}
            e_{\bf 0}(\widecheck{L}_{\bf i0}) + \widecheck{\Gamma}_{\bf i}{}^{\bf k}{}_{\bf 0} \widecheck{L}_{\bf k0} &= (\widecheck{\zeta}_{\breve{g}})_{\bf l} \delta^{\bf lk} \breve{E}_{\bf ki} \,, \\
            \label{transport_h}
            e_{\bf 0}((\widecheck{\zeta}_{\breve{g}})_{\bf i}) &= 0 \,, \\
            \label{transport_i}
            e_{\bf 0}(e_{\bf i}{}^\mu) + \widecheck{\Gamma}_{\bf i}{}^{\bf j}{}_{\bf 0} e_{\bf j}{}^\mu &= - \widecheck{\Gamma}_{\bf i}{}^{\bf 0}{}_{\bf 0} \delta_{\bf 0}{}^\mu \,,
        \end{align}
    \end{subequations}
    where $\breve{\epsilon}_{\bf ijk} := \breve{\epsilon}_{\bf 0ijk}$ and $(\breve{E},\breve{H})$ is the decomposition of the Weyl candidate $\breve{V}$ in the frame field $(e_{\bf 0},(e_{\bf i}))$ which is adapted to $\mathcal{S}_\star$ given by \Cref{lem:decomposition_Weyl_cand}.
\end{lem}

\begin{proof}
    By \Cref{lem:link_zero_quantities_CVE} $i)$, all the zero quantities for $\acute{\nabla}=\widecheck{\nabla}$, $U_{\alpha\beta} = \widecheck{L}_{\alpha\beta}$, $\zeta_\alpha = (\widecheck{\zeta}_{\breve{g}})_\alpha$ and $s=\widecheck{s}_{\breve{g}}$ vanish identically. In the gauge constructed in \Cref{prop:construction_gauge}, one has
    \begin{alignat*}{12}
        \varpi^{\bf k}{}_{\bf j0i} = 0 &\iff \eqref{transport_a} \,, &\qquad \varpi^{\bf j}{}_{\bf 00i} = 0 &\iff \eqref{transport_b} \,, &\qquad
        \varpi^{\bf 0}{}_{\bf j0i} = 0 &\iff \eqref{transport_c} \,, \\ \varpi^{\bf 0}{}_{\bf 00i} = 0 &\iff \eqref{transport_d} \,, &\qquad
        \Pi_{\bf 0i} = 0 &\iff \eqref{transport_e} \,, &\; \varrho_{\bf 0ij} = 0 &\iff \eqref{transport_f} \,, \\
        \varrho_{\bf 0i0} = 0 &\iff \eqref{transport_g} \,, &\qquad \varkappa_{\bf 0i} = 0 &\iff \eqref{transport_h} \,, &\qquad
        \widecheck{\Sigma}_{\bf 0}{}^{\bf a}{}_{\bf i} e_{\bf a}{}^\mu = 0 &\iff \eqref{transport_i} \,. \tag*{\qedhere}
    \end{alignat*}
\end{proof}

Define the unknown $\underline{u} : \mathcal{U} \to \R^m$ with $m:=90$ by
\begin{align*}
    \underline{u} := \Big(\breve{E}_{\bf 11}, \breve{E}_{\bf 22},\breve{H}_{\bf 12},\breve{H}_{\bf 21},\breve{H}_{\bf 11}, \breve{H}_{\bf 22},\breve{E}_{\bf 12},\breve{E}_{\bf 21}, \breve{E}_{\bf 31},\breve{E}_{\bf 32}, \breve{H}_{\bf 31}, \breve{H}_{\bf 32}, \\
    \widecheck{\Gamma}_{\bf i}{}^{\bf j}{}_{\bf k}, \widecheck{\Gamma}_{\bf i}{}^{\bf j}{}_{\bf 0},\widecheck{\Gamma}_{\bf i}{}^{\bf 0}{}_{\bf j},\widecheck{\Gamma}_{\bf i}{}^{\bf 0}{}_{\bf 0}, (\widecheck{\kappa}_{\breve{g}})_{\bf i}, \widecheck{L}_{\bf ij},\widecheck{L}_{\bf i0},(\widecheck{\zeta}_{\breve{g}})_{\bf i},e_{\bf i}{}^\mu \Big) \,.
\end{align*}
The frame indices in the first line are written explicitly to specify their order in the tuple, as opposed to the second line. This function $\underline{u}$ is a solution to the evolution system composed of equations \eqref{evol_EAB}-\eqref{evol_H3A} and \eqref{transport_a}-\eqref{transport_i} where $\widecheck{s}_{\breve{g}}$, $(\widecheck{\zeta}_{\breve{g}})_{\bf 0}$ and $\breve{\Theta}$ are given by \eqref{gauge_s}-\eqref{gauge_Theta}.

\bigbreak

The remaining part of this section is dedicated to writing the evolution system verified by $\underline{u}$ outside of the geometric framework. Set
\begin{align*}
    \R^4_{3\geq0} &:= \{y=(y^0,y^1,y^2,y^3) \in \R^4 \mid y^3 \geq 0 \} \,, \\
    \partial \R^4_{3\geq0} &=  \{y\in\R^4_{3\geq0} \,|\,y^3=0\} \,.
\end{align*}
The open subset $\mathcal{U}_\star$ is identified through the coordinates to an open subset of $\{ y \in \R^4_{3\geq0} \mid y^0 = 0 \}$, which is still denoted by $\mathcal{U}_\star$. In the same way, the subset $\mathcal{U}$ is identified to $(-\overline{\tau},+\overline{\tau}) \times \mathcal{U}_\star \subset \R^4_{3\geq0}$ for some $\overline{\tau} > 0$.

The evolution system on an arbitrary function $\underline{u} : \mathcal{U} \to \R^m$ whose components are denoted as follows
    \begin{align*}
        \underline{u} = \Big(E_{\bf 11},E_{\bf 22},H_{\bf 12},H_{\bf 21},H_{\bf 11},H_{\bf 22},E_{\bf 12},E_{\bf 21},E_{\bf 31},E_{\bf 32},H_{\bf 31},H_{\bf 32}, \\
        \acute{\Gamma}_{\bf i}{}^{\bf j}{}_{\bf k}, \acute{\Gamma}_{\bf i}{}^{\bf j}{}_{\bf 0},\acute{\Gamma}_{\bf i}{}^{\bf 0}{}_{\bf j},\acute{\Gamma}_{\bf i}{}^{\bf 0}{}_{\bf 0}, \kappa_{\bf i}, U_{\bf ij},U_{\bf i0},\zeta_{\bf i},e_{\bf i}{}^\mu \Big) \,,
    \end{align*}
is the following quasilinear system of partial differential equations
    \begin{subequations}
        \label{evol_system}
        \begin{align}
            \label{evol_system_EAB}
            \partial_0 E_{\bf AB} + \slashed{\epsilon}_{\bf (A|}{}^{\bf C} e_{\bf C}{}^\mu \partial_\mu H_{\bf 3|B)} - \frac{1}{2} \slashed{\epsilon}_{\bf (A|}{}^{\bf C}  e_{\bf 3}{}^\mu  \partial_\mu \left(H_{\bf |B)C}+H_{\bf C|B)} \right) + N_{\bf AB}\left(E,H\right) &= 0 \,, \\
            \partial_0 E_{\bf 3A} + \frac{1}{2} \slashed{\epsilon}^{\bf BC} e_{\bf B}{}^\mu \partial_\mu \left(H_{\bf CA} + H_{\bf AC} \right) + N_{\bf A}\left(E,H\right) &= 0 \,, \\
            \partial_0 H_{\bf AB} - \slashed{\epsilon}_{\bf (A|}{}^{\bf C} e_{\bf C}{}^\mu \partial_\mu E_{\bf 3|B)} + \frac{1}{2} \slashed{\epsilon}_{\bf (A|}{}^{\bf C}  e_{\bf 3}{}^\mu \partial_\mu \left(E_{\bf |B)C} + E_{\bf C|B)} \right) + N_{\bf AB}\left(H,-E\right) &= 0 \,, \\
            \label{evol_system_H3A}
            \partial_0 H_{\bf 3A}
            - \frac{1}{2} \slashed{\epsilon}^{\bf BC} e_{\bf B}{}^\mu \partial_\mu \left(E_{\bf CA}+E_{\bf AC}\right) +  N_{\bf A}\left(H,-E\right) &= 0 \,, \\
            \partial_0 \acute{\Gamma}_{\bf i} {}^{\bf k} {}_{\bf j} + \acute{\Gamma}_{\bf i} {^{\bf l}} {}_{\bf 0} \acute{\Gamma}_{\bf l} {}^{\bf k} {}_{\bf j} - \Theta \epsilon_{\bf j}{}^{\bf kl} H_{\bf li} - \delta_{\bf j}{}^{\bf k} U_{\bf i0} &= 0 \,, \\
            \partial_0 \acute{\Gamma}_{\bf i} {}^{\bf j} {}_{\bf 0} + \acute{\Gamma}_{\bf i} {^{\bf l}} {}_{\bf 0} \acute{\Gamma}_{\bf l} {}^{\bf j} {}_{\bf 0}  - \Theta \, \delta^{\bf jk} E_{\bf ik} - \delta^{\bf jk} U_{\bf ik} &= 0 \,, \\
            \partial_0 \acute{\Gamma}_{\bf i} {}^{\bf 0} {}_{\bf j} + \acute{\Gamma}_{\bf i} {^{\bf l}} {}_{\bf 0} \acute{\Gamma}_{\bf l} {}^{\bf 0} {}_{\bf j}  -  \Theta \, E_{\bf ij} - U_{\bf ij} &= 0 \,, \\
            \partial_0 \acute{\Gamma}_{\bf i}{}^{\bf 0}{}_{\bf 0} + \acute{\Gamma}_{\bf i}{}^{\bf j}{}_{\bf 0} \acute{\Gamma}_{\bf j}{}^{\bf 0}{}_{\bf 0} - U_{\bf i0} &= 0 \,, \\
            \partial_0 \kappa_{\bf i} + \acute{\Gamma}_{\bf i}{}^{\bf j}{}_{\bf 0} \kappa_{\bf j} - U_{\bf i0} &= 0 \,, \\
            \partial_0 U_{\bf ij} + \acute{\Gamma}_{\bf i}{}^{\bf k}{}_{\bf 0} U_{\bf kj} - \epsilon_{\bf j}{}^{\bf kl} \zeta_{\bf k} H_{\bf li} - \zeta_{\bf 0} E_{\bf ij} &= 0 \,, \\
            \partial_0 U_{\bf i0} + \acute{\Gamma}_{\bf i}{}^{\bf k}{}_{\bf 0} U_{\bf k0} - \zeta_{\bf l} \delta^{\bf lk} E_{\bf ki} &= 0 \,, \\
            \partial_0 \zeta_{\bf i} &= 0 \,, \\
            \partial_0 e_{\bf i}{}^\mu + \acute{\Gamma}_{\bf i}{}^{\bf j}{}_{\bf 0} e_{\bf j}{}^\mu + \acute{\Gamma}_{\bf i}{}^{\bf 0}{}_{\bf 0} \delta_{\bf 0}{}^\mu &= 0 \,,
        \end{align}
    \end{subequations}
where
    \begin{itemize}
        \item we recall that ${\bf A}, {\bf B}, {\bf C} \in \{{\bf 1},{\bf 2}\}$, ${\bf i}, {\bf j}, {\bf k},{\bf l} \in \{{\bf 1},{\bf 2},{\bf 3}\}$, $\mu \in \{0,1,2,3\}$,
        \item $\slashed{\epsilon}_{\bf AB}$ and $\epsilon_{\bf ijk}$ are the totally antisymmetric symbols with
        \[ \slashed{\epsilon}_{\bf 12} = +1 \,, \qquad \epsilon_{\bf 123} = +1 \,, \]
        \item $N_{\bf A}$ and $N_{\bf AB}$ are the same as in \Cref{lem:weyl_system} with $\widecheck{\Gamma}$ and $\widecheck{\kappa}_{\breve{g}}$ substituted by $\acute{\Gamma}$ and $\kappa$,
        \item $s$, $\zeta_{\bf 0}$ and $\Theta$ are the smooth functions defined by
        \begin{subequations}
        \label{redef_szeta0Theta}
        \begin{align}
            s(y^0,y^1,y^2,y^3) &:= s_\star(y^1,y^2,y^3) \,, \\
            \zeta_{\bf 0}(y^0,y^1,y^2,y^3) &:= -s_\star(y^1,y^2,y^3) y^0 \,, \\
            \label{redef_Theta}
            \Theta(y^0,y^1,y^2,y^3) &:= \Theta_\star(y^1,y^2,y^3) - s_\star(y^1,y^2,y^3) \frac{(y^0)^2}{2} \,,
        \end{align}
        \end{subequations}
        for two smooth functions $s_\star$, $\Theta_\star \in \mathcal{C}^\infty(\mathcal{U}_\star,\R)$.
    \end{itemize}

\begin{rem}
    The only differences between \eqref{weyl_system}-\eqref{transport_system} and \eqref{evol_system}, aside from the substitution of the fields, are the symmetrisation of $H_{\bf AB}$ and $E_{\bf AB}$ in \eqref{evol_E3A} and \eqref{evol_H3A} respectively and the development of the derivatives with respect to the frame field in partial derivatives.
\end{rem}
    
\noindent The evolution system \eqref{evol_system} can be recast under the form
\begin{equation}
    \label{evol_system_form}
    A^\mu (\underline{u}) \partial_\mu \underline{u} = \mathcal{Q}(y,s_\star,\Theta_\star,\underline{u}) \,,
\end{equation}
where
\begin{itemize}
    \item $\mathcal{Q}(y,s_\star,\Theta_\star,.)$ is a polynomial function of degree 2 in $\underline{u}$ with no constant term and coefficients depending smoothly on $y$, $s_\star$ and $\Theta_\star$,
    \item $(A^\mu)$ are affine functions with values into symmetric matrices of size $m$ of the form
    \begin{equation}
        \label{def_Amu}
        A^0(\underline{u}) := I + e_{\bf j}{}^0 A^{\bf j} \,, \qquad A^i(\underline{u}) := e_{\bf j}{}^i A^{\bf j} \,,
    \end{equation}
    More precisely, the constant symmetric matrices $A^{\bf j}$ are defined by
    \[ A^{\bf 1} := 
        \begin{pmatrix}
            0 & 0 & \begin{pmatrix}
                        0 & -D_{\bf 1} \\
                        C_{\bf 1} & 0
                    \end{pmatrix} & 0 \\
            0 & 0 & \begin{pmatrix}
                        D_{\bf 1} & 0 \\
                        0 & -C_{\bf 1}
                    \end{pmatrix} & 0 \\
            \begin{pmatrix}
                0 & C_{\bf 1}^\intercal \\
                -D_{\bf 1}^\intercal & 0
            \end{pmatrix} &
            \begin{pmatrix}
                D_{\bf 1}^\intercal & 0 \\
                0 & -C_{\bf 1}^\intercal
            \end{pmatrix} & 0 & 0 \\
            0 & 0 & 0 & 0
        \end{pmatrix} \,, \]
    \[ A^{\bf 2} := 
        \begin{pmatrix}
            0 & 0 & \begin{pmatrix}
                        0 & -D_{\bf 2} \\
                        C_{\bf 2} & 0
                    \end{pmatrix} & 0 \\
            0 & 0 & \begin{pmatrix}
                        D_{\bf 2} & 0 \\
                        0 & -C_{\bf 2}
                    \end{pmatrix} & 0 \\
            \begin{pmatrix}
                0 & C_{\bf 2}^\intercal \\
                -D_{\bf 2}^\intercal & 0
            \end{pmatrix} &
            \begin{pmatrix}
                D_{\bf 2}^\intercal & 0 \\
                0 & -C_{\bf 2}^\intercal
            \end{pmatrix} & 0 & 0 \\
            0 & 0 & 0 & 0
        \end{pmatrix} \,, \]
    \[ A^{\bf 3} := 
        \begin{pmatrix}
            \begin{pmatrix}
                0 & N \\
                N^\intercal & 0
            \end{pmatrix} & 0 & 0 & 0 \\
            0 & \begin{pmatrix}
                    0 & -N \\
                    -N^\intercal & 0
                \end{pmatrix} &0 & 0 \\
            0 & 0 & 0 & 0 \\
            0 & 0 & 0 & 0
        \end{pmatrix} \,, \]
    with the following matrices of size $2$ 
    \[  C_{\bf 1} := \frac{1}{2} \begin{pmatrix}
                               1 & 0 \\
                               1 & 0
                            \end{pmatrix} \,, \quad
        D_{\bf 1} := \begin{pmatrix}
                    0 & 0 \\
                    0 & 1 \\
                \end{pmatrix} \,, \quad
        C_{\bf 2} := \frac{1}{2} \begin{pmatrix}
                               0 & -1 \\
                               0 & -1
                            \end{pmatrix} \,, \]
    \[  D_{\bf 2} := \begin{pmatrix}
                    -1 & 0 \\
                    0 & 0 \\
                \end{pmatrix} \,, \quad   
        N := \frac{1}{2} \begin{pmatrix}
                                -1 & -1 \\
                                +1 & +1
                            \end{pmatrix} \,. \]
\end{itemize}

\begin{lem}
    \label{lemma_bounds_matrices}
    For all $\underline{u} \in \R^m$,
    \begin{alignat*}{9}
        \left(1-c^0(\underline{u})\right) I &\leq A^0(\underline{u}) \leq && \left(1+c^0(\underline{u})\right) I \,, \\
        -c^i(\underline{u}) \, I &\leq A^i(\underline{u}) \leq && \, c^i(\underline{u}) \, I \,,
    \end{alignat*}
    where $c^\mu(\underline{u}) := |(e_{\bf 1}{}^\mu,e_{\bf 2}{}^\mu,e_{\bf 3}{}^\mu)|_2$, $|.|_2$ denoting the euclidean norm on $\R^3$.
\end{lem}

\begin{proof}
    For all $(x,y,z) \in \R^3$, the eigenvalues of the symmetric matrix $x A^{\bf 1} + y A^{\bf 2} + z A^{\bf 3}$ are $0$ (of multiplicity $m-8$) and $\pm |(x,y)|_2/\sqrt{2}$, $\pm |(x,y,z)|_2$ (each of multiplicity $2$). Hence the result. 
\end{proof}

\subsubsection{Solving the evolution system}
\label{sec:solving_evolution_system}

Let $\mathcal{U}_\star$ be a relatively compact open subset of $\{ y \in \R^4_{3\geq 0} \, | \, y^0 = 0\}$ and define
\[ \mathcal{B}_\star := \partial \R^4_{3\geq0} \cap \mathcal{U}_\star\,, \qquad \mathcal{B} := \R_+ \times \mathcal{B}_\star \subset \partial\R^4_{3\geq0} \,. \]
We will assume that $\mathcal{B}_\star$ is connected. Let $\Theta_\star$, $s_\star \in \mathcal{C}^\infty(\mathcal{U}_\star,\R)$ be two smooth functions such that $\Theta_\star$ is a boundary defining function of $\mathcal{B}_\star$ and $f_\star := s_\star/\Theta_\star$ is uniformly bounded from above on $\mathcal{U}_\star$.

The aim of this section is to solve the evolution system \eqref{evol_system} with initial data on $\mathcal{U}_\star$ as an initial value problem if $\mathcal{B}_\star = \varnothing$ and as an initial boundary value problem otherwise, in which case suitable boundary conditions on $\mathcal{B}$ have to be determined.

To this end, we will rely on the standard theory of symmetric hyperbolic systems and more particularly to
\begin{itemize}
    \item Taylor \cite[Chapter 5]{T91} for quasilinear symmetric hyperbolic systems with initial data on closed manifolds or on $\R^n$,
    \item Rauch \cite{R85} for linear symmetric hyperbolic systems with a characteristic boundary matrix of constant multiplicity,
    \item Guès \cite{G90} and Secchi \cite{S96} for quasilinear symmetric hyperbolic systems with a characteristic boundary matrix of constant multiplicity.
\end{itemize}
The boundary conditions arising in this context are called \emph{maximal dissipative}. They are designed to compensate the energy gain from incoming characteristic fields at the boundary by the energy loss due to outgoing characteristic fields. The term `maximal' does not mean that the dissipation of energy is maximal - there may be no dissipation at all. It refers to the fact that a non-negativity property is valid on a vector subspace of the boundary which is maximal for inclusion.

However, the aforementioned results on quasilinear symmetric hyperbolic systems cannot be directly applied to our evolution system \eqref{evol_system} for two reasons. First, the system is not quite hyperbolic since $A^0(\underline{u})$ is not positive-definite for all $\underline{u} \in \R^m$. Secondly, while the boundary matrix is characteristic, it is not of constant multiplicity, not even in a neighbourhood of a point of interest in $\R^m$. These two difficulties have been overcome by Friedrich leading to \Cref{evol_solve_away} and \Cref{evol_solve_near} below.

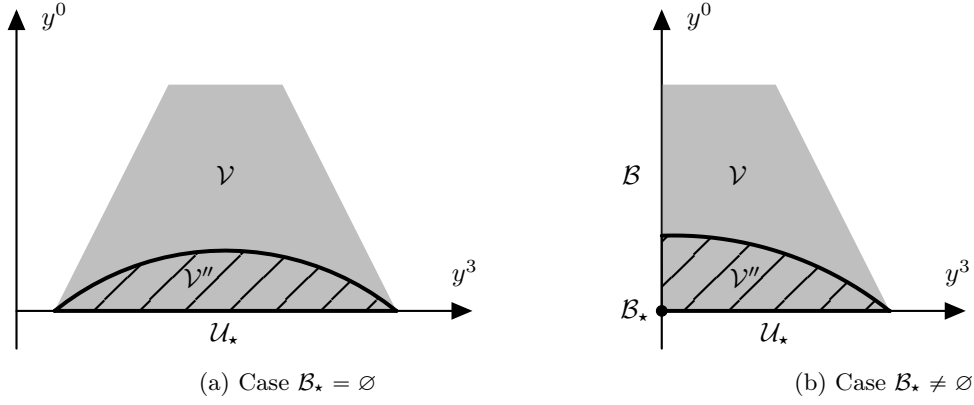
\begin{figure}
	\centering
    \hspace{0.05\textwidth}
	\begin{subfigure}{.45\textwidth}
        \begin{tikzpicture}[line cap=round,line join=round,>=triangle 45,x=1cm,y=1cm]
        \draw [->,line width=0.8pt] (0,-0.5) -- (0,4);
        \draw [->,line width=0.8pt] (0,0) -- (6,0);
        \fill[line width=0.8pt,color=gray,fill=gray,fill opacity=0.5] (0.5,0) -- (5,0) -- (3.5,3) -- (2,3) -- cycle;
        \draw[line width=0.8pt, color=black, fill=black, pattern={Lines[angle=45,distance=12]}, pattern color=black] {[smooth,samples=50,domain=0.5:5] plot(\x,{-2.76+sqrt(3.56^2-(\x-2.75)^2)})} -- (5,0) {[smooth,samples=50,domain=5:0.5] -- plot(\x,{0})} -- (0.5,-0.0011777875332374954) -- cycle;
        \draw[shift={(2.75,-2.76)},line width=1.5pt]  plot[domain=0.8875653454992822:2.254027308090511,variable=\t]({3.56*cos(\t r)},{3.56*sin(\t r)});
        \draw [line width=1.5pt] (0.5,0)-- (5,0);
        \draw (0.5,3.9) node {$y^0$};
        \draw (5.9,0.5) node {$y^3$};
        \draw (2.75,1.8) node {$\mathcal{V}$};
        \draw (2.75,-0.3) node {$\mathcal{U}_\star$};
        \draw (2.4,0.45) node {$\mathcal{V}''$};
        \end{tikzpicture}

        \caption{\centering Case $\mathcal{B}_\star = \varnothing$}
		\label{fig:RegionsNoBoundary}
	\end{subfigure}
    \hfill
    \begin{subfigure}{.45\textwidth}
        \begin{tikzpicture}[line cap=round,line join=round,>=triangle 45,x=1cm,y=1cm]
        \draw [->,line width=0.8pt] (0,-0.5) -- (0,4);
        \draw [->,line width=0.8pt] (0,0) -- (4,0);
        \fill[line width=0.8pt,color=gray,fill=gray,fill opacity=0.5] (0,3) -- (0,0) -- (3,0) -- (1.5,3) -- cycle;
        \draw[line width=0.8pt,color=black,fill=black, pattern={Lines[angle=45,distance=12]}, pattern color=black] {[smooth,samples=50,domain=0:3] plot(\x,{-3.52+sqrt(4.52^2-(\x-0.16)^2)})} -- (3,0) {[smooth,samples=50,domain=3:0] -- plot(\x,{0})} -- (0,1) -- cycle;
        \draw [shift={(0.16,-3.52)},line width=1.5pt]  plot[domain=0.8916399784412002:1.6064515663553094,variable=\t]({4.52*cos(\t r)},{4.52*sin(\t r)});
        \draw [line width=1.5pt] (0,0)-- (3,0);
        \draw [fill=black] (0,0) circle (2pt);
        \draw (0.5,3.9) node {$y^0$};
        \draw (3.9,0.5) node {$y^3$};
        \draw (-0.35,0) node {$\mathcal{B}_\star$};
        \draw (1,1.8) node {$\mathcal{V}$};
        \draw (1.5,-0.3) node {$\mathcal{U}_\star$};
        \draw (1.1,0.45) node {$\mathcal{V}''$};
        \draw (-0.4,1.8) node {$\mathcal{B}$};
        \end{tikzpicture}

		\caption{\centering Case $\mathcal{B}_\star \neq \varnothing$}
		\label{fig:RegionsBoundary}
	\end{subfigure}
	\caption{\centering Regions of local existence and local uniqueness for the evolution system \eqref{evol_system}}
\end{figure}

\begin{prop}[Solving away from the boundary]
    \label{evol_solve_away}
    Consider the case $\mathcal{B}_\star = \varnothing$. Let $\underline{u}_\star \in \mathcal{C}^\infty(\overline{\mathcal{U}}_\star,\R^m)$ such that
    \begin{equation}
        \label{eq:initial_conditions_ei0}
        e_{\bf i}{}^0 = 0 \,.
    \end{equation}
    Then there exists a solution $\underline{u} \in \mathcal{C}^\infty(\mathcal{V},\R^m)$ of the evolution system \eqref{evol_system} with initial data $\underline{u}_\star$ where $\mathcal{V}$ is a subset of $\R^4_{3\geq0}$ containing $\mathcal{U}_\star$ such as in~\Cref{fig:RegionsNoBoundary}. For this solution, the function $\Theta$ defined by \eqref{redef_Theta} is positive on $\mathcal{V}$.
    
    Furthermore, it is locally unique in the sense that if $\underline{u}' \in \mathcal{C}^\infty(\mathcal{V}',\R^m)$ is another such solution then $\underline{u}$ and $\underline{u}'$ coincide on a subset $\mathcal{V}''$ of $\mathcal{V} \cap \mathcal{V}'$ such as in~\Cref{fig:RegionsNoBoundary}.
\end{prop}

\begin{proof} We follow the ideas presented in \cite[Section 8]{F95}. The proof is divided into three steps.
    \begin{itemize}
        \item \underline{Step 1}: Reduction to a symmetric hyperbolic system on $\R^4$

        Firstly, let $\mathcal{W}_\star$ be an open set of $\R^3$ containing $\overline{\mathcal{U}}_\star$. Since $\underline{u}_\star \in \mathcal{C}^\infty(\overline{\mathcal{U}}_\star,\R^m)$, one can find a function $\widetilde{\underline{u}}_\star \in \mathcal{C}^\infty_c(\R^3,\R^m)$ such that
        \[ \widetilde{\underline{u}}_\star |_{\overline{\mathcal{U}}_\star} = \underline{u}_\star \qquad \text{and} \qquad  \widetilde{\underline{u}}_\star |_{\R^3\setminus\mathcal{W}_\star} = 0 \,. \]
        Secondly, fix a constant $c\in(0,1/2)$ and take a smooth vector-valued function $\chi \in \mathcal{C}^\infty(\R^3,\R^3)$ such that
        \[ \chi(x) = x \; \text{ if } \; |x|_2 \leq 1-2c \qquad \text{and} \qquad \sup_{x \in \R^3} |\chi(x)|_2 \leq 1-c \,, \]
        where $|.|_2$ is the euclidean norm on $\R^3$. Define for all $\underline{u} \in \R^m$
        \[ \widetilde{A}^0(\underline{u}) := A^{\bf 1} \chi_1(e_{\bf 1}{}^0,e_{\bf 2}{}^0,e_{\bf 3}{}^0) + A^{\bf 2} \chi_2(e_{\bf 1}{}^0,e_{\bf 2}{}^0,e_{\bf 3}{}^0) + A^{\bf 3} \chi_3(e_{\bf 1}{}^0,e_{\bf 2}{}^0,e_{\bf 3}{}^0) \,. \]
        Then the function $\widetilde{A}^0$
        \begin{itemize}
            \item[i)] is in $\mathcal{C}^\infty(\R^m,S_m(\R))$,
            \item[ii)]  coincides with $A^0$ on $\{ \underline{u} \in \R^m \, | \, |(e_{\bf 1}{}^0,e_{\bf 2}{}^0,e_{\bf 3}{}^0)|_2 \leq 1-2c\}$,
            \item[iii)] satisfies by \Cref{lemma_bounds_matrices}
                \begin{equation}
                    \label{bound_A0}
                    \forall \, \underline{u} \in \R^m, \quad c I \leq \widetilde{A}^0(\underline{u}) \leq (2-c) I \,.
                \end{equation}
        \end{itemize}
        It follows that
        \begin{equation}
            \label{evol_system_form_modif}
            \widetilde{A}^0(\widetilde{\underline{u}}) \partial_0 \widetilde{\underline{u}} + A^i (\widetilde{\underline{u}}) \partial_i \widetilde{\underline{u}} = \mathcal{Q}(y,s_\star,\Theta_\star,\widetilde{\underline{u}}) \,,
        \end{equation}
        is a quasilinear symmetric hyperbolic system on $\R^4$.

    \item \underline{Step 2}: Solving the symmetric hyperbolic system on $\R^4$
    
    By application of Taylor \cite[Chapter 5]{T91}, there exists $T>0$ such that for all $s>5/2$ there exists a unique solution $\widetilde{\underline{u}} \in \mathcal{C}^0((-T,T),H^s(\R^3,\R^m))$ which is bounded, that is $\sup_{(-T,T)} |\widetilde{\underline{u}}|_{H^s} < +\infty$, to the Cauchy problem constituted of \eqref{evol_system_form_modif} with initial data in $H^s(\R^3,\R^m)$.
    
    Since the initial data $\widetilde{\underline{u}}_\star \in \mathcal{C}^\infty(\R^3,\R^m)$, we deduce from Sobolev embeddings and an iterative argument that there exists a unique solution $\widetilde{\underline{u}} \in \mathcal{C}^\infty((-T,T)\times\R^3,\R^m)$ satisfying a uniform bound
    \begin{equation}
        \label{eq:uniform_bound}
        \exists \, K > 0, \forall \, y \in (-T,T)\times\R^3, \quad |\widetilde{\underline{u}}(y)| \leq K \,,
    \end{equation}
    where the constant $K$ only depends on $T$ and on the norm of the initial data $|\underline{u}_\star|_{H^s}$ for one $s>5/2$, say $s=3$.
    
    \item \underline{Step 3}: Returning to the original problem

    Thanks to the uniform bound \eqref{eq:uniform_bound} and \Cref{lemma_bounds_matrices}, one obtains lower and upper bounds for $A^i(\widetilde{\underline{u}})$, in addition of \eqref{bound_A0} ensured by step 1. By finite speed of propagation, there exists a conic subregion $\mathcal{X}$ with base $\mathcal{U}_\star$ in $[0,T) \times \mathcal{U}_\star$ on which the solution $\widetilde{\underline{u}}$ does not depend on the extension $\widetilde{\underline{u}}_\star$ of $\underline{u}_\star$ constructed in step 1. See Ringström \cite[Section 7.3]{R09} for more details.

    \begin{rem}
        Since $\widetilde{\underline{u}}_\star$ is compactly supported and $\mathcal{Q}(y,s_\star,\Theta_\star,.)$ is a polynomial function with no constant term, $\widetilde{\underline{u}}(t)$ is also compactly supported for all $t \in (-T,T)$ by finite speed of propagation.
    \end{rem}
    
    Finally, since $(e_{\bf i}{}^0)_\star = 0$ by hypothesis, there exists $T' \in (0,T)$ such that the solution $\widetilde{\underline{u}}$ satisfies $|(e_{\bf 1}{}^0,e_{\bf 2}{}^0,e_{\bf 3}{}^0)|_2 \leq 1-2c$ on $\mathcal{V} := \left([0,T'] \times \R^3 \right)\cap \mathcal{X}$. Consequently, $\underline{u} := \widetilde{\underline{u}}|_\mathcal{V} \in \mathcal{C}^\infty(\mathcal{V},\R^m)$ is solution to the evolution system \eqref{evol_system} with initial data $\underline{u}_\star$. Up to reducing $T'$, one can ensure that $1-f_\star(y^1,y^2,y^3) (y^0)^2/2$ is positive on $\mathcal{V}$ using that $f_\star$ is uniformly bounded from above on $\mathcal{U}_\star$. Then, since $\Theta_\star$ is positive on $\mathcal{U}_\star$, one deduces that $\Theta$ is positive on $\mathcal{V}$.

    Let us prove local uniqueness. Let $\underline{u}' \in \mathcal{C}^\infty(\mathcal{V}',\R^m)$ be another solution to \eqref{evol_system} on a subset $\mathcal{V}'$ with initial data $\underline{u}_\star$. Then, one has on $\mathcal{V}\cap\mathcal{V}'$
    \[ A^\mu (\underline{u}) \partial_\mu (\underline{u}-\underline{u}') = -A^\mu(\underline{u}-\underline{u}') \partial_\mu \underline{u}' + \mathcal{Q}(y,s_\star,\Theta_\star,\underline{u})-\mathcal{Q}(y,s_\star,\Theta_\star,\underline{u}') \,. \]
    Let $\mathcal{V}''$ be a subset of $\mathcal{V} \cap \mathcal{V}'$ enclosed by $\mathcal{U}_\star$ and a hypersurface which is spacelike for the evolution system with respect to $\underline{u}$ such as in \Cref{fig:RegionsNoBoundary}. Since $A^0(\underline{u}) \geq cI$ on $\mathcal{V}$ and $\mathcal{Q}(y,s_\star,\Theta_\star,.)$ is a polynomial function, one deduces by an energy estimate that $\underline{u}' = \underline{u}$ on $\mathcal{V}''$. See for instance \cite[Theorem 12.1]{K16} for more details. \qedhere
    \end{itemize}
\end{proof}

\bigbreak

\begin{prop}[Solving near the boundary]
    \label{evol_solve_near}
    Consider the case $\mathcal{B}_\star \neq \varnothing$.
    \begin{itemize}
        \item Let $d_1, d_2 \in \mathcal{C}^\infty(\mathcal{B},\R)$ be some boundary data and take a maximal dissipative boundary condition on the boundary $\mathcal{B}$ which is of the form
        \begin{subequations}
            \label{raw_bc}
            \begin{align}
                \label{raw_bc_1}
                -B_{12} \left( \frac{E_{\bf 12}+E_{\bf 21}}{2}+\frac{H_{\bf 11}-H_{\bf 22}}{2} \right) + (1-B_{11}) \frac{E_{\bf 11}-E_{\bf 22}}{2} \nonumber \\
                +(1+B_{11}) \frac{H_{\bf 12}+H_{\bf 21}}{2} &= d_1 \,, \\ 
                \label{raw_bc_2}
                B_{21} \left( \frac{H_{\bf 12}+H_{\bf 21}}{2} - \frac{E_{\bf 11}-E_{\bf 22}}{2} \right) + (1-B_{22}) \frac{H_{\bf 11}-H_{\bf 22}}{2} \nonumber \\
                -(1+B_{22}) \frac{E_{\bf 12}+E_{\bf 21}}{2} &= d_2 \,,
            \end{align}
        \end{subequations}
        where $B = (B_{ij})_{1\leq i,j \leq 2} \in \mathcal{C}^\infty(\mathcal{B}, M_2(\R))$ satisfying $B^\intercal B\leq I$.
        \item Let $\underline{u}_\star \in \mathcal{C}^\infty(\overline{\mathcal{U}}_\star,\R^m)$ be the initial data and assume that
        \begin{itemize}
            \item[i)] one has on $\mathcal{U}_\star$
            \begin{equation}
                \tag{\ref{eq:initial_conditions_ei0}}
                e_{\bf i}{}^0 = 0 \,,
            \end{equation}
            \item[ii)] one has on the corner $\mathcal{B}_\star$
            \begin{subequations}
            \label{ode_initial_data}
            \begin{alignat}{9}
                \label{ode_initial_data_a}
                \acute{\Gamma}_{\bf A}{}^{\bf 3}{}_{\bf 0} &= 0 \,, &\qquad U_{\bf A3} &= \; 0 \,, &\qquad \zeta_{\bf A} &= 0 \,, &\qquad e_{\bf A}{}^3 &= 0 \,, \\
                \label{ode_initial_data_b}
                \acute{\Gamma}_{\bf 3}{}^{\bf 3}{}_{\bf 0}  &= 0 \,, & U_{\bf 33} &= - f_\star \,, & e_{\bf 3}{}^3 &< 0 \,, 
            \end{alignat}
            \end{subequations}
            \item[iii)] $\underline{u}_\star$ satisfies on $\mathcal{B}_\star$ the analytic compatibility conditions of all orders with the boundary conditions \eqref{raw_bc}, see the last remark below the proposition.
        \end{itemize}
    \end{itemize}
    Then there exists a solution $\underline{u} \in \mathcal{C}^\infty(\mathcal{V},\R^m)$ of the evolution system \eqref{evol_system} with initial data $\underline{u}_\star$ and  boundary condition \eqref{raw_bc} on a subset $\mathcal{V} \subset \R^4_{3\geq0}$ containing $\mathcal{U}_\star$ such as in~\Cref{fig:RegionsBoundary}. This solution verifies
    \begin{equation}
        \label{eq:solve_near_frame}
        e_{\bf A}{}^3 = 0 \quad \text{and} \quad e_{\bf 3}{}^3 < 0 \quad \text{on } \mathcal{V}\cap\mathcal{B} \,,
    \end{equation}
    and the function $\Theta$ defined by \eqref{redef_Theta} on $\mathcal{V}$ is a boundary defining function of $\mathcal{V}\cap\mathcal{B}$.
    
    Furthermore, it is locally unique in the sense that if $\underline{u}' \in \mathcal{C}^\infty(\mathcal{V}',\R^m)$ is another such solution then $\underline{u}$ and $\underline{u}'$ coincide on a subset $\mathcal{V}''$ of $\mathcal{V} \cap \mathcal{V}'$ such as in~\Cref{fig:RegionsBoundary}.
\end{prop}

\begin{rems} \,
    \begin{itemize}
        \item For comparison, we recall here Friedrich's statement of the maximal dissipative boundary conditions \eqref{raw_bc} in spinorial formulation \cite[Equations (5.48)-(5.49)]{F95}
        \[ \phi_{1111} - a \phi_{0000} - c \Bar{\phi}_{0'0'0'0'} = d \,, \qquad \text{with } |a|+|c|\leq 1, \]
        where $\phi$ is the rescaled Weyl spinor and $a,c,d \in \mathcal{C}^\infty(\mathcal{B},\C)$. The boundary data is the function $d$ while the functions $a,c$ play the role of the matrix $B$. More precisely, $a$ corresponds to the antisymmetric part and the trace of $B$ while $c$ corresponds to its traceless symmetric part.
        
        \item Note that equation \eqref{eq:CVE_V} is left invariant under the transformation $V \mapsto V\star$. By \eqref{sym_dual_EH}, this translates into the invariance of \eqref{evol_system_EAB}-\eqref{evol_system_H3A} under the transformation $(E,H) \mapsto (H,-E)$. Since these equations are the ones generating the boundary term in the a priori estimates, the boundary term remains the same and thus the family of all maximally dissipative boundary conditions \eqref{raw_bc} should be stable under the transformation $(E,H) \mapsto (H,-E)$. This is indeed the case with the transformation $(B_{11},B_{12},B_{21},B_{22},d_1,d_2) \mapsto (B_{22},-B_{21},-B_{12},B_{11},-d_2,d_1)$.
        
        \item The smoothness of a solution to an initial boundary value problem requires that certain compatibility conditions (or corner conditions) between the initial data and the boundary conditions have to be satisfied on the corner between the initial hypersurface and the boundary.

        These conditions can be arranged into a hierarchy indexed by $\N_0$, according to the order of regularity. The $k$-th order compatibility conditions are obtained by differentiating $k$ times the boundary conditions with respect to time, using the evolution system as much as possible to replace the time derivatives, and eventually evaluating on the corner. This procedure can be done iteratively. The explicit form of these compatibility conditions generally becomes more lengthy and complex as the order increases. \qedhere
    \end{itemize}
\end{rems}

\begin{proof}
    We will only present a sketch of the proof, which follows the same kind of steps as the proof of \Cref{evol_solve_away}, but step 2 is now more involved. The existence of a solution for the modified quasilinear system is obtained by solving an iteration of semilinear systems and showing convergence of the sequence of these solutions. Friedrich noticed in \cite[Section 5.3]{F95} that total control of the boundary matrix for the semilinear systems can be achieved under some assumptions on the corner $\mathcal{B}_\star$.
    
\begin{itemize}
    \item \underline{Step 1}: Reduction to a symmetric hyperbolic system on $\R^4_{3\geq0}$

    The extension $\widetilde{\underline{u}}_\star$ of $\underline{u}_\star$ is taken in $\mathcal{C}^\infty_c(\R^3_{3\geq0},\R^m)$ instead of $\mathcal{C}^\infty_c(\R^3,\R^m)$ as in step 1 in the proof of \Cref{evol_solve_away}. The modification of the function $A^0$ into $\widetilde{A}^0$ is identical.

    \item \underline{Step 2}: Solving the symmetric hyperbolic system on $\R^4_{3\geq0}$

    In order to prove existence of a solution $\widetilde{\underline{u}}$ of \eqref{evol_system_form_modif} with initial data $\widetilde{\underline{u}}_\star$ and boundary condition yet to be derived, one construct a sequence $(\widetilde{\underline{u}}_k)_{k \in \N_0}$ where $\widetilde{\underline{u}}_{k+1}$ is solution to the semilinear symmetric hyperbolic system
    \begin{equation}
        \label{semi_linear_SHS}
        \widetilde{A}^0(\widetilde{\underline{u}}_k) \partial_0 \widetilde{\underline{u}}_{k+1} + A^i (\widetilde{\underline{u}}_k) \partial_i \widetilde{\underline{u}}_{k+1} = \mathcal{Q}(y,s_\star,\Theta_\star,\widetilde{\underline{u}}_k) \,,
    \end{equation}
    with initial data $\widetilde{\underline{u}}_\star$ and the boundary condition. Then a fixed point argument in suitable Sobolev spaces such as in Secchi \cite{S96} gives existence of a solution $\widetilde{\underline{u}}$. One can then prove that it is smooth as in step 2 of the proof of \Cref{evol_solve_away}.

    We will thus only focus on solving the semilinear symmetric hyperbolic systems.

    \begin{itemize}
        \item[\ding{70}] \underline{Step 2a}: Identifying the boundary matrix
        
        Let $k \in \N_0$. The outward-pointing normal vector field of the boundary $\mathcal{B}$ is $n = (0,0,0,-1)$. Using \eqref{def_Amu}, the boundary matrix is thus given by 
        \[  \left( n_0 \widetilde{A}^0(\widetilde{\underline{u}}_k) + n_i A^i(\widetilde{\underline{u}}_k) \right) = -A^3(\widetilde{\underline{u}}_k) = - e_{\bf i}{}^3[\widetilde{\underline{u}}_k] A^{\bf i}\,,  \]
        where $e_{\bf i}{}^3[\widetilde{\underline{u}}_k]$ denotes the component $e_{\bf i}{}^3$ of $\widetilde{\underline{u}}_k$.

        Assume that $k\geq1$. By hypotheses on $\Theta_\star$, $s_\star$ and definitions \eqref{redef_szeta0Theta}, one has
        \[ s = 0 \quad \text{on } \mathcal{B} \,, \qquad \zeta_{\bf 0} = 0 \quad \text{on } \mathcal{B}\,, \qquad \Theta = 0 \quad \text{on } \mathcal{B}\,. \]
        It follows that the semilinear system \eqref{semi_linear_SHS} on $\widetilde{\underline{u}}_k$ implies in particular the following ordinary differential equations on $\mathcal{B}$
        \begin{subequations}
            \begin{align}
                \label{ode_B_a}
                \partial_0 \underbracket{\acute{\Gamma}_{\bf A}{}^{\bf 3}{}_{\bf 0}} + \acute{\Gamma}_{\bf A}{}^{\bf B}{}_{\bf 0} \underbracket{e_{\bf B}{}^3} + \underbracket{\acute{\Gamma}_{\bf A}{}^{\bf 3}{}_{\bf 0}} e_{\bf 3}{}^3 &= 0 \,, \\
                \partial_0 \underbracket{U_{\bf A3}} + \acute{\Gamma}_{\bf A}{}^{\bf B}{}_{\bf 0} \underbracket{U_{\bf B3}} + \underbracket{\acute{\Gamma}_{\bf A}{}^{\bf 3}{}_{\bf 0}} U_{\bf 33} - \slashed{\epsilon}^{\bf BC} \underbracket{\zeta_{\bf B}} H_{\bf CA} &= 0 \,, \\
                \partial_0 \underbracket{\zeta_{\bf A}} &= 0 \,, \\
                \label{ode_B_d}
                \partial_0 \underbracket{e_{\bf A}{}^3} + \acute{\Gamma}_{\bf A}{}^{\bf B}{}_{\bf 0} \underbracket{e_{\bf B}{}^3} + \underbracket{\acute{\Gamma}_{\bf A}{}^{\bf 3}{}_{\bf 0}} e_{\bf 3}{}^3 &= 0 \,, \\
                \label{ode_B_e}
                \partial_0 \acute{\Gamma}_{\bf 3}{}^{\bf 3}{}_{\bf 0} + \acute{\Gamma}_{\bf 3}{}^{\bf A}{}_{\bf 0} \underbracket{\acute{\Gamma}_{\bf A}{}^{\bf 3}{}_{\bf 0}} + \acute{\Gamma}_{\bf 3}{}^{\bf 3}{}_{\bf 0} \acute{\Gamma}_{\bf 3}{}^{\bf 3}{}_{\bf 0} - U_{\bf 33} &= 0 \,, \\
                \label{ode_B_f}
                \partial_0 U_{\bf 33} + \acute{\Gamma}_{\bf 3}{}^{\bf A}{}_{\bf 0} \underbracket{U_{\bf A3}} + \acute{\Gamma}_{\bf 3}{}^{\bf 3}{}_{\bf 0} U_{\bf 33} - \slashed{\epsilon}^{\bf BC} \underbracket{\zeta_{\bf B}} H_{\bf C3} &= 0 \,, \\
                \label{ode_B_g}
                \partial_0 e_{\bf 3}{}^3 + \acute{\Gamma}_{\bf 3}{}^{\bf A}{}_{\bf 0} \underbracket{e_{\bf B}{}^3} + \acute{\Gamma}_{\bf 3}{}^{\bf 3}{}_{\bf 0} e_{\bf 3}{}^3 &= 0
            \end{align}
        \end{subequations}
        In the equations above, we dropped the notation $[\widetilde{\underline{u}}_k]$ for conciseness since they all are components of $\widetilde{\underline{u}}_k$. Moreover, we put brackets under certain components to highlight the fact that \eqref{ode_B_a}-\eqref{ode_B_d} are homogeneous. By hypothesis on the initial data \eqref{ode_initial_data_a} and assuming that $\widetilde{\underline{u}}_k$ is continuous on $\mathcal{B}$, one deduces from the Cauchy-Lipschitz theorem that
        \begin{align*}
            \acute{\Gamma}_{\bf A}{}^{\bf 3}{}_{\bf 0} &= 0 \quad \text{on } \mathcal{B} \,, \\
            U_{\bf A3} &= 0 \quad \text{on } \mathcal{B} \,, \\
            \zeta_{\bf A} &= 0 \quad \text{on } \mathcal{B} \,, \\
            e_{\bf A}{}^3 &= 0 \quad \text{on } \mathcal{B} \,.
        \end{align*}
        Assuming that $\widetilde{\underline{u}}_k$ is of class $\mathcal{C}^1$ on $\mathcal{B}$, let us differentiate \eqref{ode_B_e} and use \eqref{ode_B_e}-\eqref{ode_B_f} to obtain an ordinary differential equation for $\acute{\Gamma}_{\bf 3}{}^{\bf 3}{}_{\bf 0}$ on $\mathcal{B}$:
        \begin{align*}
            0 &= \partial_0^2 \acute{\Gamma}_{\bf 3}{}^{\bf 3}{}_{\bf 0}  + 2 \acute{\Gamma}_{\bf 3}{}^{\bf 3}{}_{\bf 0} \partial_0 \acute{\Gamma}_{\bf 3}{}^{\bf 3}{}_{\bf 0} - \partial_0 U_{\bf 33} \\
            &= \partial_0^2 \acute{\Gamma}_{\bf 3}{}^{\bf 3}{}_{\bf 0} + 2 \acute{\Gamma}_{\bf 3}{}^{\bf 3}{}_{\bf 0} \partial_0 \acute{\Gamma}_{\bf 3}{}^{\bf 3}{}_{\bf 0} +\acute{\Gamma}_{\bf 3}{}^{\bf 3}{}_{\bf 0} U_{\bf 33} \\
            &= \partial_0^2 \acute{\Gamma}_{\bf 3}{}^{\bf 3}{}_{\bf 0} + 2 \acute{\Gamma}_{\bf 3}{}^{\bf 3}{}_{\bf 0} \partial_0 \acute{\Gamma}_{\bf 3}{}^{\bf 3}{}_{\bf 0} +\acute{\Gamma}_{\bf 3}{}^{\bf 3}{}_{\bf 0} \left(\partial_0 \acute{\Gamma}_{\bf 3}{}^{\bf 3}{}_{\bf 0} + \acute{\Gamma}_{\bf 3}{}^{\bf 3}{}_{\bf 0} \acute{\Gamma}_{\bf 3}{}^{\bf 3}{}_{\bf 0}\right) \\
            0 &= \partial_0^2 \acute{\Gamma}_{\bf 3}{}^{\bf 3}{}_{\bf 0} + 3 \acute{\Gamma}_{\bf 3}{}^{\bf 3}{}_{\bf 0} \partial_0 \acute{\Gamma}_{\bf 3}{}^{\bf 3}{}_{\bf 0} +\left(\acute{\Gamma}_{\bf 3}{}^{\bf 3}{}_{\bf 0}\right)^3 \,.
        \end{align*}
        Thus there exists two functions $v,w$ on $\mathcal{B}_\star$ such that
        \[ \acute{\Gamma}_{\bf 3}{}^{\bf 3}{}_{\bf 0} = \frac{2(v+y^0)}{w+(v+y^0)^2} \quad \text{on } \mathcal{B} \,. \]
        With \eqref{ode_initial_data_b} and \eqref{ode_B_e}, one finds
        \[ \acute{\Gamma}_{\bf 3}{}^{\bf 3}{}_{\bf 0} = \frac{-f_\star y^0}{1-f_\star \frac{(y^0)^2}{2}} \quad \text{on } \mathcal{B} \,, \]
        as long as the denominator stays positive. This can be assumed up to restricting the time using that $f_\star$ is uniformly bounded from above on $\mathcal{U}_\star$. From \eqref{ode_B_f}-\eqref{ode_B_g} and \eqref{ode_initial_data_b}, it follows that
        \[ U_{\bf 33} = - f_\star \left(1-f_\star \frac{(y^0)^2}{2}\right) \quad \text{on } \mathcal{B} \,, \]
        and, using the coordinates, for any point $(y^0,y^1,y^2,0) \in \mathcal{B}$
        \[ e_{\bf 3}{}^3(y^0,y^1,y^2,0) = e_{\bf 3}{}^3(0,y^1,y^2,0) \left(1-f_\star(y^1,y^2,0) \frac{(y^0)^2}{2}\right) < 0 \,. \]

        \begin{rems} \,
            \begin{itemize}
                \item[$\bullet$] The smoothness requirement on $\widetilde{\underline{u}}_k$ on $\mathcal{B}$ is ensured if $\widetilde{\underline{u}}_k$ takes its values into Sobolev spaces of high enough order.
                \item[$\bullet$] Note that the same factor than for the conformal factor in \eqref{gauge_Theta} arises here on the boundary $\mathcal{B}$. \qedhere 
            \end{itemize}
        \end{rems}

        If $k=0$, one can choose $\widetilde{\underline{u}}_0$ verifying the equations on $\mathcal{B}$ found above.

        Consequently, for all $k \in \N_0$, the boundary matrix on $\mathcal{B}$ is given by
        \[ - A^3(\widetilde{\underline{u}}_k) = -e_{\bf i}{}^3[\widetilde{\underline{u}}_k] A^{\bf i} = \mathcal{F} A^{\bf 3} \,, \]
        where $\mathcal{F}$ is a positive function on $\mathcal{B}$ which does not depend on $k$. Hence the boundary matrix properties are essentially determined by these of $A^{\bf 3}$. The key point is that the boundary matrix is independent of the iteration number $k$.
        
        \item[\ding{70}] \underline{Step 2b}: Studying the boundary matrix to find the maximal dissipative boundary conditions

        The boundary matrix is characteristic of constant multiplicity since the kernel of $A^{\bf 3}$ has dimension
        \[ \dim \ker A^{\bf 3} = m-4>0 \,. \]
        Furthermore, $A^{\bf 3}$ has two other eigenvalues namely $-1$ and $+1$. The corresponding eigenspaces are generated respectively by
        \begin{align*}
            \underline{u}_{-,1} := (+1,-1,+1,+1,\hphantom{+}0,\hphantom{+}0,\hphantom{+}0,\hphantom{+}0,\hphantom{+}0_{m-8}) \,, \\
            \underline{u}_{-,2} := (\hphantom{+}0,\hphantom{+}0,\hphantom{+}0,\hphantom{+}0,+1,-1,-1,-1,\hphantom{+}0_{m-8}) \,,
        \end{align*}
        and
        \begin{align*}
            \underline{u}_{+,1} := (+1,-1,-1,-1,\hphantom{+}0,\hphantom{+}0,\hphantom{+}0,\hphantom{+}0,\hphantom{+}0_{m-8}) \,, \\
            \underline{u}_{+,2} := (\hphantom{+}0,\hphantom{+}0,\hphantom{+}0,\hphantom{+}0,+1,-1,+1,+1,\hphantom{+}0_{m-8}) \,.
        \end{align*}
        The outgoing (respectively ingoing) characteristic fields correspond to the eigenspaces associated to positive (respectively negative) eigenvalues of the boundary matrix. The maximal dissipative boundary conditions write under the form
        \[ \begin{pmatrix}
                E_{\bf 11}-E_{\bf 22}+H_{\bf 12}+H_{\bf 21} \\
                H_{\bf 11}-H_{\bf 22}-E_{\bf 12}-E_{\bf 21}
            \end{pmatrix} = B \begin{pmatrix}
                E_{\bf 11}-E_{\bf 22}-H_{\bf 12}-H_{\bf 21} \\
                H_{\bf 11}-H_{\bf 22}+E_{\bf 12}+E_{\bf 21}
            \end{pmatrix} + 2 \begin{pmatrix} d_1 \\ d_2 \end{pmatrix} \,, \]
        where $B \in \mathcal{C}^\infty(\mathcal{B}, M_2(\R))$ satisfying $B^\intercal B\leq I$ and $d_1$, $d_2$ represent the boundary data.

        \begin{rem}
            One can verify that the quadratic form associated to the normal matrix of the boundary is non-negative on the space of vectors of $\R^m$ satisfying the boundary condition with vanishing boundary data $d_1 = d_2 = 0$.
        \end{rem}

        \item[\ding{70}] \underline{Step 2c}: Solving the semilinear problems

        By hypothesis $iii)$, the compatibility conditions of any order hold on the corner $\mathcal{B}_\star$ for the quasilinear system \eqref{evol_system}. Thanks to hypothesis $i)$, these compatibility conditions are identical to those of the modified quasilinear system obtained in step 1. Since the initial data for the semilinear problems \eqref{semi_linear_SHS} is the same as for the modified quasilinear one, that is $\widetilde{\underline{u}}_\star$, it follows that the compatibility conditions of any order for the semilinear problems hold. Therefore, the existence of $(\widetilde
        {\underline{u}}_k)_{k\geq 1}$ can be derived iteratively by application of Rauch \cite[Theorems 9 and 10]{R85}. 
    \end{itemize}

    \item \underline{Step 3}: Returning to the original problem

    Analogous to step 3 of the proof of \Cref{evol_solve_away}. \qedhere
    \end{itemize}
\end{proof}

\begin{lem}
    Let $\underline{u} \in \mathcal{C}^\infty(\mathcal{V},\R^m)$ be the solution to the evolution system \eqref{evol_system} given by \Cref{evol_solve_away} or \Cref{evol_solve_near}. If the initial data $\underline{u}_\star \in \mathcal{C}^\infty(\overline{\mathcal{U}}_\star,\R^m)$ verifies furthermore
    \begin{equation}
        \label{eq:symmetry_EH}
        E_{\bf 12} = E_{\bf 21} \,, \qquad H_{\bf 12} = H_{\bf 21} \,,
    \end{equation}
    then this remains true on $\mathcal{V}$.
\end{lem}

\begin{proof}
    From the evolution system \eqref{evol_system}, one deduces that
    \[ \partial_0(E_{\bf 12}-E_{\bf 21}) = 0 \,, \qquad \partial_0(H_{\bf 12}-H_{\bf 21}) = 0 \,. \]
    Hence the result.
\end{proof}

\subsubsection{Propagation of the constraints}
\label{sec:propagation_constraints}

The evolution system encodes only part of the zero quantities and one still has to prove that the rest also vanish identically. This is done by exploiting the differential system \eqref{eq:diff_zero_quantities} on zero quantities.

\begin{lem}
    \label{lem:propagation_constraints}
    Let $\underline{u} \in \mathcal{C}^\infty(\mathcal{V},\R)$ be the solution to the evolution system \eqref{evol_system} given by \Cref{evol_solve_away} or \Cref{evol_solve_near} for some analytic initial data $\underline{u}_\star \in \mathcal{C}^\infty(\overline{\mathcal{U}}_\star,\R^m)$ verifying the hypotheses of the aforementioned propositions and the conditions \eqref{eq:symmetry_EH}.
    \begin{itemize}
        \item[a)] Define the following vector fields on $\mathcal{V}$
        \[ e_{\bf 0} := \partial_0 \,, \qquad e_{\bf i} := e_{\bf i}{}^\mu \partial_\mu \,. \]
        If the initial data $\underline{u}_\star$ is such that
        \begin{equation}
            \label{eq:initial_frame_field}
            \det (e_{\bf i}{}^j)_{1 \leq {\bf i},j \leq 3} \text{ nowhere vanishes on } \mathcal{U}_\star
        \end{equation}
        then $(e_{\bf a})$ is a frame field on $\mathcal{V}$.
    \end{itemize}
    The dual coframe field is denoted by $(\omega^{\bf a})$. One fixes an orientation on $\mathcal{V}$ by imposing that $(e_{\bf a})$ is a positively oriented frame field.
    \begin{itemize}
        \item[b)] Define further
        \begin{itemize}
            \item[$\bullet$] the following tensor fields on $\mathcal{V}$
            \[ \kappa := \kappa_{\bf i} \, \omega^{\bf i} \,, \qquad \zeta := \zeta_{\bf a} \, \omega^{\bf a} \,, \qquad \breve{g} := \eta_{\bf ab} \, \omega^{\bf a} \otimes \omega^{\bf b} \,, \qquad U := U_{\bf ij} \, \omega^{\bf i} \otimes \omega^{\bf j} + U_{\bf i0} \, \omega^{\bf i} \otimes \omega^{\bf 0} \,, \]
            
            \item[$\bullet$] $\acute{\nabla}$ as the unique connection such that its connection coefficients relative to the frame $(e_{\bf a})$ are given by
            \[ \left(\acute{\Gamma}_{\bf 0}{}^{\bf a}{}_{\bf b} = 0 \,, \acute{\Gamma}_{\bf i}{}^{\bf 0}{}_{\bf 0} \,, \acute{\Gamma}_{\bf i}{}^{\bf j}{}_{\bf 0} \,, \acute{\Gamma}_{\bf i}{}^{\bf 0}{}_{\bf j} \,, \acute{\Gamma}_{\bf i}{}^{\bf k}{}_{\bf j} \right) \,, \]
            
            \item[$\bullet$] $E$ and $H$ as the symmetric and trace-free 2-tensor on the distribution spanned by the family $(e_{\bf i})$ given by
            \begin{align*}
                E &:=  E_{\bf AB} \, \omega^{\bf A} \otimes \omega^{\bf B} - \delta^{\bf CD}E_{\bf CD} \, \omega^{\bf 3} \otimes \omega^{\bf 3} + E_{\bf 3A} \left( \omega^{\bf 3} \otimes \omega^{\bf A} + \omega^{\bf A} \otimes \omega^{\bf 3} \right) \,, \\
                H &:=  H_{\bf AB} \, \omega^{\bf A} \otimes \omega^{\bf B} - \delta^{\bf CD}H_{\bf CD} \, \omega^{\bf 3} \otimes \omega^{\bf 3} + H_{\bf 3A} \left( \omega^{\bf 3} \otimes \omega^{\bf A} + \omega^{\bf A} \otimes \omega^{\bf 3} \right)  \,,
            \end{align*}
            
            \item[$\bullet$] $V$ as the only Weyl candidate such that
            \[ V^{\bf 0}{}_{\bf i0j} = E_{\bf ij} \,, \qquad (\star V)^{\bf 0}{}_{\bf i0j} = H_{\bf ij} \,. \]
        \end{itemize}
        If the initial data $\underline{u}_\star$ and the smooth function $s_\star$ are such that one has on $\mathcal{U}_\star$
        \begin{alignat}{9}
            \varkappa_{\bf ia} &= 0 \,, &\qquad \mho_{\bf i} &=0 \,, &\qquad \varsigma_{\bf i} &=0 \,, &\qquad \acute{\vartheta}_{\bf iab} &=0 \,, &\qquad \aleph &=0 \,, \notag \\
            \varpi^{\bf a}{}_{\bf bij} &= 0 \,, & \Pi_{\bf ij} &= 0 \,, & \varrho_{\bf ija} &= 0 \,, &\qquad \acute{\Sigma}_{\bf i}{}^{\bf c}{}_{\bf j} &= 0 \,, &\qquad \varphi_{\bf AB0} &= 0 \,, \notag \\
            \varphi_{\bf A30} &= 0 \,, &\qquad (\star\varphi)_{\bf AB0} &= 0 \,, &\qquad (\star\varphi)_{\bf A30} &= 0 \label{eq:initial_conditions_constraints} \,,
        \end{alignat}
        then all the zero quantities defined in \Cref{def:zero_quantities} vanish identically on a neighbourhood $\mathcal{W}$ of $\mathcal{U}_\star$ in $\mathcal{V}$.
    \end{itemize}
\end{lem}

\begin{proof}
    a) By definition of the vector fields, one has on $\mathcal{V}$
    \[ \det(e_{\bf 0},e_{\bf 1},e_{\bf 2},e_{\bf 3}) = \det \begin{pmatrix}
        1 & e_{\bf i}{}^0 \\
        0 & e_{\bf i}{}^j
        \end{pmatrix} = \det (e_{\bf i}{}^j)_{{\bf i},j} \,. \]    
    Moreover, one derives from the evolution system \eqref{evol_system} that
    \[ \partial_0 \left( \det (e_{\bf i}{}^j)_{{\bf i},j} \right) + \acute{\Gamma}_{\bf k}{}^{\bf k}{}_{\bf 0}  \det (e_{\bf i}{}^j)_{{\bf i},j} = 0 \,. \]
    Thus if \eqref{eq:initial_frame_field} holds then $\det(e_{\bf 0},e_{\bf 1},e_{\bf 2},e_{\bf 3})$ nowhere vanishes on $\mathcal{V}$.

    \bigbreak

    \noindent b) Let us structure the proof into four steps.
    \begin{itemize}
        \item \underline{Step 1}: Vanishing zero quantities by definition of the fields and the evolution
            
        By definition of $s$, $\zeta_{\bf 0}$ and $\Theta$ given by \eqref{redef_szeta0Theta}, one can check that
        \[ \mho_{\bf 0} = 0 \,, \qquad \varsigma_{\bf 0} = 0 \,, \qquad \varkappa_{\bf 00} = 0 \,. \]
        Furthermore, the evolution system \eqref{evol_system} gives
        \begin{alignat*}{9}
            \varpi^{\bf a}{}_{\bf b0i} &= 0 \,, &\qquad \Pi_{\bf 0i} &= 0 \,, &\qquad  \varrho_{\bf 0ia} &= 0 \,, \\
            \varkappa_{\bf 0i} &= 0 \,, & \acute{\vartheta}_{\bf 0ab} &= 0 \,, & \acute{\Sigma}_{\bf 0}{}^{\bf a}{}_{\bf i} &= 0 \,.
        \end{alignat*}
        and
        \begin{subequations}
            \label{intermediaire}
            \begin{alignat}{4}
                \varphi_{\bf 0A3} + \mathcal{F}(\acute{\vartheta}) &= 0 \,, &\qquad
                \varphi_{\bf 0(AB)} + \mathcal{F}(\acute{\vartheta}) &= 0 \,, \\
                (\star \varphi)_{\bf 0A3} + \mathcal{F}(\acute{\vartheta}) &= 0 \,, &\qquad
                (\star \varphi)_{\bf 0(AB)} + \mathcal{F}(\acute{\vartheta}) &= 0 \,,
            \end{alignat}
        \end{subequations}
        where $\mathcal{F}(\acute{\vartheta})$ denote homogeneous terms of the zero quantity $\acute{\vartheta}_{\bf abc}$ and where
        \[ (\star \varphi)_{\bf abc} := \frac{1}{2} \breve{\epsilon}_{\bf ab}{}^{\bf de} \varphi_{\bf dec} \,. \]

        \item \underline{Step 2}: Propagation of $\varkappa_{\bf ia}=0$, $\mho_{\bf i}=0$, $\varsigma_{\bf i}=0$, $\acute{\vartheta}_{\bf iab}=0$ and $\aleph=0$

        On one hand, one has by definition of the fields and step 1,
        \begin{align*}
            2 \acute{\nabla}_{\bf [0} \varkappa_{\bf i]a} &= e_{\bf 0}(\varkappa_{\bf ia}) + \acute{\Gamma}_{\bf i}{}^{\bf k}{}_{\bf 0} \varkappa_{\bf ka} \,, \\
            2 \acute{\nabla}_{\bf [0} \mho_{\bf i]} &= e_{\bf 0}(\mho_{\bf i}) + \acute{\Gamma}_{\bf i}{}^{\bf k}{}_{\bf 0} \mho_{\bf k} \,, \\ 
            2 \acute{\nabla}_{\bf [0} \varsigma_{\bf i]} &= e_{\bf 0}(\varsigma_{\bf i}) + \acute{\Gamma}_{\bf i}{}^{\bf k}{}_{\bf 0} \varsigma_{\bf k} \,, \\
            2 \acute{\nabla}_{\bf [0} \acute{\vartheta}_{\bf i]ab} &= e_{\bf 0}(\acute{\vartheta}_{\bf iab}) - \acute{\Gamma}_{\bf i}{}^{\bf k}{}_{\bf 0} \acute{\vartheta}_{\bf kab} \,, \\
            \acute{\nabla}_{\bf 0} \aleph &= e_{\bf 0}(\aleph) \,.
        \end{align*}
        On the other hand, equations \eqref{eq:curl_varkappa}, \eqref{eq:curl_mho}, \eqref{eq:curl_varsigma}, \eqref{eq:curl_vartheta} and \eqref{eq:curl_aleph} simplify greatly by definition of the fields and step 1 to
        \begin{align*}
            2 \acute{\nabla}_{\bf [0} \varkappa_{\bf i]a} &= s \acute{\vartheta}_{\bf i0a} + \eta_{\bf a0} \mho_{\bf i} \,, \\
            2 \acute{\nabla}_{\bf [0} \mho_{\bf i]} &= 0 \,, \\
            2 \acute{\nabla}_{\bf [0} \varsigma_{\bf i]} &= - \varkappa_{\bf i0} \,, \\
            2 \acute{\nabla}_{\bf [0} \acute{\vartheta}_{\bf i]ab} &=  0 \,, \\
            \acute{\nabla}_{\bf 0} \aleph &= 0 \,.
        \end{align*}
        Hence
        \begin{subequations}
            \begin{align}
                \label{propagation_varkappa}
                \partial_0(\varkappa_{\bf ia}) + \acute{\Gamma}_{\bf i}{}^{\bf k}{}_{\bf 0} \varkappa_{\bf ka} &= s \acute{\vartheta}_{\bf i0a} + \eta_{\bf a0} \mho_{\bf i} \,, \\
                \label{propagation_mho}
                \partial_0(\mho_{\bf i}) + \acute{\Gamma}_{\bf i}{}^{\bf k}{}_{\bf 0} \mho_{\bf k} &= 0 \,, \\
                \label{propagation_varsigma}
                \partial_0(\varsigma_{\bf i}) + \acute{\Gamma}_{\bf i}{}^{\bf k}{}_{\bf 0} \varsigma_{\bf k} &= -\varkappa_{\bf i0} \,, \\
                \label{propagation_vartheta}
                \partial_0(\acute{\vartheta}_{\bf iab}) - \acute{\Gamma}_{\bf i}{}^{\bf k}{}_{\bf 0} \acute{\vartheta}_{\bf kab} &= 0 \,, \\
                \label{propagation_aleph}
                \partial_0(\aleph) &= 0 \,.
            \end{align}
        \end{subequations}
        This is a homogeneous linear system of ordinary differential equations. By the Cauchy-Lipschitz theorem, there exists a unique solution on $\mathcal{V}$ for a given initial data. Since the zero quantities $(\varkappa_{\bf ia},\mho_{\bf i},\varsigma_{\bf i},\acute{\vartheta}_{\bf iab},\aleph)$ vanish identically on $\mathcal{U}_\star$ by \eqref{eq:initial_conditions_constraints}, they vanish identically on $\mathcal{V}$.

        \item \underline{Step 3}: Decomposing $\varphi_{\bf abc}$
    
        Thanks to steps 1 and 2, one deduces that $\acute{\vartheta}_{\bf abc} = 0$. It follows that $\varphi_{\bf abc}$ is a Cotton candidate. Indeed, one has by definition of $\varphi_{\bf abc}$, see \eqref{eq:ZQ_varphi},
        \[ \varphi_{\bf (ab)c} = 0 \,, \qquad \varphi_{[\bf abc]} = 0 \,, \qquad \varphi_{\bf ab}{}^{\bf b} = - \acute{\vartheta}_{\bf dbc} V^{\bf dc}{}_{\bf a}{}^{\bf b} \,. \]
        By keeping only algebraically independent frame components of $\varphi_{\bf abc}$, it is sufficient to study
        \begin{alignat*}{4}
            \varphi_{\bf 0(AB)} \,, \qquad& \varphi_{\bf 0A3} \,, \qquad&
            (\star\varphi)_{\bf 0(AB)} \,, \qquad & (\star\varphi)_{\bf 0A3} \,, \\
            \varphi_{\bf 120} \,, \qquad& \varphi_{\bf A30} \,, \qquad&
            (\star\varphi)_{\bf 120} \,, \qquad& (\star\varphi)_{\bf A30} \,.
        \end{alignat*}
        Now the frame components on the first line vanish on $\mathcal{V}$ by \eqref{intermediaire} and $\acute{\vartheta}_{\bf abc} = 0$. The frame components on the second line will be denoted by the symbol $\overline{\underline{\varphi}}$.
    
        \item \underline{Step 4}: Propagation of $\overline{\underline{\varphi}} = 0$, $\varpi^{\bf a}{}_{\bf bij} = 0$, $\Pi_{\bf ia} = 0$, $\varrho_{\bf ija} = 0$ and $\acute{\Sigma}_{\bf i}{}^{\bf c}{}_{\bf j} = 0$
    
        On one hand, by definition of the fields and step 1,
        \begin{align*}
            3 \acute{\nabla}_{\bf [0} \varpi^{\bf a}{}_{\bf |b|ij]} &= e_{\bf 0}(\varpi^{\bf a}{}_{\bf bij}) - \acute{\Gamma}_{\bf j}{}^{\bf k}{}_{\bf 0} \varpi^{\bf a}{}_{\bf bki} + \acute{\Gamma}_{\bf i}{}^{\bf k}{}_{\bf 0} \varpi^{\bf a}{}_{\bf bkj} \,, \\
            3 \acute{\nabla}_{\bf [0} \Pi_{\bf ij]} &= e_{\bf 0}(\Pi_{\bf ij}) - \acute{\Gamma}_{\bf j}{}^{\bf k}{}_{\bf 0} \Pi_{\bf ki} + \acute{\Gamma}_{\bf i}{}^{\bf k}{}_{\bf 0} \Pi_{\bf kj} \,, \\ 
            3 \acute{\nabla}_{\bf [0} \varrho_{\bf ij]a} &= e_{\bf 0}(\varrho_{\bf ija}) - \acute{\Gamma}_{\bf j}{}^{\bf k}{}_{\bf 0} \varrho_{\bf kia} + \acute{\Gamma}_{\bf i}{}^{\bf k}{}_{\bf 0} \varrho_{\bf kja} \,, \\
            3 \acute{\nabla}_{\bf [0} \acute{\Sigma}_{\bf i}{}^{\bf a}{}_{\bf j]} &= e_{\bf 0}(\acute{\Sigma}_{\bf i}{}^{\bf a}{}_{\bf j}) - \acute{\Gamma}_{\bf j}{}^{\bf k}{}_{\bf 0} \acute{\Sigma}_{\bf k}{}^{\bf a}{}_{\bf i} + \acute{\Gamma}_{\bf i}{}^{\bf k}{}_{\bf 0} \acute{\Sigma}_{\bf k}{}^{\bf a}{}_{\bf j} \,.
        \end{align*}
        On the other hand, equations \eqref{eq:curl_varpi}, \eqref{eq:curl_Pi}, \eqref{eq:curl_varrho} and \eqref{eq:curl_Sigma}, simplify greatly by definition of the fields and steps 1 and 2 to
        \begin{align*}
            3 \acute{\nabla}_{\bf [0} \varpi^{\bf a}{}_{\bf |b|ij]} &= - \acute{\Sigma}_{\bf i}{}^{\bf f}{}_{\bf j} \acute{R}^{\bf a}{}_{\bf b0f} + S_{\bf b0}{}^{\bf ac} \varrho_{\bf ijc} - \frac{\Theta}{6} \breve{\epsilon}^{\bf k}{}_{\bf 0ij} \breve{\epsilon}{}^{\bf a}{}_{\bf b}{}^{\bf cd} \varphi_{\bf cdk} \,, \\
            3 \acute{\nabla}_{\bf [0} \Pi_{\bf ij]} &= - \frac{1}{2} \acute{\Sigma}_{\bf i}{}^{\bf l}{}_{\bf j} \acute{\Gamma}_{\bf l}{}^{\bf k}{}_{\bf 0} \kappa_{\bf k} - \frac{1}{2} \varpi^{\bf k}{}_{\bf 0ij} \kappa_{\bf k} +  \varrho_{\bf ij0} \,, \\
            3 \acute{\nabla}_{[\bf 0} \varrho_{\bf ij]a} &= -\acute{\Sigma}_{\bf i}{}^{\bf l}{}_{\bf j} \acute{\Gamma}_{\bf l}{}^{\bf k}{}_{\bf 0} U_{\bf ka} - \varpi^{\bf k}{}_{\bf 0ij} U_{\bf ka} - \frac{1}{6} \zeta_{\bf d} \breve{\epsilon}^{\bf k}{}_{\bf 0ij} \breve{\epsilon}{}^{\bf d}{}_{\bf a}{}^{\bf bc} \varphi_{\bf bck} \,, \\
            3 \acute{\nabla}_{\bf [0} \acute{\Sigma}_{\bf i}{}^{\bf d}{}_{\bf j]} &= - \varpi^{\bf d}{}_{\bf 0ij} \,.
        \end{align*}
        Hence
        \begin{subequations}
            \begin{align}
                \label{propagation_varpi}
                \partial_0(\varpi^{\bf a}{}_{\bf bij}) - \acute{\Gamma}_{\bf j}{}^{\bf k}{}_{\bf 0} \varpi^{\bf a}{}_{\bf bki} + \acute{\Gamma}_{\bf i}{}^{\bf k}{}_{\bf 0} \varpi^{\bf a}{}_{\bf bkj} &= - \acute{\Sigma}_{\bf i}{}^{\bf f}{}_{\bf j} \acute{R}^{\bf a}{}_{\bf b0f} + S_{\bf b0}{}^{\bf ac} \varrho_{\bf ijc} - \frac{\Theta}{6} \breve{\epsilon}^{\bf k}{}_{\bf 0ij} \breve{\epsilon}{}^{\bf a}{}_{\bf b}{}^{\bf cd} \varphi_{\bf cdk} \,, \\
                \label{propagation_Pi}
                \partial_0(\Pi_{\bf ij}) - \acute{\Gamma}_{\bf j}{}^{\bf k}{}_{\bf 0} \Pi_{\bf ki} + \acute{\Gamma}_{\bf i}{}^{\bf k}{}_{\bf 0} \Pi_{\bf kj} &= - \frac{1}{2} \acute{\Sigma}_{\bf i}{}^{\bf l}{}_{\bf j} \acute{\Gamma}_{\bf l}{}^{\bf k}{}_{\bf 0} \kappa_{\bf k} - \frac{1}{2} \varpi^{\bf k}{}_{\bf 0ij} \kappa_{\bf k} + \varrho_{\bf ij0} \,, \\ 
                \label{propagation_varrho}
                \partial_0(\varrho_{\bf ija}) - \acute{\Gamma}_{\bf j}{}^{\bf k}{}_{\bf 0} \varrho_{\bf kia} + \acute{\Gamma}_{\bf i}{}^{\bf k}{}_{\bf 0} \varrho_{\bf kja} &= -\acute{\Sigma}_{\bf i}{}^{\bf l}{}_{\bf j} \acute{\Gamma}_{\bf l}{}^{\bf k}{}_{\bf 0} U_{\bf ka} - \varpi^{\bf k}{}_{\bf 0ij} U_{\bf ka} - \frac{1}{6} \zeta_{\bf d} \breve{\epsilon}^{\bf k}{}_{\bf 0ij} \breve{\epsilon}{}^{\bf d}{}_{\bf a}{}^{\bf bc} \varphi_{\bf bck} \,, \\
                \label{propagation_Sigma}
                \partial_0(\acute{\Sigma}_{\bf i}{}^{\bf a}{}_{\bf j}) - \acute{\Gamma}_{\bf j}{}^{\bf k}{}_{\bf 0} \acute{\Sigma}_{\bf k}{}^{\bf a}{}_{\bf i} + \acute{\Gamma}_{\bf i}{}^{\bf k}{}_{\bf 0} \acute{\Sigma}_{\bf k}{}^{\bf a}{}_{\bf j} &= - \varpi^{\bf d}{}_{\bf 0ij} \,.
            \end{align}
        \end{subequations}
        Note that terms with $\varphi_{\bf abi}$ appearing in \eqref{propagation_varpi} and \eqref{propagation_varrho} can be expressed as linear combinations of $\overline{\underline{\varphi}}$. It only remains to derive the propagation equations for $\overline{\underline{\varphi}}$.

        On one hand,
        \begin{align*}
            \breve{g}^{\bf cd} \acute{\nabla}_{\bf c} \varphi_{\bf ABd} &= -e_{\bf 0}(\varphi_{\bf AB0}) + e_{\bf 3}(\varphi_{\bf AB3}) + \delta^{\bf CD} e_{\bf C}(\varphi_{\bf ABD}) + \acute{\Gamma} \times \overline{\underline{\varphi}} \,, \\
            \breve{g}^{\bf cd} \acute{\nabla}_{\bf c} \varphi_{\bf A3d} &= -e_{\bf 0}(\varphi_{\bf A30}) + e_{\bf 3}(\varphi_{\bf A33}) + \delta^{\bf CD}e_{\bf C}(\varphi_{\bf A3D}) + \acute{\Gamma} \times \overline{\underline{\varphi}} \,, \\
            \breve{g}^{\bf cd} \acute{\nabla}_{\bf c} (\star\varphi)_{\bf ABd} &= -e_{\bf 0}((\star\varphi)_{\bf AB0}) + e_{\bf 3}((\star\varphi)_{\bf AB3}) + \delta^{\bf CD} e_{\bf C}((\star\varphi)_{\bf ABD}) + \acute{\Gamma} \times \overline{\underline{\varphi}} \,, \\
            \breve{g}^{\bf cd} \acute{\nabla}_{\bf c} (\star\varphi)_{\bf A3d} &= -e_{\bf 0}((\star\varphi)_{\bf A30}) + e_{\bf 3}((\star\varphi)_{\bf A33}) + \delta^{\bf CD}e_{\bf C}((\star\varphi)_{\bf A3D}) + \acute{\Gamma} \times \overline{\underline{\varphi}} \,,
        \end{align*}
        where $\acute{\Gamma}\times\overline{\underline{\varphi}}$ are a sum of bilinear terms in $\left(\acute{\Gamma}_{\bf a}{}^{\bf b}{}_{\bf c},\overline{\underline{\varphi}}\right)$. With step 3, one has
        \begin{align*}
            \varphi_{\bf AB3} &= -\breve{\epsilon}_{\bf AB}{}^{\bf 03}(\star\varphi)_{\bf 033} = \slashed{\epsilon}_{\bf AB} (\star\varphi)_{\bf 0C}{}^{\bf C} = 0 \,, \\
            \varphi_{\bf A33} &= -\breve{\epsilon}_{\bf A3}{}^{\bf 0B} (\star\varphi)_{\bf 0B3} = 0 \,, \\
            (\star\varphi)_{\bf AB3} &= \breve{\epsilon}_{\bf AB}{}^{\bf 03}\varphi_{\bf 033} = \slashed{\epsilon}_{\bf AB} \varphi_{\bf 0C}{}^{\bf C} = 0 \,, \\
            (\star\varphi)_{\bf A33} &= \breve{\epsilon}_{\bf A3}{}^{\bf 0B} \varphi_{\bf 0B3} = 0 \,,
        \end{align*}
        and
        \begin{align*}
            \varphi_{\bf ABD} &= - \breve{\epsilon}_{\bf AB}{}^{\bf 03} (\star\varphi)_{\bf 03D} = \slashed{\epsilon}_{\bf AB} (\star\varphi)_{\bf 03D} = \slashed{\epsilon}_{\bf AB} (\star\varphi)_{\bf D30} \,, \\
            \varphi_{\bf A3D} &= -\breve{\epsilon}_{\bf A3}{}^{\bf 0B} (\star\varphi)_{\bf 0BD} = -\slashed{\epsilon}_{\bf A}{}^{\bf B} (\star\varphi)_{\bf 0BD} = \frac{1}{2} \slashed{\epsilon}_{\bf A}{}^{\bf B} (\star\varphi)_{\bf BD0} \,, \\
            (\star\varphi)_{\bf ABD} &= \breve{\epsilon}_{\bf AB}{}^{\bf 03} \varphi_{\bf 03D} = -\slashed{\epsilon}_{\bf AB} \varphi_{\bf 03D} = -\slashed{\epsilon}_{\bf AB} \varphi_{\bf D30} \,, \\
            (\star\varphi)_{\bf A3D} &=\breve{\epsilon}_{\bf A3}{}^{\bf 0B} \varphi_{\bf 0BD} = \slashed{\epsilon}_{\bf A}{}^{\bf B} \varphi_{\bf 0BD} = -\frac{1}{2} \slashed{\epsilon}_{\bf A}{}^{\bf B} \varphi_{\bf BD0}\,.
        \end{align*}

        On the other hand, equation \eqref{eq:div_varphi} simplifies by definition of the fields and steps 1 and 2 to
        \begin{align*}
            \breve{g}^{\bf be} \acute{\nabla}_{\bf e} \varphi_{\bf cdb} &= -\frac{1}{2} \acute{\Sigma}_{\bf i}{}^{\bf b}{}_{\bf j} \acute{\nabla}_{\bf b} V^{\bf ij}{}_{\bf cd} - \varpi^{\bf i}{}_{\bf bij} V^{\bf bj}{}_{\bf cd} - V^{\bf ij}{}_{\bf b[c} \varpi^{\bf b}{}_{\bf d]ij} + \kappa_{\bf i} \delta^{\bf ij} \varphi_{\bf cdj} + \frac{3}{2} V^{\bf ij}{}_{\bf cd} \Pi_{\bf ij} \,.
        \end{align*}
        Since $\acute{\vartheta}_{\bf abc} = 0$ by the previous steps, $\acute{\nabla}_{\bf a} \breve{\epsilon}_{\bf bc}{}^{\bf de} = 0$ and one can take the dual of the above equation to get a similar equation on $(\star\varphi)$.
        
        As a consequence, one can construct a system of partial differential equation for $\overline{\underline{\varphi}}$ of the form
        \begin{equation}
            \label{propagation_varphi}
            C^\mu(e_{\bf 1},e_{\bf 2}) \partial_\mu \overline{\underline{\varphi}} = \acute{\Gamma} \times \overline{\underline{\varphi}} + \kappa \times \overline{\underline{\varphi}} + \mathcal{F}(\varpi^{\bf a}{}_{\bf bij}, \Pi_{\bf ia}, \varrho_{\bf ija},\acute{\Sigma}_{\bf i}{}^{\bf a}{}_{\bf j}) \,,
        \end{equation}
        where
        \begin{itemize}
            \item[$\bullet$] $C^0(e_{\bf 1},e_{\bf 2})$ is an affine function of $e_{\bf 1}{}^0$ and $e_{\bf 2}{}^0$, with values into symmetric matrices and verifying $C^0(0,0) = I$,
            \item[$\bullet$] $C^j(e_{\bf 1},e_{\bf 2})$ are linear functions of $e_{\bf 1}{}^j$ and $e_{\bf 2}{}^j$ with values into symmetric matrices,
            \item[$\bullet$] $\mathcal{F}(\varpi^{\bf a}{}_{\bf bij}, \Pi_{\bf ia}, \varrho_{\bf ija},\acute{\Sigma}_{\bf i}{}^{\bf a}{}_{\bf j})$ denotes linear terms in the frame components $(\varpi^{\bf a}{}_{\bf bij},\Pi_{\bf ia},\varrho_{\bf ija},\acute{\Sigma}_{\bf i}{}^{\bf a}{}_{\bf j})$.
        \end{itemize}
        Furthermore, since $|(e_{\bf 1}{}^0,e_{\bf 2}{}^0)|_2 \leq |(e_{\bf 1}{}^0,e_{\bf 2}{}^0,e_{\bf 3}{}^0)|_2 \leq 1-2c$ on $\mathcal{V}$, one deduces that there exists $c'>0$ such that $C^0(e_{\bf 1},e_{\bf 2}) \geq c' I$.
        
        Consequently \eqref{propagation_varpi}-\eqref{propagation_Sigma} and \eqref{propagation_varphi} form a homogeneous linear symmetric hyperbolic system. Furthermore, in the case $\mathcal{B}_\star \neq \varnothing$, the boundary matrix vanishes identically since $e_{\bf A}{}^3 = 0$ on $\mathcal{V}\cap\mathcal{B}$ by \Cref{evol_solve_near}. It follows that there exists a unique solution on a potential subregion $\mathcal{W}$ of $\mathcal{V}$, neighbourhood of $\mathcal{U}_\star$ in $\mathcal{V}$, for a given initial data. Since the zero quantities $(\overline{\underline{\varphi}}, \varpi^{\bf a}{}_{\bf bij}, \Pi_{\bf ia}, \varrho_{\bf ija}, \acute{\Sigma}_{\bf i}{}^{\bf c}{}_{\bf j})$ vanish identically on $\mathcal{U}_\star$ by \eqref{eq:initial_conditions_constraints}, they vanish identically on $\mathcal{W}$. \qedhere
    \end{itemize}
\end{proof}

\begin{cor}
    \label{cor:geometry_back}
    With all the assumptions and notations of \Cref{lem:propagation_constraints}, one has
    \begin{itemize}
        \item the quintuple $(\mathcal{W},\breve{g},V,\Theta,\kappa) \in \mathscr{E}$ is a solution to the \eqref{eq:CVE} with $\Lambda=-3$,
        \item the conformal boundary $\mathscr{I}$ of $(\mathcal{W},\breve{g},V,\Theta,\kappa)$ is $\mathcal{W} \cap \mathcal{B}$,
        \item $(\Theta,\acute{\nabla},(e_{\bf a}),(y^\mu))$ is a gauge choice on $\mathcal{W}$ verifying \eqref{eq:gauge_core_properties} and \eqref{gauge_additional_properties},
        \item if $\mathscr{I}\neq\varnothing$, $(\mathcal{W},\mathscr{I},\widetilde{g}:=\Theta^{-2}\breve{g})$ is a 4-dimensional aAdS space with smooth rescaled metrics such that $(\mathcal{W}\setminus\mathscr{I},\widetilde{g})$ is solution to the \eqref{eq:VE} with $\Lambda=-3$,
        \item if $\mathscr{I}\neq\varnothing$, $e_{\bf 3}$ is its outward-pointing $\breve{g}$-unit normal.
    \end{itemize}
\end{cor}

\begin{rem}
    From now on, if $\mathscr{I}\neq\varnothing$ then one can use all the geometric notions and properties of aAdS spaces described in \Cref{sec:aAdS_spaces}.
\end{rem}

\begin{proof}
    Thanks to b) of \Cref{lem:propagation_constraints}, all the zero quantities vanish identically on $\mathcal{W}$. By application of $ii)$ of \Cref{lem:link_zero_quantities_CVE}, one deduces that $\acute{\nabla}$ is the Weyl connection $\widecheck{\nabla}$ associated to $\kappa$ with respect to $\breve{g}$, $U_{\alpha\beta}$ is its Schouten tensor $\widecheck{L}_{\alpha\beta}$, $\zeta_\alpha$ is the covector field $(\widecheck{\zeta}_{\breve{g}})_\alpha$, $s$ is the extended Friedrich scalar $\widecheck{s}_{\breve{g}}$ and $(\mathcal{W},\breve{g},V,\Theta,\kappa)$ is solution to the (CVE) with $\Lambda=-3$.

    Whether $\mathcal{B}_\star$ is empty or not, $\Theta$ is a boundary defining function of $\mathcal{V}\cap\mathcal{B}$ by \Cref{evol_solve_away} or \Cref{evol_solve_near}. Thus the conformal boundary of $(\mathcal{W},\breve{g},V,\Theta,\kappa)$ is $\mathcal{W}\cap\mathcal{B}$. Then $(\Theta,\acute{\nabla},(e_{\bf a}),(y^\mu))$ is a gauge choice on $\mathcal{W}$ verifying \eqref{eq:gauge_core_properties} and \eqref{gauge_additional_properties} by construction. Assume that $\mathscr{I}\neq\varnothing$. By application of $ii)$ of \Cref{prop:aAdS_VE}, $(\mathcal{W},\mathscr{I},\widetilde{g}:=\Theta^{-2}\breve{g})$ is a 4-dimensional aAdS space with smooth rescaled metrics such that $(\mathcal{W}\setminus\mathscr{I},\widetilde{g})$ is solution to the (VE) with $\Lambda=-3$. Finally, one deduces that $e_{\bf 3}$ is the outward-pointing $\breve{g}$-unit normal of $\mathscr{I}$ from \eqref{eq:solve_near_frame}, $e_{\bf 0}{}^3 = \partial_0 y^3 = 0$ and the fact that $(e_{\bf a})$ is a $\breve{g}$-orthonormal frame field.
\end{proof}

In order to obtain our geometric theorem, it remains to interpret geometrically the analytic boundary conditions as well as the analytic initial data and all the conditions imposed on it. This is the subject of the next two sections.

\subsection{Geometric boundary conditions}
\label{sec:geometric_bc}

In this section, we interpret geometrically a subset of the analytic boundary conditions \eqref{raw_bc} which are a priori gauge-dependent. Let us define more precisely what a geometric boundary condition is.

\begin{defi} \,
	\label{def:geometric_bc}
    \begin{itemize}
        \item A geometric boundary condition on a connected component $\mathfrak{S}$ of the conformal boundary $\mathscr{I}$ for the \eqref{eq:CVE} with $\Lambda=-3$ is a condition on the equivalence class $[(\mathfrak{S},\mathfrak{h},\mathfrak{E},\mathfrak{M},0,\kappa,\mathfrak{K},\mathfrak{z},\mathfrak{T},\mathfrak{Q})]$ of the 10-tuple induced by the (CVE) on $\mathfrak{S}$, for the equivalence relation on $\mathscr{D}$ generated by the gauge transformations of the (CVC) described in \Cref{def:transfo_dim3}.
        
        \item A geometric boundary condition on a connected component $\mathfrak{S}$ of the conformal boundary $\mathfrak{I}$ of a 4-dimensional aAdS space $(\M,\mathfrak{I},\widetilde{g})$, which is solution to the \eqref{eq:VE} with $\Lambda=-3$ on $\M\setminus\mathfrak{I}$, is a condition on the free data class $[(\mathfrak{S},\mathfrak{h},\mathfrak{t})]$, see \Cref{def:free_data_class}.
    \end{itemize}
\end{defi}

\begin{rem}
    To be invariant under confomorphisms $\Conf_\phi$, a condition must be tensorial and conformally invariant for some weight, see \Cref{sec:conformal_invariance}.
\end{rem}

\begin{lem}
    \label{lem:bijection_geometric_bc}
    The geometric boundary conditions for the (CVE) with $\Lambda=-3$ and for the (VE) with $\Lambda=-3$ in the sense of \Cref{def:geometric_bc} are in bijection.
\end{lem}

\begin{proof}
	Consider a geometric boundary condition for the (CVE) with $\Lambda = -3$ on a connected component $\mathfrak{S}$ of $\mathscr{I}$ in the sense of \Cref{def:geometric_bc}. Using the gauge transformations $\Weylchg_\chi^\perp$ and $\Weylchg_\omega^\parallel$, there exists a representative $(\mathfrak{S},\mathfrak{h},\mathfrak{E},\mathfrak{M},0,\kappa,\mathfrak{K},\mathfrak{z},\mathfrak{T},\mathfrak{Q})$ of the equivalence class such that $\kappa=0$ and $\widehat{\sigma} = 0$. Since this 10-tuple verifies equations \eqref{eq:constraints_I} of \Cref{cor:I}, one deduces that
	\[ \mathfrak{z} = \pm 1 \,, \quad \mathfrak{K}_{ij} = 0 \,, \quad \mathfrak{T}_i = 0 \,, \quad \mathfrak{Q}_{ij} = 0 \,, \quad \mathfrak{z} \mathfrak{M}_{ijk} = \mathfrak{c}_{ijk} \,, \] 
	where $\mathfrak{c}_{ijk}$ is the Cotton tensor of the Levi-Civita connection of $\mathfrak{h}$, and thus entirely determined by $\mathfrak{h}$. Furthermore, the field $\mathfrak{E}$ is the electric part of the rescaled Weyl tensor. By \eqref{eq:lim_weyl_frakt}, one has
    \[ \mathfrak{t}_{ij} = \frac{2\mathfrak{z}}{3}  \mathfrak{E}_{ij} \,. \]
    The invariance of the condition under sign change $\Signchg$ implies that it does not depend on the sign of $\mathfrak{z}$.  Therefore, the condition can be rewritten as one on $(\mathfrak{S},\mathfrak{h},\mathfrak{t})$. Finally, the invariance under the confomorphisms $\Conf_\phi$ transposes into the invariance under the equivalence relation defined in \Cref{def:equivalence_relation_pair}. Hence, one retrieves a condition on the free data class $[(\mathfrak{S},\mathfrak{h},\mathfrak{t})]$. Conversely, one can go in the other way.
\end{proof}

To achieve the goal of this section, we first select the maximal dissipative boundary conditions which can be written as tensorial boundary conditions on the distribution $\mathcal{D}^2 := \vectorspan(e_{\bf 1},e_{\bf 2})$ in \Cref{sec:tensorial_bc_distribution}. From \Cref{sec:constant_coeff} onwards, we restrict the study to boundary conditions with constant coefficients. We then transform these into conditions on the induced 10-tuple of the conformal boundary in \Cref{sec:replacing_tensors}. Finally, certain of these conditions can be interpreted geometrically, through an auxiliary system on the boundary, as either imposing the conformal class in \Cref{sec:dirichlet} (the geometric Dirichlet boundary conditions) or imposing our new boundary conditions \eqref{eq:bc_robin} in \Cref{sec:homogeneous_robin} (the geometric homogeneous Robin boundary conditions).

\subsubsection{Tensorial boundary conditions on a distribution}
\label{sec:tensorial_bc_distribution}

First, let us select the maximal dissipative boundary conditions which can be written as tensorial boundary conditions on the distribution $\mathcal{D}^2 := \vectorspan(e_{\bf 1},e_{\bf 2})$. This constitutes an initial criterion, given that a geometric boundary condition is tensorial on the conformal boundary and thus, a fortiori, tensorial on the distribution $\mathcal{D}^2$. This will be done by exploiting duality on $\mathcal{D}^2$, as presented in \Cref{sec:duality_dim_2}. 

Recall that the tensors $E_{\bf ij}$ and $H_{\bf ij}$, defined in \Cref{lem:propagation_constraints}, are symmetric and trace-free on the 3-dimensional distribution spanned by $e_{\bf 1}$, $e_{\bf 2}$ and $e_{\bf 3}$. In what follows, $E^\dagger_{\bf AB}$ and $H^\dagger_{\bf AB}$ will denote the trace-free part of the restriction of respectively $E_{\bf ij}$ and $H_{\bf ij}$ on $\mathcal{D}^2$. Observe that
\begin{align*}
    E_{\bf 12} = E_{\bf 21} = \frac{E_{\bf 12} + E_{\bf 21}}{2} &=  \frac{E^\dagger_{\bf 12} + E^\dagger_{\bf 21}}{2} = E^\dagger_{\bf 12} = E^\dagger_{\bf 21} \,, \\
    \frac{E_{\bf 11} - E_{\bf 22}}{2} &=  \frac{E^\dagger_{\bf 11} - E^\dagger_{\bf 22}}{2} = E^\dagger_{\bf 11} = - E^\dagger_{\bf 22} \,,
\end{align*}
and the same goes for $H_{\bf ij}$. Since the analytic boundary conditions \eqref{raw_bc} are affine, it is natural to search for a tensorial boundary condition on $\mathcal{D}^2$ which is affine in $(E^\dagger_{\bf AB},H^\dagger_{\bf AB})$, that is of the form
\begin{equation}
    \label{eq:ansatz_bc_tensorial}
    L_{\bf AB}{}^{\bf CD} E^\dagger_{\bf CD} + L'_{\bf AB}{}^{\bf CD} H^\dagger_{\bf CD} = q_{\bf AB} \,,
\end{equation}
where $L,L' \in \mathscr{B}(\mathcal{D}^2)$ and $q \in \mathscr{S}(\mathcal{D}^2)$, see \Cref{sec:duality_dim_2}. In order for \eqref{eq:ansatz_bc_tensorial} to impose two independent conditions such as \eqref{raw_bc}, it is necessary to require that
\begin{equation}
    \label{eq:cond_span}
    \im L \oplus \im L' = \mathscr{S}(\mathcal{D}^2) \,.
\end{equation}
Let us emphasise that the triple $(L,L',q)$ in \eqref{eq:ansatz_bc_tensorial} is defined up to the action of automorphisms of $\mathscr{S}(\mathcal{D}^2)$: if $L''\in \mathscr{B}(\mathcal{D}^2)$ is everywhere invertible then the triple $(L''\circ L,L''\circ L',L'' q)$ yields an equivalent condition. From \Cref{sec:duality_dim_2}, one can decompose $L$ and $L'$ in the basis $(\Id,\star,\sigma,\star\sigma)$:
\begin{subequations}
    \label{eq:decomp_LLprime}
    \begin{align}
    L &= \alpha \Id + \beta \star + \gamma \sigma + \delta (\star\sigma) \,, \\
    L' &= \alpha' \Id + \beta' \star + \gamma' \sigma + \delta' (\star\sigma) \,,
    \end{align}
\end{subequations}
where $\alpha,\beta,\gamma,\delta,\alpha',\beta',\gamma',\delta' \in \mathcal{C}^\infty(\mathcal{B},\R)$. By \Cref{lem:invertibility_endo}, the condition \eqref{eq:cond_span} translates into
\[ \left. 
    \begin{array}{c}
    \alpha^2+\beta^2 = \delta^2+\gamma^2 \,, \\
    (\alpha')^2+(\beta')^2 = (\delta')^2+(\gamma')^2 \,, 
    \end{array} \right\} \implies (\delta'+\beta')(\alpha+\gamma) \neq (\delta+\beta)(\alpha'+\gamma') \,. \]

\begin{lem}
    \label{lem:analytic_to_tensorial}
    The analytic boundary conditions \eqref{raw_bc} can be rewritten under the form \eqref{eq:ansatz_bc_tensorial} with $L,L'$ defined by \eqref{eq:decomp_LLprime} with
    \begin{subequations}
        \begin{alignat*}{7}
            \alpha &= -\frac{B_{21}+B_{12}}{2} \,, &\qquad \beta &= \frac{B_{11}-B_{22}}{2}-1 \,, &\qquad \gamma &= \frac{B_{21}-B_{12}}{2} \,, &\qquad \delta &= -\frac{B_{11}+B_{22}}{2} \,, \\
            \alpha' &= 1 + \frac{B_{11}-B_{22}}{2} \,, &\qquad \beta' &= \frac{B_{21}+B_{12}}{2} \,, &\; \gamma' &= \frac{B_{11}+B_{22}}{2} \,, &\qquad \delta' &= \frac{B_{21}-B_{12}}{2} \,,
        \end{alignat*}
    \end{subequations}
    and $q$ defined by
    \[ q_{\bf 12} = d_1 \,, \qquad q_{\bf 11} = d_2 \,. \]
    Furthermore, \eqref{eq:cond_span} holds.
\end{lem}

\begin{proof}
    We want the analytic boundary conditions \eqref{raw_bc} to be equivalent to $Z_{\bf AB} = 0$, where $Z \in \mathscr{S}(\mathcal{D}^2)$ is given by
    \[ Z_{\bf AB} := L_{\bf AB}{}^{\bf CD} E^\dagger_{\bf CD} + L'_{\bf AB}{}^{\bf CD} H^\dagger_{\bf CD} - q_{\bf AB} \,, \]
    for some $L,L' \in \mathscr{B}(\mathcal{D}^2)$ and $q \in \mathscr{S}(\mathcal{D}^2)$ to be determined. Note that $Z$ has two algebraic degrees of freedom, take for instance $Z_{\bf 12}$ and $Z_{\bf 11}$. Therefore equations \eqref{raw_bc} must be related to the equations $Z_{\bf 12}=0$ and $Z_{\bf 11}=0$ by an invertible transformation. More explicitly, there must exist $\lambda,\mu,\nu,\xi \in \mathcal{C}^\infty(\mathcal{B},\R)$ with $\lambda\xi-\mu\nu$ nowhere vanishing such that \eqref{raw_bc_1} is $\lambda Z_{\bf 12} + \mu Z_{\bf 11} = 0$ and that \eqref{raw_bc_2} is $\nu Z_{\bf 12} + \xi Z_{\bf 11} = 0$.

    One can always assume that $\lambda=1$, $\mu=0$, $\nu=0$ and $\xi=1$ up to taking a triple equivalent to $(L,L',q)$. Indeed, define $L''=\alpha''\Id+\beta''\star+\gamma''\sigma+\delta''(\star\sigma) \in \mathscr{B}(\mathcal{D}^2)$ for the choices
    \[ \alpha'' = \frac{\lambda+\xi}{2} \,, \quad \beta'' = \frac{\nu-\mu}{2} \,, \quad \gamma'' = \frac{\lambda-\xi}{2} \,, \quad \delta'' = \frac{\nu+\mu}{2} \,. \]
    Then the function $(\alpha'')^2 + (\beta'')^2 - (\gamma'')^2 - (\delta'')^2 = 2(\lambda\xi-\mu\nu)$ nowhere vanishes. By \Cref{lem:invertibility_endo}, $L''$ is thus everywhere invertible. Taking the action of $L''$ on the triple $(L,L',q)$ yields the desired simplification.

    By identification of $Z_{\bf 12}=0$ to $\eqref{raw_bc_1}$, one deduces that
    \[ \alpha+\gamma = -B_{12} \,, \;\; \delta-\beta = 1-B_{11} \,, \;\; \alpha'+\gamma' = 1+B_{11} \,, \;\; \delta'-\beta' = -B_{12} \,, \;\; q_{\bf 12} = d_1 \,. \]
    Similarly, identifying $Z_{\bf 11} = 0$ to \eqref{raw_bc_2} gives that
    \[ \alpha-\gamma = -B_{21} \,, \;\; \delta+\beta = -(1+B_{22}) \,, \;\; \alpha'-\gamma' = 1-B_{22} \,, \;\; \delta'+\beta' = B_{21} \,, \;\; q_{\bf 11} = d_2 \,. \]
    Hence the result. Now let us prove that \eqref{eq:cond_span} holds. By contradiction, assume that it is not verified. Then there exists a point $p \in \mathcal{B}$ such that
    \begin{align*}
        0 &= \alpha^2+\beta^2-\gamma^2-\delta^2 = 1-\det B + (B_{22}-B_{11}) \quad \text{at } p \,, \\
        0 &= (\alpha')^2+(\beta')^2-(\gamma')^2-(\delta')^2 = 1-\det B - (B_{22}-B_{11}) \quad \text{at } p \,, \\
        0 &= (\delta'+\beta')(\alpha+\gamma) - (\delta+\beta)(\alpha'+\gamma') = 1 + \det B + B_{22} + B_{11} \quad \text{at } p \,.
    \end{align*}
    One deduces that $B_{11} = B_{22} = -1$ and $B_{12}B_{21} = 0$ at $p$. However, this violates the condition $B^\intercal B \leq I$ at $p$.
\end{proof}

\subsubsection{Constant coefficients boundary conditions}
\label{sec:constant_coeff}

In the rest of \Cref{sec:geometric_bc}, we restrict to the case of constant coefficients, as was also the approach adopted by Friedrich \cite{F95}. This means that only constant matrices $B \in M_2(\R)$ are considered. Since $\sigma,\star\sigma \in \mathscr{B}(\mathcal{D}^2)$ depend on the choice of basis $(e_{\bf 1},e_{\bf 2})$ on $\mathcal{D}^2$ (see \Cref{sec:duality_dim_2}) while $L,L'$ must be independent of such choice, this implies that
\[ \gamma = \delta = \gamma' = \delta' = 0 \,. \]
By \Cref{lem:analytic_to_tensorial}, this is equivalent to impose that the matrix $B$ is symmetric and trace-free. Then, the tensorial boundary conditions on $\mathcal{D}^2$ reduce to a simple form given in the following lemma. 

\begin{lem}
    The analytic boundary conditions \eqref{raw_bc} with constant coefficients rewrite under the form
    \begin{align}
        \label{nongeometric_bc}
        -\rho \sin\chi E^\dagger_{\bf AB} - (1-\rho\cos\chi) (\star E^\dagger)_{\bf AB} & \nonumber \\ + (1+\rho \cos\chi) H^\dagger_{\bf AB} + \rho\sin\chi (\star H^\dagger)_{\bf AB} &= q_{\bf AB} \qquad \text{with } \rho \in [0,1], \chi \in [0,2\pi) \,,
    \end{align}
    where $q \in \mathscr{S}(\mathcal{D}^2)$ is freely specifiable (boundary data).
\end{lem}

\begin{rem}
    The boundary conditions \eqref{nongeometric_bc} can be written more compactly as follows
    \begin{equation}
        F_{\bf AB} + \rho \left(\cos\chi \, G_{\bf AB} + \sin \chi \, (\star G)_{\bf AB}\right) = q_{\bf AB} \,,
    \end{equation}
    where $F_{\bf AB} := H^\dagger_{\bf AB} - (\star E^\dagger)_{\bf AB}$ and $G_{\bf AB} := H^\dagger_{\bf AB} + (\star E^\dagger)_{\bf AB}$. For homogeneous boundary data, that is $q_{\bf AB} = 0$, one can check that the boundary term in the a priori energy estimate for the evolution system \eqref{evol_system} is non-positive. Indeed, using the boundary conditions and \Cref{lem:identities_dual_2}, one has
    \begin{align*}
        F_{\bf AB} F^{\bf AB} &= \rho^2 G_{\bf AB} G^{\bf AB} \,, \\
        G_{\bf AB} G^{\bf AB} &= H^\dagger_{\bf AB} H^\dagger{}^{\bf AB} + E^\dagger_{\bf AB} E^\dagger{}^{\bf AB} + 2 H^\dagger_{\bf AB} (\star E^\dagger)^{\bf AB} \,, \\
        H^\dagger_{\bf AB} (\star E^\dagger)^{\bf AB} &= \frac{1}{4} (G+F)_{\bf AB} (G-F)^{\bf AB} = \frac{1}{4} \left(G_{\bf AB} G^{\bf AB} - F_{\bf AB} F^{\bf AB} \right) \,.
    \end{align*}
    Hence
    \begin{align}
        \label{eq:energy_flux_boundary}            E^{\bf AB} \slashed{\epsilon}_{\bf (A}{}^{\bf C} H_{\bf B)C} &=  E^\dagger{}^{\bf AB} (\star H^\dagger)_{\bf AB} = -(\star E^\dagger)_{\bf AB} H^\dagger{}^{\bf AB} = - \frac{1-\rho^2}{2(1+\rho^2)} \left( H^\dagger_{\bf AB} H^\dagger{}^{\bf AB} + E^\dagger_{\bf AB} E^\dagger{}^{\bf AB} \right) \leq 0 \,.
    \end{align}
\end{rem}

\begin{proof}
    Since $B \in M_2(\R)$ is symmetric and trace-free, one has
    \[ B^\intercal B\leq I \iff (B_{11})^2 + (B_{12})^2 \leq 1 \,. \]
    Let us write $B_{11} = \rho \cos\chi$, $B_{12} = \rho \sin\chi$ with $\rho\in [0,1]$ and $\chi \in [0,2\pi)$. Then, the result follows from \Cref{lem:analytic_to_tensorial}.
\end{proof}

Equation \eqref{eq:energy_flux_boundary}, which giving the energy flux at the boundary, clearly shows that the cases with $\rho = 1$ correspond to reflective boundary conditions, while the cases with $\rho < 1$ correspond to dissipative boundary conditions. In particular, the case $\rho = 0$ maximises the loss of energy at the boundary.

\begin{lem}
    The boundary conditions \eqref{nongeometric_bc} split into the two following cases
    \begin{itemize}
        \item[i)] if $\rho \in [0,1)$ (dissipative boundary conditions)
        \begin{equation}
            \label{bc_dissipative}
            \sigma E^\dagger_{\bf AB} + (1-\sigma) (\star H^\dagger)_{\bf AB} + \lambda H^\dagger_{\bf AB} = q'_{\bf AB} \,, \quad \text{with } \sigma \in (0,1), \, \lambda \in \R \,.
        \end{equation}
        One has
        \[ \rho(\sigma,\lambda) = \sqrt{\frac{\lambda^2 + (1-2\sigma)^2}{\lambda^2+1}} \,. \]
        \item[ii)] if $\rho = 1$ (reflective boundary conditions)
        \begin{equation}
            \label{bc_reflective}
            \cos\theta \, E^\dagger_{\bf AB} + \sin\theta \, H^\dagger_{\bf AB} = q'_{\bf AB} \,, \quad \text{with } \theta \in \left[-\frac{\pi}{2},\frac{\pi}{2}\right) \,.
        \end{equation}
    \end{itemize}
    In the two cases, $q' \in \mathscr{S}(\mathcal{D}^2)$ is freely specifiable (boundary data).
\end{lem}

\begin{rems} \,
    \begin{itemize}
        \item Some of these boundary conditions are analogous to ones used in electromagnetism:
        \begin{itemize}
            \item[i)] The dissipative boundary conditions \eqref{bc_dissipative} include notably
            \begin{equation}
                \label{bc_Leontovich}
                E^\dagger_{\bf AB} = - \frac{1-\sigma}{\sigma} (\star H^\dagger)_{\bf AB} \,, \quad \text{with } \sigma \in (0,1) \,.
            \end{equation}
            These are analogous to the Leontovich boundary conditions for the Maxwell equations, see for instance \cite[§67]{LL60}. The surface impedance would be here equal to $(1-\sigma)/\sigma \in \R_+^\star$. The subcase $\sigma=1/2$, equivalent to $\rho=0$, has been studied for the Maxwell equations and the linearised Bianchi equations on the AdS space in \cite{HLSW20}.
            
            \item[ii)] The reflective boundary conditions \eqref{bc_reflective} include in particular
            \begin{equation}
                \label{bc_perfect_conductor}
                H^\dagger_{\bf AB} + \cotan \theta \, E^\dagger_{\bf AB} = 0 \,, \quad \text{with } \theta \in \left[-\frac{\pi}{2},\frac{\pi}{2}\right) \,.
            \end{equation}
            These are analogous to the perfect electromagnetic conductor boundary conditions for the Maxwell equations, see \cite{LS05}. The surface admittance would be here equal to $\cotan \theta \in [-\infty,+\infty]$.
        \end{itemize}

        \item The energy flux at the boundary given by \eqref{eq:energy_flux_boundary} rewrites in the case of the dissipative boundary conditions \eqref{bc_dissipative} as follows
        \[ E^\dagger{}^{\bf AB} (\star H^\dagger)_{\bf AB} = - \frac{\sigma(1-\sigma)}{\lambda^2 + \sigma^2+(1-\sigma)^2} \left( E^\dagger{}^{\bf AB} E^\dagger_{\bf AB} +  H^\dagger{}^{\bf AB} H^\dagger_{\bf AB} \right) \,. \qedhere \]
    \end{itemize}
\end{rems}

\begin{proof} \,
    \begin{itemize}
        \item[i)] Assume that $\rho \in [0,1)$. Then $1-\rho\cos\chi \geq 1-\rho > 0$. Applying the operator $\star$ on \eqref{nongeometric_bc} gives
        \[ -\rho \sin\chi (\star E^\dagger)_{\bf AB} + (1-\rho\cos\chi) E^\dagger_{\bf AB} + (1+\rho \cos\chi) (\star H^\dagger)_{\bf AB} - \rho\sin\chi H^\dagger_{\bf AB} = (\star q)_{\bf AB} \,. \]
        By \Cref{cor:duality_dim_2}, equation \eqref{nongeometric_bc} is equivalent to any non-trivial linear combination of itself and its dual. For the choice $u = -\rho\sin\chi$ and $v=1-\rho\cos\chi \neq 0$, one deduces that
        \[ \left( (1-\rho\cos\chi)^2+ \rho^2\sin^2\chi\right) E^\dagger_{\bf AB} + (1-\rho^2) (\star H^\dagger)_{\bf AB} - 2\rho\sin\chi H^\dagger_{\bf AB} = q''_{\bf AB}  \,, \]
        where $q''_{\bf AB} := -\rho\sin\chi \, q_{\bf AB} + (1-\rho\cos\chi)(\star q)_{\bf AB}$. One obtains \eqref{bc_dissipative} by dividing the above equality by $2(1-\rho\cos\chi)$ and by defining
        \[ \sigma := \frac{1+\rho^2-2\rho\cos\chi}{2(1-\rho\cos\chi)} \in (0,1) \,, \qquad \lambda := \frac{-\rho\sin\chi}{1-\rho\cos\chi} \in \R \,. \]
        Note that the map $(\rho,\chi) \in (0,1)\times[0,2\pi)\mapsto (\sigma,\lambda) \in \left( (0,1)\times\R \right) \setminus \left\{ (1/2,0) \right\}$ is invertible since
        \[ \rho = \sqrt{\frac{\lambda^2 + (1-2\sigma)^2}{\lambda^2+1}} \,, \qquad \rho \cos\chi = 1 - \frac{2\sigma}{1+\lambda^2}\,, \qquad \rho \sin\chi = \frac{-2\lambda\sigma}{1+\lambda^2} \,. \]
        Furthermore, the case $\rho=0$ corresponds to $(\sigma,\lambda) = (1/2,0)$.
        
        \item[ii)] Assume that $\rho = 1$ and define
        \[ \theta := 
        \frac{\chi-\pi}{2} \in \left[-\frac{\pi}{2},\frac{\pi}{2}\right) \,. \]
        Then
        \[ \sin\chi = -2\cos\theta\sin\theta \,, \qquad \cos\chi = 1-2\cos^2\theta=2\sin^2\theta-1 \,, \]
        and the boundary conditions \eqref{nongeometric_bc} rewrite under the form
        \[ 2\sin\theta \left( \cos\theta E^\dagger_{\bf AB} + \sin\theta H^\dagger_{\bf AB}\right) - 2 \cos\theta \left( \cos\theta (\star E^\dagger)_{\bf AB} + \sin\theta (\star H^\dagger)_{\bf AB}\right) = q_{\bf AB} \,. \]
        By \Cref{lem:decomp_duality_dim_2}, there exists a unique $q'_{\bf AB} \in \mathscr{S}(\mathcal{D}^2)$ such that
        \[ q_{\bf AB} = 2\sin\theta \, q'_{\bf AB} -2\cos\theta \, (\star q')_{\bf AB} \,. \]
        By \Cref{cor:duality_dim_2}, it follows that the above equation is equivalent to \eqref{bc_reflective}. \qedhere
    \end{itemize} 
\end{proof}

\subsubsection{Replacing the tensors}
\label{sec:replacing_tensors}

So far, the boundary conditions have been written in terms of the electromagnetic decomposition $(E,H)$ of the Weyl candidate $V$ in the frame field $(e_{\bf 0},(e_{\bf 1},e_{\bf 2},e_{\bf 3}))$, which is adapted to the initial hypersurface $\mathcal{U}_\star$ in $\mathcal{W}$, see \Cref{lem:decomposition_Weyl_cand}. However, this hinders the geometric interpretation of these boundary conditions since $E$ and $H$ are not tensors on the conformal boundary $\mathscr{I} = \mathcal{B} \cap \mathcal{W}$. The idea is to express them as a condition on the 10-uplet $(\mathfrak{S},\mathfrak{h},\mathfrak{E},\mathfrak{M},0,\kappa,\mathfrak{K},\mathfrak{z},\mathfrak{T},\mathfrak{Q})$ induced on $\mathscr{I}$.

\begin{lem}
    The boundary conditions \eqref{bc_dissipative} and \eqref{bc_reflective} rewrite respectively as the following conditions on the 10-tuple $(\mathfrak{S},\mathfrak{h},\mathfrak{E},\mathfrak{M},0,\kappa,\mathfrak{K},\mathfrak{z},\mathfrak{T},\mathfrak{Q}) \in \mathscr{D}$, induced by the (CVE) on $\mathscr{I}$, 
    \begin{equation}
        \label{bc_VCE_dissipative}
        -\sigma \mathfrak{E}^\dagger_{\bf AB} - (1-\sigma) (\star \mathfrak{H}^\dagger)_{\bf AB} - \lambda \mathfrak{H}^\dagger_{\bf AB} = q'_{\bf AB} \,, \quad \text{with } \sigma \in (0,1), \, \lambda \in \R \,.
    \end{equation}
    and
    \begin{equation}
        \label{bc_VCE_reflective}
        -\cos\theta \, \mathfrak{E}^\dagger_{\bf AB} - \sin\theta \, \mathfrak{H}^\dagger_{\bf AB} = q'_{\bf AB} \,, \quad \text{with } \theta \in \left[-\frac{\pi}{2},\frac{\pi}{2}\right) \,.
    \end{equation}
    where
    \begin{itemize}
        \item $\mathfrak{E}^\dagger_{\bf AB}$ and $\mathfrak{H}^\dagger_{\bf AB}$ are the trace-free part of the restriction on $\mathcal{D}^2$ of $\mathfrak{E}_{ij}$ and $\mathfrak{H}_{ij}$,
        \item $q' \in \mathscr{S}(\mathcal{D}^2)$ is freely specifiable (boundary data).
    \end{itemize}
\end{lem}

\begin{proof}
    By \Cref{cor:geometry_back}, $e_{\bf 3}$ is the outward-pointing $\breve{g}$-unit normal of $\mathscr{I}$. Hence, the frame field $(e_{\bf 3},(e_{\bf 0},e_{\bf 1},e_{\bf 2}))$ is adapted to $\mathscr{I}$. Let $(\mathfrak{S},\mathfrak{h},\mathfrak{E},\mathfrak{M},0,\kappa,\mathfrak{K},\mathfrak{z},\mathfrak{T},\mathfrak{Q}) \in \mathscr{D}$ be the 10-tuple induced on $\mathscr{I}$. By the electromagnetic decomposition \eqref{dec_Weyl_candidate}, one has
    \[ E_{\bf AB} := V^{\bf 0}{}_{\bf A0B}= -2\left(\delta^{\bf 0}{}_{\bf [0} \mathfrak{E}_{\bf B]A} - \eta_{\bf A[0} \mathfrak{E}_{\bf B]}{}^{\bf 0} \right) = - \mathfrak{E}_{\bf AB} + \delta_{\bf AB} \mathfrak{E}_{\bf C}{}^{\bf C} \,. \]
    By duality and thanks to \eqref{sym_dual_EH}, one deduces that
    \[ H_{\bf AB} := (\star V)^{\bf 0}{}_{\bf A0B} = - \mathfrak{H}_{\bf AB} + \delta_{\bf AB} \mathfrak{H}_{\bf C}{}^{\bf C} \,. \]
    Hence
    \[ E^\dagger_{\bf AB} = - \mathfrak{E}^\dagger_{\bf AB} \,, \qquad H^\dagger_{\bf AB} = - \mathfrak{H}^\dagger_{\bf AB} \,. \qedhere \]
\end{proof}

It is now possible to transpose the above analytic boundary conditions for the (CVE) into corresponding analytic boundary conditions for the (VE).

\begin{cor}
    The analytic boundary conditions \eqref{bc_VCE_dissipative} and \eqref{bc_VCE_reflective} for the (CVE) with $\Lambda=-3$ transpose into the following analytic boundary conditions for the (VE) with $\Lambda=-3$
    \begin{equation}
        \label{bc_VE_dissipative}
        \sigma' \, \mathfrak{t}^\dagger_{\bf AB} + (1-\sigma') \mathfrak{c}_{\bf 0(AB)} + \lambda' \, \mathfrak{y}^\dagger_{\bf AB} = q''_{\bf AB} \,, \quad \text{with } \sigma' \in (0,1), \lambda' \in \R \,,
    \end{equation}
    and
    \begin{equation}
        \label{bc_VE_reflective}
        \cos(\theta') \, \mathfrak{t}^\dagger_{\bf AB} - \sin(\theta') \, \mathfrak{y}^\dagger_{\bf AB} = q''_{\bf AB} \,, \quad \text{with } \theta' \in \left(-\frac{\pi}{2},\frac{\pi}{2}\right] \,.
    \end{equation}
    where
    \begin{itemize}
        \item $\mathfrak{t}^\dagger_{\bf AB}$ is the trace-free part of the restriction on $\mathcal{D}^2$ of the boundary stress-energy tensor $\mathfrak{t}_{ij}$,
        \item $\mathfrak{y}^\dagger_{\bf AB}$ is the trace-free part of the restriction on $\mathcal{D}^2$ of the Cotton-York tensor $\mathfrak{y}_{ij}$ of the conformal boundary,
        \item $\mathfrak{c}_{ijk}$ is the Cotton tensor of the conformal boundary (which is the dual tensor of the Cotton-York tensor $\mathfrak{y}_{ij}$),
        \item $q'' \in \mathscr{S}(\mathcal{D}^2)$ is freely specifiable (boundary data).
    \end{itemize}
\end{cor}

\begin{proof}
    By \eqref{eq:lim_weyl_frakt} and \eqref{cor:I}, one has
    \[ \mathfrak{t}_{ij} = \frac{2\mathfrak{z}}{3} \mathfrak{E}_{ij} \,, \qquad \mathfrak{z} = -1 \,, \qquad \mathfrak{c}_{ijk} = \mathfrak{z}  \mathfrak{M}_{ijk} \; (\text{or equivalently } \mathfrak{y}_{ij} = \mathfrak{z} \mathfrak{H}_{ij}) \,. \]
    Consequently
    \[ \mathfrak{E}^\dagger_{\bf AB} = -\frac{3}{2} \mathfrak{t}^\dagger_{\bf AB} \,, \qquad  \mathfrak{H}^\dagger_{\bf AB} = -\mathfrak{y}^\dagger_{\bf AB} \,, \]
    \[ (\star \mathfrak{H}^\dagger)_{\bf AB} = - (\star \mathfrak{y}^\dagger)_{\bf AB} = - \slashed{\epsilon}_{\bf (A}{}^{\bf C} \mathfrak{y}_{\bf B)C} = \epsilon_{\bf 0(A}{}^{\bf C} \mathfrak{y}_{\bf B)C} = -\mathfrak{c}_{\bf 0(AB)} \,. \]
    Hence the results with
    \[ \sigma' := \frac{\frac{3\sigma}{2}}{\frac{3\sigma}{2}+(1-\sigma)} \,, \qquad \lambda' := \frac{\lambda}{\frac{3\sigma}{2}+(1-\sigma)} \,, \qquad q''_{\bf AB} = \frac{1}{\frac{3\sigma}{2}+(1-\sigma)} q'_{\bf AB} \,, \]
    or
    \[ \cos(\theta') = \frac{\frac{3\cos\theta}{2}}{\sqrt{\frac{9\cos^2\theta}{4} + \sin^2\theta}} \,, \qquad \sin(\theta') = \frac{-\sin\theta}{\sqrt{\frac{9\cos^2\theta}{4} + \sin^2\theta}} \,, \]
    \[ q''_{\bf AB} = \frac{1}{\sqrt{\frac{9\cos^2\theta}{4} + \sin^2\theta}} q'_{\bf AB} \,. \qedhere \]
\end{proof}

The boundary conditions \eqref{bc_VCE_dissipative} and \eqref{bc_VCE_reflective} for the (CVE) are still not geometric in the sense of \Cref{def:geometric_bc}. They are left invariant under the tangent and normal Weyl connection changes $\Weylchg^\parallel_\omega$, $\Weylchg^\perp_\chi$ and under sign change $\Signchg$ (up to changing the sign of the boundary data) but only partially under confomorphisms $\Conf_\phi$. Indeed, even though they are invariant under conformal rescalings (up to rescaling the boundary data), they lack the invariance under isometries since they depend on the specific frame field $(e_{\bf a})$ constructed in the gauge. For the same reason, the boundary conditions \eqref{bc_VE_dissipative} and \eqref{bc_VE_reflective} for the (VE) are not geometric either.

In the next two sections, we infer geometric boundary conditions for the (VE)\footnote{The proof also gives geometric boundary conditions for the (CVE) from the analytic reflective boundary conditions \eqref{bc_VCE_reflective}, in agreement with \Cref{lem:bijection_geometric_bc}.} from the analytic reflective boundary conditions \eqref{bc_VE_reflective} by means of auxiliary systems on the conformal boundary $\mathscr{I}$.

\subsubsection{The geometric Dirichlet boundary condition}
\label{sec:dirichlet}

Consider first the case $\theta'=\pi/2$ in the reflective boundary condition \eqref{bc_VE_reflective}, that is
\begin{equation}
    \label{bc_magnetic}
    \mathfrak{y}^\dagger_{\bf AB} = -q''_{\bf AB} \,.
\end{equation}
Their corresponding analytic form is given by \eqref{raw_bc} with
\begin{equation}
    \label{eq:dirichlet_analytic}
    B = I \in M_2(\R) \,, \qquad d_1,d_2 \in \mathcal{C}^\infty(\mathcal{B},\R) \,.
\end{equation}
This is the case treated by Friedrich who proved that it amounts to imposing the conformal class on the conformal boundary, should the initial data on the corner be known, as stated in the following lemma.

\begin{prop}[Friedrich \protect{\cite[Lemma 7.1]{F95}}]
    \label{prop:partial_to_total_dirichlet_bc}
    Consider the analytic boundary conditions \eqref{raw_bc} with the choices \eqref{eq:dirichlet_analytic}.
    \begin{itemize}
        \item[i)] Let $\underline{u}_\star \in \mathcal{C}^\infty(\overline{\mathcal{U}_\star},\R^m)$ be some analytic initial data verifying all the conditions of \Cref{lem:propagation_constraints} except the analytic compatibility conditions, see $iii)$ of \Cref{evol_solve_near}. The knowledge of both the restriction on $\mathcal{B}_\star$ of $\underline{u}_\star$ and of a Lorentzian conformal class $[\mathfrak{h}]$ on a neighbourhood $\mathfrak{U}$ of $\mathcal{B}_\star$ in $\mathcal{B}$ is sufficient to compute the analytic boundary data $d_1$, $d_2$ on a neighbourhood of $\mathcal{B}_\star$ in $\mathfrak{U}$.
        
        \item[ii)] Conversely, let $d_1$, $d_2 \in \mathcal{C}^\infty(\mathfrak{U},\R)$ be some analytic boundary data on a neighbourhood $\mathfrak{U}$ of $\mathcal{B}_\star$ in $\mathcal{B}$ and let $\underline{u}_\star \in \mathcal{C}^\infty(\overline{\mathcal{U}_\star},\R^m)$ be some analytic initial data verifying all the conditions of \Cref{lem:propagation_constraints}. The knowledge of the restriction on $\mathcal{B}_\star$ of $\underline{u}_\star$ and of the analytic boundary data is sufficient to compute, on a neighbourhood of $\mathcal{B}_\star$ in $\mathfrak{U}$, the metric $\breve{\mathfrak{h}}$ induced on $\mathscr{I}=\mathcal{W}\cap\mathcal{B}\subset\mathfrak{U}$ by the metric $\breve{g}$ given by \Cref{cor:geometry_back}.
    \end{itemize}
\end{prop}

\begin{rem}
    Even if the boundary conditions \eqref{bc_magnetic} are homogeneous, the Cotton tensor of the conformal class $[\breve{\mathfrak{h}}]$ is not necessarily vanishing identically. Indeed, it vanishes identically if and only if, in addition, the initial data $\underline{u}_\star$ is such that $\breve{\mathfrak{y}}_{\bf 00}$ and $\breve{\mathfrak{y}}_{\bf 0A}$ vanish on $\mathcal{B}_\star$, or equivalently such that $H_{\bf 33}$ and $E_{\bf 3A}$ vanish on $\mathcal{B}_\star$. In that case, the conformal class is locally conformally flat by $ii)$ of \Cref{Weyl_Schouten_theorem}. Since this is the minimal geometric condition ensuring that \eqref{bc_magnetic} is homogeneous, imposing a conformal class which is locally conformally flat can naturally be interpreted as the geometric homogeneous Dirichlet boundary condition.
\end{rem}

\begin{proof} We provide details of the proof, first presented in \cite[Section 7.1]{F95}, so that this paper is self-contained.
    \begin{itemize}
        \item[i)] The key idea is that the frame field $(e_{\bf i}) = (e_{\bf 0},e_{\bf 1},e_{\bf 2})$ on $\mathscr{I} = \mathcal{W}\cap\mathcal{B}$, a priori defined from the solution $\underline{u}$ of the evolution system, does actually not depend on solving the whole evolution system (which would require a boundary data) and can be reconstructed solely with the hypotheses of $i)$.

        Let $\mathfrak{h}$ be a representative of the conformal class $[\mathfrak{h}]$ on $\mathfrak{U}$ such that, on $\mathcal{B}_\star$, the frame field $(e_{\bf i}) = (\partial_0, e_{\bf 1}{}^A \partial_A, e_{\bf 2}{}^A \partial_A)$, where $e_{\bf 1}{}^A$ and $e_{\bf 2}{}^A$ are given by the restriction on $\mathcal{B}_\star$ of the initial data $\underline{u}_\star$, is $\mathfrak{h}$-orthonormal.
        
        For each point $q \in \mathcal{B}_\star$, start a $\mathfrak{D}$-conformal geodesic $(x,\beta)$ with initial data $x_\star = q$, $\dot{x}_\star = \partial_0$ and $\beta_\star = \kappa$, where $\kappa$ is given by the restriction on $\mathcal{B}_\star$ of $\underline{u}_\star$. This forms a congruence of timelike conformal geodesics on a neighbourhood of $\mathcal{B}_\star$ in $\mathfrak{U}$. Let $\widecheck{\mathfrak{D}}$ be the Weyl connection associated to $\beta$ with respect to $\mathfrak{h}$. Set $e_{\bf 0} := \dot{x}$ and extend the vector fields $e_{\bf 1}$ and $e_{\bf 2}$ from $\mathcal{B}_\star$ by parallel propagation with respect to $\widecheck{\mathfrak{D}}$.
        
        Now that the frame field $(e_{\bf i})$ is reconstructed, one can compute the Cotton-York tensor $\mathfrak{y}_{ij}$ of $\mathfrak{h}$, see \Cref{def:cotton_york}, and eventually derive $\mathfrak{y}^\dagger_{\bf AB}$, giving the value of $q''_{\bf AB}$ and thus of the analytic boundary data.
        
        \item[ii)] The proof of $ii)$ is much more involved since it relies on an auxiliary geometric problem on the conformal boundary constituted by certain structure equations. If $(\mathscr{I},\mathfrak{h})$ is a 3-dimensional oriented smooth Lorentzian manifold and $\kappa$ is a smooth covector field on $\mathscr{I}$ then
        \begin{subequations}
            \begin{align}
                \label{eq:aux_syst_a}
                \widehat{\mathfrak{r}}^i{}_{jkl} &= 2 S_{j[k}{}^{im} \widehat{\mathfrak{l}}_{l]m} \,, \\
                \label{eq:aux_syst_b}
                2\widehat{\mathfrak{D}}_{[i} \widehat{\mathfrak{l}}_{j]k} &= - \epsilon_{ij}{}^l
                \mathfrak{y}_{lk} \,, \\
                \label{eq:aux_syst_c}
                \widehat{\mathfrak{D}}_i \mathfrak{y}^i{}_j &= 3 \kappa_i \mathfrak{y}^i{}_j \,,
            \end{align}
        \end{subequations}
        where $\widehat{\mathfrak{D}}$ is the Weyl connection associated to $\kappa$ with respect to $\mathfrak{h}$, $\widehat{\mathfrak{r}}^i{}_{jkl}$ and $\widehat{\mathfrak{l}}_{ij}$ are respectively its Riemann and Schouten tensors, $\mathfrak{y}_{ij}$ is the Cotton-York tensor of $\mathfrak{h}$, $\epsilon_{ijk}$ is the volume form. Equation \eqref{eq:aux_syst_a} is just the decomposition of the Riemann tensor \eqref{eq:decomposition_Riemann_Weyl}, equation \eqref{eq:aux_syst_b} comes from the definitions of the Cotton tensor \eqref{def:cotton_weylconn} and the Cotton-York tensor \eqref{def:cotton_york}, equation \eqref{eq:aux_syst_c} is the divergence equation on the Cotton-York tensor. The gauge invariances of this system are associated to isometries and Weyl connection changes.

        Similarly to \Cref{sec:zero_quantities}, one can introduce a set of zero quantities encoding this auxiliary system. For any 3-dimensional oriented smooth Lorentzian manifold $(\mathscr{I},\mathfrak{h})$, any smooth covector field $\kappa$ on $\mathscr{I}$, any connection $\acute{\mathfrak{D}}$ on $\mathscr{I}$, any smooth 2-tensor field $\mathfrak{U}_{ij}$ on $\mathscr{I}$ and any smooth frame field $(e_{\bf i})$ on $\mathscr{I}$, define
        \begin{subequations}
            \begin{align*}
                \varphi_{\bf i} &:= \mathfrak{h}^{\bf jk} \left( \acute{\mathfrak{D}}_{\bf j} \mathfrak{y}_{\bf ki} - \kappa_{\bf j} \mathfrak{y}_{\bf ki} \right) \,, \\
                \varpi^{\bf i}{}_{\bf jkl} &:= \acute{\mathfrak{r}}^{\bf i}{}_{\bf jkl} - 2 S_{\bf j[k}{}^{\bf im} \mathfrak{U}_{\bf l]m} \,, \\
                \Pi_{\bf ij} &:= \acute{\mathfrak{D}}_{\bf i} \kappa_{\bf j} - \acute{\mathfrak{D}}_{\bf j} \kappa_{\bf i} + 2 \mathfrak{U}_{\bf [ij]} \,, \\
                \varrho_{\bf ijk} &:= \acute{\mathfrak{D}}_{\bf i} \mathfrak{U}_{\bf jk} - \acute{\mathfrak{D}}_{\bf j} \mathfrak{U}_{\bf ik} + \epsilon_{\bf ij}{}^{\bf l} \mathfrak{y}_{\bf lk} \,, \\
                \acute{\Sigma}_{\bf i}{}^{\bf k}{}_{\bf j} &:= \acute{\mathfrak{G}}_{\bf i}{}^{\bf k}{}_{\bf j} - \acute{\mathfrak{G}}_{\bf j}{}^{\bf k}{}_{\bf i} - \langle \omega^{\bf k}, [e_{\bf i},e_{\bf j}] \rangle \,, \\
                \acute{\vartheta}_{\bf ijk} &:= \acute{\mathfrak{D}}_{\bf i} \mathfrak{h}_{\bf jk} + 2\kappa_{\bf i} \mathfrak{h}_{\bf jk} \,,
            \end{align*}
        \end{subequations}
        where $(\omega^{\bf i})$ is the dual frame field of $(e_{\bf i})$, $\acute{\mathfrak{r}}^{\bf i}{}_{\bf jkl}$ are the frame components of the Riemann tensor of $\acute{\mathfrak{D}}$ and
        \[ \acute{\mathfrak{G}}_{\bf i}{}^{\bf k}{}_{\bf j} := \langle \omega^{\bf k}, \acute{\mathfrak{D}}_{e_{\bf i}} e_{\bf j} \rangle \,. \]
        
        A gauge choice for the auxiliary system consists of choosing a Weyl connection $\widehat{\mathfrak{D}}$, an orthonormal frame field $(e_{\bf i})$ and a coordinate system $(y^i)$. Let $(\mathscr{I},\mathfrak{h},\kappa)$ be a solution to this geometric auxiliary system and assume that it admits a compact spacelike hypersurface $\mathfrak{S}_\star$. Similarly to \Cref{sec:gauge_construction}, one can construct a gauge choice $(\widecheck{\mathfrak{D}},(e_{\bf i}),(y^i)))$ on a neighbourhood $\mathfrak{U}$ of a point $p\in\mathfrak{S}_\star$ based on a congruence of timelike conformal geodesics such that
        \begin{align*}
            \mathfrak{S}_\star \cap \mathfrak{U} &= \{q \in \mathfrak{U} \, | \, y^0(q) = 0\} \,, \\
            e_{\bf A}{}^0 &= 0 \quad \text{on } \mathfrak{U}_\star \,, \\
            e_{\bf 0}{}^i &= \delta_{\bf 0}{}^i \,, \\
            \widecheck{\mathfrak{G}}_{\bf 0}{}^{\bf j}{}_{\bf i} &= 0 \,, \\
            \widecheck{\mathfrak{l}}_{\bf 0i} &= 0 \,, \\
            (\widecheck{\kappa}_{\mathfrak{h}})_{\bf 0} &= 0 \,, \\
            \mathfrak{h}_{\bf ij} &= \eta_{\bf ij} \,, \\
            \epsilon_{\bf 012} &= -1 \,.
        \end{align*}
        In this gauge, one obtains a suitable quasilinear hyperbolic system of partial differential equation for the unknown
        \[ \underline{\mathfrak{u}} := \left( \frac{1}{2} \mathfrak{y}_{\bf 00}, \mathfrak{y}_{\bf 0A}, \acute{\mathfrak{G}}_{\bf A}{}^{\bf C}{}_{\bf B}, \acute{\mathfrak{G}}_{\bf A}{}^{\bf B}{}_{\bf 0}, \acute{\mathfrak{G}}_{\bf A}{}^{\bf 0}{}_{\bf B}, \acute{\mathfrak{G}}_{\bf A}{}^{\bf 0}{}_{\bf 0}, \kappa_{\bf A}, \mathfrak{U}_{\bf AB}, \mathfrak{U}_{\bf A0}, e_{\bf A}{}^i \right) \]
        from a certain part of the zero quantities. This system will be called the auxiliary evolution system. The initial data $\underline{\mathfrak{u}}_\star$ for the auxiliary evolution system is derived by picking the adequate components of the restriction on $\mathcal{B}_\star$ of the initial data $\underline{u}_\star$ for the evolution system \eqref{evol_system}. By design of both 
        \begin{itemize}
            \item[$\bullet$] the initial data $\underline{\mathfrak{u}}_\star$,
            \item[$\bullet$] the above gauge and the one for the (CVE), see \Cref{prop:construction_gauge},
        \end{itemize}
        and by application of \Cref{lem:conf_geod_stay_conf_bound}, one deduces that the vector fields $e_{\bf 0}$, $e_{\bf 1}$, $e_{\bf 2}$ of the two systems coincide (on $\mathscr{I} = \mathcal{W} \cap \mathcal{B}$). Hence $\mathfrak{y}^\dagger_{\bf AB}$ in the above gauge is also given by $q''_{\bf AB}$ expressed in the gauge for the (CVE) which is entirely determined by the analytic boundary data $d_1$, $d_2$.
        
        \begin{rem}
            Note that the boundary data for the evolution system \eqref{evol_system} and its derivatives now play the role of source terms in the auxiliary evolution system.
        \end{rem}

        After solving the auxiliary evolution system, one can prove the propagation of the constraints, which are encoded by the rest of the zero quantities. Finally, one reconstructs the metric $\breve{\mathfrak{h}}$ by $\breve{\mathfrak{h}}_{\bf ij} := \eta_{\bf ij} \,  \omega^{\bf i} \otimes \omega^{\bf j}$. \qedhere
    \end{itemize}
\end{proof}

\subsubsection{The geometric homogeneous Robin boundary conditions}
\label{sec:homogeneous_robin}

Now consider the reflective boundary conditions \eqref{bc_VE_reflective} with $\theta' \neq \pi/2$ and $q''_{\bf AB} = 0$, that is
\begin{equation}
    \label{bc_electric}
    \mathfrak{t}^\dagger_{\bf AB} = \mu \, \mathfrak{y}^\dagger_{\bf AB} \,,
\end{equation}
where $\mu := \tan(\theta') \in \R$. Their corresponding analytic form is given by \eqref{raw_bc} with
\begin{equation}
    \label{eq:robin_analytic}
    B = \begin{pmatrix}
        \cos\chi & \sin\chi \\
        \sin\chi & -\cos\chi
        \end{pmatrix} \in M_2(\R) \, \qquad d_1 = d_2 = 0 \,,
\end{equation}
for some $\chi \in (0,2\pi)$ related to $\mu$ by an invertible relation. From \eqref{bc_electric}, which is gauge-dependent, we want to obtain our geometric reflective boundary conditions
\begin{equation}
    \tag{\ref{eq:bc_robin}}
    \mathfrak{p}_{ij} := \mathfrak{t}_{ij} - \mu \, \mathfrak{y}_{ij} = 0 \,. 
\end{equation}
Since \eqref{bc_electric} rewrite under the form
\[ \mathfrak{p}^\dagger_{\bf AB} = 0 \,, \]
it remains to study the following quantities
\[ \mathfrak{p}_{\bf A} := \mathfrak{p}_{\bf 0A} = \mathfrak{t}_{\bf 0A} - \mu \, \mathfrak{y}_{\bf 0A} \,, \qquad \mathfrak{p} := \frac{1}{2} \mathfrak{p}_{\bf 00} = \frac{1}{2} \left(\mathfrak{t}_{\bf 00} - \mu \, \mathfrak{y}_{\bf 00}\right) \,. \]
To this end, we will also use an auxiliary evolution system on the conformal boundary. Recall that both the boundary stress-energy tensor $\mathfrak{t}$ and the Cotton-York tensor $\mathfrak{y}$ are divergence-free for the Levi-Civita connection $\mathfrak{D}$ of any representative $\mathfrak{h}$ of the boundary conformal class $[\mathfrak{h}]$. Thus
\[ \mathfrak{D}_i \mathfrak{p}^i{}_j = 0 \,. \]
More generally, for any Weyl connection $\widehat{\mathfrak{D}}$ and any representative $\mathfrak{h} \in [\mathfrak{h}]$, one has
\begin{equation}
    \label{eq:div_aux_syst}
    \widehat{\mathfrak{D}}_i \mathfrak{p}^i{}_j - 3 (\widehat{\kappa}_{\mathfrak{h}})_i \mathfrak{p}^i{}_j = 0 \,,
\end{equation}  
where $\widehat{\kappa}_\mathfrak{h}$ is the covector field associated to $\widehat{\mathfrak{D}}$ with respect to $\mathfrak{h}$.

\begin{lem}
    Consider the analytic boundary conditions \eqref{raw_bc} with the choices \eqref{eq:robin_analytic}. Let $\underline{u}_\star \in \mathcal{C}^\infty(\overline{\mathcal{U}_\star},\R^m)$ be some analytic initial data satisfying all the conditions of \Cref{lem:propagation_constraints}. The quantities $\mathfrak{p}$ and $\mathfrak{p}_{\bf A}$, associated to the solution $(\mathcal{W},\breve{g},V,\Theta,\kappa)$ and the gauge choice $(\Theta,\acute{\nabla},(e_{\bf a}),(y^\mu))$ given by \Cref{cor:geometry_back}, verify on the conformal boundary $\mathscr{I} = \mathcal{W} \cap \mathcal{B}$ the following evolution system
    \begin{subequations}
        \label{eq:aux_system_fraku}
        \begin{align}
         &\quad 2e_{\bf 0}(\mathfrak{p}) - \delta^{\bf AB} e_{\bf A}(\mathfrak{p}_{\bf B}) + 3 \acute{\mathfrak{G}}_{\bf A}{}^{\bf A}{}_{\bf 0} \, \mathfrak{p} - \left(\acute{\mathfrak{G}}_{\bf A}{}^{\bf A}{}_{\bf B} - 4 \kappa_{\bf B} \right) \delta^{\bf BC} \, \mathfrak{p}_{\bf C} \nonumber \\
         &= - \acute{\mathfrak{G}}_{\bf A}{}^{\bf B}{}_{\bf 0} \delta^{\bf AC} \mathfrak{p}^\dagger_{\bf BC} \,, \\
         &\quad e_{\bf 0}(\mathfrak{p}_{\bf A}) - e_{\bf A}(\mathfrak{p}) + 3 \kappa_{\bf A} \, \mathfrak{p} + \left( \acute{\mathfrak{G}}_{\bf C}{}^{\bf C}{}_{\bf 0} \delta_{\bf AB} + \acute{\mathfrak{G}}_{\bf B}{}^{\bf 0}{}_{\bf A} \right) \delta^{\bf BC} \mathfrak{p}_{\bf C} \nonumber \\
         &= \delta^{\bf BC} e_{\bf B}(\mathfrak{p}^\dagger_{\bf AC}) - \left( \acute{\mathfrak{G}}_{\bf B}{}^{\bf C}{}_{\bf A}  - \acute{\mathfrak{G}}_{\bf E}{}^{\bf E}{}_{\bf B} \delta_{\bf A}{}^{\bf C} \right) \delta^{\bf BD} \mathfrak{p}^\dagger_{\bf CD} \,,
        \end{align}
    \end{subequations}
    where $\breve{\mathfrak{h}}$ is the metric induced by $\breve{g}$ on $\mathscr{I}$ and $\acute{\mathfrak{G}}_{\bf i}{}^{\bf j}{}_{\bf k}$ are the connection coefficients of the Weyl connection $\acute{\mathfrak{D}}$, induced by $\acute{\nabla}$ on $\mathscr{I}$, in the frame field $(e_{\bf 0},e_{\bf 1},e_{\bf 2})$.
\end{lem}

\begin{proof}
    The system \eqref{eq:aux_system_fraku} is just an expression of \eqref{eq:div_aux_syst} using the gauge properties \eqref{eq:gauge_core_properties}-\eqref{gauge_additional_properties} and the decomposition
    \[ \mathfrak{p}_{\bf AB} = \mathfrak{p}^\dagger_{\bf AB} + \mathfrak{p} \, \delta_{\bf AB} \,. \qedhere \]
\end{proof}

\noindent The analogue of \Cref{prop:partial_to_total_dirichlet_bc} in this case is the following proposition.

\begin{prop}
    \label{prop:partial_to_total_robin_bc}
    Consider the analytic boundary conditions \eqref{raw_bc} with the choices \eqref{eq:robin_analytic}. Let $\underline{u}_\star \in \mathcal{C}^\infty(\overline{\mathcal{U}_\star},\R^m)$ be some analytic initial data satisfying all the conditions of \Cref{lem:propagation_constraints}.
    \begin{itemize}
         \item[i)]  If the 4-dimensional aAdS space $(\mathcal{W},\mathcal{W}\cap\mathcal{B},\Theta^{-2}\breve{g})$ of \Cref{cor:geometry_back} verifies the geometric boundary condition \eqref{eq:bc_robin}, then $\underline{u}_\star$ satisfies
         \begin{subequations}
            \label{eq:frakp_vanish_corner}
             \begin{align}
                 \frac{2}{3} E_{33} - \mu H_{33} &= 0 \quad \text{on } \mathcal{B}_\star \,, \\
                 \frac{2}{3} H_{3A} + \mu E_{3A} &= 0 \quad \text{on } \mathcal{B}_\star \,,
             \end{align}
         \end{subequations}
         where the indices refer to the coordinates $(y^1,y^2,y^3)$ on $\mathcal{U}_\star$. 
         
        \item[ii)] Conversely, if $\underline{u}_\star$ satisfy \eqref{eq:frakp_vanish_corner} then the 4-dimensional aAdS space $(\mathcal{X},\mathcal{X}\cap\mathcal{B},\Theta^{-2}\breve{g})$, where $\mathcal{X}$ is a neighbourhood of $\mathcal{U}_\star$ in $\mathcal{W}$, verifies the geometric boundary condition \eqref{eq:bc_robin}.
    \end{itemize}
\end{prop}

\begin{rems} \,
    \begin{itemize}
        \item The conditions \eqref{eq:frakp_vanish_corner} will be interpreted as geometric compatibility conditions in \Cref{lem:geometric_compatibility_conditions}.
        \item The restriction to boundary conditions with constant coefficients does not, in fact, restrict the generality of \eqref{eq:bc_robin} since it appears to be impossible to propagate it for a non-constant $\mu$. In that regard, see also the remark following \Cref{lem:BSET_robin}. \qedhere
    \end{itemize}
\end{rems}

\begin{proof} \,
    \begin{itemize}
        \item[$i)$] If the geometric boundary condition \eqref{eq:bc_robin} holds then $\mathfrak{p}=0$, $\mathfrak{p}_{\bf A}=0$ on $\mathscr{I}$. Moreover,
        \begin{align*}
            \mathfrak{t}_{\bf 00} &= - \frac{2}{3} \mathfrak{E}_{\bf 00} = - \frac{2}{3} V^{\bf 3}{}_{\bf 030} = \frac{2}{3} E_{\bf 33} \,, \\
    	   \mathfrak{y}_{\bf 00} &= - \mathfrak{H}_{\bf 00} = - (\star V)^{\bf 3}{}_{\bf 030} = H_{\bf 33} \,, \\
           \mathfrak{t}_{\bf 0A} &= - \frac{2}{3} \mathfrak{E}_{\bf 0A} = - \frac{2}{3} V^{\bf 3}{}_{\bf 03A} = - \frac{1}{3} \slashed{\epsilon}_{\bf A}{}^{\bf B} H_{\bf 3B} \,, \\
    	   \mathfrak{y}_{\bf 0A} &= - \mathfrak{H}_{\bf 0A} = - (\star V)^{\bf 3}{}_{\bf 03A} = \frac{1}{2} \slashed{\epsilon}_{\bf A}{}^{\bf B} E_{\bf 3B} \,.
        \end{align*}
        One deduces the result using the properties of the frame field $(e_{\bf a})$.

        \item[$ii)$] Let us prove that $\mathfrak{p}$ and $\mathfrak{p}_{\bf A}$ vanish identically on a neighbourhood of $\mathcal{B}_\star$ in $\mathscr{I}$. Note that \eqref{eq:aux_system_fraku} forms a homogeneous linear symmetric hyperbolic system on $(\mathfrak{p},\mathfrak{p}_{\bf A})$. Indeed, for the following ordering of the unknowns $(\mathfrak{p},\mathfrak{p}_{\bf 1},\mathfrak{p}_{\bf 2})$, the matrices with respect to the coordinates $(y^i) = (y^0,y^1,y^2)$ are given by
        \[ D^0 = D^{\bf 0} + e_{\bf 1}{}^0 D^{\bf 1} + e_{\bf 2}{}^0 D^{\bf 2} \,, \quad D^A = e_{\bf 1}{}^A D^{\bf 1} + e_{\bf 2}{}^A D^{\bf 2} \,, \]
        where
        \[ D^{\bf 0} = \begin{pmatrix}
                 2 & 0 & 0 \\
                 0 & 1 & 0 \\
                 0 & 0 & 1
                \end{pmatrix} \,, 
            \qquad
            D^{\bf 1} = \begin{pmatrix}
                0 & -1 & 0 \\
                -1 & 0 & 0 \\
                0 & 0 & 0 
                \end{pmatrix} \,,
            \qquad
            D^{\bf 2} = \begin{pmatrix}
                0 & 0 & -1 \\
                0 & 0 & 0 \\
                -1 & 0 & 0
                \end{pmatrix}
        \]
        Consequently, there exists a neighbourhood $\mathfrak{X}$ of $\mathcal{B}_\star$ in $\mathscr{I}$ such that these quantities vanish identically on $\mathfrak{X}$ if and only if they vanish initially on $\mathcal{B}_\star$, which is equivalent to \eqref{eq:frakp_vanish_corner}. It suffices to choose a neighbourhood $\mathcal{X}$ of $\mathcal{U}_\star$ in $\mathcal{W}$ such that $\mathcal{X}\cap\mathcal{B}=\mathfrak{X}$ to obtain the result. \qedhere
    \end{itemize}
\end{proof}

\subsection{The geometric initial data problem}
\label{sec:geometric_initial_data}

Having established the geometric interpretation of certain analytic boundary conditions in \Cref{sec:geometric_bc}, let us now discuss the geometric initial data. In \Cref{sec:conditions_initial_data}, the analytic initial data and the conditions imposed on it are interpreted geometrically. This also shows how to derive the analytic initial data from the geometric one. A precise formulation of the geometric initial data problem is given in \Cref{sec:formulation_gidp}.

The next sections are devoted to the (partial) resolution of the geometric initial data problem. In \Cref{sec:conformal_method}, we cite some results using an adapted version of the conformal method. \Cref{sec:unphysical_fields} is dedicated to the study of the unphysical fields. In particular, we derive a novel result, see \Cref{thm:smooth_unphysical_fields} and \Cref{cor:smooth_unphysical_fields}, which gives necessary and sufficient conditions for the smoothness of the unphysical fields, which is part of the geometric initial data problem. This extends a previous result by K\'ann\'ar \cite{K96}.

\subsubsection{Analytic and geometric initial data}
\label{sec:conditions_initial_data}

In this section, we will use the notations of \Cref{lem:propagation_constraints}. Recall that the 4-dimensional geometry was (re)introduced by \Cref{cor:geometry_back}. The following lemma analyses all the conditions, except the compatibility conditions, imposed on the analytic initial data.

\begin{lem} 
    \label{lem:analytic_to_geometric_ID}
    Let $\underline{u}_\star \in \mathcal{C}^\infty(\overline{\mathcal{U}}_\star,\R^m)$ be some initial data verifying all the conditions of \Cref{lem:propagation_constraints}. Denote by $\breve{h}$ the metric induced by $\breve{g}$ on $\mathcal{U}_\star$. Then
    \begin{itemize}
        \item[i)] $(e_{\bf i})$ is a frame field on $\mathcal{U}_\star$ such that, on $\mathcal{B}_\star$, $e_{\bf 1}$ and $e_{\bf 2}$ are tangent while $e_{\bf 3}$ is the outward-pointing unit normal with respect to $\breve{h}$,
        \item[ii)] one has the following usual relations on the connection coefficients, see \Cref{lem:connection_coeff} and \Cref{def:extrinsic_curvature}, 
        \[ \acute{\Gamma}_{\bf i}{}^{\bf 0}{}_{\bf 0} = \kappa_{\bf i} \,, \quad \acute{\Gamma}_{\bf i}{}^{\bf k}{}_{\bf 0} \delta_{\bf kj} = \acute{\Gamma}_{\bf i}{}^{\bf 0}{}_{\bf j} \,, \quad \acute{\Gamma}_{\bf i}{}^{\bf 0}{}_{\bf j} = \acute{\Gamma}_{\bf j}{}^{\bf 0}{}_{\bf i} \,, \]
    \end{itemize}
    Define the 10-tuple $(\mathcal{U}_\star,\breve{h},E,\star H,\Theta_\star,\kappa,K,0,T,Q) \in \mathscr{D}$ where the fields $E$, $H$ and $\kappa$ are the ones defined in \Cref{lem:propagation_constraints} and
        \[ K := \acute{\Gamma}_{\bf i}{}^{\bf 0}{}_{\bf j} \, \omega^{\bf i} \otimes \omega^{\bf j} \,, \quad T := U_{\bf i0} \, \omega^{\bf i} \,, \quad Q := U_{\bf ij} \, \omega^{\bf i} \otimes \omega^{\bf j} \,, \]
        $(\omega^{\bf i})$ denoting the dual frame field of $(e_{\bf i})$. Then
    \begin{itemize}
        \item[iii)] $s_\star$ is the field $\widehat{\sigma}$, defined by \eqref{eq:def_sigma_hat}, of this 10-tuple,
        \item[iv)] this 10-tuple satisfy the \eqref{eq:CVC} with $\Lambda=-3$ and
        \begin{itemize}
            \item[a)] the gauge-independent conditions
            \begin{subequations} 
                \label{eq:gauge_independent_conditions_CVC}
                \begin{align}
                \label{eq:z_vanish}
                z &= 0 \quad \text{on } \slashed{\mathscr{I}} \,, \\
                \label{eq:K_vanish}
                K_{\perp i} + \varepsilon \frac{z}{\Psi} h_{\perp i} &= 0 \quad \text{on } \slashed{\mathscr{I}} \,,
                \end{align}
            \end{subequations}
            where a $\perp$ index denotes the normal direction to $\slashed{\mathscr{I}}$, 
            \item[b)] the gauge choices
            \begin{subequations}
                \label{eq:gauge_choices_CVC}
                \begin{align}
                \label{eq:gauge_choice_Psi}
                \Psi &\geq 0 \quad \text{on } \mathcal{S} \,, \\
                \label{eq:gauge_choice_z}
                z &= 0 \quad \text{on } \mathcal{S}\setminus\slashed{\mathscr{I}} \,, \\
                \label{eq:gauge_choice_s}
                \widehat{s}_h &= 0 \quad \text{on } \slashed{\mathscr{I}} \,.
                \end{align}
            \end{subequations}
        \end{itemize}
    \end{itemize}
\end{lem}

\begin{rems} \,
    \begin{itemize}
        \item Some of the conclusions are evident from the 4-dimensional geometry. However, it is useful to derive them solely from the analytic conditions to be sure to interpret all of these conditions.
        
        \item The gauge choices \eqref{eq:gauge_choices_CVC} for the (CVC) are the transcription of the gauge choices \eqref{conditions} for the (CVE) made at the level of initial data in the gauge construction. Moreover, the condition \eqref{eq:z_vanish} corresponds to \eqref{eq:hyp_Sigma}. Note that the gauge choice \eqref{eq:gauge_choice_Psi} (respectively \eqref{eq:gauge_choice_z}) is only possible if $\slashed{\mathscr{I}} \subset \partial\mathcal{S}$ (respectively \eqref{eq:z_vanish} holds). 
        
        \item Thanks to \eqref{eq:pseudonorm_dPsi}, the condition \eqref{eq:z_vanish} implies that $(\mathcal{U}_\star,\mathcal{B}_\star,\Theta_\star^{-2}\breve{h})$ is asymptotically hyperbolic. \qedhere
    \end{itemize}
\end{rems}

\begin{proof} \,
    \begin{itemize}
        \item[i)] The conditions \eqref{eq:initial_conditions_ei0} and \eqref{eq:initial_frame_field} correspond to $(e_{\bf i})$ being a frame field on $\mathcal{U}_\star$. Since $(e_{\bf a})$ is $\breve{g}$-orthonormal, the conditions $e_{\bf A}{}^3=0$ and $e_{\bf 3}{}^3 < 0$ on $\mathcal{B}_\star$ in \eqref{ode_initial_data} are equivalent to impose that, on $\mathcal{B}_\star$, $e_{\bf 1}$ and $e_{\bf 2}$ are tangent and $e_{\bf 3}$ is the outward-pointing unit normal with respect to $\breve{h}$.
    
        \item[ii)] The conditions $\acute{\vartheta}_{\bf i0a}=0$ and $\acute{\Sigma}_{\bf i}{}^{\bf 0}{}_{\bf j}=0$ on $\mathcal{U}_\star$ in \eqref{eq:initial_conditions_constraints} are equivalent to the relations on the connection coefficients.

        \item[iii)-iv)] The rest of the conditions in  \eqref{eq:initial_conditions_constraints}, that is to know
    \begin{alignat*}{5}
        \varkappa_{\bf ia} &= 0 \quad \text{on } \mathcal{U}_\star \,, &\qquad \mho_{\bf i}  &=0 \quad \text{on } \mathcal{U}_\star \,, &\qquad \varsigma_{\bf i} &=0 \quad \text{on } \mathcal{U}_\star \,, \\
        \acute{\vartheta}_{\bf ijk} &=0 \quad \text{on } \mathcal{U}_\star \,, &\qquad
        \aleph &=0 \quad \text{on } \mathcal{U}_\star\,, &\qquad \varpi^{\bf a}{}_{\bf bij} &= 0 \quad \text{on } \mathcal{U}_\star \,, \\
        \Pi_{\bf ij} &= 0 \quad \text{on } \mathcal{U}_\star \,, &\qquad \varrho_{\bf ija} &= 0 \quad \text{on } \mathcal{U}_\star \,, &\qquad \acute{\Sigma}_{\bf i}{}^{\bf k}{}_{\bf j} &= 0 \quad \text{on } \mathcal{U}_\star \,, \\
        \varphi_{\bf ia0} &= 0 \quad \text{on } \mathcal{U}_\star \,,
    \end{alignat*}
    are just the constraints implied by the (CVE) on $\mathcal{U}_\star$. This is equivalent to imposing that the 10-tuple defined in the lemma satisfies the (CVC) with $\Lambda=-3$ and that $s_\star$ is equal to the field $\widehat{\sigma}$ of this 10-tuple.

    Since $\Theta_\star$ is a boundary defining function of $\mathcal{B}_\star$, \eqref{eq:gauge_choice_Psi} is verified and $\slashed{\mathscr{I}} = \mathcal{B}_\star$. The fact that $z=0$ splits into \eqref{eq:z_vanish} and \eqref{eq:gauge_choice_Psi}. Since $f_\star$ is bounded from above, the field $\widehat{\sigma}$ associated to the 10-tuple vanishes on $\slashed{\mathscr{I}}$. With \eqref{eq:z_vanish} and \eqref{eq:def_sigma_hat}, one deduces \eqref{eq:gauge_choice_s}.
    
    Finally, let us look at the remaining conditions in \eqref{ode_initial_data}. By definition of $K_{ij}$ and $ii)$, the conditions $\acute{\Gamma}_{\bf A}{}^{\bf 3}{}_{\bf 0} = 0$ and $\acute{\Gamma}_{\bf 3}{}^{\bf 3}{}_{\bf 0} = 0$ on $\mathcal{B}_\star$ are equivalent to $K_{\bf 3A} = 0$ and $K_{\bf 33} = 0$ on $\mathcal{B}_\star$. Thanks to $i)$, this is the condition \eqref{eq:K_vanish} expressed with the gauge choices \eqref{eq:gauge_choices_CVC}. The condition $\zeta_{\bf A} = 0$ on $\mathcal{B}_\star$ is automatically satisfied by $i)$ and the fact that $\Theta_\star$ is a boundary defining function of $\mathcal{B}_\star$. The conditions $U_{\bf A3} = 0$ and $U_{\bf 33}+f_\star = $ on $\mathcal{B}_\star$ are ensured by \eqref{eq:CVC_j} and $iii)$. \qedhere
    \end{itemize}
\end{proof}

Let us now turn to the analytic compatibility conditions, see $iii)$ of \Cref{evol_solve_near}, in the case where $\mathcal{B}_\star \neq \varnothing$. We will only detail the geometric compatibility conditions for our geometric boundary conditions \eqref{eq:bc_robin}. Then, recall that the analytic initial data also has to satisfy \eqref{eq:frakp_vanish_corner}, see $ii)$ of \Cref{prop:partial_to_total_robin_bc}.

\begin{lem}
    \label{lem:geometric_compatibility_conditions}
    The analytic compatibility conditions and the conditions \eqref{eq:frakp_vanish_corner} can be expressed as conditions at $\mathcal{B}_\star$ on the 10-tuple of \Cref{lem:analytic_to_geometric_ID}. These form the \emph{enlarged set of compatibility conditions} for our geometric boundary conditions \eqref{eq:bc_robin}. Furthermore, the analytic compatibility conditions of order $0$ and the conditions \eqref{eq:frakp_vanish_corner} are in fact geometric, meaning that they are left invariant by the gauge transformations of the (CVC) described in \Cref{def:transfo_dim3}.
\end{lem}

\begin{rems} \,
    \begin{itemize}
        \item One can expect the higher order compatibility conditions to be also geometric although we do not address this problem in this article.
        \item The conditions \eqref{eq:frakp_vanish_corner} are part of the geometric compatibility conditions of order $0$. These do not imply higher order compatibility conditions. Indeed, if the analytic boundary condition and the enlarged set of compatibility conditions hold then the time derivatives of all order of $\mathfrak{p}$ and $\mathfrak{p}_{\bf A}$ vanish on $\mathcal{B}_\star$ thanks to \eqref{eq:aux_system_fraku}. \qedhere
    \end{itemize}
\end{rems}

\begin{proof}
    The analytic compatibility conditions of order 0 are simply the analytic boundary conditions
    \[ 0 = \mathfrak{p}^\dagger_{\bf AB} := \mathfrak{t}^\dagger_{\bf AB} - \mu \mathfrak{y}^\dagger_{\bf AB} = \frac{2}{3} E_{\bf AB}^\dagger - \mu H^\dagger_{\bf AB} \quad \text{on } \mathscr{I} \,. \]
    These are conditions at $\mathcal{B}_\star$ on the tensors $E_{ij}$ and $M_{ijk} := (\star H)_{ijk}$ of the 10-tuple defined in \Cref{lem:analytic_to_geometric_ID}. Moreover, they are gauge-independent for the gauge transformations described in \Cref{def:transfo_dim3}. Similarly, the conditions $\mathfrak{p}=0$, $\mathfrak{p}_{\bf A}=0$ are also gauge-independent conditions on the 10-tuple at $\mathcal{B}_\star$.

    The higher order analytic compatibility conditions are obtained by combining the evolution equations and the analytic boundary conditions, see the last remark below \Cref{evol_solve_near}. Thanks to the gauge properties \eqref{eq:gauge_core_properties} and \eqref{gauge_additional_properties}, all the 4-dimensional fields can be expressed with fields in the 10-tuple.
\end{proof}

\subsubsection{Formulation of the geometric initial data problem}
\label{sec:formulation_gidp}

In light of \Cref{sec:conditions_initial_data} as well as \Cref{prop:CVC_VC}, the geometric initial data problem is formulated as follows: find a smooth Riemannian metric $\widetilde{h}_{ij}$ and a smooth symmetric 2-tensor $\widetilde{K}_{ij}$ on the interior $\mathring{\mathcal{S}}$ of a compact 3-dimensional smooth manifold $\mathcal{S}$ with non-empty boundary $\partial\mathcal{S}$ such that
\begin{itemize}
    \item[i)] $(\mathcal{S},\partial\mathcal{S},\widetilde{h}_{ij})$ is an asymptotically hyperbolic space, in the sense of \Cref{def:aH}, with smooth rescaled metrics,
    
    \item[ii)] $(\mathring{\mathcal{S}},\widetilde{h}_{ij},\widetilde{K}_{ij})$ is a solution to the \eqref{eq:VC} with $\Lambda=-3$,
    
    \item[iii)] for any boundary defining function $\Psi$ of $\partial\mathcal{S}$, any smooth covector field $\kappa$ on $\mathcal{S}$ and any smooth scalar field $z$ on $\mathcal{S}$ which vanishes on $\partial\mathcal{S}$,
    \begin{itemize}
        \item[a)] the field $P_{ij} := \Psi \widetilde{K}_{ij}$ extends smoothly to $\partial\mathcal{S}$ and verifies $P_{\perp i} = 0$ on $\partial\mathcal{S}$, 
        \item[b)] the unphysical fields defined by \eqref{CVC_def_Q}-\eqref{CVC_def_M} extend smoothly to $\partial\mathcal{S}$,
    \end{itemize}
    
    \item[iv)] on each connected component $\mathcal{C}$ of $\partial\mathcal{S}$, the corresponding geometric compatibility conditions are verified.
\end{itemize}

\begin{rems} \,
    \begin{itemize}
        \item Condition $iii)a)$ is equivalent to both the smoothness of the unphysical field $K_{ij}$ defined by \eqref{CVC_def_K} and the geometric condition \eqref{eq:K_vanish}. Note that \eqref{eq:z_vanish} is verified by the choice of $z$.
        \item The smoothness of the unphysical fields is referred as the `asymptotically simple' hypothesis by Friedrich in \cite{F95}. \qedhere
    \end{itemize}
\end{rems}

In \Cref{sec:conformal_method}, we cite certain results of Andersson and Chru\'sciel addressing $i)$-$iii)a)$, when the trace $\widetilde{K}$ is constant, by means of an adapted version of the conformal method. Under some hypotheses, K\'ann\'ar derived the conditions under which a solution given by the conformal method verifies $iii)b)$. These will be stated in \Cref{sec:unphysical_fields}, prior to the presentation of a new result extending K\'ann\'ar's theorem, see \Cref{cor:smooth_unphysical_fields}. This corollary provides the necessary and sufficient conditions for a solution to $i)$-$ii)$ (regardless of the method used to solve this problem) to satisfy $iii)$.

\subsubsection{The conformal method}
\label{sec:conformal_method}

Let us first recall the conformal method to solve the (VC) on a 3-dimensional closed smooth manifold $\widetilde{\mathcal{S}}$. The method consists in decomposing the fields $\widetilde{h}_{ij}$ and $\widetilde{K}_{ij}$ under the following form
\begin{subequations}
	\label{eq:decomp_tilde}
	\begin{align}
		\label{eq:decomp_htilde}
		\widetilde{h}_{ij} &= \theta^4 h_{ij} \,, \\
		\label{eq:decomp_Ktilde}
		\widetilde{K}_{ij} &= \frac{\widetilde{K}}{3} \theta^4 h_{ij} + \theta^4 (\mathcal{K}_h X)_{ij} + \theta^{-2} \phi_{ij} \,,
	\end{align}
\end{subequations}
where
\begin{itemize}
	\item $\theta$ is a smooth positive function on $\widetilde{\mathcal{S}}$,
	\item $h_{ij}$ is a smooth Riemannian metric on $\widetilde{\mathcal{S}}$,
	\item $\widetilde{K}$ is a smooth scalar field on $\widetilde{\mathcal{S}}$,
	\item $\mathcal{K}_h$ is the conformal Killing operator of $h$ defined by
	\[ (\mathcal{K}_h X)_{ij} := D_i X_j + D_j X_i - \frac{2}{3} D_k X^k h_{ij} \,, \]
	\item $X^i$ is a smooth vector field on $\widetilde{\mathcal{S}}$ defined up to conformal Killing fields,
	\item $\phi_{ij}$ is a \emph{TT-tensor}, that is a symmetric 2-tensor which is trace-free and divergence-free with respect to the metric $h$,	
\end{itemize}
Note that \eqref{eq:decomp_Ktilde} is an orthogonal decomposition of $\widetilde{K}_{ij}$ into its pure-trace part, its pure-divergence part and its TT part. The factors in powers of $\theta$, notably the unusual one in front of the conformal Killing operator, are designed so that the decomposition is conformally invariant, see later below. The (VC) are then equivalent to the following coupled elliptic equations on $(\theta,X^i)$
\begin{subequations}
	\label{eq:VC_conf_method}
	\begin{align}
		\label{eq:conf_scalar_wave_theta}
		\Delta_h \theta - \frac{r}{8} \theta - \frac{1}{12} \theta^5 \left( \widetilde{K}^2 - 3\Lambda \right) &= - \frac{1}{8} \theta^{-3} \left| \theta^4 \mathcal{K}_h X + \theta^{-2} \phi \right|^2_h \,, \\
		\label{eq:divAgrad_X}
		D^i \left( \theta^6 (\mathcal{K}_h X)_{ij} \right) &= \frac{2}{3} \theta^6 (d\widetilde{K})_j \,.
	\end{align}
\end{subequations}
Both the decompositions \eqref{eq:decomp_tilde} and the elliptic system \eqref{eq:VC_conf_method} are left invariant under the following transformation
\[ \left(\theta,h_{ij},\widetilde{K},X^i,\phi_{ij}\right) \mapsto \left(\omega^{-1/2} \theta,\omega^2 h_{ij},\widetilde{K},X^i,\omega^{-1} \phi_{ij}\right) \,, \]
for any smooth positive function $\omega$. The free data is thus composed of the scalar field $\widetilde{K}$ and of the class $[(h,\phi)]$ for the same equivalence relation than in \Cref{def:equivalence_relation_pair}. An important case is when $\widetilde{K}$ is constant: from \eqref{eq:divAgrad_X} one deduces that $\mathcal{K}_h X = 0$ and it just remains to solve \eqref{eq:conf_scalar_wave_theta}.

Furthermore, the same operators can be used in order to construct TT-tensors. Indeed, for any symmetric smooth 2-tensor $A_{ij}$, one can define a TT-tensor $\phi_{ij}$ under the form
\[ \phi_{ij} := (\mathcal{K}_h X)_{ij} + A_{ij} \,, \]
where $X^i$ is the solution, defined up to conformal Killing fields of $h$, of the elliptic equation
\[ D_i \left(\mathcal{K}_h\right)^{ij} = - D_i A^{ij} \,. \]

\bigbreak

Let us now turn to the cases where the initial data set intersects the conformal boundary $\mathscr{I}$. This occurs when $\Lambda < 0$ and for hyperboloidal initial data sets when $\Lambda=0$. Consider a 3-dimensional compact smooth manifold $\mathcal{S}$ with boundary $\partial\mathcal{S} \neq \varnothing$. There are a few problems when trying to apply the conformal method in this context. First, note that if $\Psi := \theta^{-2}$ is a boundary defining function of $\partial\mathcal{S}$ then $\theta$ blows up on the boundary. Thus equation \eqref{eq:conf_scalar_wave_theta} is not suitable for setting a boundary condition. With that in mind, if one chooses the unknown $\Psi$ instead then one has to deal with the following singular equation  
\begin{equation}
    \label{eq:eqdiff_Psi}
    \Psi \Delta_h \Psi + \frac{r}{4} \Psi^2 - \frac{3}{2} |d\Psi|^2_h  + \frac{1}{6} \left(\widetilde{K}^2-3\Lambda \right) = \frac{1}{4} \Psi^4 \left|\tf_{\widetilde{h}} \widetilde{K} \right|_h^2 \,. 
\end{equation}
Moreover, in our case, one expects from \eqref{CVC_def_K} that on any compact neighbourhood of a boundary point $p \in \partial\mathcal{S}$, the field $\widetilde{K}_{ij}$ verifies the following asymptotic behaviours
\begin{equation}
    \label{eq:behav_Ktilde}
    \left(\tf_{\widetilde{h}} \widetilde{K}\right)_{ij} = \grando{\Psi^{-1}} \,, \qquad\widetilde{K} = - 3z + \grando{\Psi} \,. 
\end{equation}
Note that the first asymptotic behaviour does not match with the factors in powers of $\theta$ in the decomposition \eqref{eq:decomp_Ktilde}.

These problems require to modify the method. This was done by Andersson, Chru\'sciel and Friedrich \cite{ACF92} in the case where $\Lambda \leq 0$ and $\widetilde{K}_{ij}$ is pure trace. In that case, note that the trace $\widetilde{K}$ is then constant by the momentum constraint. Their results were later extended by Andersson and Chru\'sciel \cite{AC94,AC96} to the case where $\Lambda = 0$ and $\widetilde{K}$ is a non-zero constant. Finally, K\'ann\'ar \cite{K96} transposed the theorems to $\widetilde{K}$ constant and $\Lambda < \widetilde{K}^2/3$. For our purpose, a constant trace $\widetilde{K}$ implies that it vanishes due to the geometric condition \eqref{eq:z_vanish}. We summarise the relevant\footnote{Observe that our case $(\Lambda,\widetilde{K}) = (-3,0)$ is totally equivalent to $(\Lambda,\widetilde{K}) = (0,\pm3)$.} results in the following theorem.

\begin{thm}[Andersson and Chru\'sciel \protect{\cite[Theorems 3.2, 4.7 and 2.8]{AC96}}]
    \label{thm:conformal_method}
    Let $\mathcal{S}$ be a 3-dimensional compact smooth manifold with non-empty boundary $\partial\mathcal{S}$. 
    \begin{itemize}
        \item[i)] The adapted conformal method: Let $(h_{ij},\phi_{ij})$ be a free data pair where
        \begin{itemize}
            \item[$\bullet$] $h_{ij}$ is a smooth Riemannian metric on $\mathcal{S}$,
            \item[$\bullet$] $\phi_{ij}$ is a smooth TT-tensor on the interior $\mathring{\mathcal{S}}$ for the metric $h_{ij}$ such that, for any (and thus all) boundary defining function $x$ of $\partial\mathcal{S}$, $x^2 \phi_{ij}$ extends smoothly to $\partial\mathcal{S}$ and verifies $x^2\phi_{\perp i} = 0$ on $\partial\mathcal{S}$.
        \end{itemize}
        Take a boundary defining function $x$ of $\partial\mathcal{S}$ and define the smooth Riemannian metric $\overline{h}_{ij} := x^{-2} h_{ij}$ on the interior $\mathring{\mathcal{S}}$. Then there exists a unique uniformly bounded, uniformly bounded away from zero, locally $\mathcal{C}^2$ function $\theta$ solution to
        \begin{subequations}
            \begin{alignat}{4}
                \label{eq:eqdiff_theta}
                \Delta_{\overline{h}} \theta - \frac{\overline{r}}{8} \theta - \frac{3}{4} \theta^5 &= -\frac{1}{8} \theta^{-7} \left|x^3 \phi\right|_h^2 \quad && \text{on } \mathring{\mathcal{S}} \,, \\
                \label{eq:bc_theta}
                \theta &= \left|dx\right|^{1/2}_h && \text{on } \partial\mathcal{S} \,.
            \end{alignat}
        \end{subequations}
        Even though $\Psi := x\theta^{-2}$ is polyhomogeneous and of class $\mathcal{C}^3$ on $\mathcal{S}$ for generic free data pairs, there exists large sets of non-generic free data pairs for which it is smooth up to the boundary. In that case, the fields
        \begin{subequations}
        \label{eq:physical_fields_conf_method}
        \begin{align}
            \widetilde{h}_{ij} &:= \theta^4 \, \overline{h}_{ij} = \Psi^{-2} h_{ij} \,, \\
            \widetilde{K}_{ij} &:= \theta^{-2} \left(x\phi_{ij}\right) = \Psi \phi_{ij} \,,
        \end{align}  
        \end{subequations}
        satisfy the conditions $i)$-$iii)a)$ of the geometric initial data problem as defined in \Cref{sec:formulation_gidp}.

        \item[ii)] Construction of TT-tensors: Let $h_{ij}$ be a smooth Riemannian metric on $\mathcal{S}$ and let $x$ a boundary defining function of $\partial\mathcal{S}$. For any smooth symmetric 2-tensor $A_{ij}$ on $\mathcal{S}$ such that
        \[ h^{ij} A_{ij} = 0 \quad \text{on } \mathcal{S} \,, \qquad A_{\perp i} = 0 \quad \text{on } \partial\mathcal{S} \,, \]
        there exists a polyhomogeneous vector field $X^i$ on $\mathcal{S}$ solution to
        \[ D_i \left(\mathcal{K}_h X\right)^{ij} = - D_i \left( x^{-2} A^{ij} \right) \quad \text{on } \mathcal{S} \,. \]
        The vector field $X^i$ is unique up to smooth solutions of the homogeneous equation. Then the field
        \[ \phi_{ij} := x^{-2} A_{ij} + \left(\mathcal{K}_h X\right)_{ij} \]
        is a TT-tensor on $\mathring{\mathcal{S}}$ for the metric $h_{ij}$ such that $x^2\phi_{ij}$ is polyhomogeneous, of class $\mathcal{C}^1$ on $\mathcal{S}$ and verifies $x^2\phi_{\perp i} = 0$ on $\partial\mathcal{S}$. Moreover, if $X^i$ is smooth on $\mathcal{S}$ then so is $x^2\phi_{ij}$. In particular, this is the case if $x^{-2} A_{ij}$ is smooth on $\mathcal{S}$.
    \end{itemize}
\end{thm}

\subsubsection{On the unphysical fields}
\label{sec:unphysical_fields}

Let us now focus on the condition $iii)$ of the geometric initial data problem defined in \Cref{sec:formulation_gidp}. For fields arising from the conformal method as described in \Cref{thm:conformal_method}, the following theorem ensures that $iii)$ holds under some hypotheses.

\begin{thm}[K\'ann\'ar \protect{\cite[Theorem 2]{K96}}]
    \label{thm:Kannar}
    Let $\mathcal{S}$ be a (3-dimensional compact) smooth manifold (with non-empty boundary) diffeomorphic to the closed unit ball of $\R^3$. Let $h_{ij}$ be a smooth Riemannian metric on $\mathcal{S}$ and let $S_{ij}$ be a smooth trace-free symmetric 2-tensor on $\mathcal{S}$. Take $x$ a boundary defining function of $\partial\mathcal{S}$ such that
    \begin{itemize}
        \item $\left|dx\right|_h = 1$ in a neighbourhood of $\partial\mathcal{S}$,
        \item $\overline{r} = -6 + \grando{x^3}$ where $\overline{r}$ is the scalar curvature of $\overline{h}_{ij} := x^{-2} h_{ij}$.
    \end{itemize}
    If the restriction of $S_{ij}$ in the tangential directions to $\partial\mathcal{S}$ is pure-trace, that is
    \begin{equation}
        \label{eq:Sdagger}
        S^\dagger_{AB} = 0 \quad \text{on } \partial\mathcal{S} \,,
    \end{equation}
    and if $\partial\mathcal{S}$ is totally geodesic for the Levi-Civita connection $D$ of $h$, that is
    \begin{equation}
        \label{eq:totally_geod}
        \slashed{\mathfrak{K}}_{AB} = 0 \,,
    \end{equation}
    then
    \begin{itemize}
        \item[i)] the TT-tensor $\phi_{ij}$, constructed by $ii)$ of \Cref{thm:conformal_method} out of $h_{ij}$ and $x$ and $A_{ij} := x S_{ij}$, is such that $x\phi_{ij}$ is smooth on $\mathcal{S}$,
        \item[ii)] the conformal factor $\Psi := x\theta^{-2}$, constructed by $i)$ of \Cref{thm:conformal_method} out of $(h_{ij},\phi_{ij})$ and $x$, is smooth on $\mathcal{S}$,
        \item[iii)] the consequent fields $\widetilde{h}_{ij}$, $\widetilde{K}_{ij}$ defined by \eqref{eq:physical_fields_conf_method} satisfy $iii)b)$ of the geometric initial data problem, see \Cref{sec:formulation_gidp}, in addition of $i)$-$iii)a)$ already ensured by \Cref{thm:conformal_method}.
    \end{itemize}
\end{thm}

\begin{rem}
    The choice of $x$ in \Cref{thm:Kannar} is more restrictive than the gauge in the original version of the theorem \cite[Theorem 2]{K96}. However, this is actually the choice used in K\'ann\'ar's proof. This makes the comparison with \Cref{thm:smooth_unphysical_fields} easier.
\end{rem}

\Cref{thm:Kannar} is quite restrictive for two reasons: the topology of the manifold is completely imposed and the field $\widetilde{K}_{ij}$ is more regular than expected since it is smooth up to the boundary. In what follows, we derive the necessary and sufficient conditions in the general case for a solution to $i)$-$ii)$ of the geometric initial data problem to satisfy $iii)$. Let us start with a technical lemma.

\begin{lem}
    \label{lem:momentum_asymp}
    Let $(\mathcal{S},\slashed{\mathfrak{I}},\widetilde{h})$ be a 3-dimensional asymptotically hyperbolic space with smooth rescaled metrics and let $\widetilde{K}_{ij}$ be a smooth symmetric 2-tensor field on $\mathcal{S}\setminus\slashed{\mathfrak{I}}$ verifying the momentum constraint
    \[ \widetilde{D}_i \widetilde{K}^i{}_j - \widetilde{D}_j \widetilde{K}^i{}_i = 0 \quad \text{on } \mathcal{S}\setminus\slashed{\mathfrak{I}} \,. \]
    For any boundary defining function $\Psi$ of $\slashed{\mathfrak{I}}$, if the field $P_{ij} := \Psi \widetilde{K}_{ij}$ extends smoothly to $\slashed{\mathfrak{I}}$ then it verifies
    \begin{subequations}
        \begin{alignat}{9}
            \label{eq:asymptotic_Paperp}
            \Psi^{-1} P_{A\perp} &= \slashed{\mathfrak{D}}_A P_B{}^B - \slashed{\mathfrak{D}}_B P^B{}_A &\quad \text{on } \slashed{\mathfrak{I}} \,, \\
            \label{eq:asymptotic_Pperpperp}
            \Psi^{-1} P_{\perp\perp} &= e_\perp \left( P \right) - \slashed{\mathfrak{K}}^{AB} P_{AB} &\quad \text{on } \slashed{\mathfrak{I}} \,,
        \end{alignat}
    \end{subequations}
    where
    \begin{itemize}
        \item[$\bullet$] $\slashed{\mathfrak{h}}$ is the metric induced on (a connected component of) $\slashed{\mathfrak{I}}$ by the rescaled metric $h := \Psi^2 \widetilde{h}$, $\slashed{\mathfrak{D}}$ is the Levi-Civita connection of $\slashed{\mathfrak{h}}$, $\slashed{\mathfrak{K}}$ is the extrinsic curvature on (a connected component of) $\slashed{\mathfrak{I}}$ of the Levi-Civita connection $D$ of $h$,
        \item[$\bullet$] $\perp$ denotes the outward-pointing normal of $\slashed{\mathfrak{I}}$ with respect to $h$ while $A,B$ are tangential indices,
        \item[$\bullet$] $P := h^{ij} P_{ij}$ is the trace of $P_{ij}$ with respect to $h$.
    \end{itemize}
\end{lem}

\begin{rems} \,
    \begin{itemize}
        \item If $P_{ij} := \Psi \widetilde{K}_{ij}$ extends smoothly to $\slashed{\mathfrak{I}}$ then the trace $\widetilde{K} := \widetilde{h}^{ij} \widetilde{K}_{ij}$ can be smoothly rescaled by $\Psi^{-1}$. Indeed, one has $\Psi^{-1} \widetilde{K} = P \in \mathcal{C}^\infty(\mathcal{S},\R)$.

        \item Equations \eqref{eq:asymptotic_Paperp}-\eqref{eq:asymptotic_Pperpperp} imply that $P_{i\perp} = 0$ on $\slashed{\mathfrak{I}}$. Since $P_{ij}$ is smooth up to the conformal boundary and by application of \Cref{lem:dividing_bdf}, $\Psi^{-1} P_{i\perp}$ are well defined on $\slashed{\mathfrak{I}}$, their value being given by the equations. \qedhere
    \end{itemize}
\end{rems}

\begin{proof}
    By definition of $P_{ij}$ and \Cref{prop_transition_weyl_conn}, one has
    \[ 0 = \widetilde{D}_i \widetilde{K}^i{}_j - \widetilde{D}_j \widetilde{K}^i{}_i = \Psi \left( D_i P^i{}_j - D_j P \right) - 2 P_{jk} h^{kl} (d\Psi)_l \quad \text{on } \mathcal{S}\setminus\slashed{\mathfrak{I}} \,. \]
    Since $P_{ij}$ and $h_{ij}$ are smooth on the whole of $\mathcal{S}$, this equation extends up to $\slashed{\mathfrak{I}}$. Take a frame field $(e_{\bf i}) = (e_\perp,(e_{\bf A}))$ adapted to the conformal boundary $\slashed{\mathfrak{I}}$ (see \Cref{sec:adapted_frame_fields}) with $e_\perp$ the outward-pointing $h$-unit normal. Evaluating the above equation on $\slashed{\mathfrak{I}}$ yields
    \begin{equation}
        \label{eq:Piperp_vanish}
        P_{\bf i\perp} = 0 \quad \text{on } \slashed{\mathfrak{I}} \,.
    \end{equation}
    Now, one can first divide the equation by $\Psi$ before evaluating on $\slashed{\mathfrak{I}}$. This gives
    \[ 2 \Psi^{-1} P_{\bf j\perp} = e_{\bf j}(P) - D_{\bf i} P^{\bf i}{}_{\bf j} \quad \text{on } \slashed{\mathfrak{I}} \,. \]
    Using \eqref{eq:Piperp_vanish}, one has
    \begin{align*}
        e_{\bf A}(P) &= e_{\bf A}(P_{\bf B}{}^{\bf B}) = \slashed{\mathfrak{D}}_{\bf A} P_{\bf B}{}^{\bf B} \,, \\
        D_{\bf i} P^{\bf i}{}_{\bf A} &= D_\perp P_{\perp \bf A} + D_{\bf B} P^{\bf B}{}_{\bf A} = - \Psi^{-1} P_{\bf A\perp} + \slashed{\mathfrak{D}}_{\bf B} P^{\bf B}{}_{\bf A} \,, \\
        D_{\bf i} P^{\bf i}{}_\perp &= D_\perp P_{\perp\perp} + D_{\bf B} P^{\bf B}{}_\perp = - \Psi^{-1} P_{\perp\perp} - \slashed{\mathfrak{K}}_{\bf B}{}^{\bf A} P^{\bf B}{}_{\bf A} \,.
    \end{align*}
    Hence equations \eqref{eq:asymptotic_Paperp}-\eqref{eq:asymptotic_Pperpperp}.
\end{proof}

\begin{thm}
    \label{thm:smooth_unphysical_fields}
    Let $(\mathcal{S},\slashed{\mathfrak{I}},\widetilde{h})$ be a 3-dimensional aH space with smooth rescaled metrics and let $\widetilde{K}_{ij}$ be a smooth symmetric 2-tensor field on $\mathcal{S}\setminus\slashed{\mathfrak{I}}$ such that $(\mathcal{S}\setminus\slashed{\mathfrak{I}},\widetilde{h},\widetilde{K})$ is a solution of the \eqref{eq:VC} for $\Lambda = -3$. Then the following two statements are equivalent
    \begin{itemize}
        \item[i)] for any (and thus all) boundary defining function $\Psi$ of $\slashed{\mathfrak{I}}$
        \begin{itemize}
            \item[a)] the field $P_{ij} := \Psi \widetilde{K}_{ij}$ extends smoothly to $\slashed{\mathfrak{I}}$ and verifies
            \begin{equation}
                \label{eq:normal_derivative_Pdagger}
                \left( D_\perp P\right)^\dagger_{AB} = s_h P_{AB}^\dagger \quad \text{on } \slashed{\mathfrak{I}} \,,
            \end{equation}
            \item[b)] all connected components $\mathfrak{S}$ of $\slashed{\mathfrak{I}}$ are umbilical with
            \begin{equation}
                \label{eq:umbilical_condition}
                \slashed{\mathfrak{K}}_{AB} = s_h \;  \slashed{\mathfrak{h}}_{AB} \,,
            \end{equation}
        \end{itemize}
        where
        \begin{itemize}
            \item[$\bullet$] $D$ is the Levi-Civita connection of the rescaled metric $h := \Psi^2 \widetilde{h}$, $s_h$ is the non-extended 3-dimensional Friedrich scalar field of $h$ defined by \eqref{eq:def_Friedrich_scalar_dim3_nonextended},
            \item[$\bullet$] $\slashed{\mathfrak{h}}$ is the metric induced by $h$ on $\mathfrak{S}$, $\slashed{\mathfrak{D}}$ is the Levi-Civita connection of $\slashed{\mathfrak{h}}$, $\slashed{\mathfrak{K}}$ is the extrinsic curvature of $D$ on $\mathfrak{S}$,
            \item[$\bullet$] $\perp$ denotes the outward-pointing normal of $\slashed{\mathfrak{I}}$ with respect to $h$ while $A,B$ are tangential indices,
            \item[$\bullet$] $X^\dagger_{AB}$ is the trace-free part of the restriction of a 2-tensor $X_{ij}$ on $\slashed{\mathfrak{I}}$,
        \end{itemize}
        
        \item[ii)] for any boundary defining function $\Psi$ of $\slashed{\mathfrak{I}}$, any smooth covector field $\kappa$ on $\mathcal{S}$ and any smooth scalar field $z$ on $\mathcal{S}$ which vanishes on $\slashed{\mathfrak{I}}$, one has
        \begin{itemize}
            \item[a)] the fields $K_{ij}$, $Q_{ij}$, $E_{ij}$, $T_i$, $M_{ijk}$ defined by \eqref{CVC_def} extend smoothly to $\slashed{\mathfrak{I}}$,
            \item[b)] the 10-tuple $(\mathcal{S},h,E,M,\Psi,\kappa,K,z,T,Q)$ verifies the geometric conditions \eqref{eq:gauge_independent_conditions_CVC}.
        \end{itemize}
    \end{itemize}
\end{thm}

\begin{rem}
    From the transformation rules \eqref{eq:transfo_s_dim3} and these of the extrinsic curvature described in \Cref{sec:extrinsic_curvature}, one deduces that if $\eqref{eq:umbilical_condition}$ holds for one rescaled metric then for all rescaled metric $h$ and for all smooth covector field $\kappa$
    \[ (\widehat{\slashed{\mathfrak{K}}}_h)_{AB} = \widehat{s}_h  \;  \slashed{\mathfrak{h}}_{AB} \,, \]
    where $\widehat{D}$ is the Weyl connection associated to $\kappa$ with respect to $h$, $\widehat{\slashed{\mathfrak{K}}}_h$ is the extrinsic curvature of $\widehat{D}$ with respect to $h$ on $\mathcal{C}$, $\widehat{s}_h$ is the extended 3-dimensional Friedrich scalar field defined by \eqref{eq:Friedrich_scalar_dim3} and $\slashed{\mathfrak{h}}$ is the metric induced by $h$ on $\mathcal{C}$.
\end{rem} 

\begin{cor}
    \label{cor:smooth_unphysical_fields}
    Consider a 3-dimensional compact smooth manifold $\mathcal{S}$ with non-empty boundary $\partial\mathcal{S}$. Let $\widetilde{h}_{ij}$ be a smooth Riemannian metric on $\mathring{\mathcal{S}}$ and let $\widetilde{K}_{ij}$ be a smooth symmetric 2-tensor on $\mathring{\mathcal{S}}$ such that $i)$ and $ii)$ of the geometric initial data problem are verified, see \Cref{sec:formulation_gidp}. Then $iii)$ of the geometric initial data problem holds if and only if for any (and thus all) boundary defining function $\Psi$ of $\partial\mathcal{S}$, the conditions described in $i)$ of \Cref{thm:smooth_unphysical_fields} hold.
\end{cor}

\begin{rems} \,
    \begin{itemize}
        \item \Cref{cor:smooth_unphysical_fields} implies \Cref{thm:Kannar}. From $|dx|^2_h=1$ on a neighbourhood of $\partial\mathcal{S}$ and \eqref{eq:bc_theta}, one has $\theta=1+\grando{x^k}$ for some integer $k\geq1$. Let us analyse the terms in \eqref{eq:eqdiff_theta}:
        \begin{align*}
            \Delta_{\overline{h}} \theta &= x^2 \left( \Delta_h \theta - h^{-1}(dx,d\theta) x^{-1} \right) = \grando{x^k} \,, \\
            -\overline{r}\theta -6 \theta^5 &=(6+\grando{x^3})(1+\grando{x^k}) - 6(1+\grando{x^k}) = \grando{x^{\min(3,k)}} \,, \\
            \theta^{-7} \left|x^3\phi\right|^2_h &= \theta^{-7} x^4 \left|x\phi\right|^2_h = \grando{x^4} \,.
        \end{align*}
        Hence $k=3$. It follows that the function $\Psi:=x\theta^{-2}$ verifies $\left|d\Psi\right|^2_h - 1 = \grando{x^3} = \grando{\Psi^3}$. Then \eqref{eq:eqdiff_Psi} gives
        \[ s_h = \frac{1}{2\Psi} \left( \left|d\Psi\right|^2_h - 1\right) + \frac{\Psi^3}{12} \theta^{-4} \left| S \right|^2_h = \grando{\Psi^2} \,. \]
        In particular, the gauge choice of $x$ in \Cref{thm:Kannar} implies that $s_h = 0$ on $\partial\mathcal{S}$. Finally, one deduces that \eqref{eq:normal_derivative_Pdagger} and \eqref{eq:umbilical_condition} reduce to \eqref{eq:Sdagger} and \eqref{eq:totally_geod} respectively.

        \item The hypotheses in $i)$ of \Cref{thm:smooth_unphysical_fields} are not sufficient to prevent logarithmic terms in solutions constructed by the conformal method as described in \Cref{thm:conformal_method}. A detailed analysis of the logarithmic terms is given in \cite[Section 4]{AC94}. In particular, the solutions constructed by the conformal method are smooth on $\mathcal{S}$ if and only if certain of the first logarithmic terms, given by \cite[Equations (4.13), (4.14) and (4.28)]{AC94} for the choice of $x$ as in \Cref{thm:Kannar}, vanish on the boundary $\partial\mathcal{S}$. From this, one can deduce the minimal hypotheses on the field $A_{ij}$ so that the solutions are smooth up to the boundary under the hypotheses $i)$ of \Cref{thm:smooth_unphysical_fields}. \qedhere
    \end{itemize}
\end{rems}

\begin{proof}[Proof of \Cref{thm:smooth_unphysical_fields}] \, \\
    $\Longrightarrow$ Assume that $i)$ is verified. Let $\Psi$ be a boundary defining function of $\slashed{\mathfrak{I}}$, let $\kappa$ be a smooth covector field on $\mathcal{S}$ and let $z$ be a smooth scalar field on $\mathcal{S}$ which vanishes on $\slashed{\mathfrak{I}}$. 
    \begin{itemize}
        \item The field $K_{ij}$ defined by \eqref{CVC_def_K} and the geometric conditions \eqref{eq:gauge_independent_conditions_CVC}: The condition \eqref{eq:z_vanish} holds by hypothesis on the field $z$. Moreover, the function $z/\Psi$ extends smoothly on $\slashed{\mathfrak{I}}$ by \Cref{lem:dividing_bdf}. By assumption, the field $P_{ij} := \Psi \widetilde{K}_{ij}$ extends smoothly on $\slashed{\mathfrak{I}}$ and the rescaled metric $h_{ij} := \Psi^2 \widetilde{h}_{ij}$ is smooth on the whole of $\mathcal{S}$. Therefore, the field $K_{ij}$ defined by \eqref{CVC_def_K} extends smoothly on $\slashed{\mathfrak{I}}$. Furthermore, one has $P_{i\perp} = 0$ on $\slashed{\mathfrak{I}}$ by \Cref{lem:momentum_asymp}, which is equivalent to \eqref{eq:K_vanish} by \eqref{CVC_def_K}.
    \end{itemize}
    To prove the smoothness of the other fields in \eqref{CVC_def}, we will work in a frame field $(e_{\bf i}) = (e_\perp,(e_{\bf A}))$ adapted to the conformal boundary $\slashed{\mathfrak{I}}$ (see \Cref{sec:adapted_frame_fields}) with $e_\perp$ the outward-pointing $h$-unit normal. Since $(\mathcal{S},\slashed{\mathfrak{I}},\widetilde{h})$ is an aH space and $\Psi$ is a boundary defining function of $\slashed{\mathfrak{I}}$, one has $(d\Psi)_\perp = -1$ on $\slashed{\mathfrak{I}}$. The idea is to apply \Cref{lem:dividing_bdf} for each field in a specific order and to decompose the arising conditions in the above frame field. The number of such conditions can be reduced by exploiting the algebraic symmetries.
    \begin{itemize}
        \item The field $Q_{ij}$ defined by \eqref{CVC_def_Q}: since $K_{ij}$, $h_{ij}$, $\kappa_i$, $z$ and $\Psi$ are all smooth on $\mathcal{S}$, the smoothness of $Q_{ij}$ is equivalent by \Cref{lem:dividing_bdf} to
        \[ -\widehat{D}_i (\widehat{\zeta}_h)_j - \kappa_i (\widehat{\zeta}_h)_j + z K_{ij} + \widehat{\sigma} h_{ij} = 0 \quad \text{on } \slashed{\mathfrak{I}} \,. \]
        Note that the antisymmetric part of the tensor field on the left hand side is equal to $-\Psi \widehat{D}_{[i} \kappa_{j]}$. Thus, it vanishes on $\slashed{\mathfrak{I}}$. Therefore the above condition is equivalent to
        \begin{alignat*}{12}
            -\widehat{D}_{\bf A} (\widehat{\zeta}_h)_{\bf B} - \kappa_{\bf A} (\widehat{\zeta}_h)_{\bf B} + \widehat{s}_h h_{\bf AB} &= 0 \quad \text{on } \slashed{\mathfrak{I}}  &\; \iff \; && (\widehat{\slashed{\mathfrak{K}}}_h)_{\bf AB} &= \widehat{s}_h \; \slashed{\mathfrak{h}}_{\bf AB} \,, \\
            -\widehat{D}_{\bf A} (\widehat{\zeta}_h)_\perp - \kappa_{\bf A} (\widehat{\zeta}_h)_\perp + \widehat{s}_h h_{\bf A\perp} &= 0 \quad \text{on } \slashed{\mathfrak{I}} &\; \iff \; && 0 &= 0 \,, \\
            -\widehat{D}_\perp (\widehat{\zeta}_h)_\perp - \kappa_\perp (\widehat{\zeta}_h)_\perp + \widehat{s}_h h_{\perp\perp} &= 0 \quad \text{on } \slashed{\mathfrak{I}} &\; \iff \; && e_\perp((d\Psi)_\perp) &= \widehat{s}_h \quad \text{on } \slashed{\mathfrak{I}} \,.
        \end{alignat*}
        The first condition is ensured by \eqref{eq:umbilical_condition} thanks to the remark below the theorem. The second condition is trivial. Finally, note that equation \eqref{eq:CVC_k} now holds on the whole of $\mathcal{S}$ by continuity. Differentiating it with respect to $e_\perp$ and evaluating it on $\slashed{\mathfrak{I}}$ gives the third condition.

        \item The field $T_i$ defined by \eqref{CVC_def_T}: since $K_{ij}$, $h_{ij}$, $\kappa_i$, $z$ and $\Psi$ are all smooth on $\mathcal{S}$, the smoothness of $T_i$ is equivalent by \Cref{lem:dividing_bdf} to
        \[ -(dz)_i + K_{ij} h^{jk} (\widehat{\zeta}_h)_k = 0 \quad \text{on } \slashed{\mathfrak{I}} \,. \]
        Using the frame field, this can be rewritten under the form
        \[ K_{\bf i\perp} - \frac{(dz)_\perp}{(d\Psi)_\perp} h_{\bf i\perp} = 0 \quad \text{on } \slashed{\mathfrak{I}} \,, \]
        which is verified by \eqref{eq:K_vanish} and L'Hôpital's rule.

        \item The field $M_{ijk}$ defined by \eqref{CVC_def_M}: since $K_{ij}$, $h_{ij}$ and $\kappa_i$ are all smooth on $\mathcal{S}$, the smoothness of $M_{ijk}$ is equivalent by \Cref{lem:dividing_bdf} to
        \begin{equation}
            \label{eq:aux_M}
            G_{ijk}(K,\kappa) + h_{k[i} h^{lm} G_{j]lm}(K,\kappa) = 0 \quad \text{on } \slashed{\mathfrak{I}} \,.
        \end{equation}
        Note that \eqref{eq:CVC_d} now holds on the whole of $\mathcal{S}$ by continuity. Hence the condition rewrites as
        \begin{equation}
            \label{eq:aux_cond_M}
            G_{ijk}(K,\kappa) - 2 h_{k[i} T_{j]} = 0 \quad \text{on } \slashed{\mathfrak{I}} \,. 
        \end{equation}
        By \eqref{CVC_def_K} and \eqref{eq:def_F_G}, one has on $\mathcal{S}$
        \[ G_{ijk}(K,\kappa) = G_{ijk}(P,\kappa) - 2 h_{k[i} \left( \left( d\frac{z}{\Psi} \right) \vphantom{\kappa}_{j]} - \kappa_{j]} \frac{z}{\Psi} \right) \,. \]
        Using \Cref{prop_transition_weyl_conn}, one has on $\mathcal{S}$
        \[ G_{ijk}(P,\kappa) = G_{ijk}(P,0) + 2h_{k[i} P_{j]l} h^{lm} \kappa_m = 2 D_{[i} P_{j]k} + 2h_{k[i} \widetilde{K}_{j]l} h^{lm} \Psi \kappa_m \,. \]
        Moreover, from \eqref{CVC_def_T} and \eqref{CVC_def_K}, one has on $\mathcal{S}$
        \begin{align*}
            T_i &= - \left(d\frac{z}{\Psi}\right)_i + \kappa_i \frac{z}{\Psi} + \frac{1}{\Psi} \left( K_i{}^j - \frac{z}{\Psi} \delta_i{}^j \right) (\widehat{\zeta}_h)_j \\
            &= - \left(d\frac{z}{\Psi}\right)_i + \kappa_i \frac{z}{\Psi} + \widetilde{K}_{ik} h^{jk} (\widehat{\zeta}_h)_j \,.
        \end{align*}
        Hence, \eqref{eq:aux_cond_M} resumes to
        \begin{equation}
            \label{eq:aux_cond_M_2}
        	2 D_{[i} P_{j]k} - 2 h_{k[i} \widetilde{K}_{j]l} h^{lm} (d\Psi)_m = 0 \quad \text{on } \slashed{\mathfrak{I}} \,.
        \end{equation}
        It is clear from \eqref{eq:aux_M} that the tensor field on the left hand side is a Cotton candidate on $(\mathcal{S},h)$, see \Cref{def:cotton_candidate}. Thus \eqref{eq:aux_cond_M_2} is equivalent to
        \begin{subequations}
            \begin{align}
            \label{eq:aux_a}
            2 D_{\bf [A} P_{\bf B]C} + 2 h_{\bf C[A} \widetilde{K}_{\bf B]\perp} = 0 \quad \text{on } \slashed{\mathfrak{I}} \,, \\
            \label{eq:aux_b}
            2 D_{\bf [\perp} P_{\bf A]B} + 2 h_{\bf B[\perp} \widetilde{K}_{\bf A]\perp} = 0 \quad \text{on } \slashed{\mathfrak{I}} \,.
            \end{align}
        \end{subequations}
        First consider \eqref{eq:aux_a}. By counting algebraic freedoms, it is equivalent to its trace. Furthermore, since $P_{\bf i\perp} = 0$ on $\slashed{\mathfrak{I}}$, one deduces that $D_{\bf A} P_{\bf BC} = \slashed{\mathfrak{D}}_{\bf A} P_{\bf BC}$. Hence, \eqref{eq:aux_a} is equivalent to
        \[ \widetilde{K}_{\bf A\perp} = 2 h^{\bf BC} \slashed{\mathfrak{D}}_{\bf [A} P_{\bf B]C} = \slashed{\mathfrak{D}}_{\bf A} P_{\bf B}{}^{\bf B} - \slashed{\mathfrak{D}}_{\bf B} P^{\bf B}{}_{\bf A} \quad \text{on } \slashed{\mathfrak{I}} \,. \]
        This is equation \eqref{eq:asymptotic_Paperp} which holds by \Cref{lem:momentum_asymp}.
        
        Now consider \eqref{eq:aux_b}. Using $P_{\bf i\perp} = 0$ on $\slashed{\mathfrak{I}}$ and \eqref{eq:umbilical_condition}, it is equivalent to
        \[ \widetilde{K}_{\perp\perp} h_{\bf AB} = D_\perp P_{\bf AB} - s_h P_{\bf AB} \quad \text{on } \slashed{\mathfrak{I}} \,. \]
        Taking the trace gives
        \[ 2 \widetilde{K}_{\perp\perp} = e_{\perp}\left(h^{\bf AB} P_{\bf AB}\right) -s_h \, h^{\bf AB} P_{\bf AB} \,. \]
        Using again $P_{\bf i\perp} = 0$ on $\slashed{\mathfrak{I}}$, one has
        \[ h^{\bf AB} P_{\bf AB} = P \quad \text{on } \slashed{\mathfrak{I}} \,, \]
        and
        \[ e_\perp(h^{\bf AB} P_{\bf AB}) = e_\perp(P) - h^{\bf A\perp} e_\perp(P_{\bf A\perp}) - h^{\perp\perp} e_{\perp}(P_{\perp\perp}) - P_{\bf i\perp} e_\perp(h^{\bf i\perp}) = e_\perp(P) + \widetilde{K}_{\perp\perp} \quad \text{on } \slashed{\mathfrak{I}} \,. \]
        Hence, one retrieves \eqref{eq:asymptotic_Pperpperp} which is again verified by \Cref{lem:momentum_asymp}. Taking the trace-free part of \eqref{eq:aux_b} gives \eqref{eq:normal_derivative_Pdagger} ensured by $i)a)$.

        \item The field $E_{ij}$ defined by \eqref{CVC_def_E}: since $K_{ij}$, $h_{ij}$, $\kappa_i$ and $Q_{ij}$ are all smooth on $\mathcal{S}$, the smoothness of $E_{ij}$ is equivalent by \Cref{lem:dividing_bdf} to
        \[ X_{ij} := Q_{ij} - \widehat{l}_{ij} - F_{ij}(K) = 0 \quad \text{on } \slashed{\mathfrak{I}} \,. \]
        Note that the tensor field $X_{ij}$ is symmetric and trace-free on $(\mathcal{S},h)$. By multiplying \eqref{eq:CVC_f}, which holds on $\mathring{\mathcal{S}}$, by the boundary defining function $\Psi$, one has
        \[ 2 (\widehat{\zeta}_h)_l S_{k[i}{}^{lm} X_{j]m} = -\Psi \left( 2\widehat{D}_{[i} Q_{j]k} - 2 K_{k[i} T_{j]} - zM_{ijk} \right) \quad \text{on } \mathring{\mathcal{S}} \,.\]
        Since all the fields appearing in the above equation now extends smoothly, it also holds true on $\slashed{\mathfrak{I}}$ giving
        \[ 2 (\widehat{\zeta}_h)_l S_{k[i}{}^{lm} X_{j]m} = 2 (\widehat{\zeta}_h)_{[i} X_{j]k} - 2 h_{k[i} X_{j]m} h^{lm} (\widehat{\zeta}_h)_l = 0 \quad \text{on } \slashed{\mathfrak{I}} \,. \]
        In particular,
        \begin{alignat*}{9}
            2 (\widehat{\zeta}_h)_{[\perp} X_{\bf A]B} - 2 h_{\bf B[\perp} X_{\bf A]m} h^{\bf lm} (\widehat{\zeta}_h)_{\bf l} &= 0 \quad \text{on } \slashed{\mathfrak{I}} &\; \iff \;&& X_{\bf AB} = 0 \quad \text{on } \slashed{\mathfrak{I}} \,, \\
            2 (\widehat{\zeta}_h)_{\bf [A} X_{\bf B]C} - 2 h_{\bf C[A} X_{\bf B]m} h^{\bf lm} (\widehat{\zeta}_h)_{\bf l} &= 0 \quad \text{on } \slashed{\mathfrak{I}} &\; \iff \;&& X_{\bf A\perp} = 0 \quad \text{on } \slashed{\mathfrak{I}} \,.
        \end{alignat*}
        Therefore, one has $X_{ij} = 0$ on $\slashed{\mathfrak{I}}$.
    \end{itemize}

    \bigbreak

    \noindent $\Longleftarrow$ Assume that $ii)$ is verified. Let $\Psi$ be a boundary defining function of $\slashed{\mathfrak{I}}$ and take $\kappa=0$, $z=0$. Equation \eqref{CVC_def_K} gives that
    \[ P_{ij} := \Psi \widetilde{K}_{ij} = K_{ij} \qquad \text{on } \mathcal{S}\setminus\slashed{\mathfrak{I}} \,. \]
    Since the field $K_{ij}$ extends smoothly on $\slashed{\mathfrak{I}}$ by $ii)a)$, so does the field $P_{ij}$. The condition \eqref{eq:K_vanish}, ensured by $ii)b)$, implies that $P_{i\perp} = 0$ on $\slashed{\mathfrak{I}}$.

    Using $ii)a)$ and \Cref{lem:dividing_bdf}, one deduces from \eqref{CVC_def_Q} and \eqref{CVC_def_M} that
    \begin{align*}
        D_i (d\Psi)_j - s_h h_{ij} &= 0 \quad \text{on } \slashed{\mathfrak{I}} \,, \\
        G_{ijk}(K,0) + h_{k[i} h^{lm} G_{j]lm}(K,0) &= 0 \quad \text{on } \slashed{\mathfrak{I}} \,.
    \end{align*}
    In the same manner than above, the first equation implies in particular \eqref{eq:umbilical_condition} and then the second equation implies in particular \eqref{eq:normal_derivative_Pdagger}. Hence the result.
\end{proof}

\subsection{Main theorem}
\label{sec:main_thm}

\subsubsection{Precise version of the main theorem}

The precise version of our main result on the geometric local existence and uniqueness for the vacuum Einstein equations \eqref{eq:VE} with $\Lambda<0$, $n=3$ and our new geometric boundary conditions \eqref{eq:bc_robin} write as follows.

\begin{thm}[Main theorem, precise version]
    \label{thm:main_precise}
    Let
    \begin{itemize}
    	\item $\mathcal{S}$ be a 3-dimensional oriented and compact smooth manifold with non-empty boundary $\partial\mathcal{S}$,
    	\item $\widetilde{h}$ be a smooth Riemannian metric on the interior $\mathring{\mathcal{S}}$,
    	\item $\widetilde{K}_{ij}$ be a smooth symmetric 2-tensor on the interior $\mathring{\mathcal{S}}$,
    	\item for each connected component $\mathcal{C}$ of $\partial\mathcal{S}$, $\mu_\mathcal{C}$ be a real number,
    \end{itemize}
    such that
    \begin{itemize}
        \item[H1.] $(\mathcal{S},\partial\mathcal{S},\widetilde{h})$ is asymptotically hyperbolic, in the sense of \Cref{def:aH}, with smooth rescaled metrics,
        
        \item[H2.] $(\mathring{\mathcal{S}},\widetilde{h}_{ij},\widetilde{K}_{ij})$ is a solution to the vacuum constraint equations for the normalised negative cosmological constant $\Lambda = -3$, that is
        \begin{alignat*}{4}
            \widetilde{r} + (\widetilde{K}_i{}^i)^2 - \widetilde{K}_{ij} \widetilde{K}^{ij} &= -6  && \quad \text{on } \mathring{\mathcal{S}} \,, \\
            \widetilde{D}_i \widetilde{K}^i{}_j - \widetilde{D}_j \widetilde{K}_i{}^i &= 0  && \quad \text{on } \mathring{\mathcal{S}} \,,
        \end{alignat*}
        where $\widetilde{D}$ is the Levi-Civita connection of $\widetilde{h}$ and $\widetilde{r}$ is its scalar curvature,
        
        \item[H3.] for any boundary defining function $\Psi$ of $\partial\mathcal{S}$, the field $P_{ij} := \Psi \widetilde{K}_{ij}$ extends smoothly to $\partial\mathcal{S}$ and, for all connected components $\mathcal{C}$ of $\partial\mathcal{S}$, one has
        \begin{subequations}
            \begin{alignat*}{4}
                \slashed{\mathfrak{K}}_{AB} &= s_h \, \slashed{\mathfrak{h}}_{AB} && \quad \text{on } \mathcal{C} \,, \\
                \left( D_\perp P\right)^\dagger_{AB} &= s_h \, P_{AB}^\dagger && \quad \text{on } \mathcal{C} \,,
            \end{alignat*}
        \end{subequations}
        where $h := \Psi^2 \widetilde{h}$ is the rescaled metric associated to $\Psi$, $\slashed{\mathfrak{K}}$ is the extrinsic curvature on $\mathcal{C}$ of the Levi-Civita connection $D$ of $h$, $s_h$ is the non-extended 3-dimensional Friedrich scalar defined by \eqref{eq:def_Friedrich_scalar_dim3_nonextended}, $\slashed{\mathfrak{h}}$ is the metric induced by $h$ on $\mathcal{C}$ and $X{}^\dagger_{AB}$ denotes the trace-free part of the restriction of a 2-tensor $X_{ij}$ on $\mathcal{C}$,
 
        \item[H4.] on each connected component $\mathcal{C}$ of $\partial\mathcal{S}$, the enlarged set of compatibility conditions depending on $\mu_\mathcal{C}$, as defined in \Cref{lem:geometric_compatibility_conditions}, are verified,
    \end{itemize}
    Then there exists a triple $(\mathcal{M},\mathcal{S}_\star,\widetilde{g})$ where $\M$ is a 4-dimensional time-oriented and oriented smooth manifold with corners, $\mathcal{S}_\star$ is a boundary hypersurface of $\M$ and $\widetilde{g}$ is a smooth Lorentzian metric  on $\mathring{\M}\cup\mathring{\mathcal{S}_\star}$ such that
    \begin{itemize}
        \item[i)] $(\mathcal{M},\partial\M\setminus\mathring{\mathcal{S}_\star},\widetilde{g})$ is asymptotically Anti-de Sitter, in the sense of \Cref{def_aAdS}, with smooth rescaled metrics,
        \item[ii)] all boundary hypersurfaces other than $\mathcal{S}_\star$ intersect $\mathcal{S}_\star$ orthogonally with respect to any (and thus all) rescaled metric $g$,
        \item[iii)] $(\mathring{\M}\cup\mathring{\mathcal{S}_\star},\widetilde{g})$ is a solution to the vacuum Einstein equations for $\Lambda = -3$, that is
        \[ \widetilde{R}_{\alpha\beta} = - 3 \widetilde{g}_{\alpha\beta} \quad \text{on } \mathring{\M}\cup\mathring{\mathcal{S}_\star} \,, \]
        where $\widetilde{R}_{\alpha\beta}$ is the Ricci tensor of $\widetilde{g}$,
        \item[iv)] the inward-pointing unit normal of $\mathcal{S}_\star$ is future-directed, 
        \item[v)] there exists an orientation preserving diffeomorphism $\phi : \mathcal{S} \to \mathcal{S}_\star \subset \M$ such that
        \begin{itemize}
            \item[a)] $\phi$ maps isometrically $(\mathring{\mathcal{S}},\widetilde{h})$ to $\mathring{\mathcal{S}_\star}$ equipped with the metric induced by $\widetilde{g}$,
            \item[b)] $\phi_\star\widetilde{K}_{ij}$ is the second fundamental form of $\mathring{\mathcal{S}_\star}$ in $(\mathring{\M}\cup\mathring{\mathcal{S}_\star},\widetilde{g})$,
        \end{itemize}
        \item[vi)] for all connected component $\mathfrak{S}$ of $\partial\M\setminus\mathring{\mathcal{S}_\star}$, the following geometric boundary conditions are verified
            \begin{equation}
                \tag{\ref{eq:bc_robin}}
                \mathfrak{t}_{ij} = \mu_{
                \phi^{-1}(\mathfrak{S}\cap\mathcal{S}_\star)} \, \mathfrak{y}_{ij} \quad \text{on } \mathfrak{S} \,,
            \end{equation}
        where $\mathfrak{t}$ is the boundary stress-energy tensor and $\mathfrak{y}$ is the Cotton-York tensor both associated to any representative of the conformal class $[\mathfrak{h}]$ on $\mathfrak{S}$ defined by \Cref{def_aAdS}.
    \end{itemize}
    In addition, local uniqueness holds: if $(\M',\mathcal{S}_\star',\widetilde{g}')$ is another such triple verifying $i)-vi)$ then there exists a third triple $(\M'',\mathcal{S}_\star'',\widetilde{g}'')$ verifying $i)-vi)$  and two time-orientation and orientation preserving embeddings $\psi : \M'' \xhookrightarrow{} \M$, $\psi' : \M'' \xhookrightarrow{} \M'$ such that
    \begin{itemize}
        \item $\psi \circ \phi'' = \phi$ and $\psi' \circ \phi'' = \phi'$ where $\phi$, $\phi'$, $\phi''$ denote the diffeomorphisms described by $v)$ of the three triple,
        \item $\psi^\star \widetilde{g} = \widetilde{g}''$ and $(\psi')^\star \widetilde{g}' = \widetilde{g}''$.
    \end{itemize}
\end{thm}

\begin{figure}
    \centering
    \begin{tikzpicture}
        \pgfmathsetmacro{\persp}{0.35}; 
    	\pgfmathsetmacro{\Ri}{0.8} 
    	\pgfmathsetmacro{\Re}{2.5} 
    	\pgfmathsetmacro{\H}{3} 
        \pgfmathsetmacro{\Tr}{2.5*\Re} 

        \fill[color=lightgray,even odd rule]
            (\Re/2.5+\Tr,0) ellipse ({\Ri} and {\Ri*\persp})
            (\Tr,0) ellipse ({\Re} and {\Re*\persp});
        \draw[dashed] (\Tr+\Re/2.5,0) ellipse ({\Ri} and {\Ri*\persp});
        \draw (\Tr+\Re,0) arc (0:-180:{\Re} and {\Re*\persp});
        \draw[dashed] (\Tr+\Re,0) arc (0:180:{\Re} and {\Re*\persp});
        \draw (\Tr-\Re,0)--(\Tr-\Re,\H);
        \draw (\Tr+\Re,0)--(\Tr+\Re,\H);
        \draw[dashed] (\Tr+\Re/2.5-\Ri,0)--(\Tr+\Re/2.5-\Ri,\H);
        \draw[dashed] (\Tr+\Re/2.5+\Ri,0)--(\Tr+\Re/2.5+\Ri,\H);

        \fill[color=lightgray,even odd rule]
            (\Re/2.5,0) ellipse ({\Ri} and {\Ri*\persp})
            (0,0) ellipse ({\Re} and {\Re*\persp});
        \draw (\Re/2.5,0) ellipse ({\Ri} and {\Ri*\persp});
        \draw (0,0) ellipse ({\Re} and {\Re*\persp});

        \draw [->,line width=1.2pt] (\Tr/2-\Re*0.8,-\Re*\persp) to[bend right] (\Tr/2+\Re*0.8,-\Re*\persp);
        \path[->] (\Tr/2-\Re*0.8,-\Re*\persp) to[bend right] node[midway,below,inner sep=5pt] {$\phi$} (\Tr/2+\Re*0.8,-\Re*\persp);

        \draw (\Tr-\Re/2.2,\H*0.7) node {$\mathcal{M}$};
        \draw (\Tr-\Re/2.2,0) node {$\mathcal{S}_\star$};
        \draw (-\Re/2.2,0) node {$\mathcal{S}$};
        \draw (\Tr+\Re/2.5+\Ri/2,\H*0.7) node {$\mathfrak{S}_2$};
        \draw (\Tr+\Re+\Ri/2,\H*0.7) node {$\mathfrak{S}_1$};
        \draw (\Re/2.5-\Ri,\Ri*\persp*1.5) node {$\mathcal{C}_2$};
        \draw (\Re*0.8,\Re*\persp) node {$\mathcal{C}_1$};
    \end{tikzpicture}
    \caption{Illustration of \Cref{thm:main_precise}}
\end{figure}

\bigbreak

\begin{rems} \,
	\begin{itemize}
        \item The first equation in $H3$ implies in particular that all connected components $\mathcal{C}$ of $\partial\mathcal{S}$ are umbilical in the conformal structure $(\mathcal{S},[h])$, where $[h]$ is the class of rescaled metrics.
        
		\item One can modify the above theorem to impose an inhomogeneous Dirichlet boundary condition as in Friedrich \cite{F95} on some of the connected components of the conformal boundary, provided that the compatibility conditions are modified accordingly.
        
		\item Since the boundary conditions are reflective, one can construct a past development by changing the time orientation.
        
		\item From the above theorem, one can construct a unique maximal development following Choquet-Bruhat and Geroch \cite{CBG69} and Sbierski \cite{S16}. \qedhere
	\end{itemize}	
\end{rems}

\begin{proof}
    Equipped with all the results of \Cref{sec:CVE,sec:solving_evolution_system,sec:geometric_bc,sec:geometric_initial_data}, the proof mainly follows standard lines. Let us detail the local existence part.
	\begin{itemize}
		\item\underline{Step 1}: Construction of the 10-tuple
		
		Let $\Psi$ be a boundary defining function of $\partial\mathcal{S}$, let $\kappa_i$ be a smooth covector field on $\mathcal{S}$ and let $z=0\in\mathcal{C}^\infty(\mathcal{S},\R)$. By \Cref{cor:smooth_unphysical_fields}, the hypotheses H1-H3 imply that the 10-tuple $(\mathcal{S},h,E,M,\Psi,\kappa,K,0,\allowbreak T,Q)$, where the fields are defined by \eqref{CVC_def}, is in $\mathscr{D}$, is solution to the \eqref{eq:CVC} and verifies the geometric boundary conditions \eqref{eq:gauge_independent_conditions_CVC}. By choice of $\Psi$ and $z$, it also verifies \eqref{eq:gauge_choice_Psi} and \eqref{eq:gauge_choice_z}. Using the gauge transformation $\Weylchg_\omega^\parallel$ defined in \Cref{def:transfo_dim3}, one can assume that \eqref{eq:gauge_choice_s} while maintaining \eqref{eq:gauge_choice_Psi} and \eqref{eq:gauge_choice_z}.
		
		\item\underline{Step 2}: Construction of the patches and of the analytic initial data $\underline{u}_\star$
		
		For any point $p \in \mathcal{S}$, one can construct a coordinate system $(y^i) = (y^1,y^2,y^3)$ and an $h$-orthonormal frame field $(e_{\bf i})=(e_{\bf 1},e_{\bf 2},e_{\bf 3})$ on a neighbourhood $\mathcal{U}_\star(p)$ of $p$ such that
        \begin{itemize}
            \item[\ding{70}] they are smooth on the closure $\overline{\mathcal{U}_\star(p)}$,
            \item[\ding{70}] $y^3 \geq 0$ on $\mathcal{U}_\star(p)$ and $\partial\mathcal{S} \cap \mathcal{U}_\star(p) = \{ q \in \mathcal{U}_\star(p) \, | \, y^3(q) = 0 \}$,
            \item[\ding{70}] $\partial\mathcal{S} \cap \mathcal{U}_\star(p)$ is connected, that is $\mathcal{U}_\star(p)$ intersects at most one connected component $\mathcal{C}$ of the boundary $\partial\mathcal{S}$,
            \item[\ding{70}] if $\partial\mathcal{S} \cap \mathcal{U}_\star(p) \neq \varnothing$, the vector field $e_{\bf 3}$ is its outward-pointing unit normal with respect to $h$ ($e_{\bf 1}$ and $e_{\bf 2}$ are thus tangent).
        \end{itemize}
        By compactness of $\mathcal{S}$, the family $(\mathcal{U}_\star(p))_{p \in \mathcal{S}}$ admits a finite subcover. Let $\mathcal{U}_\star$ be a subset of this finite subcover. Construct the analytic initial data $\underline{u}_\star \in \mathcal{C}^\infty(\overline{\mathcal{U}_\star},\R^m)$ from the 10-tuple of step 1 as follows:
        \begin{itemize}
            \item[\ding{70}] $E_{\bf AB}$, $E_{\bf 3B}$, $H_{\bf AB}$, $H_{\bf 3A}$, $\kappa_{\bf i}$ are components of the already defined tensors with respect to the frame field $(e_{\bf i})$ and $e_{\bf i}{}^j := \langle dy^j,e_{\bf i}\rangle$,
            \item[\ding{70}] $\acute{\Gamma}_{\bf i}{}^{\bf j}{}_{\bf k}$ are the frame connection coefficients of the Weyl connection associated to $\kappa$ with respect to $h$,
            \item[\ding{70}] $\acute{\Gamma}_{\bf i}{}^{\bf 0}{}_{\bf j} := K_{\bf ij}$, $\acute{\Gamma}_{\bf i}{}^{\bf j}{}_{\bf 0} = \delta^{\bf jk} K_{\bf ik}$, $\acute{\Gamma}_{\bf i}{}^{\bf 0}{}_{\bf 0} := \kappa_{\bf i}$, $U_{\bf i0} := T_{\bf i}$, $U_{\bf ij} := Q_{\bf ij}$, $e_{\bf i}{}^0 := 0$ and $\zeta_{\bf i} := (d\Psi)_{\bf i} + \Psi \kappa_{\bf i}$.
        \end{itemize}
        Let $\Theta_\star$ be the restriction on $\mathcal{U}_\star$ of $\Psi$ and $s_\star$ be the restriction on $\mathcal{U}_\star$ of the field $\widehat{\sigma}$, defined by \eqref{eq:def_sigma_hat}, of the 10-tuple of step 1. By \eqref{eq:gauge_choice_s} and \Cref{lem:dividing_bdf}, $f_\star := s_\star/\Theta_\star$ is smooth and bounded from above on $\mathcal{U}_\star$.
		
		\item \underline{Step 3}: Evolution system and propagation of the constraints
		
		If $\mathcal{B}_\star \neq \varnothing$, take the analytic boundary condition given by \eqref{eq:robin_analytic} with $\mu$ equal to $\mu_{\mathcal{C}}$ for the unique connected component $\mathcal{C}$ of $\partial\mathcal{S}$ intersecting $\mathcal{U}_\star$. In all cases, the analytic initial data $\underline{u}_\star$ verifies all the conditions of \Cref{lem:propagation_constraints} by step 2 and H4. Applying \Cref{evol_solve_away}/\Cref{evol_solve_near} and \Cref{lem:propagation_constraints} and \Cref{cor:geometry_back} yields a quintuple $(\mathcal{W},\breve{g},V,\Theta,\kappa) \in \mathscr{E}$ solution to the (CVE) with $\Lambda=-3$ and of conformal boundary $\mathscr{I} = \mathcal{W}\cap\mathcal{B}$. Moreover, if $\mathcal{B}_\star\neq\varnothing$, $(\mathcal{W},\mathscr{I},\Theta^{-2}\breve{g})$ is an aAdS space. 
        
		\item \underline{Step 4}: Geometric boundary conditions and patching up

        Assume $\mathcal{B}_\star \neq\varnothing$. By H4, the analytic initial data $\underline{u}_\star$ satisfies \eqref{eq:frakp_vanish_corner}. Using $ii)$ of \Cref{prop:partial_to_total_robin_bc}, one can assume that $(\mathcal{W},\mathscr{I},\Theta^{-2}\breve{g})$ satisfy the geometric boundary condition \eqref{eq:bc_robin} for $\mu=\mu_\mathcal{C}$ as in step 3, up to reducing $\mathcal{W}$.

        Following standard procedures, one can patch up the local solutions on the different open subsets of the finite cover using local uniqueness of \Cref{evol_solve_away}/\Cref{evol_solve_near}. This yields a quintuple $(\mathcal{M},g,V,\Theta,\kappa)\in \mathscr{E}$ solution to the (CVE) with $\Lambda=-3$. Then, one can verify that $(\mathcal{M},\mathcal{S},\Theta^{-2}g)$ satisfy $i)-vi)$ for the good choices of time-orientation and orientation. \qedhere	
	\end{itemize}   
\end{proof}

\subsubsection{Further comments on the main theorem}
\label{sec:comments_main_thm}

\bigskip

\noindent $\bullet$ Positive mass theorems for $\Lambda < 0$ and $n=3$

\bigskip

Recall that the 3-dimensional hyperbolic space $(\mathbb{H}^3,b)$ can be defined by taking the open unit ball of $\R^3$ equipped with the following Riemannian metric
\[ b := \frac{4}{(1-r^2)^2} \delta \,, \]
where $\delta$ is the Euclidian metric on $\R^3$ and $r$ is the radius in spherical coordinates. Its conformal boundary is the unit circle. One can easily construct a FG gauge using the boundary defining function $\rho := (1-r)/(1+r)$, giving
\[ b = \frac{1}{\rho^2} \left( d\rho^2 + \frac{(1-\rho^2)^2}{4} \dsphere{2} \right) \,. \]

There have been several proofs of the positive mass theorem in the case of $\Lambda < 0$ and $n=3$: see Wang \cite[Theorem 3.3]{W01},
Chru\'sciel and Herzlich \cite[Theorem 1.5]{CH03}. In these theorems, they consider asymptotically hyperbolic spaces for which the least restrictive definition is as follows: an asymptotically hyperbolic space is a pair $(\mathcal{S},\widetilde{h})$, where $\mathcal{S}$ is a 3-dimensional smooth manifold with boundary $\partial\mathcal{S}$ and $\widetilde{h}$ is a Riemannian metric on $\mathring{\mathcal{S}}$, such that there exists two compacts $K \Subset \mathring{\mathcal{S}}$, $K' \Subset \mathbb{H}^3$ and a diffeomorphism $\phi : \mathcal{S}\setminus K \to \mathbb{H}^3 \setminus K'$ verifying $\phi_\star h - b = \grando{\rho}$. Hence $\phi_\star h$ is a perturbation of the hyperbolic metric near the conformal boundary which is diffeomorphic to $\mathbb{S}^2$. Moreover, the required decay rate is quite restrictive since it is at the level of the boundary stress-energy tensor in the expansion of the initial metric induced by the FG expansion. This is thus a much stronger condition than our definition of aH spaces, see \Cref{def:aH}. Consequently, the geometric initial data of \Cref{thm:main_precise} do not verify the hypotheses of the positive mass theorem, except the trivial case of the hyperbolic space.

\bigbreak

\noindent $\bullet$ The orthogonality hypothesis between $\mathcal{S}_\star$ and $\mathscr{I}$

\bigskip

One restriction of \Cref{thm:main_precise} is that the embedded initial hypersurface $\mathcal{S}_\star=\phi(\mathcal{S})$ has to be orthogonal to the conformal boundary for all rescaled metrics by $ii)$. If this is not the case, observe that it cannot be asymptotically hyperbolic in the sense of \Cref{def:aH} because
\[ h^{-1}(d\Psi,d\Psi) = 1 + z^2 \neq 1 \quad \text{on } \partial\mathcal{S}_\star \,, \]
by \eqref{eq:pseudonorm_dPsi}. With regard to (partially) solving the geometric initial data problem in this case, the conformal method appears to be sufficiently robust. Note that the asymptotic behaviour of the field $\widetilde{K}_{ij}$ is still given by \eqref{eq:behav_Ktilde}. Furthermore, we strongly expect that an analogue of \Cref{thm:smooth_unphysical_fields} hold, generalising a result of K\'ann\'ar \cite[Theorem 1]{K96}.

\begin{figure}
    \centering
    \begin{tikzpicture}
        \pgfmathsetmacro{\size}{4}
        \pgfmathsetmacro{\racc}{0.9*\size}
        \draw[dashed] (0,0) -- (0,\size);
        \draw (-0.5,\size/2) node {$\mathscr{I}$};
        \draw plot[domain=0:\size,variable=\x]({\x},{0.3*\x*(\x/\size-2)});
        \draw (3/4*\size,-1/2.7*\size) node {$\mathcal{S}$};
        \draw [fill=black] (0,0) circle (1pt);
        \draw (-0.5,0) node {$\partial\mathcal{S}$};
        \draw[blue] plot[domain=0.01:\racc,variable=\x]({\x},{exp(\size/\racc^2)*exp(\size/(\x*(\x-2*\racc)))*0.3*\x*(\x/\size-2)});
        \draw[blue] plot[domain=\racc:\size,variable=\x]({\x},{0.3*\x*(\x/\size-2)});
        \draw[blue] (\size/2,-\size/10) node {$\mathcal{S}'$};
        \draw[dashed] (0,0) -- (\size,\size);
        \draw[dashed] (0,0) -- (\size,-\size);
        \fill[color=gray,fill=gray,fill opacity=0.5] (0,0) -- (\size,\size) -- (\size,-\size) -- cycle;
        \draw (3/4*\size,\size/4) node {$\widetilde{\mathcal{D}}$};
    \end{tikzpicture}
    \caption{Modification of the initial data}
    \label{fig:modification_ID}
\end{figure}
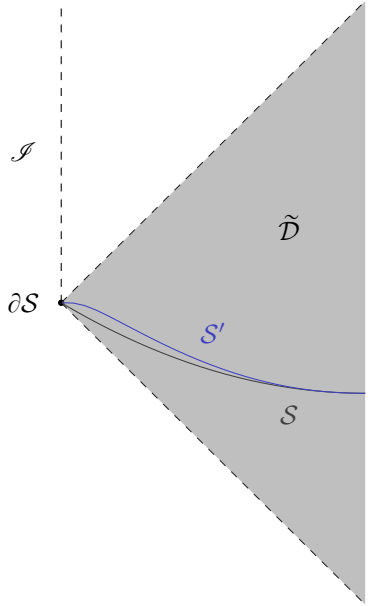

As for the geometric initial data problem with an initial hypersurface which is not orthogonal to the conformal boundary, one could either attempt to adapt the proof of \Cref{thm:main_precise} or to modify the initial data. For the former idea, one could modify the gauge constructed in \Cref{prop:construction_gauge} by taking $e_{\bf 0}$ a unit transverse vector field to the hypersurface which is tangent to $\mathscr{I}$ on $\partial\mathcal{S}_\star$. Observe that, from $\acute{\vartheta}_{\bf 00i}=0$, one has $e_{\bf 0}(\breve{g}_{\bf 0i})=0$. For the latter idea, one could:
\begin{itemize}
    \item[\ding{70}] First solve the vacuum Einstein equations from the initial data set $(\mathring{\mathcal{S}},\widetilde{h}_{ij}, \widetilde{K}_{ij})$ by the classical results \cite{CB52,CBG69} giving a maximal globally hyperbolic development $(\widetilde{\mathcal{D}},\widetilde{g}_{\alpha\beta})$.
    \item[\ding{70}] Modify the hypersurface $\mathring{\mathcal{S}}$ in $\widetilde{\mathcal{D}}$ outside a compact set into a new spacelike hypersurface with mean curvature $\widetilde{K}'$ tending to $0$ at spatial infinity. Attach $\partial\mathcal{S}$ to it in order to construct a manifold with boundary denoted by $\mathcal{S}'$.
    \item[\ding{70}] Apply \Cref{thm:main_precise} to the initial data set composed of $\mathcal{S}'$ and $\widetilde{h}'_{ij}$, $\widetilde{K}'_{ij}$ the first and second fundamental forms of $\mathring{\mathcal{S}}'$ in $(\widetilde{\mathcal{D}},\widetilde{g}_{\alpha\beta})$.
\end{itemize}
The construction is illustrated in \Cref{fig:modification_ID}, in the conformal picture and with embedded hypersurfaces for simplicity. The timelike conformal boundary $\mathscr{I}$ resulting from \Cref{thm:main_precise} is represented for clarity. For this method to work, one has to prove a propagation result of both the smoothness of the unphysical fields and of the compatibility conditions from $\mathcal{S}$ to $\mathcal{S}'$.

\bigbreak

\noindent $\bullet$ Link with the boundary conditions for the Teukolsky equations

\bigskip

The geometric homogeneous Dirichlet boundary condition translate into certain boundary conditions for the Teukolsky equations on aAdS black holes, see for instance \cite[Section 6.3]{HLSW20} for AdS and \cite[Proposition 2.2]{GH24_part_I} for SAdS. In what follows, we will determine the boundary conditions corresponding to the homogeneous Robin boundary conditions \eqref{eq:bc_robin}. Construct a null frame $(e^\bullet_{\bf a})$ for the rescaled metric by setting
\[ e^\bullet_{\bf A} := e_{\bf A} \,, \quad e^\bullet_{\bf 3} := e_{\bf 0} - e_{\bf 3} \,, \quad e^\bullet_{\bf 4} := e_{\bf 0} + e_{\bf 3} \,.\]
In the null decomposition of the Weyl tensor, see \cite[Equations (4.8)]{CK90}, one has
\[ \alpha_{\bf AB} := W(e^\bullet_{\bf A},e^\bullet_{\bf 4},e^\bullet_{\bf B},e^\bullet_{\bf 4}) \,, \qquad \underline{\alpha}_{\bf AB} := W(e^\bullet_{\bf A},e^\bullet_{\bf 3},e^\bullet_{\bf B},e^\bullet_{\bf 3}) \,. \]
These are symmetric trace-free smooth 2-tensors on the distribution $\mathcal{D}^2 = \vectorspan(e_{\bf 1},e_{\bf 2})$, that is they are in $\mathscr{S}(\mathcal{D}^2)$ as defined in \Cref{sec:duality_dim_2}. By definition of the rescaled Weyl tensor and the electromagnetic decomposition, see \Cref{lem:decomposition_Weyl_cand}, one deduces that
\[ \alpha_{\bf AB} = -2 \Theta \left( E^\dagger_{\bf AB} + (\star H^\dagger)_{\bf AB} \right) \,, \qquad \underline{\alpha}_{\bf AB} = -2 \Theta \left( E^\dagger_{\bf AB} - (\star H^\dagger)_{\bf AB} \right) \,. \]
Hence the boundary conditions \eqref{bc_electric} rewrite under the form
\begin{align*}
    \mathfrak{p}^\dagger_{\bf AB} = 0 \;\; \text{on } \mathscr{I} &\iff \frac{2}{3} E^\dagger_{\bf AB} - \mu H^\dagger_{\bf AB} = 0 \;\; \text{on } \mathscr{I} \\
    &\iff \lim_{\Theta\to0} \Theta^{-1} \left( \frac{2}{3} (\alpha+\underline{\alpha})_{\bf AB} - \mu \left(\star(\alpha-\underline{\alpha})\right)_{\bf AB} \right) = 0 \\
    &\iff \lim_{\Theta\to0} \Theta^{-1} \left( \frac{2}{3} \left(\star(\alpha+\underline{\alpha})\right)_{\bf AB} - \mu (\alpha-\underline{\alpha})_{\bf AB} \right) = 0 \,.
\end{align*}
In the limit case $\mu\to+\infty$, one retrieves \cite[Equation (31)]{GH24_part_I}, that is\footnote{The discrepancy in powers of $\Theta$, with $-1$ instead of $-3$, is easily explained by our use of a frame field for the rescaled metric instead of the physical metric.}
\[ \lim_{\Theta\to 0} \Theta^{-1} \left(  \alpha - \underline{\alpha}\right)_{\bf AB} = 0 \,. \]
The second boundary condition follows from the second Bianchi identity and, more precisely, from the evolution equations \eqref{evol_EAB} and \eqref{evol_HAB}. One has
\begin{align*}
     &e_{\bf 0}\left(\mathfrak{p}^\dagger_{\bf AB}\right) = 0 \;\; \text{on } \mathscr{I} \\
     \iff & e_{\bf 0}\left(\frac{2}{3} E^\dagger_{\bf AB} - \mu H^\dagger_{\bf AB} \right) = 0 \;\; \text{on } \mathscr{I} \\
    \iff & e_{\bf 3} \left( \frac{2}{3} (\star H^\dagger)_{\bf AB} + \mu (\star E^\dagger)_{\bf AB} \right) - \slashed{\epsilon}_{\bf (A|}{}^{\bf C} e_{\bf C} \left( \frac{2}{3} H_{\bf 3|B)} + \mu E_{\bf 3|B)} \right) \\
    &\qquad = \frac{2}{3} N^\dagger_{\bf AB}(E,H) - \mu N^\dagger_{\bf AB}(H,-E) \;\; \text{on } \mathscr{I} \\
    \iff & e_{\bf 3} \left( \frac{2}{3} (\star H^\dagger)_{\bf AB} + \mu (\star E^\dagger)_{\bf AB} \right) - \slashed{\epsilon}_{\bf (A|}{}^{\bf C} e_{\bf C} \left( \frac{2}{3} H_{\bf 3|B)} + \mu E_{\bf 3|B)} \right) \\
    &\qquad = N^\dagger_{\bf AB}\left(\frac{2}{3}E- \mu H,\frac{2}{3} H + \mu E \right)  \;\; \text{on } \mathscr{I} \,.
\end{align*}
Note that the geometric homogeneous Robin boundary condition also implies that
\begin{align*}
    \frac{2}{3} H_{\bf 3B} + \mu E_{\bf 3B} &= 0 \;\; \text{on } \mathscr{I} \,, \\
    \frac{2}{3} E_{\bf 33} - \mu H_{\bf 33} &= 0 \;\; \text{on } \mathscr{I} \,.
\end{align*}
Since $e_{\bf A}$ are tangent to $\mathscr{I}$, one deduces that
\[ e_{\bf 3} \left( \Theta^{-1} \left(\frac{2}{3} (\alpha-\underline{\alpha})_{\bf AB} + \mu  (\star(\alpha+\underline{\alpha}))_{\bf AB} \right) \right) = L^{(\mu)}_{\bf AB}(\alpha,\underline{\alpha}) \quad \text{on } \mathscr{I} \,, \]
for some linear function $L^{(\mu)}: \mathscr{S}(\mathcal{D}^2) \times \mathscr{S}(\mathcal{D}^2) \to \mathscr{S}(\mathcal{D}^2)$. Finding the explicit expression of $L^{(\mu)}$ would require a further study but one can expect simplifications using the boundary condition \eqref{eq:bc_robin}, the gauge properties \eqref{eq:gauge_independent_conditions_CVC}-\eqref{eq:gauge_choices_CVC} and \eqref{eq:umbilical_condition}.

\addcontentsline{toc}{section}{References}
\bibliographystyle{plain} 
\bibliography{references.bib}

\begin{thebibliography}{10}

\bibitem{ALN20}
Ian~M. Anderson, Thomas Leistner, and Pawel Nurowski.
\newblock Explicit ambient metrics and holonomy.
\newblock {\em J. Differential Geom.}, 114(2):193--242, 2020.

\bibitem{A06}
Michael~T. Anderson.
\newblock On the uniqueness and global dynamics of {A}d{S} spacetimes.
\newblock {\em Classical Quantum Gravity}, 23(23):6935--6953, 2006.

\bibitem{AC94}
Lars Andersson and Piotr~T. Chru\'sciel.
\newblock On ``hyperboloidal'' {C}auchy data for vacuum {E}instein equations and obstructions to smoothness of scri.
\newblock {\em Comm. Math. Phys.}, 161(3):533--568, 1994.

\bibitem{AC96}
Lars Andersson and Piotr~T. Chru\'sciel.
\newblock Solutions of the constraint equations in general relativity satisfying ``hyperboloidal boundary conditions''.
\newblock {\em Dissertationes Math. (Rozprawy Mat.)}, 355:100, 1996.

\bibitem{ACF92}
Lars Andersson, Piotr~T. Chru\'sciel, and Helmut Friedrich.
\newblock On the regularity of solutions to the {Y}amabe equation and the existence of smooth hyperboloidal initial data for {E}instein's field equations.
\newblock {\em Comm. Math. Phys.}, 149(3):587--612, 1992.

\bibitem{AM84}
Abhay Ashtekar and Anne Magnon.
\newblock Asymptotically anti-de {S}itter space-times.
\newblock {\em Classical Quantum Gravity}, 1(4):L39--L44, 1984.

\bibitem{AIS78}
S.~J. Avis, C.~J. Isham, and D.~Storey.
\newblock Quantum field theory in anti-de {S}itter space-time.
\newblock {\em Phys. Rev. D (3)}, 18(10):3565--3576, 1978.

\bibitem{B08}
Alain Bachelot.
\newblock The {D}irac system on the anti-de {S}itter universe.
\newblock {\em Comm. Math. Phys.}, 283(1):127--167, 2008.

\bibitem{B11}
Alain Bachelot.
\newblock The {K}lein-{G}ordon equation in the anti-de {S}itter cosmology.
\newblock {\em J. Math. Pures Appl. (9)}, 96(6):527--554, 2011.

\bibitem{B99}
Danny Birmingham.
\newblock Topological black holes in anti-de {S}itter space.
\newblock {\em Classical Quantum Gravity}, 16(4):1197--1205, 1999.

\bibitem{BR11}
Piotr Bizo\ifmmode~\acute{n}\else \'{n}\fi{} and Andrzej Rostworowski.
\newblock Weakly turbulent instability of anti--de sitter spacetime.
\newblock {\em Phys. Rev. Lett.}, 107:031102, Jul 2011.

\bibitem{BF82bis}
Peter Breitenlohner and Daniel~Z. Freedman.
\newblock Positive energy in anti-de {S}itter backgrounds and gauged extended supergravity.
\newblock {\em Phys. Lett. B}, 115(3):197--201, 1982.

\bibitem{BF82}
Peter Breitenlohner and Daniel~Z. Freedman.
\newblock Stability in gauged extended supergravity.
\newblock {\em Ann. Physics}, 144(2):249--281, 1982.

\bibitem{CHK19}
Diego~A. Carranza, Adem~E. Hursit, and Juan~A. Valiente~Kroon.
\newblock Conformal wave equations for the {E}instein-tracefree matter system.
\newblock {\em Gen. Relativity Gravitation}, 51(7):Paper No. 88, 39, 2019.

\bibitem{CK18}
Diego~A. Carranza and Juan~A. Valiente~Kroon.
\newblock Construction of anti--de {S}itter--like spacetimes using the metric conformal {E}instein field equations: the vacuum case.
\newblock {\em Classical Quantum Gravity}, 35(24):245006, 34, 2018.

\bibitem{CS22}
Athanasios Chatzikaleas and Arick Shao.
\newblock A gauge-invariant unique continuation criterion for waves in asymptotically anti--de~{S}itter spacetimes.
\newblock {\em Comm. Math. Phys.}, 395(2):521--570, 2022.

\bibitem{CS24}
Athanasios Chatzikaleas and Jacques Smulevici.
\newblock Nonlinear periodic waves on the {E}instein cylinder.
\newblock {\em Anal. PDE}, 17(7):2311--2378, 2024.

\bibitem{CB74}
Yvonne Choquet-Bruhat.
\newblock Global solutions of the constraints equations on open and closed manifolds.
\newblock {\em Gen. Relativity Gravitation}, 5(1):49--60, 1974.

\bibitem{CBG69}
Yvonne Choquet-Bruhat and Robert Geroch.
\newblock Global aspects of the {C}auchy problem in general relativity.
\newblock {\em Comm. Math. Phys.}, 14:329--335, 1969.

\bibitem{CK90}
D.~Christodoulou and S.~Klainerman.
\newblock Asymptotic properties of linear field equations in {M}inkowski space.
\newblock {\em Comm. Pure Appl. Math.}, 43(2):137--199, 1990.

\bibitem{C86}
Demetrios Christodoulou.
\newblock Global solutions of nonlinear hyperbolic equations for small initial data.
\newblock {\em Comm. Pure Appl. Math.}, 39(2):267--282, 1986.

\bibitem{C08}
Demetrios Christodoulou.
\newblock {\em Mathematical problems of general relativity. {I}}.
\newblock Zurich Lectures in Advanced Mathematics. European Mathematical Society (EMS), Z\"urich, 2008.

\bibitem{CK94}
Demetrios Christodoulou and Sergiu Klainerman.
\newblock {\em The global nonlinear stability of the {M}inkowski space}, volume~41 of {\em Princeton Mathematical Series}.
\newblock Princeton University Press, Princeton, NJ, 1993.

\bibitem{C20}
Piotr~T. Chru\'sciel.
\newblock {\em Geometry of black holes}, volume 169 of {\em International Series of Monographs on Physics}.
\newblock Oxford University Press, Oxford, 2020.

\bibitem{CH03}
Piotr~T. Chru\'sciel and Marc Herzlich.
\newblock The mass of asymptotically hyperbolic {R}iemannian manifolds.
\newblock {\em Pacific J. Math.}, 212(2):231--264, 2003.

\bibitem{DH06}
Mihalis Dafermos and Gustav Holzegel.
\newblock Dynamic instability of solitons in 4+1-dimensional gravity with negative cosmological constant.
\newblock Unpublished, available at \url{https://www.dpmms.cam.ac.uk/~md384/ADSinstability.pdf}, 2006.

\bibitem{dH09}
Sebastian de~Haro.
\newblock Dual gravitons in {$\rm AdS_4/CFT_3$} and the holographic {C}otton tensor.
\newblock {\em J. High Energy Phys.}, (1):042, 30, 2009.

\bibitem{dHSS01}
Sebastian de~Haro, Kostas Skenderis, and Sergey~N. Solodukhin.
\newblock Holographic reconstruction of spacetime and renormalization in the {A}d{S}/{CFT} correspondence.
\newblock {\em Comm. Math. Phys.}, 217(3):595--622, 2001.

\bibitem{EK15}
Alberto Enciso and Niky Kamran.
\newblock A singular initial-boundary value problem for nonlinear wave equations and holography in asymptotically anti-de {S}itter spaces.
\newblock {\em J. Math. Pures Appl. (9)}, 103(4):1053--1091, 2015.

\bibitem{EK19}
Alberto Enciso and Niky Kamran.
\newblock Lorentzian {E}instein metrics with prescribed conformal infinity.
\newblock {\em J. Differential Geom.}, 112(3):505--554, 2019.

\bibitem{FG85}
Charles Fefferman and C.~Robin Graham.
\newblock Conformal invariants.
\newblock In {\em \'Elie Cartan et les math\'ematiques d'aujourd'hui - Lyon, 25-29 juin 1984}, number S131 in Ast\'erisque, pages 95--116. Soci\'et\'e math\'ematique de France, 1985.

\bibitem{FG12}
Charles Fefferman and C.~Robin Graham.
\newblock {\em The ambient metric}, volume 178 of {\em Annals of Mathematics Studies}.
\newblock Princeton University Press, Princeton, NJ, 2012.

\bibitem{FAS25}
Francisco Fernández-Álvarez and José M.~M. Senovilla.
\newblock Gravitational radiation at infinity with negative cosmological constant and {A}d{S}$_4$ holography, 2025.

\bibitem{CB52}
Yvonne Four\`es-Bruhat.
\newblock Th\'eor\`eme d'existence pour certains syst\`emes d'\'equations aux d\'eriv\'ees partielles non lin\'eaires.
\newblock {\em Acta Math.}, 88:141--225, 1952.

\bibitem{F81}
Helmut Friedrich.
\newblock On the regular and the asymptotic characteristic initial value problem for einstein’s vacuum field equations.
\newblock {\em Proc. Roy. Soc. London Ser. A}, 375:169--184, 1981.

\bibitem{F82}
Helmut Friedrich.
\newblock On the existence of analytic null asymptotically flat solutions of {E}instein's vacuum field equations.
\newblock {\em Proc. Roy. Soc. London Ser. A}, 381(1781):361--371, 1982.

\bibitem{F86}
Helmut Friedrich.
\newblock On the existence of {$n$}-geodesically complete or future complete solutions of {E}instein's field equations with smooth asymptotic structure.
\newblock {\em Comm. Math. Phys.}, 107(4):587--609, 1986.

\bibitem{F95}
Helmut Friedrich.
\newblock Einstein equations and conformal structure: existence of anti-de {S}itter-type space-times.
\newblock {\em J. Geom. Phys.}, 17(2):125--184, 1995.

\bibitem{F03}
Helmut Friedrich.
\newblock Conformal geodesics on vacuum space-times.
\newblock {\em Comm. Math. Phys.}, 235(3):513--543, 2003.

\bibitem{F09}
Helmut Friedrich.
\newblock Initial boundary value problems for {E}instein's field equations and geometric uniqueness.
\newblock {\em Gen. Relativity Gravitation}, 41(9):1947--1966, 2009.

\bibitem{GH23}
Olivier Graf and Gustav Holzegel.
\newblock Mode stability results for the {T}eukolsky equations on {K}err-anti--de {S}itter spacetimes.
\newblock {\em Classical Quantum Gravity}, 40(4):Paper No. 045003, 43, 2023.

\bibitem{GH24_part_I}
Olivier Graf and Gustav Holzegel.
\newblock Linear stability of schwarzschild-anti-de sitter spacetimes i: The system of gravitational perturbations, 2024.

\bibitem{GH24_part_II}
Olivier Graf and Gustav Holzegel.
\newblock Linear stability of schwarzschild-anti-de sitter spacetimes ii: Logarithmic decay of solutions to the teukolsky system, 2024.

\bibitem{GH04}
C.~Robin Graham and Kengo Hirachi.
\newblock The ambient obstruction tensor and {$Q$}-curvature.
\newblock In {\em Ad{S}/{CFT} correspondence: {E}instein metrics and their conformal boundaries}, volume~8 of {\em IRMA Lect. Math. Theor. Phys.}, pages 59--71. Eur. Math. Soc., Z\"{u}rich, 2005.

\bibitem{GL91}
C.~Robin Graham and John~M. Lee.
\newblock Einstein metrics with prescribed conformal infinity on the ball.
\newblock {\em Adv. Math.}, 87(2):186--225, 1991.

\bibitem{G90}
Olivier Gu\`es.
\newblock Probl\`eme mixte hyperbolique quasi-lin\'eaire caract\'eristique.
\newblock {\em Comm. Partial Differential Equations}, 15(5):595--645, 1990.

\bibitem{GS24}
Simon Guisset and Arick Shao.
\newblock On counterexamples to unique continuation for critically singular wave equations.
\newblock {\em J. Differential Equations}, 395:223--261, 2024.

\bibitem{HT85}
Marc Henneaux and Claudio Teitelboim.
\newblock Asymptotically anti-de {S}itter spaces.
\newblock {\em Comm. Math. Phys.}, 98(3):391--424, 1985.

\bibitem{H12}
Gustav Holzegel.
\newblock Well-posedness for the massive wave equation on asymptotically anti-de {S}itter spacetimes.
\newblock {\em J. Hyperbolic Differ. Equ.}, 9(2):239--261, 2012.

\bibitem{HLSW20}
Gustav Holzegel, Jonathan Luk, Jacques Smulevici, and Claude Warnick.
\newblock Asymptotic properties of linear field equations in anti--de~{S}itter space.
\newblock {\em Comm. Math. Phys.}, 374(2):1125--1178, 2020.

\bibitem{HS16}
Gustav Holzegel and Arick Shao.
\newblock Unique continuation from infinity in asymptotically anti--de {S}itter spacetimes.
\newblock {\em Comm. Math. Phys.}, 347(3):723--775, 2016.

\bibitem{HS17}
Gustav Holzegel and Arick Shao.
\newblock Unique continuation from infinity in asymptotically anti--de {S}itter spacetimes {II}: {N}on-static boundaries.
\newblock {\em Comm. Partial Differential Equations}, 42(12):1871--1922, 2017.

\bibitem{HS23}
Gustav Holzegel and Arick Shao.
\newblock The bulk-boundary correspondence for the {E}instein equations in asymptotically anti-de {S}itter spacetimes.
\newblock {\em Arch. Ration. Mech. Anal.}, 247(3):Paper No. 56, 77, 2023.

\bibitem{HS12}
Gustav Holzegel and Jacques Smulevici.
\newblock Self-gravitating {K}lein-{G}ordon fields in asymptotically anti-de-{S}itter spacetimes.
\newblock {\em Ann. Henri Poincar\'e}, 13(4):991--1038, 2012.

\bibitem{HS13kads}
Gustav Holzegel and Jacques Smulevici.
\newblock Decay properties of {K}lein-{G}ordon fields on {K}err-{A}d{S} spacetimes.
\newblock {\em Comm. Pure Appl. Math.}, 66(11):1751--1802, 2013.

\bibitem{HS13}
Gustav Holzegel and Jacques Smulevici.
\newblock Stability of {S}chwarzschild-{A}d{S} for the spherically symmetric {E}instein-{K}lein-{G}ordon system.
\newblock {\em Comm. Math. Phys.}, 317(1):205--251, 2013.

\bibitem{HS14}
Gustav Holzegel and Jacques Smulevici.
\newblock Quasimodes and a lower bound on the uniform energy decay rate for {K}err-{A}d{S} spacetimes.
\newblock {\em Anal. PDE}, 7(5):1057--1090, 2014.

\bibitem{HW15}
Gustav Holzegel and Claude~M. Warnick.
\newblock The {E}instein-{K}lein-{G}ordon-{A}d{S} system for general boundary conditions.
\newblock {\em J. Hyperbolic Differ. Equ.}, 12(2):293--342, 2015.

\bibitem{HW14}
Gustav~H. Holzegel and Claude~M. Warnick.
\newblock Boundedness and growth for the massive wave equation on asymptotically anti-de {S}itter black holes.
\newblock {\em J. Funct. Anal.}, 266(4):2436--2485, 2014.

\bibitem{ISTY00}
C.~Imbimbo, A.~Schwimmer, S.~Theisen, and S.~Yankielowicz.
\newblock Diffeomorphisms and holographic anomalies.
\newblock {\em Classical Quantum Gravity}, 17(5):1129--1138, 2000.
\newblock Strings '99 (Potsdam).

\bibitem{IW04}
Akihiro Ishibashi and Robert~M. Wald.
\newblock Dynamics in non-globally-hyperbolic static spacetimes. {III}. {A}nti-de {S}itter spacetime.
\newblock {\em Classical Quantum Gravity}, 21(12):2981--3013, 2004.

\bibitem{K96}
J\'anos K\'ann\'ar.
\newblock Hyperboloidal initial data for the vacuum {E}instein equations with cosmological constant.
\newblock {\em Classical Quantum Gravity}, 13(11):3075--3084, 1996.

\bibitem{K22}
Leonhard M.~A. Kehrberger.
\newblock The case against smooth null infinity {I}: {H}euristics and counter-examples.
\newblock {\em Ann. Henri Poincar\'e}, 23(3):829--921, 2022.

\bibitem{K04}
Satyanad Kichenassamy.
\newblock On a conjecture of {F}efferman and {G}raham.
\newblock {\em Adv. Math.}, 184(2):268--288, 2004.

\bibitem{LL60}
L.~D. Landau and E.~M. Lifshitz.
\newblock {\em Electrodynamics of continuous media}, volume Vol. 8 of {\em Course of Theoretical Physics}.
\newblock Pergamon Press, Oxford; Addison-Wesley Publishing Co., Inc., Reading, MA, 1960.
\newblock Translated from the Russian by J. B. Sykes and J. S. Bell.

\bibitem{LB82}
Claude~R. LeBrun.
\newblock {${\cal H}$}-space with a cosmological constant.
\newblock {\em Proc. Roy. Soc. London Ser. A}, 380(1778):171--185, 1982.

\bibitem{LN10}
Thomas Leistner and Pawel Nurowski.
\newblock Ambient metrics for {$n$}-dimensional {$pp$}-waves.
\newblock {\em Comm. Math. Phys.}, 296(3):881--898, 2010.

\bibitem{LS05}
I.~V. Lindell and A.~H. Sihvola.
\newblock Perfect electromagnetic conductor.
\newblock {\em J. Electromagn. Waves Appl.}, 19(7):861--869, 2005.

\bibitem{LK12}
Christian L\"{u}bbe and Juan~A. Valiente~Kroon.
\newblock The extended conformal {E}instein field equations with matter: the {E}instein-{M}axwell field.
\newblock {\em J. Geom. Phys.}, 62(6):1548--1570, 2012.

\bibitem{M98}
Juan Maldacena.
\newblock The large {$N$} limit of superconformal field theories and supergravity.
\newblock {\em Adv. Theor. Math. Phys.}, 2(2):231--252, 1998.

\bibitem{MR13}
Maciej Maliborski and Andrzej Rostworowski.
\newblock Time-periodic solutions in an einstein ads--massless-scalar-field system.
\newblock {\em Phys. Rev. Lett.}, 111:051102, Aug 2013.

\bibitem{M96}
Richard Melrose.
\newblock {\em Differential analysis on manifolds with corners}.
\newblock 1996.
\newblock Unfinished book, available at \url{https://math.mit.edu/~rbm/book.html}.

\bibitem{M20}
Georgios Moschidis.
\newblock A proof of the instability of {A}d{S} for the {E}instein-null dust system with an inner mirror.
\newblock {\em Anal. PDE}, 13(6):1671--1754, 2020.

\bibitem{M23}
Georgios Moschidis.
\newblock A proof of the instability of {A}d{S} for the {E}instein-massless {V}lasov system.
\newblock {\em Invent. Math.}, 231(2):467--672, 2023.

\bibitem{P60}
Roger Penrose.
\newblock Republication of: {C}onformal treatment of infinity.
\newblock {\em Gen. Relativity Gravitation}, 43(3):901--922, 2011.

\bibitem{R85}
Jeffrey Rauch.
\newblock Symmetric positive systems with boundary characteristic of constant multiplicity.
\newblock {\em Trans. Amer. Math. Soc.}, 291(1):167--187, 1985.

\bibitem{R09}
Hans Ringstr\"om.
\newblock {\em The {C}auchy problem in general relativity}.
\newblock ESI Lectures in Mathematics and Physics. European Mathematical Society (EMS), Z\"urich, 2009.

\bibitem{ST12}
Olivier Sarbach and Manuel Tiglio.
\newblock Continuum and discrete initial-boundary value problems and einstein’s field equations.
\newblock {\em Living Rev. Relativ.}, 15(1):9, August 2012.

\bibitem{S16}
Jan Sbierski.
\newblock On the existence of a maximal cauchy development for the einstein equations: a dezornification.
\newblock {\em Ann. Henri Poincar\'e}, 17:301--329, 2016.

\bibitem{S96}
Paolo Secchi.
\newblock Well-posedness of characteristic symmetric hyperbolic systems.
\newblock {\em Arch. Rational Mech. Anal.}, 134(2):155--197, 1996.

\bibitem{S21}
Arick Shao.
\newblock The near-boundary geometry of {E}instein-vacuum asymptotically anti--de {S}itter spacetimes.
\newblock {\em Classical Quantum Gravity}, 38(3):Paper No. 034001, 55, 2021.

\bibitem{S01}
Kostas Skenderis.
\newblock Asymptotically anti-de {S}itter spacetimes and their stress energy tensor.
\newblock In {\em Quantization, gauge theory, and strings, {V}ol. {I} ({M}oscow,2000)}, pages 394--402. Sci. World, Moscow, 2001.

\bibitem{SS00}
Kostas Skenderis and Sergey~N. Solodukhin.
\newblock Quantum effective action from the {A}d{S}/{CFT} correspondence.
\newblock {\em Phys. Lett. B}, 472(3-4):316--322, 2000.

\bibitem{T91}
Michael~E. Taylor.
\newblock {\em Pseudodifferential operators and nonlinear {PDE}}, volume 100 of {\em Progress in Mathematics}.
\newblock Birkh\"auser, 1991.

\bibitem{K16}
Juan~A. Valiente~Kroon.
\newblock {\em Conformal methods in general relativity}.
\newblock Cambridge Monographs on Mathematical Physics. Cambridge University Press, Cambridge, 2016.

\bibitem{V12}
Andr\'as Vasy.
\newblock The wave equation on asymptotically anti de {S}itter spaces.
\newblock {\em Anal. PDE}, 5(1):81--144, 2012.

\bibitem{W01}
Xiaodong Wang.
\newblock The mass of asymptotically hyperbolic manifolds.
\newblock {\em J. Differential Geom.}, 57(2):273--299, 2001.

\bibitem{W13}
Claude~M. Warnick.
\newblock The massive wave equation in asymptotically {A}d{S} spacetimes.
\newblock {\em Comm. Math. Phys.}, 321(1):85--111, 2013.

\end{thebibliography}

\end{document}